\documentclass[10pt]{article}
\usepackage[usenames,dvipsnames,svgnames,table]{xcolor}
\usepackage[colorlinks=true,
            linkcolor=black,
            urlcolor=blue,
            citecolor=blue]{hyperref}

\usepackage[square,numbers]{natbib}
\usepackage[normalem]{ulem}

\usepackage{enumitem}
\usepackage{sectsty,amsmath, amsfonts, amsthm, amssymb, mathtools,dsfont,mathrsfs}
\usepackage{euscript,graphicx,hyperref,indentfirst}
\usepackage{tocvsec2}
\usepackage{fullpage}
\usepackage{subcaption}
\usepackage{float}	
\usepackage{forest, siunitx, booktabs}
\usepackage{bbm}
\usepackage{etoolbox}
\usepackage{centernot}
\usepackage{graphics}
\usepackage[titletoc]{appendix}
\usepackage{fancyhdr}
\usepackage{courier}
\usepackage{soul}
\usepackage{ulem}
\usepackage{xcolor}
\usepackage[utf8]{inputenc}
\usepackage{calc}
\usepackage[multiple]{footmisc}
\usepackage[titletoc]{appendix}
\usepackage{yfonts}
\usepackage{tabulary}
\usepackage{booktabs}

\numberwithin{equation}{section}

\makeatletter
\def\namedlabel#1#2{\begingroup
    #2%
    \def\@currentlabel{#2}%\item $R(X)=\E[\tilde{Y}|X]$ is a true regressor.
    \phantomsection\label{#1}\endgroup
}
\newcommand{\specificthanks}[1]{\@fnsymbol{#1}}% Inserts a specific \thanks symbol
\newcommand{\lrarrow}{\mathrel{\mathpalette\lrarrow@\relax}}
\newcommand{\lrarrow@}[2]{%
  \vcenter{\hbox{\ooalign{%
    $\m@th#1\mkern6mu\rightarrow$\cr
    \noalign{\vskip1pt}
    $\m@th#1\leftarrow\mkern6mu$\cr
  }}}%
}
\makeatother
\newcommand*\justify{%
  \fontdimen2\font=0.4em% interword space
  \fontdimen3\font=0.2em% interword stretch
  \fontdimen4\font=0.1em% interword shrink
  \fontdimen7\font=0.1em% extra space
  \hyphenchar\font=`\-% allowing hyphenation
}

\renewcommand{\texttt}[1]{%
  \begingroup
  \ttfamily
  \begingroup\lccode`~=`/\lowercase{\endgroup\def~}{/\discretionary{}{}{}}%
  \begingroup\lccode`~=`[\lowercase{\endgroup\def~}{[\discretionary{}{}{}}%
  \begingroup\lccode`~=`.\lowercase{\endgroup\def~}{.\discretionary{}{}{}}%
  \catcode`/=\active\catcode`[=\active\catcode`.=\active
  \justify\scantokens{#1\noexpand}%
  \endgroup
}
\newcommand{\oper}{\mathcal{E}}
\newcommand{\pulloper}{\bar{\mathcal{E}}}

\newcommand{\TT}{\mathcal{T}}

\renewcommand{\P}{\mathbb{P}}
\newcommand{\E}{\mathbb{E}}
\newcommand{\R}{\mathbb{R}}

\newcommand{\1}{\mathbbm{1}}
\newcommand{\logit}{\text{logit}}

\newcommand{\indep}{\perp\!\!\!\!\perp}

\newcommand{\cL}{\mathcal{L}}
\newcommand{\cG}{\mathcal{G}}
\newcommand{\fL}{\mathfrak{L}}
\newcommand{\MICe}{\text{MIC$_e$}}
\newcommand{\MICstar}{\text{MIC$_*$}}

\newcommand{\intP}{\tilde{P}}
\newcommand{\vhat}{\hat{v}}

\newcommand{\cp}{{\cal P}}

\newcommand{\RR}{\mathbb{R}}
\newcommand{\PP}{\mathbb{P}}
\newcommand{\eps}{\varepsilon}
\newcommand{\Chi}{\mathcal{X}}

\newcommand{\tP}{\tilde{P}}
\newcommand{\cP}{\mathcal{P}}

\newcommand{\ME}{ \text{\tiny \it ME}}
\newcommand{\CE}{ \text{\tiny \it CE}}
\newcommand{\vpdp}{v^{\ME}}
\newcommand{\empvme}{\hat{v}^{\ME}}
\newcommand{\vme}{\vpdp}
\newcommand{\vce}{v^{\CE}}

\newcommand{\vpdpP}{v^{\ME,\mathcal{P}}}
\newcommand{\vceP}{v^{\CE,\mathcal{P}}}

\def\clap#1{\hbox to 0pt{\hss#1\hss}}

\newtheorem*{proposition*}{Proposition}
\newtheorem{proposition}{Proposition}[section]
\newtheorem*{theorem*}{Theorem}
\newtheorem{theorem}{Theorem}[section]
\newtheorem{corollary}{Corollary}[section]
\newtheorem*{corollary*}{Corollary}
\newtheorem*{definition*}{Definition}
\newtheorem{definition}{Definition}[section]
\newtheorem{lemma}{Lemma}[section]
\newtheorem{example}{Example}[section]
\newtheorem*{lemma*}{Lemma}
\newtheorem{remark}{Remark}[section]
\newtheorem*{remark*}{Remark}
\newtheorem*{mtheorem*}{Main Theorem}

\title{Stability theory of game-theoretic group feature explanations for machine learning models}

\author{Alexey Miroshnikov \thanks{Emerging Capabilities Research Group, Discover Financial Services Inc., Riverwoods, IL} \textsuperscript{,$\!\!\!$}
\thanks{first author, alexeymiroshnikov@discover.com, ORCID:0000-0003-2669-6336} \and Konstandinos Kotsiopoulos \textsuperscript{\specificthanks{1},}\thanks{ kostaskotsiopoulos@discover.com, ORCID:0000-0003-2651-0087} \and
Khashayar Filom\textsuperscript{\specificthanks{1},}\thanks{khashayarfilom@discover.com, ORCID:0000-0002-6881-4460}
\and Arjun Ravi Kannan\textsuperscript{\specificthanks{1},}\thanks{ arjunravikannan@discover.com, ORCID:0000-0003-4498-1800}  }

\date{}
    
\begin{document}
\maketitle
\vspace{-0.4 in}

\begin{abstract}
In this article, we study feature attributions of Machine Learning (ML) models originating from linear game values and coalitional values defined as operators on appropriate functional spaces. The main focus is on random games based on the conditional and marginal expectations. The first part of our work formulates a stability theory for these explanation operators by establishing certain bounds for both marginal and conditional explanations. The differences between the two games are then elucidated, such as showing that the marginal explanations can become discontinuous on some naturally-designed domains, while the conditional explanations remain stable. In the second part of our work, group explanation methodologies are devised based on game values with coalition structure, where the features are grouped based on dependencies. We show analytically that grouping features this way has a stabilizing effect on the marginal operator on both group and individual levels, and allows for the unification of marginal and conditional explanations. Our results are verified in a number of numerical experiments where an information-theoretic measure of dependence is used for grouping.
\end{abstract}

\vspace{5pt}

{ \footnotesize {\bf keywords}: ML interpretability, explanation operator, game value, Radon-Nikodym derivative, mutual information.}

\vspace{5pt}


\section{Introduction}

The use of Machine Learning (ML) models has become widespread due to their dominance over traditional statistical techniques. In particular, contemporary ML models have a complex structure which allows for a higher predictive power and the capability of processing a larger number of attributes. Having a complex model structure, however, comes at the expense of increased difficulty of interpretability\footnote{We use the words interpretability (interpretation) and explainability (explanation) interchangeably. However, the methods discussed in this paper primarily deal with post-hoc explanations derived from model's results; for details on interpretable models vs post-hoc explanations see \cite{Hall-Gill}.}. This, in turn, may raise concerns of model trustworthiness and create other issues if not appropriately managed.

Explaining the outputs of complex ML models (such as ensemble trees or neural nets) has applications in several fields. Predictive models, and strategies that rely on such models, are sometimes subject to laws and regulations, such as the Equal Credit Opportunity Act\nocite{ECOA}. The latter requires financial institutions to notify consumers who have been declined or negatively impacted by a credit decision of the main factors that contributed to that decision. Another application is in medicine, where ML models are used to predict the likelihood of a certain disease or a medical condition, or the result of a medical treatment \cite{Ji2021,Teneggi2023}. Model interpretations (or explanations) can then be used to make judgments regarding the most contributing factor affecting the likelihood of the disease or the choice of the most optimal treatment; for instance, see \cite{Elshawi et al}.

The objective of a model explainer is to quantify the contribution of each predictor to the value of a predictive model $f$ trained on the data $(X,Y)$, where $X \in \R^n$ are predictors and $Y$ is a response variable. Many post-hoc explanations (in the ambient settings) are based on the pair $(X,f)$. However, there are numerous methods that rely on the structure of the model, its implementation, and even the sequence of algorithmic steps that led to the construction of such a model. 

There is a comprehensive body of research that discusses approaches for construction of post-hoc explainers, as well as self-explainable models. Some of the notable works on this topic are \cite{Friedman} on Partial Dependence Plots (PDP), \cite{Ribeiro et al} on Local Interpretable Model-agnostic Explanations (LIME), \cite{LundbergLee} on Shapley additive explanations (SHAP) based on Shapley value \cite{Shapley}, \cite{Hu et al} on locally interpretable models based on data partitioning, \cite{Vaughan et al} on explainable neural networks, \cite{Alvarez-Melis, Elton} on self-explainable models, among others.

Many promising interpretability techniques utilize ideas from cooperative game theory for constructing explainers using game values with appropriately designed games adopted to a machine learning setting \cite{Strumbelj2014,LundbergTreeSHAP,Wang-Lundberg,ChenLShapley2019,Miroshnikov2022,Teneggi2023,Saabasdoc,Cohen2023}. In this setting, given a model $f$, the features $X=(X_1,\dots,X_n)$ are viewed as $n$ players playing a random cooperative game $v(S;X,f)$, a set function on the subsets of indices $S\subseteq \{1,2,\dots,n\}$. Here the randomness comes from the features. While the literature considers games that are observations of random games, in our paper we will view them as random variables which enables us to perform rigorous analysis.

Two of the most notable games based on the pair $(X,f)$ are the conditional and marginal games\footnote{ In the literature, the conditional and marginal games are typically defined as functions of an observation $x$ instead of $X$, which makes the corresponding deterministic games.}
\begin{equation*}
    \vce(S):=\E[f(X)|X_S], \quad \vpdp(S):=\E[f(x_S,X_{-S})]\big|_{x_s=X_S}, \quad S\subseteq N:=\{1,2,\dots,n\}
\end{equation*}
where $X_S:=(X_{i_1},\dots,X_{i_k})$, $S=\{i_1,\dots,i_k\}$, and $-S:=N \setminus S$.
These are motivated by the corresponding deterministic games introduced in \cite{LundbergLee} and discussed in \cite{Janzing2020, Sundararajan2019}. 
For other examples of appropriate games used in ML setting see the works of \cite{LundbergTreeSHAP, Wang-Lundberg, ChenLShapley2019, Miroshnikov2022, Teneggi2023}.

A game value $(N,v) \mapsto h[N,v]\in\RR^n$ is a quantification of feature contributions to the model's output when $v \in \{\vce,\vpdp\}$. Intuitive explanations have been proposed in \cite{Sundararajan2019,Janzing2020, Chen-Lundberg} on how to interpret the game values based on each game.  Roughly speaking, conditional game values explain predictions $f(X)$ viewed as a random variable, while marginal game values explain the transformations occurring in the model $f(x)$, sometimes called mechanistic explanations \citep{Elton}. The work of \cite{Chen-Lundberg} intuitively describes conditional explanations, also known as observational, as consistent with the data (``true-to-the-data'') and marginal explanations, also known as interventional, as consistent with the model (``true-to-the-model'').  Some of the articles that describe implementation of Shapley values or their approximates for the above games are \cite{LundbergLee,LundbergIntTreeShap,aas2021,Filom2023,Kotsiopoulos2023}.

In this article, we study model explanations based on linear game values defined as linear operators
\begin{equation*}
\bar{\oper}^{\CE}[f; h,X] : f \mapsto h[N,v^{\CE}(\cdot;X,f)]  \quad \text{and} \quad \bar{\oper}^{\ME}[f; h, X] : f \mapsto h[N,v^{\ME}(\cdot;X,f)]
\end{equation*}
on appropriate functional spaces. We investigate the continuity of these operators which illuminates the differences between the two games. The heuristic concepts of ``true-to-the-model'' and ``true-to-the-data'' as discussed in the work of \citep{Chen-Lundberg} inspired us to introduce the rigorous notion of consistency in explanations with respect to a probability measure (in the space of features) using continuity arguments.

We show that conditional explanations are continuous in the space of models $f \in L^2(P_X)$, where $P_X$ is the pushforward measure. As a consequence, any two models, with similar inputs and predictions, will have similar explanations. For this reason, we define explanations to be $P_X$-consistent (or conditionally-consistent) if they are continuous in $L^2(P_X)$; see Section 2. Similarly, we show that marginal explanations are continuous in a different space $L^2(\tilde{P}_X)$, where $\tilde{P}_X := \frac{1}{2^n}\sum_{S\subseteq N}P_{X_S}\otimes P_{X_{-S}}$. For this reason, these explanations are considered $\tilde{P}_X$-consistent (or marginally-consistent) and (by design) encode the input-output relationship in the model. We show that in special cases, where the marginal explanations are also continuous in $L^2(P_X)$, the bound on these explanations may grow indefinitely as the strength of the dependencies increases, which serves as a precursor of instability.

Given the above formalism, and  partially motivated by the discussions in \cite{Janzing2020, Sundararajan2019, Kumar2020,MartingroupShapley}, we next state the issues associated with marginal and conditional explanations, which we attempt to resolve in our work. 

\begin{itemize}
\item [$(i)$] It is well-known from the Rashomon effect \cite{Breiman2001} that under predictor dependencies distinct models that approximate the same data well can have different representations \cite{Fisher2019}. Consequently, the marginal explanations for models with similar predictions may vary significantly, while conditional ones will be similar. Theoretically, it means that the marginal explanations may not be continuous in the space $L^2(P_X)$, while conditional ones are; see \S \ref{sec::obser_inter_expl}. This property may have an adverse impact in practical applications where models are periodically retrained or when different models are trained on the same data. Moreover, it also has an adverse effect for assessing global feature importance during the modelling process \cite{Fisher2019}.

\item [$(ii)$] In light of the curse of dimensionality, computing conditional game values is typically infeasible when the predictor dimension is large, which is the case in many applications; see \cite{Hastie et al}. There are several methods that attempt to approximate conditional games; see \cite{aas2021,olsen}. Others replace the game with one that attempts to mimic the conditioning such as in the case of the path-dependent TreeSHAP algorithm \cite{LundbergTreeSHAP,LundbergIntTreeShap}, or assume predictor independence such as KernelSHAP \cite{LundbergLee}, which effectively results in estimating the marginal explanations. The aforementioned methods have limited success and no theoretical guarantees on estimation accuracy. 

\item [$(iii)$] Given a model with highly dependent features, additive explainers spread any meaningful contributions (of latent variables) across dependent components, which can lead to rendering their individual explanations extremely minuscule; see \citep{Kumar2020,aas,MartingroupShapley}, and \S \ref{sec::examples}. As a consequence, this affects the ranking of individual features based on their explanations. Some informative dependent features may be ranked lower in the list, while less informative, and often independent, features are ranked higher.

\end{itemize}

In this article, we address issues $(i)$-$(iii)$ by studying the continuity of suitably defined feature explanation operators. We design group explainers, which utilize predictor groups to output contribution values of both predictor groups as well as single predictors within a group. To address the aforementioned issues, we group predictors by dependencies using an information-theoretic approach of \cite{Reshef16} and then investigate the analytical properties of corresponding group explainers in the context of operator continuity. We show that explainers based on quotient game values  or game values with coalitional structure, such as the Owen value \cite{Owen} or Two-step Shapley value \cite{Kamijo2009}, allow for the unification of the marginal and conditional approaches, which as a consequence provides a remedy (or mitigates) the instability of marginal explanations in $L^2(P_X)$. We also show that predictor grouping alleviates the issue of contribution splitting. To our knowledge a rigorous treatment of explanations in a functional analytic setting has never been done before. We believe our work can provide the proper language for understanding when to employ the aforementioned games.

Designing explainers based on predictor groups has been discussed before in \cite{aas}, where the groups are formed based on linear dependencies. The authors of \cite{aas} observe that forming groups by dependencies alleviates the inconsistencies between the marginal and conditional approaches. The work of \cite{MartingroupShapley} focuses on quotient game explainers and provides a practical perspective on their implementation. The groups there are formed by feature knowledge rather than dependencies, and the conditional game is approximated by the method outlined in \cite{aas2021}. Those works mainly focus on the Shapley value and consider deterministic empirical games and investigate practical aspects of grouping. Motivated by the aforementioned articles, our work focuses on rigorous analysis of group explainers which we believe complements the works of \cite{aas,MartingroupShapley}. Below is a brief summary of technical results presented in our paper that address issues $(i)$-$(iii)$.

\vspace{5pt}

\noindent{\bf Summary of key technical results}

\begin{itemize}

\item We set up game-theoretic explainers based on a game value $h[\cdot,\cdot]$ in the form \eqref{lingameform} as operators. We show that the conditional operator $f\mapsto\bar{\oper}^{\CE}[f;X,h]$ associated with a game value $h$ and predictors $X$ is continuous in $L^2(P_X)$, while the marginal operator $\bar{\oper}^{\ME}$ is continuous in $L^2(\tilde{P}_X)$; see Theorems \ref{prop::condoperator} and \ref{prop::margoperator}. We show that the marginal operator can become ill-posed or unbounded in $L^2(P_X)$; see Theorem \ref{thm::margoperatorunbound}. We define $P_X$-consistent and $\tilde{P}_X$-consistent explanations as those continuous in $L^2(P_X)$ and $L^2(\intP_X)$, respectively. We also establish procedures for proper extensions of game values to non-cooperative games such as marginal and conditional; see Lemmas \ref{lmm:extlingame}. We will present various conditions for well-posedness and $P_X$-consistency of the marginal explanations by discussing the absolute continuity of $\intP_X$ with respect to $P_X$ and the corresponding Radon-Nikodym derivative (if it exists) which encodes the strength of dependencies in features; see Theorems \ref{prop::margoperatorwellpos}, \ref{thm::margoperatorunbound}, and Proposition \ref{prop::cond_marg_value_bound}.

\item Given a partition $\cP=\{S_1,S_2,\dots,S_m\}$ of predictor indices and a  linear game value $h$, we consider each union $X_{S_j}$ as a player and assign its contribution to be the quotient game value $h_j[M,v^{\mathcal{P}}]$, where $v^{\mathcal{P}}(A):=v\big( \cup_{j \in A} S_j\big)$ is the quotient game on $M:=\{1,2,\dots,m\}$, with $v \in \{\vpdp,\vce\}$.
 We show that if the unions $X_{S_1},X_{S_2},...,X_{S_m}$ are independent, the quotient game values for marginal and conditional games coincide, which implies the continuity of marginal quotient explainers in $L^2(P_X)$, making them both $\intP_X$-consistent and $P_X$-consistent; see Lemma \ref{lmm::boundquot} and Proposition \ref{prop::quotgameexpl},  which also consider the case where the union independence is dropped. Moreover, the complexity of the quotient game explainer is $O(2^m)$ which can be significantly lower than $O(2^n)$, the complexity of the single-feature explainer. We also study trivial group explainers obtained via summation of explanations over the elements of $\mathcal{P}$ in the context of additive models; see Proposition \ref{prop::additivesumshap}.

\item We design single feature explainers that consider predictor dependencies by utilizing coalitional game values, maps in the form $(N,v,\mathcal{P}) \mapsto g[N,v,\mathcal{P}]$ where the coalitional structure is encoded in the partition $\mathcal{P}=\{S_1,\dots,S_m\}$ of predictors. We introduce a novel two-step representation formula for coalitional game values consisting of two game values (applied to games played across and within groups, respectively) and a family of intermediate games. Many known game values such as the Owen value, the Banzhaf-Owen value, the symmetrical Banzhaf value, and the two-step Shapley value admit such a representation. This representation allows for the construction of a large collection of coalitional game values with desirable properties, such as efficiency and quotient game property \eqref{quotientgame}; see Lemmas \ref{lmm::efficg2step} and \ref{lmm::qp2step}$(i)$-$(iii)$. We show that such explanations are stable in a finer space than $L^2 (\intP_X)$; see Proposition \ref{prop::extra1} and Corollary \ref{corr:extra3}. Furthermore, under union independence, \eqref{quotientgame} allows us to unify conditional and marginal approaches for trivial group explainers associated with the partition $\mathcal{P}$; see Proposition \ref{prop::unifcoalexpl}. Finally, we generalize game values with a two-step formulation to recursive game values under a generic partition tree; see \S\ref{sec::parttree1}.

\item To form groups of predictors based on dependencies, which effectively reduces the number of explainable components, we propose a variable hierarchical clustering technique that employs a state-of-the-art measure of dependence called the maximal information coefficient, a regularized version of mutual information introduced in \cite{Reshef16}. This method allows for a practical construction of a partition $\cP=\{S_1,S_2,\dots,S_m\}$ of predictor indices $N=\{1,2,\dots,n\}$ based on dependencies present in the joint distribution of $X$. We utilize the clustering and provide numerical examples that illustrate the stabilization effect of grouping.

\end{itemize}

\noindent {\bf Structure of the paper.}
In \S\ref{sec::MLI}, we introduce the requisite concepts such as the conditional and marginal games, the notion of game value, and in particular, the Shapley value. We also provide a review of the relevant literature. The conditional and marginal game operators are set up in \S\ref{sec::obser_inter_expl}. Theorems \ref{prop::condoperator}-\ref{thm::margoperatorunbound} address the stability of conditional and marginal explanations and highlight their differences; the relevant proofs can be found in \ref{app::gameops} of the appendix. \S\ref{sec::Group_expl} introduces and investigates various types of group explainers in great detail. The section starts with an extension of game values to non-cooperative games in \S\ref{sec::prelimgames}, and proceeds to trivial group explainers (\S\ref{sec::TGExplainers}), quotient game explainers (\S\ref{sec::qgameexpl}), and group explainers based on a coalitional structure. Results on their stability and computational complexity are presented. Section \S \ref{sec::Grouping} provides an outline of variable hierarchical clustering via the maximal information coefficient and its application to a synthetic dataset with dependencies. Next, in \S\ref{sec::examples}, we provide examples that illustrate the theoretical aspects outlined in  \S\ref{sec::Group_expl} on both synthetic and real-world data. A conclusion is outlined in \S\ref{sec::conclusion}. The paper finishes with an appendix containing all technical proofs, and the generalization of the two-step representation to recursive explanations under a generic partition tree.

\section{Preliminaries}\label{sec::MLI}

\subsection{Notation and hypotheses}\label{sec::notation}
Throughout this article, we consider the joint distribution $(X,Y)$, where $X=(X_1,X_2,\dots,X_n) \in \R^n$ are the predictors, and $Y$ is a response variable with values in $\RR$ (not necessarily a continuous random variable). Let the trained model, which estimates the true regressor ${\E}[Y|X=x]$, be denoted by $f(x)$. We assume that all random variables are defined on the common probability space $(\Omega,\mathcal{F},\PP)$, where $\Omega$ is a sample space, $\mathcal{F}$ a $\sigma$-algebra of sets, and $\PP$ a probability measure. We let $P_X$ be a pushforward measure of $X$ on $\RR^n$ and its support be denoted by $\Chi:={\rm supp}(P_X)$. Similarly, we denote $\Chi_i:={\rm supp} (P_{X_i})$,  $i \in \{1,2,\dots,n\}$.

Let $S \subseteq N$. Let $X_S$ denote the set of features $X_i$ with $i \in S$ and let $\Chi_S$ denote its support, where we ignore the predictors' ordering to improve readability. We say that the predictors $X_S=\{X_i\}_{i \in S}$ are independent if $P_{X_S}=\prod_{i \in S} \otimes P_{X_i}$. 
Let $\mathcal{P}=\{S_1,S_2,\dots,S_m\}$ be a partition of predictor indices $\{1,2,\dots,n\}$. We say that the group predictors $X_{S_1}, X_{S_2}, \dots, X_{S_m}$ are independent if $P_X=\prod_{j=1}^m \otimes P_{X_{S_j}}$.

Given $\epsilon>0$, the $(L^p,\epsilon)$-Rashomon set of models about $f_*$  is defined to be the ball of radius $\epsilon$ around a given model $f_*$  in the space $L^p (P_X)$, that is, $\{f \in L^p(P_X): \E[|f_*(X)-f(X)|^p] \leq \epsilon^p \}$. This is a modified version of the definition in \cite{Fisher2019} which also incorporates the distance from the response variable $Y$ to $f_*(X)$. Finally, the collection of Borel functions on $\RR^n$ is denoted by  $\mathcal{C}_{\mathcal{B}(\RR^n)}$.

Let $X=(X_1,\dots,X_n)$ and $Z=(Z_1,\dots,Z_m)$ be random vectors. Let $D(\cdot,\cdot)$ be a metric on the space of Borel probability measures $\mathscr{P}_k(R^{m+n})$ with $k$-th finite moment, for some $k \geq 0$. We say that $X$ and $Z$ are $(D,\epsilon)$-weakly independent if $D(P_{(X,Z)},P_X \otimes P_Z) \leq \epsilon$.

\subsection{Explainability and game theory}\label{subsec::preliminaries}

The objective of a (local) model explainer $E(x; f, X)=(E_1,\dots,E_n)$ is to quantify the contribution of each predictor $X_i$, $i\in\{1,\dots,n\}$, to the value of a predictive model $f \in \mathcal{C}_{\mathcal{B}(\RR^n)}$ at a data instance $x \sim P_X$.

Many promising interpretability techniques utilize ideas from cooperative game theory for constructing explainers. A cooperative game with $n$ players is a set function $v$ that acts on a set of size $n$, say $N=\{1,2,\dots,n\}$, and satisfies $v(\varnothing)=0$. A game value is a map $v \mapsto h[N,v]\in \R^n$ that determines the worth of each player. See \S\ref{sec::prelimgames} for more details. 

In the ML setting,  the features $X=(X_1,X_2,\dots,X_n)$  are viewed as $n$ players in an appropriately designed game $S \mapsto v(S;x,X,f)$ associated with the observation $x\sim P_X$, random features $X$, and model $f$. The game value $h[{\color{blue}N,}v]$ then assigns the contributions of each respective feature to the total payoff $v(N;x,X,f)$ of the game at the data instance $x$.

Two of the most notable games in the ML literature are given by
\begin{equation}\label{margcondgamedet}
    \vce_*(S; x,X,f)=\E[f(X)|X_S=x_S], \quad \vme_*(S;x,X,f)=\E[f(x_S,X_{-S})]
\end{equation}
with
\begin{equation*}
\vce_*(\varnothing; x, X, f)=\vpdp_*(\varnothing; x, X, f)=\E[f(X)]
\end{equation*}
introduced in \cite{Strumbelj2014,LundbergLee} in the context of the Shapley value \citep{Shapley}

\begin{equation}\label{shapform}
\varphi_i[N,v] = \sum_{S \subseteq N \backslash\{i\}} \frac{s!(n-s-1)!}{n!} [ v(S \cup \{i\}) - v( S ) ], \quad  s=|S|, \,  n=|N|.
\end{equation}

The value $\varphi$ satisfies the axioms of symmetry, linearity and the aforementioned efficiency property  (see \hyperref[axiom:SP]{(SP)}, \hyperref[axiom:LP]{(LP)} and \hyperref[axiom:EP]{(EP)} in Appendix \ref{app::gameaxioms}). The efficiency property, most appealing to the ML community, allows for a disaggregation of the payoff $v(N)$ into $n$ parts that represent a contribution to the game by each player: $\sum_{i=1}^n \varphi_{i}[N,v] = v(N).$

The games defined in \eqref{shapgame} are not cooperative since they do not satisfy the condition  $v(\varnothing)=0$. 
In such a case, the efficiency property reads as $\sum_{i=1}^n \varphi_{i}[N,v] = v(N)-v(\varnothing)$. See \S\ref{sec::prelimgames} for a careful treatment of game values for non-cooperative games.

In this paper, to study game-theoretical explainers in their entirety, we consider random conditional and marginal games given by
\begin{equation}\label{shapgame}
    \vce(S; X,f)=\E[f(X)|X_S], \quad \vme(S;X,f)=\E[f(x_S,X_{-S})]\big|_{x_s=X_S}
\end{equation}
which are well-defined for $f \in \mathcal{C}_{\mathcal{B}(\RR^n)}$ with $\E[|f(X)|] < \infty$ and related to the deterministic ones in \eqref{margcondgamedet} via $\vce=\vce_*|_{x=X}$ and $\vme=\vme_*|_{x=X}$. For these games, the corresponding Shapley values $\varphi[N,\vce]$ and $\varphi[N,\vme]$, respectively, are random vectors in $\RR^n$.

\begin{remark} \rm
The deterministic and random Shapley explainers are trivially related as follows:
\[
\varphi[N,\vce]=\E[\varphi[N,\vce]|X]=\varphi[N,\vce_*](x)\big|_{x=X}  \quad 
\text{and} \quad
\varphi[N,\vme]=\E[\varphi[N,\vme]|X]=\varphi[N,\vme_*](x)\big|_{x=X}.
\]
where we used the fact that $\varphi[N,\vce],\varphi[N,\vme]$ are measurable with respect to $\sigma(X)$.
\end{remark}

This motivates the following definition of a generic random explainer.
\begin{definition}\label{randexplainer}
Let $X=(X_1,\dots,X_n)$ be predictors. Suppose $E(\cdot;\cdot,X)$ is a model explainer defined for every $f \in \mathcal{C}_{\mathcal{B}(\RR^n)}$ and $x \in \Chi$. Suppose the map $x \mapsto E(x; f,X) \in \RR^n$ is  Borel. The random model explainer induced by $E$ is defined by $\pulloper[f;E,X]:=E(X;f,X)$, $f \in \mathcal{C}_{\mathcal{B}(\RR^n)}$.
\end{definition}

Notice that the map $x \mapsto E(x; f,X)$ in the definition above takes values in $\Bbb{R}^n$ but $\pulloper[f;E,X]$ is a random vector of dimension $n$.

\begin{definition}[\bf consistency]\label{def::consistency}
Let $X=(X_1,\dots,X_n), E(\cdot;\cdot,X)$ and $\pulloper[\cdot;E,X]$ be as in Definition \ref{randexplainer}. Suppose that $\pulloper[f;E,X] \in L^2(\Omega,\mathcal{F},\PP)^n$ for every $f \in L^2(\RR^n,\mu)$ where $\mu$ is a Borel probability measure on $\R^n$. We say that  $E$ (and $\pulloper$) is $\mu$-consistent if $f \mapsto \pulloper[f;E,X]$ is locally Lipshitz continuous, that is, for every $f_*\in L^2(\mu)$ there exists a constant $c=c_{f_*} \geq 0$ such that
\[
\| \pulloper[f]-\pulloper[f_*]\|_{L^2(\P)} \leq c_{f_*} \|f-f_*\|_{L^2(\mu)}, \quad \forall f \in L^2(\mu).
\]
\end{definition}

The consistency condition guarantees that models that are similar in $L^2(\mu)$, in the sense they are close in $L^2(\mu)$, have similar explanations (up to a scaling constant determined by the bound). For instance, suppose $\mu=P_X$  and $c_{f_*}=1$. Then, if $\|f_*-f\|_{L^2(P_X)} \leq \epsilon$, that is, the predictions of $f$ and $f_*$ are close to one another within $\epsilon$, then their explanations are also close to each other within $\epsilon$. We further note that if $\pulloper$ is linear, then $\mu$-consistency is equivalent to the global Lipshitz continuity with $c(f_*)=\|\pulloper\|$ for each $f_*\in L^2(\mu)$.

\begin{remark} \rm
In principle, one can replace the $L^2$ spaces in Definition \ref{def::consistency} with the spaces $L^p(\Omega,\mathcal{F},\PP)$ and $L^p(\RR^n,\mu)$, respectively. 
\end{remark}

In what follows, when the context is clear, we suppress the explicit dependence of $v(S; X,f)$, where $v \in \{\vce,\vme\}$, on $X$ and $f$. Furthermore, we will refer to values $\varphi_i[N,\vpdp]$ and $\varphi_i[N,\vce]$ as marginal and conditional Shapley values.

\subsection{Relevant works on individual feature attributions}

When it comes to explanations, there are global methods such as PDP (Partial Dependence Plots) \cite{Friedman} or BETA (Black Box Explanations through Transparent Approximations) \cite{Lakkaraju2017} which describe the overall effect of features, as well as local methods such as the rule-based method Anchors \cite{Ribeiro2018anchors}, or LIME (Local Interpretable Model-agnostic Explanation) \cite{Ribeiro et al} and SHAP (SHapley Additive exPlanations) \cite{LundbergLee} which provide individualized feature attributions to explain a single prediction based on the game-theoretic work of Shapley \cite{Shapley}. 

Game-theoretic explainability methods, such as in \cite{LundbergLee}, often utilize the marginal and conditional games in \eqref{shapgame} in their deterministic rather than random form; the former can be obtained by conditioning the formulas in \eqref{shapgame} on an individual sample, i.e. $X_S=x_S$. The games are often referred to as  interventional and observational respectively, which are terms borrowed from the causality literature community. Strictly speaking, however, the interventional game is based on the direct acyclic graph associated with the feature distribution and properly defined using the $do$ operator \cite{Pearl2000}. Only under certain conditions do the marginal and interventional games coincide \cite{Zhao2019}. For this reason, we refer to the games in their probabilistic context.

There has been a collection of noted articles devoted to the difference between marginal and conditional Shapley values, a topic which is at the heart of our paper. Articles \cite{Janzing2020} and \cite{Sundararajan2019} argued that the marginal Shapley value is appropriate as an explanation of the model (in other words, explaining the input-output process) as it satisfies the so called null-player property, meaning all predictors that are not explicitly used by the model are attributed zero value. Article \cite{Chen-Lundberg} replied to the criticism with the statement that the conditional Shapley value is true-to-the-data and the marginal one is true-to-the-model and the two games have different objectives. Their introduced definitions, while being somewhat intuitive, lack a more rigorous definition.

Computing a game value of the empirical marginal game is computationally intensive for large number of predictors and nearly infeasible for the conditional one. In practice, the marginal game is approximated by the empirical marginal one defined in \eqref{empmarggame} via averaging across a background dataset (which ideally should be the dataset that includes both training and test samples). Interventional TreeSHAP algorithm \cite{LundbergIntTreeShap} is an optimized  algorithm for tree-based models which takes a background dataset as an input. It estimates the marginal Shapley explanations, where the accuracy depends on the size of the dataset; see \cite[Lemma 2.3]{Kotsiopoulos2023}. 

Paper \cite{LundbergLee} introduced the KernelSHAP model-agnostic algorithm which assumes independence of predictors and attempts to approximate the marginal Shapley value by employing variational principles. Paper \cite{LundbergTreeSHAP} introduces the path-dependent TreeSHAP algorithm that replaces the conditional game with one that mimics the conditioning by utilizing the model's tree structure, which produces an ad-hoc approximation of the conditional game value. There are also model-specific methods for estimating Shapley values for neural networks such as the DeepSHAP algorithm \cite{chen2021}.

For a proper estimation of conditional Shapley values see the notable work of \cite{aas2021}, where the approximation is done via non-parametric vine copulas. Furthermore, in \cite{olsen} the authors employ variational autoencoders with arbitrary conditioning for models with dependent features. 
To our knowledge, a rigorous error analysis has not been carried out for these methods, making the error bounds unknown.

The work of \cite{Kotsiopoulos2023} follows the ideas of \cite{Strumbelj2014} and develops a collection of sampling methods for group and coalitional game values for the marginal game. Moreover, it addresses the statistical analysis of these estimations, providing rigorous error bounds for each method.

The work of \cite{Filom2023} designs an algorithm that computes marginal explanations for the CatBoost ML algorithm. The explanations avoid the use of the background dataset and directly utilize the model's internal structure to compute marginal game values for a large class of linear game values and coalitional game values. This method is computationally fast and statistically very precise. More importantly, unlike path dependent TreeSHAP, it is implementation invariant \cite{Sundararajan2017, Filom2023}.

Grouping predictors to construct explainers has been discussed before in \cite{aas,MartingroupShapley}. In \cite{aas}, explanations of unions of predictors are constructed using the KernelSHAP \citep{LundbergLee} method; there, the unions are formed by dependencies (using correlation-based methods) and the Shapley value contributions are obtained via summation of single feature explanations across the groups. It has been observed by the authors of \cite{aas} that forming groups by dependencies alleviates the inconsistencies between the marginal explanations and the data.

In \cite{MartingroupShapley} the authors focus on constructing conditional game explainers using grouping and provide a practical perspective. The groups there are treated as players, which leads to the computation of the Shapley value for the quotient conditional game, and the groups are formed by feature knowledge rather than dependencies; the conditional game here is approximated by the method outlined in \cite{aas2021}. Our work is partially motivated by the studies in \cite{aas} and confirms some of their findings. Furthermore, our work is much more general and is applicable to game values other than Shapley.

\section{Conditional and marginal game operators}
\label{sec::obser_inter_expl}

In our work, the game $\vce$ is referred to as conditional and $\vpdp$ as  marginal; see \eqref{shapgame} for definitions. 
If predictors $X$ are independent, the two games coincide. In the presence of dependencies, however, the games are very different. The conditional game explores the data by taking into account dependencies, while the marginal game explores the model $f$ in the space of its inputs, ignoring the dependencies. Strictly speaking, the conditional game is determined by the probability measure $P_X$, while the marginal game is determined by the product probability measures $P_{X_{S}} \otimes P_{X_{-S}}$, $S \subseteq N$.

The explanations based on these two games  have been addressed in the works  \cite{Sundararajan2019,Janzing2020,Chen-Lundberg,Miroshnikov2022}. These works illustrate that, for certain types of models, the conditional Shapley explanations are consistent with observations  while the marginal ones are consistent with the model.

Building upon the aforementioned works, we offer our viewpoint by introducing operators based on the two games whose outputs are explanations viewed as random variables. This construction allows us to better understand the relationships between explanations, the data, and the model; and will be used later to discuss some stability questions motivating group explainers.

An appealing property of the marginal and conditional games is that of linearity with respect to models. Specifically, given random features $X=(X_1,X_2,\dots,X_n)$ and two continuous models $f,g$ we have
\[
v(S;X,\alpha \cdot f+g)=\alpha \cdot v(S;X,f)+v(S;X,g), \quad v\in\{\vce,\vme\}.
\]
If the game value $h[N,v]$ is also linear, the linearity extends to explanations
\[
h[N,v(S;X,\alpha \cdot f+g)]=\alpha \cdot h[N,v(S;X,f)]+h[N,v(S;X,g)], \quad  v\in\{\vce,\vme\}
\]
on the space of continuous models. To extend the marginal and conditional games to a more general class of models, we consider equivalence classes of models $L^2(\mu)$ for an appropriate Borel probability measure $\mu$, on which the games are well-defined maps. Once the spaces are defined, the linearity of explanations provides a natural approach to obtaining explanations of certain ML ensembles (such as sums of trees) because the construction of explanations focuses on each single term of the ensemble, simplifying the process of determining the appropriate game for a given case.

\subsection{Stability theory of single feature explainers based on linear game values} \label{subsec::single_feat_stability}

We begin the discussion by introducing  linear operators associated with the conditional game 
and then investigating their properties.  
A necessary ingredient for constructing such an operator is a linear game value which allows quantifying the contribution of each feature. For simplicity, in this section, we work with  the linear game value $h$ in the (marginalist) form
\begin{equation}\label{lingameform}
h_i[N,v]=\sum_{S \subseteq N\setminus\{i\}} w(S,n) \big[ v(S\cup\{i\}) - v(S) \big], \quad i \in N=\{1,2,\dots,n\}.
\end{equation}
Such game values are determined by weights $w(S,n)$ where $S$ is a proper subset of $N$.
Notice that the Shapley value \eqref{shapform} is of the form above. Indeed, game values of this form satisfy desirable properties such as linearity \hyperref[axiom:EP]{(LP)} and the null-player property \hyperref[axiom:NPP]{(NPP)} (cf. Appendix \ref{app::gameaxioms}).

\begin{definition}\label{def::condoperator} Let $h$ be a game value as in \eqref{lingameform} and $X=(X_1,\dots,X_n)$ be defined on $(\Omega,\mathcal{F},\PP)$. 

\begin{itemize}
  \item [(i)] The conditional game operator $
\oper^{\CE}: L^2(\Omega,\mathcal{F},\P) \to  L^2(\Omega,\mathcal{F},\P)^n$
associated with $h,X$ is defined by 
\begin{equation}\label{condoperator}
\oper_i^{\CE}[Z; h,X]:=\sum_{S \subseteq N\setminus\{i\}} w(S,n) 
\big[ \E[Z|X_{S\cup\{i\}}] - \E[Z|X_{S}] \big], \quad i \in N,
\end{equation}
where we set $\E[Z|X_{\varnothing}]:=\E[Z]$.

\item [$(ii)$] The pullback conditional game operator $\bar{\oper}^{\CE}: L^2(P_X) \to  L^2(\Omega,\mathcal{F},\P)^n$ associated with $h,X$ is defined by 
\begin{equation*}%\label{cond_operator_pull}
\bar{\oper}^{\CE}[f;h,X]:=h[N,\vce(\cdot ;X,f)].
\end{equation*}
\end{itemize}
\end{definition}

For the ease of notation, throughout this section we denote the Hilbert space $L^2(\Omega,\mathcal{F},\P)$ by $L^2(\PP)$ and  assume that $X=(X_1,\dots,X_n)$ is a random vector defined on $(\Omega,\mathcal{F},\P)$.

\begin{theorem}[\bf properties]\label{prop::condoperator} Let h, X and 
$\oper^{\CE}$ be as in Definition \ref{def::condoperator}. Then:
\begin{itemize}

\item [$(i)$] $\oper^{\CE}_i$ is a bounded linear, self-adjoint operator satisfying
\begin{equation}\label{condoperatorstab}
\|\oper^{\CE}_i[Z;h,X]\|_{L^2(\PP)} \leq \Big(\sum_{S \subseteq N\setminus\{i\}} |w(S,n)|\Big) \|Z\|_{L^2(\P)}.
\end{equation}

\item [$(ii)$] Let $X_i \in L^2(\P)$. If $X_i \indep X_{N\backslash \{i\}}$ and \hyperref[axiom:NN]{(NN)} holds, then $\|\oper^{\CE}_i\|=\sum_{S \subseteq N\setminus\{i\}} |w(S,n)|$.

\item [$(iii)$] 
$\{Z\in L^2(\PP):\, Z \indep X\}  \subseteq \{Z\in L^2(\PP):\, \E[Z|X_{S \cup \{i\}}] = \E[Z|X_{S}], \,\, S \subseteq N \setminus \{i\} \} \subseteq {\rm Ker}(\oper_i^{\CE})$, $i\in N$.

\item [$(iv)$] 
$\{Z\in L^2(\PP):\, Z \indep X\}\subseteq \{Z\in L^2(\PP): \E[Z|X]=const \,\, \text{$\P$-a.s.}  \} 
\subseteq {\rm Ker}(\oper^{\CE}).$

\item[$(v)$] ${\rm Ker}(\oper^{\CE})=\{Z\in L^2(\PP): \E[Z|X]=const \,\, \text{$\P$-a.s.}  \} $
if $h$ satisfies axiom \hyperref[axiom:TPG]{(TPG)}.

\item [$(vi)$] 
If $h$ satisfies the efficiency property \hyperref[axiom:EP]{(EP)}, then $\sum_{i=1}^n \oper_i^{\CE}(Z) = \E[Z|X]-\E[Z]$.
\end{itemize}
\end{theorem}
\begin{proof}
See Appendix \ref{app::prop::condoperator_proof}.
\end{proof}

\begin{remark}\rm An immediate consequence of Theorem \ref{prop::condoperator}$(i)$-$(iv)$ is the following stronger inequality
\begin{equation*}%\label{condoperatorstab}
\|\oper^{\CE}_i[Z;h,X]\|_{L^2(\PP)} \leq \Big(\sum_{S \subseteq N\setminus\{i\}} |w(S,n)|\Big) \|Z-\E[Z]\|_{L^2(\P)}.
\end{equation*}
\end{remark}

We next present two corollaries to Theorem \ref{prop::condoperator}; see Appendix \ref{app::prop::condoperator_proof} 
for their proofs.
\begin{corollary}\label{corr::cond_operator_cons}
Let $Y \in L^2(\Omega,\mathcal{F},\P)$ and $h$ be as in \eqref{lingameform}. Set $\epsilon := Y - \E[Y|X]$. Then
\begin{equation}\label{dataconsist}
\oper^{\CE}[Y ; h,X]=\oper^{\CE}[ \E[Y|X]; h,X], \quad \oper^{\CE}[ \epsilon; h,X] =0.
\end{equation}
\end{corollary}
\begin{proof}
Follows immediately from that fact that $\epsilon\in {\rm Ker}(\oper^{\CE})$ due to Theorem \ref{prop::condoperator}$(iv)$.
\end{proof}

Equation \eqref{dataconsist}  states that if the regressor is independent of the noise, then the conditional explanations of the model and the response variable coincide.

\begin{corollary}\label{corr::cond_operator_cont} Let $h$, $X$, $\bar{\oper}^{\CE}$ be as in Definition \ref{def::condoperator}.
\begin{itemize}
\item [$(i)$] The operator $\bar{\oper}^{\CE}$ is a bounded linear operator satisfying
\[
\|\bar{\oper}^{\CE}[f_1;h,X]-\bar{\oper}^{\CE}[f_2;h,X]\|_{L^{2}(\P)^n}\leq C\|f_1-f_2\|_{L^2(P_X)}.
\]
Here $C:=\sqrt{n}\max_{i}(C_i)$ where $C_i$ is the constant on the right-hand side of \eqref{condoperatorstab}.

\item [$(ii)$]  For a game value $h$ of the form \eqref{lingameform} 
  which satisfies \hyperref[axiom:NN]{(NN)} and the efficiency property \hyperref[axiom:EP]{(EP)},  the Lipschitz inequality from $(i)$ can be improved as
\begin{equation}\label{effvalbound}
\|\bar{\oper}^{\CE}[f_1;h,X]-\bar{\oper}^{\CE}[f_2;h,X]\|_{{L^{2}(\P)^n}}\leq \|f_1-f_2\|_{L^2(P_X)}.
\end{equation}

\item [$(iii)$] 
 One has ${\rm Ker}(\bar{\oper}^{\CE})\supseteq\{f\in L^2(P_X):f=const \,\, \text{$P_X$-a.s.}\}$ with equality achieved if $h$ satisfies \hyperref[axiom:TPG]{(TPG)}.
\end{itemize}
\end{corollary}
\begin{proof}
See Appendix \ref{app::prop::condoperator_proof}.
\end{proof}

\begin{remark}
\rm
Arguments in  \ref{app::prop::condoperator_proof}  also show that  
$\|\oper^{\CE}[Z_1;h,X]-\oper^{\CE}[Z_2;h,X]\|_{L^{2}(\P)^n}\leq C\|Z_1-Z_2\|_{L^2(\P)}$
with $C$ as above; and if conditions $w(S,n)\geq 0$ and the efficiency are satisfied, the inequality may be sharpened as  
$\|\oper^{\CE}[Z_1;h,X]-\oper^{\CE}[Z_2;h,X]\|_{L^{2}(\P)^n}\leq \|\E[Z_1-Z_2|X]\|_{L^2(\P)}\leq \|Z_1-Z_2\|_{L^2(\P)}$.
For details of those arguments, see the proof of Corollary \ref{corr::cond_operator_cont} in Appendix \ref{app::prop::condoperator_proof}.
\end{remark}

Corollary  \ref{corr::cond_operator_cont} implies that for two distinct models $f_1(x)$, $f_2(x)$ that approximate the data well, the  conditional explanations are consistent with those of the data.

We next take a similar approach in constructing an operator based on the marginal game. To choose an appropriate space of models, note that for any bounded $f \in \mathcal{C}_{\mathcal{B}(\RR^n)}$ the marginal game is given by
\[
\vpdp(S;X;f)= \int f(X_S,x_{-S}) P_{X_{-S}}(dx_{-S}), \quad S \subseteq N,
\]
which implies that 
\begin{equation*}%\label{margspaceprob}
\E\big[\vpdp(S;X;f)\big]= \int f(x_S,x_{-S}) [P_{X_S} \otimes P_{X_{-S}}](d x_S, dx_{-S}).
\end{equation*}
Since the marginal explanations based on the game value \eqref{lingameform} are  linear combinations of $\vpdp(S;X;f)$, $S \subseteq N$, natural domains for the marginal operator are the spaces $L^q(\intP_X)$, $q\geq 1$, with the corresponding co-domains being $L^q(\P)$, where
\begin{equation}\label{Ptilde}
\intP_X:=\frac{1}{2^n}\sum_{S\subseteq N} P_{X_S} \otimes P_{X_{-S}}
\end{equation}
with the corresponding $L^q$-norm
\[
 \|f\|^q_{L^q(\intP_X)} := \frac{1}{2^n} \sum_{S \subseteq N}  \int f^q(x_S, x_{-S})   [P_{X_S} \otimes P_{X_{-S}}](dx_S,dx_{-S}),
\]
where we ignore the variable ordering in $f$ to ease the notation, and we assign $P_{X_{\varnothing}} \otimes P_X=P_X \otimes P_{X_{\varnothing}}=P_X$. 
In what follows, we develop the $L^2$-theory for the marginal explanations.

\begin{definition}\label{def::margoperator} Let $h$, $X$ be as in Definition \ref{def::condoperator}. The marginal game operator $
\bar{\oper}^{\ME}: L^2(\intP_X) \to  L^2(\Omega,\mathcal{F},\P)^n$
associated with $h,X$ is defined by 
\begin{equation}\label{margoperator}
\bar{\oper}^{\ME}[f;h,X]:=h[N,\vpdp(\cdot\,;X,f)].
\end{equation}
\end{definition}

\begin{theorem}[\bf properties]\label{prop::margoperator} 

Let $X$, $h$, $f$, and $(\bar{\oper}^{\ME},L^2(\tP_X))$ be as in Definition \ref{def::margoperator}. Then:
\begin{itemize}

\item [$(i)$] $\bar{\oper}_i^{\ME}$ is a well-defined, bounded linear operator satisfying

\begin{equation*}%\label{margoperatorstab}
\begin{aligned}
\|\bar{\oper}^{\ME}_i[f;h,X]\|_{L^2(\PP)} & \leq 2^{\frac{n+1}{2}}\Big(\sum_{S \subseteq N \setminus \{i\}} w^2(S,n)\Big)^{\frac{1}{2}} \|f\|_{L^2(\intP_X)}. 
\end{aligned}
\end{equation*}

\item [$(ii)$] $\{f \in L^2(\tP_X): f=const\,\,  \text{$\tP_X$-a.s.}\} \subseteq {\rm Ker}(\bar{\oper}^{\ME})$.

\item [$(iii)$]  If axiom \hyperref[axiom:TPG]{(TPG)} holds, then ${\rm Ker}(\bar{\oper}^{\ME}) \subseteq\{f \in L^2(\tP_X): f=const\,\,  \text{$P_X$-a.s.}\}$.

\item [$(iv)$] If axiom \hyperref[axiom:TPG]{(TPG)} holds and $\tP_X\ll P_X$, ${\rm Ker}(\bar{\oper}^{\ME}) = \{f \in L^2(\tP_X): f=const\,\,  \text{$\tP_X$-a.s.}\}$.

\item [$(v)$] If $f(x)=f(x_{N \setminus \{i\} })$ for some $i \in N$, then $i$ is a null player for $\vpdp(\cdot; X,f)$.

\item [$(vi)$] $\{f \in L^2(\tP_X): f(x)=f(x_{N\backslash \{i\}}) \text{  $\tP_X$-a.s.}\} \subseteq {\rm Ker}(\bar{\oper}_i^{\ME})$.

\item [$(vii)$]  If $h$ satisfies the efficiency property \hyperref[axiom:EP]{(EP)}, then 
$\sum_{i=1}^n \bar{\oper}_i^{\ME}[f] = f(X)-\E[f(X)]$.
\end{itemize}
\end{theorem}
\begin{proof}
See Appendix \ref{app::prop::margoperator_proof}.
\end{proof}

\begin{remark}\rm An immediate consequence of Theorem \ref{prop::margoperator}$(i)$-$(ii)$ is the following stronger inequality
\begin{equation*}%\label{margoperatorstab}
\|\bar{\oper}^{\ME}_i[f;h,X]\|_{L^2(\PP)} \leq 2^{\frac{n+1}{2}}\Big(\sum_{S \subseteq N \setminus \{i\}} w^2(S,n)\Big)^{\frac{1}{2}} \|f-\tilde{f}_0\|_{L^2(\intP_X)}, \quad \tilde{f}_0 := \E_{x \sim \intP_X}[f(x)]. 
\end{equation*}
\end{remark}

\begin{lemma}\label{lemm::discont} Let $X$, $h$ be as in Definition \ref{def::margoperator}. Let $f_1,f_2 \in L^2(\intP_X)$ satisfy $f_1(x) - f_2(x) = \sum_{i=1}^n \eta_i(x_i)$. Suppose that the weights in \eqref{margoperator} satisfy \hyperref[axiom:NVA]{(NVA)}. Then 
\begin{equation*} %\label{margconsist}
\|\bar{\oper}^{\ME}[f_1;h,X]-\bar{\oper}^{\ME}[f_2;h,X]\|_{L^{2}(\P)^n} \geq C \big( \|f_1 - f_2\|_{L^{2}( \intP_{X})}-|\E[f_1(X)-f_2(X)]| \big)
\end{equation*}
for some $C>0$ independent of $f_1,f_2$ provided that $\eta_i\in L^2(P_{X_i})$ for each $i$.
\end{lemma}

\begin{proof}
See Appendix  \ref{app::lemm::discont}.
\end{proof}

Theorem \ref{prop::margoperator}$(i)$ states that the marginal operator is bounded in $L^2(\intP_X)$ and hence the marginal explanations are continuous in $L^2(\intP_X)$. In addition, Lemma \ref{lemm::discont}  guarantees (in special cases) that models that are far apart in $L^2(\intP_X)$ yield marginal explanations that are far apart. Under dependencies in predictors, however, two models that are close in $L^2(P_X)$ may yield (as we will see) marginal explanations that are far apart in $L^2(\intP_X)$, which may cause  the map  $(X,f)\mapsto \bar{\oper}^{\ME}[f;X]$ to be unbounded on some other domains; see the discussion below in \S \ref{sec::instab_marg}.

\begin{remark}
\rm The theory we developed in \S \ref{subsec::single_feat_stability} views explanations as maps from a space of models to a space of random variables. While the intuitive notions of true-to-the-model and true-to-the-data introduced in \cite{Chen-Lundberg} are not equivalent to the continuity in $L^2(\intP_X)$ and $L^2(P_X)$, respectively, they are related. Roughly speaking, for explanations to be true-to-the-data, it is necessary for the explanation map to be continuous in $L^2(P_X)$, and to be true-to-the-model continuity in $L^2(\intP_X)$ is required. Below we present a simple example illustrating that marginal explanations depend on the model representation, while the conditional ones do not.
\end{remark}

\begin{example}\label{ex::marg_instab} \rm
Let $X=(X_1,X_2,X_3)$  with $\E[X_i]=0$. Suppose that $X_i=Z + \epsilon_i$, $\epsilon_i\sim \mathcal{N}(0,\delta)$, $i\in\{1,2\}$, for some small $\delta > 0$, where $Z \sim \mathcal{N}(0,1)$. Also suppose that
$\epsilon_1,\epsilon_2,Z,X_3$ are independent, and let the response variable be
\begin{equation*}%\label{model1}
Y  = f_0(X) := X_1 + X_2 + X_3.
\end{equation*}

Note that there are many good models defined on $\Chi=\widetilde{\Chi}=\RR^3$ that represent the same data in $L^2$ sense. For instance, consider
\begin{equation*}
  f_{\alpha}(x)=(1+\alpha)x_1+(1-\alpha)x_2+x_3, \quad x \in \RR^3, \, \alpha \in [0,1],
\end{equation*}
in which case the response variable can be expressed by 
\[
Y = f_{\alpha}(X) + \epsilon_{\alpha}, \quad \|\epsilon_{\alpha}\|_{L^2(\P)} \leq \sqrt{2}\delta,
\]
where $\epsilon_{\alpha}:=\alpha(\epsilon_2-\epsilon_1)$. 
Note that the models satisfy: 
\begin{equation*}%\label{distinct}
 f_{\alpha} \in L^2(\tP_X), \quad \|f_{\alpha}-f_{\beta}\|_{L^2(P_X)} = \sqrt{2} \delta |\alpha-\beta|, \quad |\alpha-\beta| \leq \|f_{\alpha}-f_{\beta}\|_{L^2(\intP_X)}<\infty.
\end{equation*}

Consider next the conditional explanations based on Shapley value $h=\varphi$. Direct computations of the explanations for the response variable give:
\begin{equation*}%\label{dataexpl}
\begin{aligned}
\oper^{\CE}_1[Y; \varphi,X]&=\frac{1}{2}\big( 2 X_1 + \E[X_2|X_1] -  \E[X_1|X_2] \big)=X_1 + O(\delta),\\
\oper_2^{\CE}[Y; \varphi,X]&=\frac{1}{2}\big( 2 X_2 + \E[X_1|X_2] -  \E[X_2|X_1] \big)=X_2 + O(\delta),\\
\oper^{\CE}_3[Y; \varphi,X]&=X_3.
\end{aligned}
\end{equation*}
Using the fact that $\epsilon_{\alpha}$ and $X_3$ are independent, we obtain the explanations of the residuals to be
\begin{equation*}%\label{data_shap}
\begin{aligned}
\oper^{\CE}_i[\epsilon_{\alpha}; \varphi,X]&=\frac{\alpha}{2}\big( \epsilon_1-\epsilon_2 \pm (\E[\epsilon_1|X_1] + \E[\epsilon_2|X_2]) \big)=O(\delta), \, i\in\{1,2\} \quad \text{and} \quad \oper^{\CE}_3[\epsilon_{\alpha}; \varphi,X]=0.\\
\end{aligned}
\end{equation*}
Then, employing the linearity of $\oper^{\CE}$, the conditional Shapley explanations for $f_{\alpha}$ are then given by

\begin{equation}\label{ex::obs_shap}
\begin{aligned}
\oper^{\CE}_1[f_{\alpha}(X); \varphi,X]  =X_1 + O(\delta), \quad \oper^{\CE}_2[f_{\alpha}(X); \varphi,X] =X_2 + O(\delta), \quad \oper^{\CE}_3[f_{\alpha}(X); \varphi,X]=X_3.
\end{aligned}
\end{equation}

Furthermore, for any two models $f_{\alpha}$ and $f_{\beta}$, we have
\[
 \|\bar{\oper}^{\CE}_i[f_{\alpha};\varphi,X]-\bar{\oper}^{\CE}_i[f_{\beta}; \varphi,X]\|_{L^2(\P)} \leq (\alpha-\beta) O(\delta), \quad i \in \{1,2,3\},
\]
where $|O(\delta)|\leq 3 \delta$. Thus, as $\delta \to 0$, we get the same conditional explanations in the limit for all models $f_{\alpha}$.

On the other hand, computing marginal expectations, we obtain
\begin{equation}\label{ex::marg_shap}
\bar{\oper}^{\ME}_1[f_{\alpha};\varphi,X] = (1+ \alpha) X_1, \quad \bar{\oper}^{\ME}_2[f_{\alpha};\varphi,X] = (1-\alpha)  X_2,  \quad \bar{\oper}^{\ME}_3[f_{\alpha};\varphi,X] = X_3.
\end{equation}
Furthermore, for any two models $f_{\alpha}$ and $f_{\beta}$, we have
\begin{equation*}%\label{ex::marg_shap_est}
 \bar{\oper}^{\ME}_i[f_{\alpha};\varphi,X]-\bar{\oper}^{\ME}_i[f_{\beta}; \varphi,X]=\pm(\alpha-\beta)X_i=\pm(\alpha-\beta)Z+O(\delta), \quad i \in \{1,2\} \quad \text{in \, $L^2(\P)$}.
\end{equation*}

Comparing equations \eqref{ex::obs_shap} and \eqref{ex::marg_shap}, we see that the conditional Shapley values for predictors $X_1,X_2$ are independent of the representative model up to small additive noise, while that is not the case for the marginal ones. 
\end{example}

\subsection{Stability of marginal explanations on a space equipped with $L^2(P_X)$-norm}\label{sec::instab_marg}

The objective of this section is to investigate when the marginal explanations behave as the conditional ones. That is, we will determine when we can expect that the two models that to have similar predictions have similar marginal explanations, and how the dependencies in features impact dissimilarity. To answer these questions, it is necessary to investigate the stability of marginal explanations on a space equipped with $L^2(P_X)$-norm.

If one attempts to equip the space $L^2(\intP_X)$ with the $L^2(P_X)$-norm, then the marginal game operator may not always be well-defined or bounded; see Theorem \ref{prop::margoperatorwellpos} and Theorem \ref{thm::margoperatorunbound}. To understand this, define the following space: 
\begin{equation}\label{datasubspace}
H_X := \bigg( \Big\{[f]: [f]=\big\{\tilde{f}: \text{$\tilde{f}=f$ $P_X$-a.s. and } \int |\tilde{f}(x)|^2 \intP_X(dx) < \infty \big\}\Big\}, \, \| \cdot \|_{L^2(P_X)} \bigg) \hookrightarrow L^2(P_X).
\end{equation}

Note that either $H_X$  contains exactly the same elements as $L^2(\intP_X)$ or some elements of $L^2(\intP_X)$ are placed in the same equivalence class of $H_X$. Strictly speaking, $H_X$  is a quotient space of $L^2(\intP_X)$ modulo $H^0_X:=\{f\in L^2(\intP_X): \|f\|_{L^2(P_X)}=0\}$ equipped with the $L^2(P_X)$-norm; keep in mind that, since $P_X \ll \intP_X$, if $f_1=f_2$ $\intP_X$-a.s., then $f_1=f_2$ $P_X$-almost surely.  

It turns out, as the lemma below states, that the absolute continuity of $\intP_X$ with respect to $P_X$ is a necessary and sufficient condition for the marginal game to be a well-defined map on $H_X$.
\begin{lemma}\label{lmm::marg_game_wellposed}
Let $X=(X_1,\dots,X_n)$ be the predictors. The map $f \in H_X \mapsto \{\vme(S;X,f)\}_{S \subseteq N} \in (L^2(\P))^{2^n}$ is well-defined if and only if $\intP_X \ll P_X$. Consequently, $(\bar{\oper}^{\ME}[\cdot;h,X],H_X)$ is well-defined for every linear game value $h$ if and only if $\intP_X \ll P_X$.
\end{lemma}
\begin{proof}
See Appendix \ref{app::lmm::marg_game_wellposed}. 
\end{proof}

In other words, the above lemma states that if the density of $\tilde{P}_X$  with respect to $P_X$  exists, then the game value as an operator on $H_X$ is well-defined. A geometric consequence of the above lemma is given in Remark \ref{remark::geometry}.

\begin{remark}\label{remark::geometry} \rm
The continuity condition $\intP_X \ll P_X$ can be related to the shape of the support of $P_X$. Indeed, it requires
${\rm{supp}}(\intP_X)={\rm{supp}}(P_X)$. It is not hard to show that, conversely, this condition implies $\intP_X \ll P_X$ when $P_X$ admits a density function which is Lebesgue a.e. positive on ${\rm{supp}}(P_X)$. On the other hand, assumptions on the shape of ${\rm{supp}}(P_X)$ can preclude the continuity $\intP_X \ll P_X$. E.g., if the support is a lower-dimensional compact subset of $\Bbb{R}^n$ without any product structure, then $\intP_X \not\ll P_X$. See Appendix \ref{app::measures_relationship} for details.
\end{remark}

Given the lemma above, it is not surprising that the absolute continuity also comes up with regard to the marginal operator.

\begin{theorem}[\bf well-posedness]\label{prop::margoperatorwellpos}
Let $X=(X_1,\dots,X_n)$ be the predictors, and $h$  defined in \eqref{lingameform}. 
\begin{itemize}
  \item [$(i)$] Suppose $\intP_X \ll P_X$. Then $H_X \cong(L^2(\intP_X),\|\cdot\|_{L^2(P_X)})$ and $(\bar{\oper}^{\ME}[\cdot;h,X],H_X)$ acting via the formula \eqref{margoperator} is well-defined.

  \item [$(ii)$] Suppose  $\intP_X \not \ll P_X$. Then, for each $[f] \in H_X$ there exist $f_1,f_2 \in [f]$, such that $\|f_1 - f_2 \|_{L^2(\tP_X)}\neq 0$. Consequently, $H_X \cong(L^2(\intP_X)/H_X^0,\|\cdot\|_{L^2(P_X)})$, and $(\bar{\oper}^{\ME}[\cdot;h,X],H_X)$ is well-defined if and only if $H^0_X=\big\{f \in L^2(\tP_X): \|f\|_{L^2(P_X)}=0 \big\} \subseteq {\rm Ker}(\bar{\oper}^{\ME}[\cdot;h,X],L^2(\tP_X))$.

\end{itemize}
\end{theorem}
\begin{proof}
See Appendix \ref{app::prop::margoperatorwellpos}.
\end{proof}

Part $(ii)$ of Theorem \ref{prop::margoperatorwellpos} states that, even if $\intP_X \not \ll P_X$,  the marginal operator on $H_X$  may still be well-defined if $H_X^0$  is in the kernel of the marginal operator on $L^2(\tilde{P}_X)$ since functions in equivalence classes of $H_X^0$ when plugged into the formula \eqref{lingameform} yield zero explanations. In such a situation,  the linear combination of terms $\vme(S;f,X)$ encoded by $h$ gives rise to a well-defined map on $H_X$ even though at least one assignment 
$[f]\mapsto \vme(S;f,X)$ should be ill-posed, as according to Lemma \ref{lmm::marg_game_wellposed}.

\begin{example}\label{ex::marg_illposed} \rm
 Consider $h=\varphi$. Let $X=(X_1,X_2)$ satisfy $X_2=g(X_1)+Z$ where $Z$ is a bounded random variable independent of $X_1$ and $g$ is continuous. Suppose that the supports of $X_1,X_2$ are $\Chi_1=\Chi_2=\RR$, and that $|Z| \leq M$. In this case, $\Chi \subseteq \{(x_1,x_2): x_1\in \RR, |x_2-g(x_1)| \leq M\}$ where $\Chi$ is the support of $(X_1,X_2)$, and hence the complement $\Chi^C$ is a non-empty open set. Pick any open rectangle $R=(a,b) \times (c,d) \subset \Chi^C$ and set $f_R(x):=\1_{R}(x)\in L^2(\tP_X)$. Then, using the fact $P_X(R)=0$, we obtain \text{$\PP$-a.s.} 
\[
\vpdp(\varnothing;X,f_R)=\vpdp(\{1,2\};X,f_R)=0 \quad 
\]
 and
\[
\vpdp(\{1\};X,f_R)=\P(X_2 \in (c,d)) \1_{(a,b)}(X_1), \quad \vpdp(\{2\};X,f_R)=\P(X_1 \in (a,b)) \1_{(c,d)}(X_2).
\]
Then, we obtain \text{$\PP$-a.s.}
\[
% \bar{\oper}^{\ME}_1[f_R;\varphi,X]=-\bar{\oper}^{\ME}_2[f_R;\varphi,X] = 
\varphi_1[\vme(\cdot;X,f_R)]=-\varphi_2[\vme(\cdot;X,f_R)] = \frac{1}{2} \big( \P(X_2 \in (c,d)) \1_{(a,b)}(X_1) - \P(X_1 \in (a,b)) \1_{(c,d)}(X_2) \big)
\]
and hence, recalling that $(a,b)\subset \Chi_1$, $(c,d)\subset \Chi_2$ 
and $(a,b)\times(c,d)\subset\Chi^C$, we have 
\[
\|\varphi_i[\vme(\cdot;X,f_R)]\|^2_{L^2(\P)}=\frac{1}{4} \big( \P(X_2 \in (c,d))^2 \P(X_1 \in (a,b)) + \P(X_1 \in (a,b))^2 \P(X_2 \in (c,d)) \big)>0.
\]
Note that $f_R \in H_X$ satisfies $\|f_R\|_{H_X}=\|f_R\|_{L^2(P_X)}=0$, and hence $f_R \in [0]$-equivalence class of $H_X$. Since the marginal Shapley formula for $0$ and $f_R$ yields different outputs, the operator $(\bar{\oper}^{\ME},H_X)$ is ill-posed. 
\end{example}

The above discussion motivates us to focus our investigation on the case $\intP_X \ll P_X$. In this case, the Radon-Nikodym derivative of $\intP_X$ with respect to $P_X$ exists and encodes information about feature dependencies. The following lemma, which will be helpful for our analysis, provides a representation of the Radon-Nikodym derivative and the space $L^2(\intP_X)$.

\begin{lemma}\label{lmm::_marg_value_bound_px}
Let $X \in \RR^n$ be predictors. Suppose  $\tilde{P}_X \ll P_X$. Let $r:=\frac{d \intP_X}{dP_X}$. Then $L^2(\intP_X)$ can be identified with the weighted 
$L^2$-space $L^2_r(P_X)$ where 
\begin{equation}\label{radon_decomp}
  r = \tfrac{1}{2^n} \sum_{S \subseteq N}r_S\geq \frac{1}{2^{n-1}}, \quad \text{where} \quad 0 \leq r_S:=\tfrac{d P_{X_{S}} \otimes P_{X_{-S}}}{d P_X}\in L^1(P_X) \quad \text{with} \quad \|r_S\|_{L^1(P_X)}=1.
\end{equation}
\end{lemma}
\begin{proof}
See Appendix \ref{app::lmm::_marg_value_bound_px}.
\end{proof}

We next establish the conditions when the marginal game is continuous, that is, bounded. This will help us to determine when the marginal operator on $H_X$ is bounded.

\begin{lemma}[\bf game boundedness]\label{lmm::marggameunbound}
Suppose $\intP_X \ll P_X$. Let $r$, $r_S$ be as in Lemma \ref{lmm::_marg_value_bound_px}.

\begin{itemize}

  \item [$(i)$] Suppose that $r=\frac{d \intP_X }{dP_X} \in L^{\infty}(P_X)$, which is equivalent to
\begin{equation}\label{boundradon}\tag{BG}
[P_{X_S} \otimes P_{X_{-S}}](A \times B) \leq M \cdot P_{(X_S,X_{-S})}(A \times B), \,\, A \in \mathcal{B}(\RR^{|S|}), \, B \in \mathcal{B}(\RR^{|-S|})
\end{equation}
for any $S \subseteq N$ and some $M \geq 0$. Then the map $f \in H_X \mapsto \vme(S;X,f) \in L^2(\P)$, $S \subseteq N$, is bounded. 

\item [$(ii)$] Let $\varnothing \neq S \subset N$. Suppose that either
\begin{equation}\label{blowupcondgame}\tag{UG1}
\small \sup \bigg\{ \frac{[P_{X_S} \otimes P_{X_{-S}}](A \times B)}{P_X(A \times B)} \cdot P_{X_{-S}}(B), \,\, A \in \mathcal{B}(\RR^{|S|}), \, B \in \mathcal{B}(\RR^{|-S|}), P_X(A \times B)>0 \bigg\} = \infty.
\end{equation}
or the non-negative, well-defined Borel function 
\begin{equation}\label{blowupcondgame2}\tag{UG2}
\rho(x_S) := \int r^{1/2}_S(x_S,x_{-S}) P_{X_{-S}}(dx_{-S})
\end{equation}
with values in $\R \cup \{\infty\}$ is not $P_{X_S}$-essentially bounded. 

Then the map $f \in H_X  \mapsto \vme(S;X,f) \in L^2(\P)$ is unbounded.
\end{itemize}
\end{lemma}
\begin{proof}
See Appendix \ref{app::lmm::marggameunbound}.
\end{proof}

\begin{theorem}[\bf game value boundedness]\label{thm::margoperatorunbound}

Let $X,h$ be as in Theorem \ref{prop::margoperatorwellpos}. Suppose $\intP_X \ll P_X$, and let $r$ and $r_S$ be as in Lemma \ref{lmm::_marg_value_bound_px}.

\begin{itemize}

\item [$(i)$]
Suppose \eqref{boundradon} holds. Then $H_X = L^2(P_X)$ and for $f \in L^2(P_X)$ 
\begin{equation}\label{marg_value_bound}
  \| \bar{\oper}^{\ME}_i[f;h,X] \|_{L^2(\PP)} \leq \Big( 1+2 \cdot \max_{S \subseteq N} \|r_S-1\|_{L^{\infty}(P_X)} \Big) \Big(\sum_{S \subseteq N\setminus \{i\}} |w(S,n)| \Big) \| f \|_{L^2(P_X)}
\end{equation}
Consequently, $(\bar{\oper}^{\ME},H_X)$ is  bounded.

\item [$(ii)$] Suppose there exist two distinct indices $i,j \in N$ such that
\begin{equation}\label{blowupcond}\tag{UO}
\sup \bigg\{ \frac{[P_{X_i} \otimes P_{X_j}](A \times B)}{P_{(X_i,X_j)}(A \times B)} \cdot P_{X_j}(B), \,\,\, A,B \in \mathcal{B}(\RR), P_{(X_i,X_j)}(A \times B)>0 \bigg\} = \infty.
\end{equation}
Suppose that the weights in \eqref{margoperator} satisfy the non-negativity condition \hyperref[axiom:NN]{(NN)} and
\begin{equation}\label{weightsblowup}
 \sum_{S\subseteq N \setminus\{i,j\}} w(S,n) > 0.
\end{equation}
Then $(\bar{\oper}_i^{\ME},H_X)$, $(\bar{\oper}_j^{\ME},H_X)$, and $(\bar{\oper}^{\ME},H_X)$ are unbounded linear operators. 
\end{itemize}
\end{theorem}

\begin{proof}
See Appendix \ref{app::thm::margoperatorunbound}.
\end{proof}

\noindent
In \S\ref{sec::Group_expl}, after a careful treatment of  more general game values, the theorem above on stability will be extended to explainers with a coalition structure; see Proposition \ref{prop::coalqgameprop}.

\begin{remark}\rm
  When $r_S=1, \ \forall S\subseteq N$ (that is the predictors are independent  and hence $\intP_X=P_X$), the bound in \eqref{marg_value_bound} becomes that of \eqref{condoperatorstab}. See Appendix \ref{app::measures_relationship}.
\end{remark}

% \begin{remark}\label{BCisSharp}
% \rm To see how condition \eqref{blowupcond} for the unboundedness of the marginal operator emerges naturally, consider the case
% where $h$ is the Shapley value $\varphi$ for which the weights $w(S,n)$ are known (cf. \eqref{shapform}). 
% With $R=A\times B$ as in \eqref{blowupcond}, and setting $f_R(x):=\1_R(x_i,x_j)$,
% computations similar to those in Appendix \ref{app::prop::margoperatorwellpos} show that 
% $$
% \frac{\|\bar{\oper}_i^{\ME}[f_R]\|_{L^2(\PP)}^2}{\|f_R\|_{L^2(P_X)}^2}
% =\frac{1}{4}\frac{[P_{X_i} \otimes P_{X_j}](R)}{P_{(X_i,X_j)}(R)}\big(P_{X_i}(A)+P_{X_j}(B)\big)+O(1)
% $$
% as $A$ and $B$ vary among Borel subsets of $\RR$ with $P_{(X_i,X_j)}(A\times B)>0$. See Appendix \ref{app::on_condition_uo} for details.
% \end{remark}

\begin{remark}\label{BCisSharp}\rm
In Appendix \ref{app::on_condition_uo}, we shall explain how condition \eqref{blowupcond} for the unboundedness of the marginal operator emerges naturally by considering the case  of $h$ being the Shapley value $\varphi$ for which the weights $w(S,n)$ are known (cf. \eqref{shapform}).
\end{remark}

\begin{remark} \rm
In Theorem \ref{thm::margoperatorunbound}, we showed that if $r=\frac{d \intP_X}{dP_X}$ exists and belongs to $ L^{\infty}(P_X)$, then $H_X=L^2(P_X)$. It turns out that the converse is true as well. That is, if $\intP_X \ll P_X$ and $H_X=L^2(P_X)$, then $r \in L^{\infty}(P_X)$; see Lemma \ref{suppl::rn_noundedness}.
\end{remark}

Theorem \ref{thm::margoperatorunbound} suggests that there are two regimes for well-defined marginal explanations. In the first one, the explanations are bounded but the Lipschitz bound increases as the strength of dependencies increases. In the second one, the marginal operator is unbounded. Below are two examples that illustrate both cases.

\begin{example}\label{ex::marg_unbound_energy}\rm
Let $f(x) = \frac{1}{\sqrt{\delta}}(x_1 - x_2)$, $\delta>0$. Let $X=(X_1,X_2)$  with $\E[X_i]=0$. Let $X_i=Z + \epsilon_i$, $\epsilon_i\sim \mathcal{N}(0,\delta^2)$, $i\in\{1,2\}$, where $Z \sim \mathcal{N}(0,1)$, and $\epsilon_1,\epsilon_2,Z$ are independent.

First, note that $f(X)=\frac{1}{\sqrt{\delta}}(\epsilon_1-\epsilon_2)$ and hence, by independence of $\epsilon_1$ and $\epsilon_2$, we obtain
\[
\|f\|^2_{L^2(P_X)}=Var(f(X))= \delta^{-1} \cdot (Var(\epsilon_1)+Var(\epsilon_2))=2 \delta.
\]
Then, since  $\bar{\oper}_1^{\ME}[f;\varphi,X]=\frac{1}{\sqrt{\delta}}X_1, \bar{\oper}_2^{\ME}[f;\varphi,X]=-\frac{1}{\sqrt{\delta}} X_2$,
we conclude
\[
\|\bar{\oper}^{\ME}[f;\varphi,X]\|^2_{L^2(\P)^2} = \frac{2}{\delta}+O(\delta).
\]
Thus, as $\delta \to 0^+$, $\|f\|_{L^2(P_X)} \to 0$, but $\|\bar{\oper}^{\ME}[f;\varphi,X]\|_{L^2(\P)^2} \to \infty$.
\end{example}

\begin{example} \rm
Let $Y \sim \exp(1)$ and $Z\sim \mathcal{N}(0,1)$. Let $X=(X_1,X_2)$ be a random vector with values in $\R^2$ such that $P_X=\tfrac{1}{2}P_{(Y,Y)}+\tfrac{1}{2} P_Z \otimes P_Z$. By design, $\intP_X \ll P_X$ and hence the marginal Shapley value is a well-defined operator on $L^2(P_X)$. Take $t \in \R_+$ and define a square $R^t := [t-1,t] \times [t,t+1] =: I_1^t \times I_2^t$. Then, since $\lim_{t\to +\infty}\frac{(P_Y(I_j^t))^2}{P_Z(I_j^t)}=\infty$, $j\in\{1,2\}$, we have
\[
\lim_{t \to +\infty}\frac{[P_{X_1} \otimes P_{X_2}](R^t)}{P_X(R^t)}\cdot P_{X_j}(I_j^t) =\infty, \quad j \in \{1,2\}
\]
which by Theorem \ref{thm::margoperatorunbound} implies that the marginal Shapley value on $L^2(P_X)$ is unbounded.
\end{example}

The absolute continuity condition also allows to express the Wasserstein distance of the two probability measures using the Radon-Nikodym derivative, explaining how the latter controls the strength of dependencies among the predictors.

\begin{lemma} Let $X=(X_1,\dots,X_n)\in L^1(\P)$ be the predictors. Let $r$, $r_S$ be as in Lemma \ref{lmm::_marg_value_bound_px}.
\begin{equation*}%\label{W1_dist} 
W_1( \intP_X, P_{X}) \leq  \int |x| \cdot |r(x)-1| \, P_{X}(dx) \leq \frac{1}{2^n}\sum_{S\subseteq N} \int |x| \cdot |r_S(x)-1| \, P_{X}(dx) < \infty
\end{equation*}
\end{lemma} 
\begin{proof}
Follows from Lemma \ref{lmm::_marg_value_bound_px}, Lemma \ref{app::lmm::W1_bound}, and the triangle inequality.
\end{proof}

The above lemma illustrates that dependencies are controlled by the Radon-Nikodym derivative. When $r=1$,  then $r_S=1$, $S \subseteq N$ and the two measures coincide. When $r$ deviates from $1$, the dependencies start to impact the distance. As a consequence, the marginal and conditional explanations start  to differ from one another. The estimate on this difference is discussed below in the special case when the Radon-Nikodym derivative is bounded.

\begin{lemma}\label{lmm::cond_marg_value_bound}
Suppose  $\tilde{P}_X \ll P_X$. Let $r$, $r_S$ be as in Lemma \ref{lmm::_marg_value_bound_px}. Suppose $r \in L^2(P_X)$.
\begin{itemize}

  \item [$(i)$] Let $f \in L^2_{r^2}(P_X)$. Then, for every $S \subseteq N$
    \begin{equation*}%\label{app::cond_marg_game_bound}
 \E\big[ (\vce(S;X,f) - \vme(S;X,f))^2 \big] \leq \| (r_S-1) \cdot f \|^2_{L^2(P_X)} < \infty.
  \end{equation*}

  \item [$(ii)$] Let $f \in L^2_{r^2}(P_X)$. Let $h[N,v]$ be a game value in the form \eqref{lingameform}. Then
  \begin{equation*}%\label{app::cond_marg_value_bound}
  \begin{aligned}
  &\Big(\E\big[ \big(h[N,\vce(\cdot;X,f)] - h[N,\vme(\cdot;X,f)] \big)^2 \big]\Big)^{1/2} \\
  &\leq \sum_{S \subseteq N \setminus\{i\} } |w(S,n)| \Big( \| (r_S-1) \cdot f \|_{L^2(P_X)} + \| (r_{S \cup \{i\}}-1) \cdot f \|_{L^2(P_X)} \Big).
  \end{aligned}
  \end{equation*}
\end{itemize}
\end{lemma}
\begin{proof}
See Appendix \ref{app::lmm::cond_marg_value_bound}.
\end{proof}

As a corollary we obtain the following approximation result.
\begin{proposition}[\bf approximation]\label{prop::cond_marg_value_bound}
Let $h$, $X$ be as in Definition \ref{def::condoperator} and $r$, $r_S$ as in Lemma \ref{lmm::_marg_value_bound_px}. Suppose $\tilde{P}_X \ll P_X$ with $r=\frac{d \intP_X}{ d P_X } \in L^{\infty}(P_X)$. Then $H_X = L^2(P_X)$, $r_S \in L^{\infty}(P_X)$, $S \subseteq N$, and for $f \in L^2(P_X)$ 
\begin{equation*}%\label{cond_marg_value_bound}
 \| \bar{\oper}^{\CE}_i[f;h,X] - \bar{\oper}^{\ME}_i[f;h,X] \|_{L^2(\PP)} \leq 2 \cdot \Big( \max_{S \subseteq N} \|r_S-1\|_{L^{\infty}(P_X)} \Big) \Big(\sum_{S \subseteq N\setminus \{i\}} |w(S,n)| \Big) 
  \| f \|_{L^2(P_X)}, \,\, i \in N.
\end{equation*}
\end{proposition}
\begin{proof}
See Appendix \ref{app::prop::cond_marg_value_bound}.
\end{proof}

\begin{remark}\rm 
It is crucial to point out that $\mu$-consistency of explanations is merely a stability (continuity) requirement with the Lipschitz bound determining the relative scale between explanation differences and the differences of associated models. Thus, the three criteria that are useful for the design of explanations are: a) $\mu$-consistency which determines the type of similarity of explanations, b) the Lipschitz bound which determines relative scaling of explanations and models, and c) the game which determines the ``shape'' of explanations.
\end{remark}

\vspace{5pt}

\noindent{\bf Global feature importance.} The above analysis extends to global feature attributions inherited from game values as follows. Given a game value $h$ and predictors $X=(X_1,\dots,X_n)$ define the global conditional and marginal attributions by
\[
 \beta(v,X,f):=\{\beta_i(v,X,f)\}_{i\in N}, \quad \beta_i:= \|h_i[N,v(\cdot;X,f)] \|_{L^2(\P)}, \,\, v \in \{\vce,\vme\}.
\]
Then, according to Corollary \ref{corr::cond_operator_cont} and Theorem \ref{prop::margoperator}, the global explanations satisfy the continuity condition $|\beta(v,X,f_1-f_2)| \leq C\|f_1-f_2\|_{L^2(\mu)}$, $f_1,f_2 \in L^2(\mu)$, with $\mu=P_X$ when $v=\vce$ and $\mu=\intP_X$ when $v=\vme$ for some $C=C(h)$. Furthermore, if $h$ satisfies conditions of Corollary \ref{corr::cond_operator_cont}$(ii)$ and $v=\vce$, then $C=1$.

\vspace{5pt}

\noindent{\bf Conclusion.} The difference between conditional and marginal explanations can be summarized as follows: 
\begin{itemize}
   \item [(1)] $f_n \to f$ in $L^2(P_X)$ implies $\bar{\oper}^{\CE}[f_n] \to \bar{\oper}^{\CE}[f]$.
 
   \item [(2)] $ f_n \to f$ in $L^2(\intP_X)$ implies $\bar{\oper}^{\ME}[f_n] \to \bar{\oper}^{\ME}[f]$. 

   \item [(3)] $f_n,f \in L^2(\tP_X)$ and $f_n \to f$ in $L^2(P_X)$ in general does not imply $\bar{\oper}^{\ME}[f_n] \rightarrow \bar{\oper}^{\ME}[f]$. 
\end{itemize}

Results of this subsection on the stability of conditional or marginal explanations will be vastly generalized in \S\ref{sec::Group_expl} for more general game values that are not necessarily in the form of \eqref{lingameform}; see Proposition \ref{prop::quotgameexpl}$(iii)$ and Proposition \ref{prop::coalqgameprop}. In \S\ref{sec::Group_expl}, we alleviate the instability of marginal explanations discussed in this section through constructing group explainers that can unify marginal and conditional approaches.

\subsection{Splitting of explanation energy on dependencies}\label{subsec::splitting}

We next provide an example that showcases that model's \textit{energy} (in the sense of its squared norm) is split on conditional explanations and some of it is dissipated. To see this, recall that the efficiency property puts a constraint on the vector $\bar{\oper}^{\CE}[f;h,X]$; its components should add up to $f(X)-\E[f(X)]$. As a consequence, in light of Corollary \ref{corr::cond_operator_cont}$(ii)$, the energy of the conditional explanation vector is bounded by that of the (centered) model:
\begin{equation}\label{energyineq}
\|\bar{\oper}^{\CE}[f;h,X]\|^2_{L^2(\PP)^n}=\sum_{i=1}^n \|\bar{\oper}_i^{\CE}[f;h,X]\|^2_{L^2(\PP)}  \leq \|f-f_0\|^2_{L^2(P_X)}, \quad f_0 := \E[f(X)]. 
\end{equation}

By contrast, in view of the Rashomon effect \citep{Breiman2001} and Theorem \ref{prop::margoperatorwellpos}, the energy of the model can be significantly lower than that of the marginal explanations; see Example \ref{ex::marg_unbound_energy}.

It is worth mentioning that, when the game value $h$ is efficient, then the independence of explanations (both marginal or conditional) leads to energy conservation. In general, for conditional explanations, we have the following result on the energy conservation.
\begin{lemma}\label{lemm::cond_operator_split}
Let $h$ be an efficient game value in the form \eqref{lingameform} with $w(S,n)\geq 0$. The equality in \eqref{energyineq} is achieved if and only if
\begin{equation}\label{conservenergycond}
\langle f(X) - \bar{\oper}^{\CE}_i[f;h,X], \bar{\oper}^{\CE}_i[f;h,X]  \rangle_{L^2(\P)} = 0, \quad \forall i \in N.
\end{equation}
\end{lemma}
\begin{proof}
See Appendix \ref{app::prop::condoperator_proof}.
\end{proof}

\begin{example}\label{ex::energycons}\rm 
Suppose $X=(X_1,\dots,X_n)$ are independent and let $f(x):=\sum_{j=1}^n g_j(x_j)$. Then for each $i \in N$ 
\[
\bar{\oper}^{\CE}_i[f;\varphi,X]=\bar{\oper}^{\ME}_i[f;\varphi,X]=\sum_{j=1}^n \bar{\oper}^{\ME}_i[g_j;\varphi,X] = \bar{\oper}^{\ME}_i[g_i;\varphi,X] = g_i(X_i)-\E[g_i(X_i)]
\]
where we used the independence of predictors, linearity of  $\bar{\oper}^{\ME}$,  null-player property, and efficiency property (see Theorem \ref{prop::margoperator}  for details). Thus, the conditional  (and marginal) explanations are independent and hence
\[
\|f-f_0\|^2_{L^2(P_X)}= Var(f(X))=\sum_{i=1}^n Var(g_i(X))= \|\bar{\oper}^{\ME}[f;\varphi,X]\|^2_{L^2(\P)^n} = \|\bar{\oper}^{\CE}[f;\varphi,X]\|^2_{L^2(\P)^n}.
\]
\end{example}

We next show that, under dependencies, model's energy can be dissipated on explanations.

\begin{example}\label{ex::energy_dissip_dep}\rm
Let $X=(X_1,\dots,X_n)$ and $f(x)$ be globally Lipschitz. Suppose that $X_i=Z + \epsilon_i$, where $Z \in L^2(\P)$ is a latent variable, and $\epsilon_i \sim \mathcal{N}(0,\delta)$, for each $i \in N$. 

Define $\bar{f}(x):= \frac{1}{n}\sum_{i=1}^n f(x_i,x_i,\dots,x_i)$. Since $\bar{f}$ is a symmetric function (that is, the order of the input components does not matter) and the Shapley value $\varphi$ is a symmetric game value, we must have
\[
\oper^{\CE}_i[ \bar{f}(X);\varphi,X]=\oper^{\CE}_j[ \bar{f}(X);\varphi,X], \quad i,j\leq n.
\]
Since $f$ is globally Lipschitz, by the linearity of $\oper^{\CE}$, we obtain for each $i \in N$
\[
\oper^{\CE}_i[ f(X);\varphi,X]=\oper^{\CE}_i[ \bar{f}(X);\varphi,X] + O(\delta) \quad \text{in $L^2(\P)$},
\]
% where $\delta := (\delta_1,\dots,\delta_n)$. 
Then, the last two equality imply that for $i,j\leq n$
\[
\bar{\oper}^{\CE}_i[ f;\varphi,X]=\bar{\oper}^{\CE}_j[ f;\varphi,X] + O(\delta), \quad \text{in $L^2(\P)$}.
\]
Hence, by efficiency of $\varphi$, for every $j \leq n$ we obtain
\[
f(X) - f_0 = \sum_{i=1}^n \bar{\oper}^{\CE}_i[ f;\varphi,X] = n \cdot \oper^{\CE}_j[ f(X);\varphi,X] + O(\delta) \quad \text{in $L^2(\P)$},
\]
where $f_0:=\E[f(X)]$. This implies that
\[
\| \bar{\oper}^{\CE}[f;\varphi,X] \|_{L^2(\P)^n} = \frac{1}{\sqrt{n}} \| f - f_0 \|_{L^2(P_{X})} + O(\delta). 
\]
\end{example}

The following example illustrates that  \eqref{conservenergycond} can be violated even if the predictors are independent. Here, the main cause of the energy dissipation is interaction in between the variables in the model (meaning that the model fails to be additive, unlike Example \ref{ex::energycons}).

\begin{example}\label{ex::energy_dissip_int}
\rm Suppose $X=(X_1,\dots,X_n)$ are independent and let $f(x)=\prod_{i=1}^n x_i$. Let $h=\varphi$. Then, using independence of predictors, direct computations of the Shapley value for yield
\[
\bar{\oper}^{\ME}_i[f;\varphi,X] = \bar{\oper}^{\CE}_i[f;\varphi,X]  = \frac{1}{n} \big( f(X)-f_0 \big), \quad i \in N, \quad f_0=\E[f(X)].
\]
Thus, we conclude
\[
\|\bar{\oper}^{\ME}[f;\varphi,X]\|_{L^2(\P)^n}=\|\bar{\oper}^{\CE}[f;\varphi,X]\|_{L^2(\P)^n} = \frac{1}{\sqrt{n}} \|f-f_0\|_{L^2(P_X)}.
\] 
\end{example}

\subsection{Relevant works in the context of stability}

The theory we developed in \S \ref{subsec::single_feat_stability} views explanations as maps from a space of models to a space of random variables. 
In this context, for an explanation map to be true-to-the-data we require it to be continuous in $L^2(P_X)$ while to be true-to-the-model we require it to be continuous in $L^2(\intP_X)$. These are satisfied by the marginal and conditional game values, respectively. Below is a discussion of other relevant works in the context of the stability theory presented in \S \ref{subsec::single_feat_stability}.

The interventional TreeSHAP algorithm \cite{LundbergIntTreeShap} produces an approximation of the Shapley value for the marginal game in the case of tree-based models, where the estimation error depends on the size of the background dataset; see \cite[Lemma 2.3]{Kotsiopoulos2023}. Therefore, these values are approximately ``true-to-the-model'' when viewed as random maps (which is achieved by replacing the observation $x$ in the explainer with $X$).

The path-dependent TreeSHAP method \cite{LundbergTreeSHAP}, however, relies on the implementation as it is shown in \cite{Filom2023}. In particular, the authors of \cite{Filom2023} construct a predictive model which can be represented by two distinct statistically similar trees. Consequently, this algorithm cannot be true-to-the-data as it violates the continuity in $L^2(P_X)$, nor is it true-to-the-model\footnote{It follows that the path-dependent TreeSHAP method (in the presence of dependencies) is an ad-hoc approximation of the conditional Shapley value rather than an approximation as claimed by the authors.}. Indeed, the path-dependent TreeSHAP value is not a well-defined map on the space of models. The continuity can be achieved only if the tree structure itself is incorporated in the functional space.

The article \cite{Filom2023} also shows that the ``eject'' variant of TreeSHAP introduced in \cite{Campbell2022} is also implementation dependent, which implies that the method (in the population limit) is neither true-to-the model nor true-to-the data given our definitions. 

Not all explanation methods that rely on the model's internal structure are ill-posed or unstable. For instance, the explanation technique developed in \cite{Filom2023} for a Catboost ML model is implementation invariant, does not require a background dataset, and relies only on the internal parameters of the Catboost model. The authors show that the algorithm approximates the marginal game values, making the method approximately true-to-the-model, and the estimation error depends on the size of the training set.

The paper \cite{ChenLShapley2019} introduces the so called L-Shapley value for structured data, where the observation $x\sim X$ is augmented with a graph in the feature space and $Y$ is a discrete random variable. The paper designs an information-theoretic game via the cross entropy between the distribution of $p(y|x)$ and $p(y|x_S)$. The objective is to approximate the Shapley value for this game by utilizing the graph structure, which allows to reduce complexity and remove weak interactions of features. Since the aforementioned game considers the conditional distribution of response variables, replacing $x_S$ with $X_S$ in $p(y|x_S)$ leads to continuity of (random) explanations in $L^1(\PP)$. If one adjusts the game to depend on the model $f$, where $Y=f(X)$, then one can obtain continuity in $L^1(P_X)$ of the map $f \mapsto p(f(X)|X_S)$, in turn leading to game values that are true-to-the-data. Furthermore, the L-Shapley value is not immune to contribution splitting (see \S \ref{subsec::splitting}) due to the property of additivity, which is not ameliorated by the information-theoretic setup of the game.

The work \cite{Kumar2020} discusses the Shapley value for the marginal and conditional games and provide certain perceived criticisms, some of which are addressed in our work. Specifically, they discuss the issue of a proxy predictor in the context of the conditional game, and illustrate that the attributions are different when a predictor is dropped from a regressor. As we show in \S \ref{sec::TGExplainers}, mutual information can be employed to group together proxy predictors. Incorporating these groups into the game itself can help mitigate contribution splits. Using the grouping methodology, the aforementioned issue is mitigated because the information of a group containing proxy predictors will not affect the conditional explanation when proxies are removed. 

The papers \cite{Sundararajan2019,Janzing2020,Kumar2020} points out the differences of the conditional and marginal game values by illustrating how the conditional Shapley value assigns a non-zero attribution to predictors not explicitly used by the model and that the marginal Shapley value considers ``out-of-distribution'' samples to assign attributions. While these are portrayed as criticisms, they are however properties of the corresponding explanations that align with the definitions of true-to-the-data and true-to-the-model,  respectively. Furthermore, the authors of  \cite{Kumar2020} present the additivity axiom (not to be confused with efficiency) for the sum of two games as being useful only in the context when the models themselves are additive. However, our theory indicates that the linearity of the Shapley value allows for explanations to be extended from additive to non-additive models while preserving the natural linearity property. In the literature, there are also other games (e.g. the baseline games) associated with ML models \cite{Janzing2020, Sundararajan2019,Covert2022}, but they are outside the scope of our study.

While our work investigates the stability of Shapley value explanations for each given instance, these can also be adapted for global feature importance; see Section \S \ref{subsec::single_feat_stability}. The Rashomon effect will impact the results when using the marginal Shapley value for feature importance. Meanwhile, the conditional explanations are shown to be independent of the model structure as they explain the response variable itself; see Corollary \ref{corr::cond_operator_cont}. This in turn implies that the corresponding feature importance will be unaffected by the Rashomon effect. A highly relevant work to this is the paper \cite{Fisher2019} that describes another global variable importance technique that seeks to address the Rashomon effect. Due to the existence of many models that approximate the data well, the authors define a collection of models called the $\epsilon$-Rashomon set containing models that have similar predictive power within some threshold $\epsilon$. Then, the work seeks to find a measure of global feature importance for the entire $\epsilon$-Rashomon set. To accomplish this, the authors define a feature importance, called model reliance (MR), by evaluating the ratio of losses between models with that feature with and without noise (this can be viewed as evaluating the expected loss when switching off and on the predictor). This ratio is then calculated for all models in the $\epsilon$-Rashomon set and the final global importance, called model class reliance (MCR), is defined as the interval with the minimum MR and maximum MR for the given predictor. The authors provide estimators for MR and MCR, with corresponding error bounds. In contrast, we suggest constructing predictor groups based on dependencies. For independent groups, the marginal and conditional explanations coincide (see Proposition \ref{prop::quotgameexpl}), which means that evaluating the former takes into account both the Rashomon effect and the issue of contribution splitting among proxy predictors. Thus, in our case we tackle the Rashomon effect not via a collection of models, but by unifying the marginal and conditional explanations in the context of game values.

\subsection{Game value extensions to non-cooperative games}\label{sec::prelimgames}

In this  subsection, we discuss possible extensions of generic linear game values, which are not necessarily in the form \eqref{lingameform}, to non-cooperative games such as marginal and conditional ones. Recall that a cooperative game with $n$ players is a set function $v$ that acts on a finite set of players $N \subset \mathbb{N}$ and satisfies $v(\varnothing)=0$. 
Typically, $N=\{1,2,\dots,n\}$. 
Recall that a game value is a map $(N,v) \mapsto h[N,v]\in \R^n$ that determines the worth of each player $i \in N$ of a game $v$. A set $T \subseteq N$ is called a \textit{carrier} of $v$ if $v(S)=v(S \cap T)$ for all $S \subseteq N$. 

In what follows, we shall repeatedly refer to the properties of game values outlined in Appendix \ref{app::gameaxioms} such as linearity \hyperref[axiom:LP]{(LP)}, efficiency property \hyperref[axiom:EP]{(EP)}, 
null-player property \hyperref[axiom:NPP]{(NPP)} etc.

Let $V_0$ be the set of all cooperative games with finitely many players.
Let us next consider set functions with finite carriers that violate the condition $v(\varnothing)=0$. To this end, let us denote the collection of such games by
\begin{equation}\label{noncoopgames}
V = \{ (N,v): \,\,  v(\varnothing)\in \RR, \quad v(S)=\tilde{v}(S), \,\, S \subseteq N,\quad  |S|\geq 1, \,\,\text{for some} \,\, (N,\tilde{v}) \in V_0 \}.
\end{equation}

One way to construct an extension of a linear game value to $V$ is to introduce a new player and turn a non-cooperative game $v$ with $n$ players into a cooperative one with $n+1$ players. Another approach is to incorporate the value $v(\varnothing)$ into the values of the extension. Here, we consider the latter approach to avoid dealing with an extra player.

In what follows, for each $v \in V$, the cooperative game $\tilde{v}$ denotes its projection 
onto $V_0$ as in \eqref{noncoopgames} (it agrees with $v$ on non-empty sets). Given a linear game value $h$, we seek an extension $\bar{h}$ to $V$ that satisfies:

\begin{enumerate}\label{extgamevalhyp}
    \item [(E1)] $\bar{h}[N,\tilde{v}]=h[N,\tilde{v}]$ for $(N,\tilde{v}) \in V_0$,
    \item [(E2)] $\bar{h}$ is linear on $V$.
\end{enumerate}

\begin{lemma}\label{lmm:extlingame}
Let $h$ be a linear game value. An extension $\bar{h}$ satisfying (E1)-(E2) has the representation:
\begin{equation}\label{extrepr}
{\bar{h}_i[N,v]=h_i[N,\tilde{v}] + \gamma_i v(\varnothing), \quad i \in N=\{1,2,\dots,n\},}
\end{equation}
{ where $\{\gamma_i\}_{i=1}^n$ are constants that depend on $N$}. Furthermore, any game in the form \eqref{extrepr} satisfies properties (E1)-(E2). In addition, if $h$ is symmetric, then $\bar{h}$ is symmetric if and only if $\gamma_i=\gamma_j$, for each $i,j \in N$.
\end{lemma}

\begin{proof} % proof is checked.
First, suppose $\bar{h}$ has the form \eqref{extrepr}. Let $(N,v)$ be any non-cooperative game 
and $(N,\tilde{v})$ the cooperative game which is its projection. Then
\[
\bar{h}_i[N,\tilde{v}]=h_i[N,\tilde{v}]+\gamma_i \tilde{v}(\varnothing)=h_i[N,\tilde{v}].
\]
The function $\bar{h}$ defined above is clearly linear if  $h$ is.

Suppose next that $\bar{h}$ is an extension of $h$ that satisfies (E1)-(E2). Take $(N,v) \in V$. Let $u_0$ denote a non-cooperative game satisfying $u_0(\varnothing)=1$ and $u_0(S)=0$ for any non-empty $S\subseteq N$. Observe that $v$ can be expressed as $v=\tilde{v}+v(\varnothing)u_0$. Then, using (E1) and (E2), we conclude
\[
\bar{h}[N,v]=h[N,\tilde{v}]+v(\varnothing)\bar{h}[N,u_0].
\]
Setting  $\gamma_i = \bar{h}_i[N,u_0]$, we obtain \eqref{extrepr}.

Finally, suppose that $h$ is symmetric and let $\bar{h}$ be its extension. Let $\pi$ be any permutation of $N$. Then
\[
\bar{h}_i[N,\pi v]=h_i[N,\widetilde{\pi v}] + \gamma_i \left(\pi v(\varnothing)\right)=h_{\pi(i)}[N,\tilde{v}] + \gamma_i v(\varnothing),
\]
where we used the fact that $\widetilde{\pi v}=\pi \tilde{v}$ and $\pi v(\varnothing)=v(\varnothing)$.
Then $\bar{h}_i[\pi v]=\bar{h}_{\pi(i)}[v]$ if and only if $\gamma_i=\gamma_{\pi(i)}$. Since $\pi$ is arbitrary, we conclude that $\bar{h}$ is symmetric if and only if $\gamma_i=\gamma_j$ for all $i,j$.
\end{proof}

For example, consider the Shapley value $\varphi$ defined in \eqref{shapform}. 
The same formula can be applied to non-cooperative games to construct an extension.  
In that case, one has $\gamma_i(\varphi,n)=-\frac{1}{n}$  and the extension satisfies $\bar{\varphi}_i[N,v]=\varphi_i[N,\tilde{v}]-\frac{1}{n}v(\varnothing)$. The efficiency property for the extension then reads as
\begin{equation*}%\label{effext}
\sum_{i=1}^n\bar{\varphi}_i[N,v] = v(N) - v(\varnothing).
\end{equation*}

Another well-known value is the \textit{Banzhaf value} \citep{Banzhaf1965} given by
\begin{equation}\label{banzhafform}
Bz_i[N,v]=\sum_{S \subseteq N\backslash\{i\}}\frac{1}{2^{n-1}} \big[ v(S \cup \{i\}) - v(S)\big], \quad i \in N.
\end{equation}
The Banzhaf value assumes that every player is equally likely to enter any coalition unlike the Shapley value that assumes that players are
equally likely to join coalitions of the same size,
and that all coalitions of a given size are equally likely.  There is only one property that differs between the two values. Shapley value satisfies the efficiency property \eqref{effp}, while Banzhaf value satisfies the \textit{total power property} \eqref{tpp} instead.

To extend the Banzhaf value using the same formula, we set $\gamma_i(Bz,n)=-\frac{1}{2^{n-1}}$. In this case, the extension $\overline{Bz}$ has exactly the representation \eqref{banzhafform} and the total power property \hyperref[axiom:TPP]{(TPP)} reads as
\begin{equation*}\label{tppext}
\sum_{i=1}^n\overline{Bz}_i[N,v] = \text{TPP}(Bz[N,\tilde{v}])-\frac{n}{2^{n-1}}v(\varnothing).
\end{equation*}

Given a game value $h$ and its extension $\bar{h}$, we will abuse the notation and write
\begin{equation*}%\label{abuse1}
\bar{h}(X; f,v):=\bar{h}[N,v( \cdot \,; X,f)], \quad v \in \{\vce,\vpdp\}.
\end{equation*}

\begin{definition}
Let $h$ be a linear game value and $\bar{h}$ its extension. We say that $\bar{h}$ is centered if $\bar{h}[N,c]=0$ for any constant non-cooperative game $(N,c)\in V$.
\end{definition}

Notice that the extensions of Shapley and Banzhaf values we introduced above are centered. 

\begin{lemma}\label{lmm::unitgame}
Let $h$ be a linear game value and $\bar{h}$ its extension with $\gamma=\{\gamma_i\}_{i=1}^n$ as in \eqref{extrepr}. Let $u$ denote a unit, non-cooperative game, that is, $u(S)=1$ for all $S \subseteq N$ and any   
$N$. Then
\begin{itemize}
\item [$(i)$] $\bar{h}$ is centered if and only if $\gamma=-h[N,\tilde{u}]$.

\item [$(ii)$] $\bar{h}$ is centered if and only if  $\bar{h}[N,v]=h[N,(v-v(\varnothing)u)]$.

\item [$(iii)$] 
 If $h$ has the form 
\begin{equation*}
h_i[N,\tilde{v}] = \sum_{S\subseteq N\setminus\{i\}} w(i,N,S) \big[\tilde{v}(S\cup\{i\})-\tilde{v}(S)\big], \quad i \in N,
\end{equation*}
where $w(i,N,S)$ ($i\in N$, $S \subseteq N$) 
are constants,
then  it extends by the same formula to a centered game value for non-cooperative games:
\begin{equation}\label{specialgameval}
\bar{h}_i[N,v] = \sum_{S\subseteq N\setminus\{i\}} w(i,N,S) \big[v(S\cup\{i\})-v(S)\big], \quad i \in N.
\end{equation}

\item [$(iv)$] Let $X=(X_1,\dots,X_n)$ be the predictors and $f$ a model. Let $f_0=\E[f(X)]$. Then for $v \in \{\vce,\vpdp\}$
\begin{equation}\label{extprobgame}
\bar{h}(X;f,v)=h(X;f-f_0,v)+f_0 \bar{h}[N,u].
%=h(X;f-f_0,v)+f_0 (h[N,\tilde{u}]+\gamma).
\end{equation}
As a consequence, if $\bar{h}$ is centered, then $\bar{h}(X;f,v)=h(X;f-f_0,v)$, $v \in \{\vce,\vpdp\}$.
\end{itemize}
\end{lemma}

\begin{proof}\rm
Let $c$ be any constant non-cooperative game, that is, for some constant $c_0$, $c(S)=c_0$ for all $S \subseteq N$. Then, by \eqref{extrepr}, we have  $\bar{h}[N,c]=\bar{h}[N,c_0 u]=c_0\bar{h}[N,u]=c_0 (h[N,\tilde{u}]+\gamma)$. This proves $(i)$. The statement $(ii)$ follows from (E1)-(E2) and $(i)$.

Suppose $\bar{h}$ has the form \eqref{specialgameval}. Then for any constant non-cooperative game $c$ and each $i \in N$, we have $c(S\cup\{i\})=c(S)$, $S \subseteq N$, and hence $\bar{h}_i[N,c]=0$.
This proves $(iii)$.

Let $v=\vce$. For any $S \subseteq N$ we have
\[
\vce(S;X,f)=\E[f(X)|X_S]=\E[f(X)-f_0|X_S]+f_0 = \vce(S;X,f-f_0)+f_0 u(S)
\]
and hence using the linearity of $\bar{h}$ and the representation \eqref{extrepr} we obtain
\[
\bar{h}(X;f,\vce)=h[N,\vce(\cdot; X, f-f_0)] +f_0\bar{h}[N,u],
\]
which establishes \eqref{extprobgame} for $v=\vce$. The proof of \eqref{extprobgame} for $v=\vpdp$ is similar.
\end{proof}

See \ref{app::gameaxioms}  and \ref{app::gamevalrepr} in the appendix for more on game values.

\section{Group explainers with coalition structures}\label{sec::Group_expl}

In this section, we construct explainers that quantify predictor contributions to the model output by considering predictor unions. In particular, given predictors $X \in \RR^n$ and a partition $\mathcal{P}=\{S_1,S_2,\dots,S_m\}$ of $N$, our objective is to utilize the partition $\mathcal{P}$ to explain the contribution of each predictor $X_i$ under a coalition structure as well as the contribution of each group $X_{S_j}$. We will refer to such explainers as explainers with a coalition structure. Predictor groups are formed based on dependencies, which allows for the reduction of the predictor dimensionality and constructs explanations which unify the marginal and conditional approaches.

%%%%%%%%%%%%%%%%%%%%%%%%%%%%%%%%%%%%%%%%%%%%%%%%%%%%%%%%%%%%%%%%%%%%%%%%%%%%%%%
% Part 4.1: Extensions
%%%%%%%%%%%%%%%%%%%%%%%%%%%%%%%%%%%%%%%%%%%%%%%%%%%%%%%%%%%%%%%%%%%%%%%%%%%%%%%

\subsection{ Trivial group explainers} \label{sec::TGExplainers}

A game value determines the worth of each individual player while group explainers are capable of quantifying the attribution of any subset of players. In this subsection, we present a simple way of constructing group explainers that work for any partition of players. A more sophisticated approach appears in \S\ref{sec::qgameexpl}.

%%%%%%%%%%%%%%%%%%%%%%%%%%%%%%%%%%%%%%%%%%%%%%%%%%%%%%%%%%%%%%%%%%%%%%%%%%%%%%%%%%%%%%%%%%%%%%%%%%%%%%%%%%%%%%%%%%%%%%%%%%%%%%%%%%%%%%%%%%%%%%%%%%%%%%%%%%%%%%
% Trivial explainer definition
%%%%%%%%%%%%%%%%%%%%%%%%%%%%%%%%%%%%%%%%%%%%%%%%%%%%%%%%%%%%%%%%%%%%%%%%%%%%%%%
%%%%%%%%%%%%%%%%%%%%%%%%%%%%%%%%%%%%%%%%%%%%%%%%%%%%%%%%%%%%%%%%%%%%%%%%%%%%%%%

\begin{definition}\label{def::trivgexp}
Let $X=(X_1,\dots,X_n)$ be the predictors, $f$ a model, and $\mathcal{P}=\{S_1,S_2,\dots,S_m\}$ a partition of predictors. Let $h$ be a linear game value and $\bar{h}$ its extension. A trivial group explainer based on $(\bar{h},\mathcal{P})$ is defined by
\begin{equation}\label{gengroupexplainer}
\bar{h}_{S_j}(X;f,v)=\sum_{i \in S_j} \bar{h}_i(X; f,v), \quad S_j \in \mathcal{P}, \quad v \in \{\vce,\vpdp\}.
\end{equation}
\end{definition}

%%%%%%%%%%%%%%%%%%%%%%%%%%%%%%%%%%%%%%%%%%%%%%%%%%%%%%%%%%%%%%%%%%%%%%%%%%%%%%%%%%%%%%%%%%%%%%%%%%%%%%%%%%%%%%%%%%%%%%%%%%%%%%%%%%%%%%%%%%%%%%%%%%%%%%%%%%%%%%
% Proposition of properties of trivial explainers
%%%%%%%%%%%%%%%%%%%%%%%%%%%%%%%%%%%%%%%%%%%%%%%%%%%%%%%%%%%%%%%%%%%%%%%%%%%%%%%
%%%%%%%%%%%%%%%%%%%%%%%%%%%%%%%%%%%%%%%%%%%%%%%%%%%%%%%%%%%%%%%%%%%%%%%%%%%%%%%

Recall that if predictors  $X$ are independent, then the conditional and marginal game coincide, which implies that $\bar{h}[\vce]=\bar{h}[\vpdp]$. This is in general no longer true if the independence is dropped. \\

\begin{example}\label{1stExample} \rm %(AM): I prefer non-italic in examples.
For example, consider the model:
\begin{equation*}%\label{nonunifexp}
X=(X_1,X_2,X_3), \quad f(X)=X_2X_3,  \quad \mathcal{P}=\{\{1,2\},\{3\}\},
\end{equation*}
with $(X_1,X_2)$ independent of $X_3$ and $\E[X]=0$. Computing the marginal explanations for $h=\varphi$ gives
\[
\bar{\varphi}_{\{1,2\}}(X;f,\vpdp)=\bar{\varphi}_{\{3\}}(X;f,\vpdp)=\frac{1}{2}X_2X_3,
\]
while the conditional ones are given by
\[
\bar{\varphi}_{\{1,2\}}(X;f,\vce)=\frac{1}{2}X_2X_3 - \frac{1}{6} \E[X_2|X_1] X_3, \quad \bar{\varphi}_{\{3\}}(X;f,\vce)=\frac{1}{2}X_2X_3 + \frac{1}{6} \E[X_2|X_1] X_3.
\]
Consequently, the two types of explanations differ whenever $\E[X_2|X_1] \neq 0$.
\end{example}

Suppose next that the partition $\mathcal{P}$ yields independent unions $X_{S_1},X_{S_2},\dots,X_{S_m}$. The question we are interested in is what the relationship between the group explanations $\bar{h}_{S_j}(X; \vce , f)$ and $\bar{h}_{S_j}(X; \vpdp, f)$ is.

\begin{proposition}\label{prop::additivesumshap}
Let $X$, $\mathcal{P}$, $h$, $\bar{h}$ be as in Definition \ref{def::trivgexp}. Let $\bar{h}$ be centered. Suppose that $X_{S_1},X_{S_2},\dots,X_{S_m}$ are independent and $f$ is additive across the elements of $\mathcal{P}$,
\begin{equation}\label{additiveP}
f(X)=\sum_{j \in M} f_j(X_{S_j}), \quad m=|\mathcal{P}|, \quad  M=\{1,2,\dots,m\}.
\end{equation}

\begin{itemize}

  \item [(i)]  The games $\vce(\cdot \,;X,f), \vpdp(\cdot\,;X,f)$ can be expressed as follows:
\begin{equation}\label{decomposition}
\begin{aligned}
  \vce(\cdot;X,f )   & = \sum_{j\in M} \vce_j(\cdot;X,f)-(m-1)\E[f(X)]u \\
  \vpdp(\cdot\,;X,f) & = \sum_{j\in M} \vpdp_j (\cdot;X,f) - (m-1)\E[f(X)]u
\end{aligned}
\end{equation}
where  $u$ is the unit, non-cooperative game and 
\begin{equation}\label{decomposed game}
\vce_j(S;X,f) := \vce(S\cap S_j;X,f), \quad 
\vpdp_j(S;X,f) := \vpdp(S\cap S_j;X,f), \quad j \in M.
\end{equation}

\item [(ii)] If $h$ satisfies \hyperref[axiom:NPP]{(NPP)}, then
\begin{equation}\label{decompgame}
  \bar{h}_{i}[N,\vce] = \bar{h}_{i}[N,\vce_j], \quad \bar{h}_{i}[N,\vpdp] = \bar{h}_{i}[N,\vpdp_j], \quad i\in S_j,
\end{equation}
while for $i \notin S_j$ we have $\bar{h}_{i}[N,\vce_j]=\bar{h}_{i}[N,\vme_j]=0$.

\item [(iii)] If $h$ satisfies \hyperref[axiom:EP]{(EP)} and \hyperref[axiom:NPP]{(NPP)}, then

\begin{equation*}%\label{sumshapunifadditiv}
 \bar{h}_{S_j}(X;f,\vce) = \bar{h}_{S_j}(X;f,\vpdp)=f_j(X_{S_j})-\E[f_j(X_{S_j})]
\end{equation*}
Consequently, $\bar{h}_{S_j}(X;f,\vce)$ and $\bar{h}_{S_j}(X;f,\vme)$ are bounded and continuous on $L^2(P_X)$.
\end{itemize}

\end{proposition}

\begin{proof}
For any $S\subseteq N$ one has
\begin{equation*}
\begin{split}
\vce(S;X,f)
&=\sum_{ j\in M }\E[f_j(X_{S_j})|X_{S \cap S_j}]
=\sum_{ j\in M }\E\Big[f(X)-
\sum_{ k\in M, k\neq j }f_k(X_{S_k})\Big|X_{S \cap S_j}\Big]\\
&=\sum_{ j\in M }\E[ f(X)|X_{S \cap S_j}]
-\sum_{ j\in M }\sum_{ k\in M, k\neq j }\E[f_k(X_{S_k})]\\
&=\sum_{j\in M} \vce(S\cap S_j;X,f)-
(|M|-1)\Big(\sum_{k\in M}\E[f_k(X_{S_k})]\Big)\\
&=\sum_{j\in M} \vce_j(S;X,f)-(m-1)\E[f(X)].
\end{split}
\end{equation*}
Notice that on the first line we used 
$\E[f_j(X_{S_j})|X_S]=\E[f_j(X_{S_j})|X_{S\cap S_j}]$ 
which is due to the following fact: if $W,Y,Z$ are vectors of random variables on the same probability space with $(W,Y)$ independent of $Z$, then 
$\E[W|Y,Z]=\E[W|Y]$.

Similarly, for the marginal game, we have the following which concludes the proof of $(i)$.
\begin{equation*}
\begin{split}
\vpdp(S;X,f) &=\E[f(x_{S},X_{-S})]|_{x_S=X_{S}}
=\sum_{j\in M} \E[f_j(x_{S_j \cap S},X_{S_j \setminus S})]|_{x_{S_j \cap S}=X_{S_j \cap S}}\\   
&=\sum_{j\in M} \bigg(\E[f(x_{S_j \cap S},X_{-(S_j\cap S)})]|_{x_{S_j \cap S}=X_{S_j \cap S}}
-\E\big[\sum_{ k\in M, k\neq j }f_k(X_{S_k})\big]\bigg)\\
&=\sum_{j\in M}\E[f(x_{S_j \cap S},X_{-(S_j\cap S)})]|_{x_{S_j \cap S}=X_{S_j \cap S}}-\sum_{ j\in M }\sum_{ k\in M, k\neq j }\E[f_k(X_{S_k})]\\
&=\sum_{j\in M} \vpdp(S\cap S_j;X,f)-
(|M|-1)\big(\sum_{k\in M}\E[f_k(X_{S_k})]\big)\\
&=\sum_{j\in M} \vpdp_j(S;X,f)-(m-1)\E[f(X)].
\end{split}    
\end{equation*}
Applying the centered extension $\bar{h}$ to either of the games appearing in \eqref{decomposition}, the contribution of the constant game $(m-1)\E [f(X)]u$ would be zero. Thus
\begin{equation*}
\bar{h}_i[N,\vce(\cdot;X,f)]
=\sum_{j\in M} \bar{h}_i[N,\vce_j(\cdot;X,f)]=\sum_{j\in M} h_i[N,\vce_j(\cdot;X,f-f_0)],    
\end{equation*}
where $f_0:=\E[f(X)]$. 
Notice that we have used Lemma \ref{lmm::unitgame}$(iv)$ for the last equality.

Now if $h$  satisfies \hyperref[axiom:NPP]{(NPP)}, the terms 
$h_i[N,\vce_j(\cdot;X,f-f_0)]$ vanish unless $i\in S_j$ because $S_j$ is clearly a carrier for the cooperative game $\big(N,\vce_j(\cdot;X,f-f_0)\big)$. Therefore, if $i\in S_{j_*}$, the last equation simplifies to $\bar{h}_i[N,\vce(\cdot;X,f)]=\bar{h}_i[N,\vce_{j_*}(\cdot;X,f)]$. This establishes $(ii)$ for the conditional game. The proof for the marginal game is similar. 

Finally, suppose $h$ satisfies \hyperref[axiom:EP]{(EP)} as well. Again, we only present the proof for conditional game. 
Invoking $(ii)$ one has
\begin{equation*}
\bar{h}_{S_j}(X;f,\vce)=\sum_{i \in S_j}\bar{h}_i[N,\vce_j(\cdot;X,f)]=
\sum_{i \in S_j} h_i[N,\vce_j(\cdot;X,f-f_0)]. 
\end{equation*}
Since $S_j$ is a carrier for the cooperative game $\big(N,\vce_j(\cdot;X,f-f_0)\big)$, by the efficiency property 
 the last summation is equal to $\vce_j(N;X,f-f_0)=\vce(S_j;X,f-f_0)$.
One can easily check that the last term coincides with $f_j(X_{S_j})-\E[f_j(X_{S_j})]$.
\end{proof}

Proposition \ref{prop::additivesumshap} considers models that are additive across the partition $\mathcal{P}=\{S_1,\dots,S_m\}$. Proposition \ref{prop::additivesumshap}$(i)$ states that $\vce$ and $\vme$ can be expressed (up to a constant) as a sum of games $\vce_j$  and $\vme_j$, $j\in M$, respectively, each having a carrier $S_j$. Consequently, Proposition \ref{prop::additivesumshap}$(ii)$ implies that, if \hyperref[axiom:NPP]{(NPP)} holds, the marginal and conditional attributions of $i\in S_j$  are equal to  attributions $\bar{h}_i$ at the games $\vce_j$  and $\vme_j$, respectively, with $i \notin S_j$ contributing zero to the total payoff of those games. Finally, Proposition \ref{prop::additivesumshap}$(iii)$ implies that the trivial group explanations for marginal and conditional games are both continuous  on $L^2 (P_X)$ if \hyperref[axiom:EP]{(EP)} and \hyperref[axiom:NPP]{(NPP)} hold; hence the Rashomon effect does not impact the group explanations $\bar{h}_{S_j}[N,\vme]$.

\begin{remark} \rm
The method  of constructing \eqref{gengroupexplainer} was explored in \cite{aas} for additive models and proved to be effective in producing Shapley-based 
explanations consistent with both the model and data.
However, if the assumption of model additivity across the partition $\mathcal{P}$, made in Proposition \ref{prop::additivesumshap}, is dropped, the marginal and conditional trivial group explanations will in general no longer be equal as illustrated by Example \ref{1stExample}.
\end{remark}

\begin{corollary}\label{corr_additivecomplex}
 Let $\mathcal{P}$, $f$, $h$, $\bar{h}$, and $u$ be as in Proposition \ref{prop::additivesumshap}. For $j \in M$, define the games on the set of players $S_j$ by
    \[
    \nu_j^{\CE}(S) := \vce( S\, ;X_{S_j},f_j),\quad \nu_j^{\ME}(S) := \vpdp( S \,;X_{S_j},f_j), \quad S\subseteq S_j.
    \]
If $h$ satisfies \hyperref[axiom:NPP]{(NPP)} and 
\hyperref[axiom:CDP]{(CDP)}, then for each $i \in S_j$
\begin{equation}\label{CDPeffect}
\bar{h}_i[N,\vce(\cdot;X,f)]=\bar{h}_{i}[S_j,\nu_j^{\CE}], \quad \bar{h}_i[N,\vpdp(\cdot;X,f)]=\bar{h}_{i}[S_j,\nu_j^{\ME}].
\end{equation}
\end{corollary}

\begin{proof}
We have shown in \eqref{decompgame} that $\bar{h}_i[N,\vce(\cdot;X,f)]$ and $\bar{h}_i[N,\vpdp(\cdot;X,f)]$  coincide respectively with $\bar{h}_{i}[N,\vce_j]$ and $\bar{h}_{i}[N,\vpdp_j]$ provided that $i\in S_j$ and \hyperref[axiom:NPP]{(NPP)} holds.
The games $(N,\vce_j)$ and $(N,\vpdp_j)$, defined in \eqref{decomposed game}, have $S_j$ as a carrier. Here, $\nu_j^{\CE}$ and $\nu_j^{\ME}$ are their restrictions to games with $S_j$ as the set of players. The game values for these restrictions should agree with the original game values if $h$ satisfies \hyperref[axiom:CDP]{(CDP)}.
\end{proof}

If a model is additive across the elements of the partition $\mathcal{P}$ with independent components, Corollary \ref{corr_additivecomplex} states that the marginal and conditional game values for $i \in S_j$  can be expressed as the corresponding game values for the universe of players $S_j$ and the restrictions $\vce|_{S_j}$ and $\vme|_{S_j}$, respectively, when \hyperref[axiom:NPP]{(NPP)} and 
\hyperref[axiom:CDP]{(CDP)} hold. If one drops \hyperref[axiom:CDP]{(CDP)}, then \eqref{CDPeffect} is no longer guaranteed; instead, one can expect only  \eqref{decompgame}. 

Consequently, the corollary implies that the computation of $\bar{h}_i[N,\vce]$, $i \in S_j$, can be done in time $O(2^{|S_j|})$ and hence the complexity of group explanations $\bar{h}_{S_j}[N,\vce]$, and that of single feature explanations in $S_j$, is $O(2^{|S_j|})$ rather than $O(2^{|N|})$, the complexity of a generic linear game value stated in Lemma \ref{lmm::lingamevalrepr}. Similar conclusion holds for marginal explanations.

%%%%%%%%%%%%%%%%%%%%%%%%%%%%%%%%%%%%%%%%%%%%%%%%%%%%%%%%%%%%%%%%%%%%%%%%%%%%%%%
%%%%%%%%%%%%%%%%%%%%%%%%%%%%%%%%%%%%%%%%%%%%%%%%%%%%%%%%%%%%%%%%%%%%%%%%%%%%%%%
% Quotient game explainers
%%%%%%%%%%%%%%%%%%%%%%%%%%%%%%%%%%%%%%%%%%%%%%%%%%%%%%%%%%%%%%%%%%%%%%%%%%%%%%%
%%%%%%%%%%%%%%%%%%%%%%%%%%%%%%%%%%%%%%%%%%%%%%%%%%%%%%%%%%%%%%%%%%%%%%%%%%%%%%%

% \subsection{Quotient game explainers}\label{sec::qgameexpl}

\subsection{Group explainers based on quotient games}\label{sec::qgameexpl}

The trivial group explainers obtained by \eqref{gengroupexplainer} are based on 
single-feature explanations and do not utilize the structure imposed by the partition $\mathcal{P}$. Constructing explainers that explicitly incorporate the coalition structure of $\mathcal{P}$ might be advantageous when the partition is based on dependencies. In the case when predictors within each union $X_{S_j}$ share significant amount of mutual information, the change in value of one of the predictors causes 
a certain 
change in value of other predictors in the union, and thus, 
the predictors within the union ``act in agreement'' with one another.

%%%%%%%%%%%%%%%%%%%%%%%%%%%%%%%%%%%%%%%%%%%%%%%%%%%%%%%%%%%%%%%%%%%%%%%%%%%%%%%
%%%%%%%%%%%%%%%%%%%%%%%%%%%%%%%%%%%%%%%%%%%%%%%%%%%%%%%%%%%%%%%%%%%%%%%%%%%%%%%
% Quotient game definition
%%%%%%%%%%%%%%%%%%%%%%%%%%%%%%%%%%%%%%%%%%%%%%%%%%%%%%%%%%%%%%%%%%%%%%%%%%%%%%%
%%%%%%%%%%%%%%%%%%%%%%%%%%%%%%%%%%%%%%%%%%%%%%%%%%%%%%%%%%%%%%%%%%%%%%%%%%%%%%%

To design explainers of unions with the partition in mind, we make use of quotient games. 
This is the content of this subsection.
\begin{definition}\label{def::quotientgame}
Given a cooperative game $(N,v)$ with $N=\{1,2,\dots,n\}$ 
and a partition $\mathcal{P}=\{S_1,S_2,\dots,S_m\}$ of $N$, the {\it quotient game} $(M,v^{\mathcal{P}})$, where $M=\{1,2,\dots,m\}$,
is defined by
\begin{equation*}
v^{\mathcal{P}}(A):=v\big( \cup_{j \in A} S_j\big), \quad \quad A \subseteq M.
\end{equation*}
For non-cooperative games, we adapt the same definition; note that one always has $v^{\mathcal{P}}(\varnothing)=v(\varnothing)$.  

\end{definition}

By design, the quotient game is played by the unions; that is, the game $v^{\mathcal{P}}$ is obtained by restricting $v$ to unions $S_j\in \mathcal{P}$ by viewing the elements of the partition $\mathcal{P}$ as players. The complexity of the quotient game value $h[M,v^{\mathcal{P}}]$ is of the order $2^{|\cP|} \cdot O(v)$, where $O(v)$ stands for the complexity of the game evaluation for any $S \subseteq N$; this fact follows directly from the representation formula \eqref{lingamevalrepr} of Lemma \ref{lmm::lingamevalrepr}. This motivates us to define explanations of predictor unions using quotients.

%%%%%%%%%%%%%%%%%%%%%%%%%%%%%%%%%%%%%%%%%%%%%%%%%%%%%%%%%%%%%%%%%%%%%%%%%%%%%%%
%%%%%%%%%%%%%%%%%%%%%%%%%%%%%%%%%%%%%%%%%%%%%%%%%%%%%%%%%%%%%%%%%%%%%%%%%%%%%%%%%%%%%%%% Quotient game explainer definition
%%%%%%%%%%%%%%%%%%%%%%%%%%%%%%%%%%%%%%%%%%%%%%%%%%%%%%%%%%%%%%%%%%%%%%%%%%%%%%%
%%%%%%%%%%%%%%%%%%%%%%%%%%%%%%%%%%%%%%%%%%%%%%%%%%%%%%%%%%%%%%%%%%%%%%%%%%%%%%%

\begin{definition}\label{def::qgameexplainer}
Let $X$, $f$, $\mathcal{P}$, $h$, $\bar{h}$ be as in Definition \ref{def::trivgexp}. The quotient game explainer based on $(\bar{h},\mathcal{P})$ is defined by
\begin{equation*}%\label{qgameexplainer}
\bar{h}_{S_j}^{\mathcal{P}}(X;f,v):=\bar{h}_j[M,v^{\mathcal{P}}(\cdot\,; X,f)], \quad S_j \in \mathcal{P}, \quad v \in \{\vce,\vpdp\}.
\end{equation*}
\end{definition}

For the quotient explainers we have the following result.

\begin{lemma}\label{lmm::boundquot}
Let $X,f$, $\mathcal{P}$, $h$, $\bar{h}$ be as in Definition \ref{def::trivgexp}  and suppose that $\bar{h}$ is centered. Then 
\begin{itemize}
\item [(i)] For $v\in \{\vce,\vpdp\}$ we have
\[
    \bar{h}_{S_j}^{\mathcal{P}}(X;f,v) = h_{S_j}^{\mathcal{P}}(X;f-f_0,v), \quad f_0 := \E[f(X)].
\]

\item [(ii)] The quotient marginal and conditional explanations satisfying the following bounds:
\begin{equation*}
\| \bar{h}_j[\vceP(\cdot;X,f)] \|_{L^2(\PP)} \leq \bar{C} \|f\|_{L^2(P_X)}, \quad \| \bar{h}_j[\vpdpP(\cdot;X,f)] \|_{L^2(\PP)} \leq 2^{|\mathcal{P}|} \bar{C} \|f\|_{L^2(\tilde{P}_{X,\cP})}
\end{equation*}
where $\bar{C}=\bar{C}(\bar{h},|\cP|,j)$ and $\tilde{P}_{X,\cP}:=\frac{1}{2^m}\sum_{A \subseteq M, Q_A=\cup_{j\in A}S_j} P_{X_{Q_A}} \otimes P_{X_{-Q_A}} \ll \tilde{P}_{X}$.

Consequently, if the Radon-Nikodym derivative  $r_{Q_A}=\tfrac{d P_{X_{Q_A}} \otimes P_{X_{-Q_A}}}{d P_X}$, $Q_A=\cup_{j \in A} S_j$, exists and belongs to $L^\infty(P_X)$ for each $A \subseteq M$, the map $f \in L^2(P_X) \mapsto \bar{h}_{S_j}^{\mathcal{P}}(X;f,\vme)$ is well-defined and bounded for each $j \in M$.

\item [(iii)] Let $h$ have the form \eqref{lingameform} and satisfy the properties outlined in Corollary \ref{corr::cond_operator_cont}$(ii)$, including \hyperref[axiom:EP]{(EP)}. Then
\begin{equation}\label{eff_quotient_bound}
\sum_{j=1}^m \| \bar{h}_{S_j}^{\mathcal{P}}(X;\vce,f)\|^2_{L^2(\PP)}  \leq \|f-f_0\|^2_{L^2(P_X)} \leq \|f\|^2_{L^2(P_X)}, \quad f_0 := \E[f(X)]. 
\end{equation}
\end{itemize}
\end{lemma}

\begin{proof}

Note that for any $A \subseteq M$
\[
\vceP(A;X,f-f_0)=\E[f(X)-f_0|X_{\cup_{j \in A}S_j}]=\vceP(A;X,f)-f_0 u(A)
\]
(where $u$ is the unit non-cooperative game) and hence 
\[
    \bar{h}_j[\vceP(\cdot;X,f)]=h_j[M,\vceP(\cdot;X,f-f_0)],
\]
which proves $(i)$ for $v=\vce$. The proof of $(i)$ for $v=\vpdp$ is similar.

Next, let $\gamma=\{\gamma_j\}_{j=1}^m$ be the constants that come up in extending $h$ to non-cooperative games on $M$ as in Lemma \ref{lmm:extlingame}. Thus, by Lemma \ref{lmm::lingamevalrepr}, the extension $\bar{h}$ satisfies the following growth condition for any $(M,w)$:
\begin{equation*}%\label{growthgameval}
|\bar{h}_j[M,w]| \leq |\gamma_j| \cdot w(\varnothing) + C \big(\sum_{A \subseteq M, A\ne \varnothing} |w(S)| \big), \quad j\in M,
\end{equation*}
where $C=C(h,M,j)$ is a constant that depends on $h$, $M$, and $j$.

Then, setting  $w=v^{\CE,\mathcal{P}}$, for any $f \in L^2(P_X)$ we have
\[
\|\bar{h}_{S_j}^{\mathcal{P}}(X;f,\vce)\|_{L^2(\PP)}\leq |\gamma_j|\cdot |f_0|  +  C \sum_{\substack{A \subseteq M\\ A\ne \varnothing}} \|\E[f(X)|X_{\cup_{j \in A}S_j}]\|_{L^2(\PP)} \leq  \max(|\gamma_j|,C) \|f\|_{L^2(P_X)}.
\]

Similarly, setting  $w=v^{\ME,\mathcal{P}}$ and $Q_A=\cup_{j \in A}S_j$, $A \subseteq M$, for any $f \in L^2(\intP_X)$ we have
\[
\|\bar{h}_{S_j}^{\mathcal{P}}(X;f,\vme)\|_{L^2(\PP)}\leq |\gamma_j|\cdot |f_0| + C \sum_{\substack{A \subseteq M\\ A\ne \varnothing}} \|f\|_{L^2(P_{X_{Q_A}} \otimes P_{X_{-Q_A}})} \leq  2^{|\mathcal{P}|}\max(|\gamma_j|,C)  \|f\|_{L^2(\intP_{X,\mathcal{P}})}.
\]
Then, setting $\bar{C}=\max(\gamma_j,C)$ and using $(ii)$ together with the two  inequalities above proves $(ii)$.
The proof of $(iii)$ follows the steps in the proof of Corollary \ref{corr::cond_operator_cont}$(ii)$.
\end{proof}

When the elements of the partition $\mathcal{P}$ are independent, we have $\tilde{P}_{X,\cP}=P_X$, and hence the following stability result.

\begin{proposition}[\bf approximation]\label{prop::quotgameexpl} 
Let $X,f$, $\mathcal{P}$, $h$, $\bar{h}$ be as in Definition \ref{def::trivgexp}  and suppose that $\bar{h}$ is centered. 
\begin{itemize}

\item [(i)] 

Suppose $r_{Q_A}=\tfrac{d P_{X_{Q_A}} \otimes P_{X_{-Q_A}}}{d P_X}$, $Q_A=\cup_{j \in A} S_j$, exists and belongs to 
$L^\infty(P_X)$ for each $A \subseteq M$. Then $\big(L^2(\intP_{X,\cP}),\|\cdot\|_{L^2(P_X)}\big) = L^2(P_X)$ and for $f \in L^2(P_X)$
\begin{equation*}%\label{qgameunif}
\bar{h}_{S_j}^{\mathcal{P}}(X;f,\vce)=\bar{h}_{S_j}^{\mathcal{P}}(X;f,\vme) + \mathcal{I}(f,\{r_{Q_A}\}_{A \subseteq M}) \quad in \quad L^2(\P),
\end{equation*}
with the error term $\mathcal{I}$ satisfying the bound
\[
 \|\mathcal{I}(f,\{r_{Q_A}\}_{A \subseteq M})\|_{L^2(\P)} \leq C(h,\cP) \cdot \Big( \max_{A \subseteq M} \|r_{Q_A}-1\|_{L^{\infty}(P_X)} \Big) \cdot \| f \|_{L^2(P_X)}.
\]

\item [(ii)] If $X_{S_1},X_{S_2},\dots,X_{S_m}$ are independent, the marginal and conditional games coincide, i.e. $\vceP=\vpdpP$, and hence
\begin{equation*}%\label{qgameunif}
\bar{h}_{S_j}^{\mathcal{P}}(X;f,\vpdp)=\bar{h}_{S_j}^{\mathcal{P}}(X;f,\vce), \quad S_j \in \mathcal{P}.
\end{equation*}

\item [(iii)]
If $X_{S_1},X_{S_2},\dots,X_{S_m}$ are  independent, $h$ satisfies \hyperref[axiom:EP]{(EP)} and \hyperref[axiom:NPP]{(NPP)} properties and $f$ 
is additive across $\mathcal{P}$ as in \eqref{additiveP}, then
\begin{equation}\label{qgamevssums}
 \bar{h}_{S_j}^{\mathcal{P}}(X;f,v)=\bar{h}_{S_j}(X;f,v),\quad S_j \in \mathcal{P}, \quad v \in \{\vce,\vpdp\}.
\end{equation}

\end{itemize}
\end{proposition}

\begin{proof}

Let $\intP_{X,\cP}$ be as in Lemma \ref{lmm::boundquot}$(i)$. Suppose the Radon-Nykodym derivative $r_{Q_A}$ exists and belongs to $L^{\infty}(P_X)$ for each $A \subseteq M$. Then $(L^2(\intP_{X,\cP}),\|\cdot\|_{L^2(P_X)}) = L^2(P_X)$. The remaining part of the statement $(i)$ follows directly from Lemma \ref{lmm::_marg_value_bound_px}, Lemma \ref{lmm::lingamevalrepr}, and Lemma \ref{app::lmm::cond_marg_value_bound}$(i)$.

Let $A \subseteq M$. Let $S_A:=\cup_{j \in A}S_j$. Then, by the independence of $X_{S_A}$ and $X_{-S_A}$ we have
% \E[f(X)|X_{S_A}]
\begin{equation*} \label{pdpcequot}
\vceP(A;X,f)=\E[f(X_{S_A},X_{-S_A})|X_{S_A}]=\E[f(x_{S_A},X_{-S_A})]|_{x_{S_A}=X_{S_A}}=\vpdpP(A;X,f).
\end{equation*}
Since $A\subseteq M$ was arbitrary, the two quotient games coincide. This together with Lemma \ref{lmm::boundquot} implies $(ii)$.

Suppose now that $f$ is additive across the elements of $\mathcal{P}$ and that $h$ is linear and satisfies the efficiency property. Let $f_0=\E[f(X)]$. First, let us assume that $f_0=0$. Then $\vce$ and $\vpdp$ are cooperative games and hence the quotient games are as well. 
Note that for any $A \subseteq M$ we have
\[
\vpdpP(A;X,f)=\sum_{j \in A}f_j(X_{S_j})+\sum_{k \notin A}\E[f_k(X_{S_k})]=\sum_{j \in A}(f_j(X_{S_j})-\E[f_j(X_{S_j})]).
\]
In particular, when $A$ is a singleton $\{j\}$ one has $\vpdpP(\{j\};X,f)=f_j(X_{S_j})-\E[f_j(X_{S_j})]$.
We deduce that in general 
$\vpdpP(A;X,f)=\sum_{j\in A}\vpdpP(\{j\};X,f)$.

This proves that $\vpdpP$ is a non-essential cooperative game. Then, $(ii)$ together with Lemma \ref{lmm::nonessential} implies
\[
h_j[M,\vceP(\cdot;X,f)]=h_j[M,\vpdpP(\cdot;X,f)]=\vpdpP(\{j\};X,f)=f_j(X_{S_j})-\E[f_j(X_{S_j})].
\]
Then, the above relationship and Proposition \ref{prop::additivesumshap}$(ii)$ give $(iii)$ for the case $f_0=0$.

Now for a general $f$, applying the result just obtained to $f-f_0$ yields:
$$
 h_{S_j}^{\mathcal{P}}(X;f-f_0,v)=\bar{h}_{S_j}(X;f-f_0,v),\quad S_j \in \mathcal{P}, \quad v \in \{\vce,\vpdp\}.
$$
The left-hand side is equal to $\bar{h}_{S_j}^{\mathcal{P}}(X;f,v)$ by Lemma \ref{lmm::boundquot}$(i)$ while the right-hand side is the same as $\bar{h}_{S_j}(X;f,v)$ because translating $f$ by a constant does not alter $\bar{h}_{S_j}(X;f,v)$
according to Proposition \ref{prop::additivesumshap}$(iii)$. 
Consequently, we obtain 
\eqref{qgamevssums}.
\end{proof}

Proposition \ref{prop::quotgameexpl}$(i)$ states that the quotient game explanations always coincide for games $\vce$ and $\vpdp$ whenever the unions are independent, and hence always leads to explanations that are continuous in $L^2(P_X)$ and thus $P_X$-consistent in the sense of Definition \ref{def::consistency}. Proposition \ref{prop::quotgameexpl}$(ii)$ states that when a model is additive across elements of the partition $\mathcal{P}$ then the quotient and trivial group explainers actually coincide, while, in general, such equality is not guaranteed as illustrated in Example \ref{1stExample}. Furthermore, when each union is treated as a player, $2^{|\mathcal{P}|}\cdot O(v)$ becomes an upper bound for the complexity of computing the quotient game explainer obtained from $v\in\{\vce,\vme\}$. In particular, when $|\mathcal{P}|=O(\log(n^{1-\delta}))$, the complexity becomes linear.

\noindent{\bf Implications on feature importance.} Proposition \ref{prop::quotgameexpl}$(i)$ has direct implications on the processes that make use of quantifying local feature attributions (or global ones, as described in Section \S \ref{subsec::single_feat_stability}) to make data-informed decisions. Suppose the response variable satisfies $Y=f_*(X)$ for some reference model $f_*\in L^2(P_X)$ and $h$ is an efficient game value in the form \eqref{lingameform}. Then, for all models in the $(L^2,\epsilon)$-Rashomon set about $f_*$,  their marginal group explanations $h_{S_j}^{\mathcal{P}}[N,\vme]$ will differ by $\epsilon$ for global explanations (and in an $L^2$-sense for local ones) from those of $f_*$, no matter the functional representation of the models, even if some predictors are dropped from consideration. This implication addresses the issues discussed in \cite{Kumar2020,Janzing2020} and \cite{Sundararajan2019}, and is an alternative solution to the global feature importance method discussed in \cite{Fisher2019} when the predictors in each group $S_j$ are strongly dependent.

In a real-life setting, for example, financial institutions are required by the Equal Credit Opportunity Act \cite{ECOA} to inform customers on which factors impacted an adverse credit decision. Using group feature attributions based on the marginal quotient game leads to explanations that are true-to-the-data. Consequently, if a customer applies at different times when distinct models are used to assess credit risk, the explanations generated from those models will be similar, which guarantees explanation consistency throughout time. This also means that those explanations provide high fidelity information to the customer on what actions to take for obtaining credit in the future.

\vspace{5pt}

Using the quotient game approach, there are certain considerations one must take into account:
\begin{enumerate}[label=(\alph*)]
\item if the partition is changed the game values have to be recomputed; 
\item knowing the quotient game values does not help with computation of single feature explanations, which are expensive computationally;
\item even if single feature explanations are known, the trivial and quotient game explanations in general are not equal (see
Example \ref{1stExample}); this case causes loss of continuity of marginal, trivial group explanations with respect to models in $L^2(P_X)$ when dependencies are present.
\end{enumerate}

The aforementioned difficulties can be overcome, when explainers are constructed with the help of coalition values 
that utilize the partition structure $\mathcal{P}$ for computation of single players. We discuss such explainers in the next section.

\subsection{Explainers based on games with coalition structure}\label{sec::ExCoal}

%%%%%%%%%%%%%%%%%%%%%%%%%%%%%%%%%%%%%%%%%%%%%%%%%%%%%%%%%%%%%%%%%%%%%%%%%%%%%%%
%%%%%%%%%%%%%%%%%%%%%%%%%%%%%%%%%%%%%%%%%%%%%%%%%%%%%%%%%%%%%%%%%%%%%%%%%%%%%%%
% Intro to coaltional values
%%%%%%%%%%%%%%%%%%%%%%%%%%%%%%%%%%%%%%%%%%%%%%%%%%%%%%%%%%%%%%%%%%%%%%%%%%%%%%%
%%%%%%%%%%%%%%%%%%%%%%%%%%%%%%%%%%%%%%%%%%%%%%%%%%%%%%%%%%%%%%%%%%%%%%%%%%%%%%%

A more advanced way to design explainers with the partition in mind is to employ cooperative game theory with coalition structure, in which the objective is to compute the payoffs of players in a game where players form unions acting in agreement within the union. 

The games with coalitions were introduced by \cite{Aumann1974} and later many more researchers contributed to the development of this subject. Some of the notable works are \cite{Owen}, \cite{Owen1982}, \cite{Tijs1981}, \cite{Dubey1981}, \cite{Amer1995}, \cite{Alonso2002}, \cite{Albizuri2004}, \cite{Casas2003}, \cite{Vidal-Puga2012}. See also the work by \citet{Lorenzo-Freire} containing a detailed exposition on  games with coalitions.

\begin{definition}
Let $N \subset \mathbb{N}$ and $\mathcal{P}=\{S_1,S_2,\dots,S_m\}$ be a partition of $N$. A {\it coalitional value} $g$ is a map that assigns to every game with a coalition structure $(N,v,\mathcal{P})$ a vector
\[
g[N,v,\mathcal{P}]=\{g_i[N,v,\mathcal{P}]\}_{i \in N}
\]
where $g_i[N,v,\mathcal{P}]$ denotes the payoff for the player $i \in N$.
\end{definition}
Note that any game value $h[N,v]$ could be viewed as a coalitional value that has no explicit dependence on the partition $\mathcal{P}$. Furthermore, the map $(N,v) \mapsto g[N,v,\bar{N}]$, where  $\bar{N}=\{\{i\}: i \in N\}$ denotes the partition containing singletons, induces a game value. 
Properties of game values such as linearity \hyperref[axiom:LP]{(LP)}, efficiency \hyperref[axiom:EP]{(EP)} etc. (see Appendix \ref{app::gameaxioms}) extend to coalitional game values in an obvious way.

Some notable (non-trivial) coalitional values are the Owen value and the Banzhaf-Owen value respectively defined by
\begin{equation}\label{BzOw}
\begin{aligned}
Ow_i[N,v,\mathcal{P}] &= \sum_{R \subseteq M \setminus \{j\} } \sum_{T \subseteq S_j \setminus \{i\} } 
\frac{r!(m-r-1)!}{m!} \frac{t!(s_j-t-1)!}{s_j!}\big[ v(Q\cup T \cup \{i\})-v(Q\cup T) \big]\\
BzOw_i[N,v,\mathcal{P}] &= \sum_{R \subseteq M \setminus \{j\} } \sum_{T \subseteq S_j \setminus \{i\} } 
\frac{1}{2^{m-1}} \frac{1}{2^{s_j-1}} \big[ v(Q\cup T \cup \{i\})-v(Q\cup T) \big]\\
\end{aligned}
\end{equation}
where $i\in S_j$, $t=|T|$, $s_j=|S_j|$, $r=|R|$ and  $Q=\cup_{r\in R}S_r$. The difference between the two values is that the Owen value satisfies the efficiency property, while the Banzhaf-Owen value satisfies the total power property. In addition, the Owen value for partitions consisting of singletons is  the Shapley value \eqref{shapform}, while for such partitions the Banzhaf-Owen value  is the Banzhaf value \eqref{banzhafform}. These properties can be verified directly.

%%%%%%%%%%%%%%%%%%%%%%%%%%%%%%%%%%%%%%%%%%%%%%%%%%%%%%%%%%%%%%%%%%%%%%%%%%%%%%%
%%%%%%%%%%%%%%%%%%%%%%%%%%%%%%%%%%%%%%%%%%%%%%%%%%%%%%%%%%%%%%%%%%%%%%%%%%%%%%%
% Extension of coalitional values
%%%%%%%%%%%%%%%%%%%%%%%%%%%%%%%%%%%%%%%%%%%%%%%%%%%%%%%%%%%%%%%%%%%%%%%%%%%%%%%
%%%%%%%%%%%%%%%%%%%%%%%%%%%%%%%%%%%%%%%%%%%%%%%%%%%%%%%%%%%%%%%%%%%%%%%%%%%%%%%

To extend linear coalitional values to games that fail to satisfy $v(\varnothing)=0$, one can carry out the same program as in  \S\ref{sec::prelimgames}. Given a linear coalitional value $g$, we seek an extension $\bar{g}$ to $V$ that satisfies:

\begin{enumerate}\label{extcoalgamevalhyp}
    \item [(E1$'$)] $\bar{g}[N,\tilde{v},\mathcal{P}]=g[N,\tilde{v},\mathcal{P}]$ for $(N,\tilde{v}) \in V_0$.
    \item [(E2$'$)] $\bar{g}$ is linear on $V$.
\end{enumerate}

\begin{lemma}[\bf extension]\label{lmm:extlincoalgame}
Let $g$ be a linear coalitional value. An extension $\bar{g}$ satisfying (E1$'$)-(E2$'$) has the representation:
\begin{equation}\label{extreprcoalval}
\bar{g}_i[N,v,\mathcal{P}]=g_i[N,\tilde{v},\mathcal{P}] + \gamma_i v(\varnothing), \quad i \in N,
\end{equation}
where $\gamma=\{\gamma_i\}_{i=1}^n$ are constants that depend on $N$ and $\cP$. Furthermore, any coalitional value in the form \eqref{extreprcoalval} satisfies (E1$'$)-(E2$'$). In addition, if $g$ is symmetric, then $\bar{g}$ is symmetric if and only if $\gamma_i=\gamma_j$, $i,j \in N$.
\end{lemma}

\begin{proof}
The proof follows the same steps as those in the proof of Lemma \ref{lmm:extlingame}.
\end{proof}

Abusing the notation as before, for each $i \in N$ we write
\begin{equation}\label{abuse2}
\bar{g}_i(X; v,\mathcal{P},f):=\bar{g}_i[N,v(\cdot\,;X,f),\mathcal{P}], \quad v \in \{\vce,\vpdp\}.    
\end{equation}

Group explainers based on coalitional values can similarly be defined either via sums or quotient games.
\begin{definition}\label{def::coalexpl}
Let $X$, $\mathcal{P}$ be as in Definition \ref{def::trivgexp}. Let $g$ be a linear coalitional value and $\bar{g}$ its extension. The trivial and quotient game explainers based on $(\bar{g},\mathcal{P})$ and $v \in \{\vce,\vpdp\}$ are defined by
\begin{equation*}%\label{coalexplainers}
\bar{g}_{S_j}(X;v,\mathcal{P},f)=\sum_{i \in S_j} \bar{g}_i[N,v(\cdot \,; X,f),\mathcal{P}], \quad \bar{g}_{S_j}^{\mathcal{P}}(X;v,f)=\bar{g}_j[M,v^{\mathcal{P}}(\cdot\,; X,f), \bar{M}], \quad S_j \in \mathcal{P}.
\end{equation*}
\end{definition}

\begin{definition}
Let $g$ be a linear coalitional value and $\bar{g}$ its extension. We say that $\bar{g}$ is centered if $\bar{g}[N,c,\mathcal{P}]=0$ for any constant non-cooperative game $(N,c)\in V$ and any partition $\cP$.
\end{definition}

\begin{lemma}\label{lmm::unitgamecoalval}
Let $g$ be a linear coalitional value and $\bar{g}$ its extension with $\gamma$ as in \eqref{extreprcoalval}. Let $u$ denote a unit, non-cooperative game, that is, $u(S)=1$ for all $S \subseteq N$ and any $N$. Then
\begin{itemize}
\item [$(i)$] $\bar{g}$ is centered if and only if $\gamma=-g[N,\tilde{u},\mathcal{P}]$.

\item [$(ii)$] If $\bar{g}$ is centered then $\bar{g}[N,v,\mathcal{P}]=g[N,(v-v(\varnothing)u),\cP]$

\item [$(iii)$] If $\bar{g}$ has the marginalist form
\begin{equation}\label{specialcoalval}
\bar{g}_i[N,v,\mathcal{P}] = \sum_{S\subseteq N\setminus\{i\}} w(N,\mathcal{P},S) \big(v(S\cup\{i\})-v(S)\big), \quad i \in N,
\end{equation}
where $w(N,\mathcal{P},S)$ are constants, then it is centered.

\item [$(iv)$] Let $X \in \RR^n$ be predictors and $f$ a model. Let $f_0=\E[f(X)]$. Then for $v \in \{\vce,\vpdp\}$
\begin{equation*}%\label{extprobcoalgame}
\bar{g}(X;f,v)=g(X;f-f_0,v)+f_0 \bar{g}[N,u,\mathcal{P}]=
g(X;f-f_0,v)+f_0 (g[N,\tilde{u},\mathcal{P}]+\gamma).
\end{equation*}
As a consequence, if $\bar{g}$ is centered, then  for each $i \in N$ and $S_j \in \mathcal{P}$
\[
\bar{g}(X;f,v,\mathcal{P})=g(X;f-f_0,v,\mathcal{P}), \quad \bar{g}^{\mathcal{P}}_{S_j}(X;f,v)=g^{\mathcal{P}}_{S_j}(X;f-f_0,v), \quad v \in \{\vce,\vpdp\}.
\]
\end{itemize}
\end{lemma}
\begin{proof}
The proof follows the same steps as those in the proof of Lemma \ref{lmm::unitgame}.
\end{proof}

%%%%%%%%%%%%%%%%%%%%%%%%%%%%%%%%%%%%%%%%%%%%%%%%%%%%%%%%%%%%%%%%%%%%%%%%%%%%%%%
%%%%%%%%%%%%%%%%%%%%%%%%%%%%%%%%%%%%%%%%%%%%%%%%%%%%%%%%%%%%%%%%%%%%%%%%%%%%%%%
% Continuity for coalitional value explainers
%%%%%%%%%%%%%%%%%%%%%%%%%%%%%%%%%%%%%%%%%%%%%%%%%%%%%%%%%%%%%%%%%%%%%%%%%%%%%%%
%%%%%%%%%%%%%%%%%%%%%%%%%%%%%%%%%%%%%%%%%%%%%%%%%%%%%%%%%%%%%%%%%%%%%%%%%%%%%%%

The results of  Theorem \ref{prop::condoperator}, Theorem \ref{prop::margoperator}, and Theorem \ref{thm::margoperatorunbound} can be extended to coalitional values.

\begin{proposition}[\bf properties]\label{prop::coalqgameprop}  
Let $X$, $\mathcal{P}$ be as in Definition \ref{def::trivgexp}. Let $g$ be a linear coalitional value and $\bar{g}$ its extension. 

\begin{itemize}
\item [$(i)$] The linear map $f \mapsto \bar{\oper}^{\CE}[f;\bar{g},X]:=\bar{g}(X;\vce,\cP,f)$ is bounded and hence continuous on $L^2(P_X)$. 

\item [$(ii)$] The linear map $f \mapsto \bar{\oper}^{\ME}[f;\bar{g},X]:=\bar{g}(X;\vpdp,\cP,f)$ is bounded and hence continuous on $L^2(\intP_X)$.

\item [$(iii)$] 
Suppose $g$ is of the form \eqref{lincoalgamevalrepr}, and for any $i \in N$ and subset $S$ such that $\gamma(i, N,\mathcal{P},S)\neq 0$,
the Radon-Nikodym derivative  $r_S=\tfrac{d P_{X_{S}} \otimes P_{X_{-S}}}{d P_X}$ exists and belongs to $L^\infty(P_X)$. Then $f \mapsto \bar{\oper}^{\ME}[f;\bar{g},X]=\bar{g}(X;\vpdp,\cP,f)$ defines a bounded linear operator on 
$L^2(P_X)$ as well.

\item[$(iv)$] Suppose $(\bar{\oper}^{\ME}[\cdot;\bar{g},X],H_X)$ is well-defined. Let $g$ have the form \eqref{specialcoalval} with $w(N,\mathcal{P},S)>0$. Suppose there exists distinct $i,j \in \{1,2,\dots,n\}$  for which \eqref{blowupcond} holds. 
Then $(\bar{\oper}^{\ME}[\cdot;\bar{g},X],H_X)$ is unbounded.

\item [$(v)$] If $\bar{g}$ is centered, then the conclusions of Proposition \ref{prop::quotgameexpl}$(i)$, under group independence, hold for the quotient game explanations.  In particular, the linear map $f \mapsto \bar{g}_{S_j}^{\mathcal{P}}(X;f,v)$, $\{\vce,\vpdp\}$,  is bounded and hence continuous on $L^2(P_X)$.

\end{itemize}
\end{proposition}

\begin{proof}
The properties $(i)$ and $(ii)$ follow from Lemma \ref{lmm::lingamevalrepr} and the fact that for any $S \subseteq N=\{1,2,\dots,n\}$
\[
\|\vce(S;X,f)\|_{L^2(\PP)} \leq \|f\|_{L^2(P_X)}, \quad \|\vpdp(S;X,f)\|_{L^2(\PP)} \leq 2^n\|f\|_{L^2(\intP_X)}.
\]
The property $(iii)$ follows from Lemma \ref{lmm::lingamevalrepr}, Lemma \ref{lmm::_marg_value_bound_px}, and the definition of the Radon-Nikodym derivative.
The property $(iv)$ follows from Proposition \ref{prop::coalvalblowup} (take $k$ to be $i$) and the fact that $w(N,\mathcal{P},S)>0$, while $(v)$ can be obtained following the steps in the proof of Proposition \ref{prop::quotgameexpl}$(i)$.
\end{proof}

\begin{remark}\rm
Proposition \ref{prop::coalqgameprop}$(iii)$ can be generalized to linear coalitional game values that are not necessarily 
in the  form of \eqref{specialcoalval}; see Proposition \ref{prop::coalvalblowup}.
\end{remark}

%%%%%%%%%%%%%%%%%%%%%%%%%%%%%%%%%%%%%%%%%%%%%%%%%%%%%%%%%%%%%%%%%%%%%%%%%%%%%%%
%%%%%%%%%%%%%%%%%%%%%%%%%%%%%%%%%%%%%%%%%%%%%%%%%%%%%%%%%%%%%%%%%%%%%%%%%%%%%%%
% Two-step formulation
%%%%%%%%%%%%%%%%%%%%%%%%%%%%%%%%%%%%%%%%%%%%%%%%%%%%%%%%%%%%%%%%%%%%%%%%%%%%%%%
%%%%%%%%%%%%%%%%%%%%%%%%%%%%%%%%%%%%%%%%%%%%%%%%%%%%%%%%%%%%%%%%%%%%%%%%%%%%%%%

\subsection{Coalitional explainers with two-step formulation}\label{subsubsec::twostepprop}%\noindent{\bf Two-step property.}

Having quotient game explanations does not provide one with single feature explanations, which sometimes are desirable. Unlike game values, coalitional values may allow for a more efficient way of computing single feature explanations. This is the case, for example, when a coalitional value can be obtained using a two-step procedure: first by playing a quotient-like game and then a game inside the union. This consequently affects the structure of the coalitional value, which in turn improves the exposition on stability (\S 4.3.3) and allows for extending coalitional values to recursive ones for generic partition trees (\S \ref{sec::parttree1}).

\begin{definition}\label{def::2stepprop}
Let $g[N,v,\mathcal{P}]$ be a linear coalitional value. We say that $g$ satisfies a two-step formulation if for any  $N \subset \mathbb{N}$ and its partition $\cP=\{S_1,\dots,S_m\}$ there exists a linear symmetric game value $h^{(1)}$, a linear game value $h^{(2)}$, and games $v^{(1)},v^{(2)},\dots, v^{(m)}$ played respectively on $S_1$, $S_2,\dots,S_m$ 
(all dependent on $(N,v,\cP)$), such that
\begin{equation*}
g_i[N,v,\cP] = h_{i}^{(2)}[S_j,v^{(j)}], \,\, i \in S_j \quad \text{where} \quad  v^{(j)}(T)=h_j^{(1)}[M, \hat{v}_T ], \,\, T \subseteq S_j,\,\, M=\{1,2,\dots,m\},
\end{equation*}
and $\{\hat{v}_T=\hat{v}_T(N,v,\cP)\}_{j\in M,T\subseteq S_j}$ is a family of intermediate games on $M$ satisfying
\begin{itemize}
\item [$(i)$] for any permutation $\pi:M \to M$
\begin{equation*}
\hat{v}_{T}\left(N,v, \{S_{\pi(k)}\}_{k=1}^n \right) = \pi^{-1} \hat{v}_{T}\left(N,v, \{S_{k}\}_{k=1}^n \right), \quad T \subseteq S_j;
\end{equation*}
\item [(ii)] $\hat{v}_{S_j}(N,v,\cP)=v^{\cP}$ for any $j\in M$;
\item [(iii)] $\hat{v}_T(N,v,\{N\})(\{1\})=v(T)$ for any $T \subseteq N$.
\end{itemize}
\end{definition}

In the definition above, property  $(i)$ ensures that $g[N,u,\mathcal{P}]$ is independent of the ordering of the sets in $\mathcal{P}$. Property $(ii)$ requires that when $T\in \mathcal{P}$ the intermediate game is the quotient game played on the partition elements. Property $(iii)$ requires that for the grand coalition structure 
the total payoff of the intermediate game associated with $T$ is equal to the payoff of $v$ on $T$.

The following lemma shows that $h^{(1)}$  and $h^{(2)}$ can be recovered from $g$ up to multiplicative constants. Conversely, as we shall see later, imposing certain conditions on $h^{(1)}$  and $h^{(2)}$ can result in desirable properties of the coalitional value.

\begin{lemma}\label{lmm::twostepconncoal}
Let $g$ be a coalitional value with a two-step formulation with $h^{(1)}$, $h^{(2)}$ as in Definition \ref{def::2stepprop}. 
\begin{itemize}
  \item [(i)]  $g[N,v,\{N\}] = \alpha h^{(2)}[N,v]$, where $\alpha =h^{(1)}_1[\{1\},\tilde{u}]$.

  \item [(ii)] Suppose $h^{(2)}[\{i\}, \tilde{u}] = \beta$ for any $i \in \mathbb{N}$. Then $g[N,v,\bar{N}] = \beta h^{(1)}[N,v]$. 

\end{itemize}
\end{lemma}
\begin{proof}
Suppose $h^{(1)}$ satisfies $h^{(1)}_1[\{1\},\tilde{u}]=\alpha$. Then for the grand coalition $\cP=\{N\}$ we have
\begin{equation*}%\label{twostepgrand}
g_i[N,v,\{N\}] = h_i^{(2)}[N,v^{(1)}], \,\, v^{(1)}(T) = h_1^{(1)}[\{1\},\hat{v}_T(N,v,\{N\})]=h^{(1)}[\{1\},\tilde{u}] \cdot \hat{v}_T(N,v,\{N\})(\{1\})=\alpha v(T).
\end{equation*}

Next, suppose  $h^{(2)}[\{i\}, \tilde{u}] = \beta$ for any $i \in \mathbb{N}$. Then by Definition \ref{def::2stepprop}, for the partition $\cP=\bar{N}$  
consisting of singletons we have
\begin{equation*}%\label{twostepsingleton}
g_i[N,v,\bar{N}] = h_i^{(2)}[\{i\},v^{(i)}] = h_i^{(2)}[\{i\},\tilde{u}] \cdot v^{(i)}(\{i\}) = \beta v^{(i)}(\{i\})=\beta h^{(1)}_i[N,\hat{v}_{\{i\}}(N,v,\bar{N})]=\beta h_{i}^{(1)}[N,v].
\end{equation*}
This proves the lemma.
\end{proof}

  In the setting of Definition \ref{def::2stepprop}, $g$ can have infinitely many representations, such as via rescaling. Specifically, if $g$ is a coalitional value with a two-step formulation based on $h^{(1)},h^{(2)}$, then for any $\alpha \neq 0$ 
  \[
  g_i[N,v,\cP]=h_{\alpha,i}^{(2)}[S_j,v^{(j)}], \quad v^{(j)}(T)=h_{\alpha,j}^{(1)}[M,\hat{v}_T], \quad i \in S_j,
  \]
  where $h^{(1)}_{\alpha}:=\alpha h^{(1)}$ and $h^{(2)}_{\alpha}:=(1/\alpha)h^{(2)}$. Under some mild conditions, however, $g$ has a distinct representation in terms of re-normalized game values $h_*^{(1)}$ and $h_*^{(2)}$, and some  scaling constant $\alpha_*$; see Lemma \ref{lmm::canonical}.

\vspace{5pt}

\noindent{\bf Examples.} Before we proceed, let us provide several examples of game-values that have two-step formulation. First, consider the two-step Shapley defined in \citet{Kamijo2009} and given by 
\begin{equation}\label{2shap}
TSh_i[N,v,\mathcal{P}] = \varphi_i[S_j,v] + \frac{1}{|S_j|} \big( \varphi_j[M,v^{\mathcal{P}}]  - v(S_j)\big), \quad i \in S_j.
\end{equation}
Define the intermediate game $\hat{v}$ as follows. For each non-empty $T \subseteq S_j$ define 
the game $\hat{v}_T$ by 
\begin{equation}\label{tshinterm}
\hat{v}_T(A) =  \frac{|T|}{|S_j|}v^{\cP}(A) + |A|\big( v(T) - \frac{|T|}{|S_j|}v^{\cP}( \{j\}) \big), \,\, A\subseteq M.
\end{equation}
Let the game values be $h^{(1)}=h^{(2)}=\varphi$. Then for each $j \in M$ and $i \in S_j$ we have
\begin{equation}\label{2steptsh}
TSh_i[N,v,\cP]  = \varphi_i[S_j,v^{(j)}_{TSh}], \quad v^{(j)}_{TSh}(T) = \varphi_j[M,\hat{v}_T]= v(T)+\frac{|T|}{|S_j|} \big( \varphi_j[M,v^{\cP}]-v(S_j) \big), \,\,T \subseteq S_j. 
\end{equation}

Two other examples are the Owen and Banzhaf-Owen values (see \eqref{BzOw}). 
For each $T \subseteq S_j$ define the intermediate game by
\begin{equation}\label{modqgame}
\hat{v}_T(A)=v^{\cP|T}(A) := \1_{\{ j \notin A \}} v^{\cP}(A) + \1_{\{ j \in A \}} v( \cup_{ k \in A \setminus \{j\} } S_k \cup T), \,\, A\subseteq M.
\end{equation}
Then for each $j \in M$, $i \in S_j$, and $T \subseteq S_j$
\begin{equation*}%\label{2stepOwBz}
\begin{aligned}
Ow_i[N,v,\cP]&=\varphi_i[S_j,v^{(j)}_{Ow}],&  v^{(j)}_{Ow}(T)&=\varphi_j[M,v^{\cP|T}]\\
BzOw_i[N,v,\cP] & =Bz_i[S_j,v^{(j)}_{BzOw}],&  v^{(j)}_{BzOw}(T)&=Bz_j[M,v^{\cP|T}].
\end{aligned}
\end{equation*}
Notice that one has $h^{(1)}=h^{(2)}=\varphi$  (see \eqref{shapform}) in the first two-step formulation while $h^{(1)}=h^{(2)}=Bz$ (see \eqref{banzhafform}) in the second.

More generally, the coalitional value in the form
\begin{equation}\label{gencoalval}
g_i[N,v,\cP]=\sum_{A\subseteq M \setminus \{j\}} \sum_{T \subseteq S_j \setminus \{i\}} w^{(1)}(A,M,j) w^{(2)}(T,S_j,i) \Big( v(Q_A \cup T \cup \{i\}) -v ( Q_A \cup T)\Big), i\in S_j
\end{equation}
where $Q_A=\cup_{\alpha\in A}S_{\alpha}$ can be expressed via two-step formulation where $\hat{v}_T(A)=v^{\cP|T}(A)$ and $h^{(1)},h^{(2)}$ are linear game values in the form \eqref{lingameform} with weights $w^{(1)}$ and $w^{(2)}$, respectively. 

\begin{remark}\rm It follows directly from the definition that a coalitional value $g_i[N,v,\cP]$ with the two-step formulation has complexity $2^{|S_j|+|\cP|}\cdot O(v)$, which can be significantly lower than that of game values. Note also that the two-step Shapley value for singletons is the Shapley value. Thus, in light of \eqref{2steptsh}, the (empirical) quotient game explanations can be re-used for the computation of single feature explanations based on the two-step Shapley value, which in turn lowers the complexity to $O(n (2^{|\cP|}+2^{|S_j|})) \cdot O(v)$, where $v \in \{\hat{v}_*^{\CE},\hat{v}_*^{\ME}\}$ is the estimator of the deterministic conditional and marginal games defined in \eqref{margcondgamedet}. This, however, is not true for the Owen and Banzhaf values since single feature explanations have to be computed without the use of the quotient game values, which gives the complexity stated earlier.
\end{remark}

\subsubsection{Bounds for marginal coalitional values with two-step representation}\label{sec::extra}

Let $g$ be a linear coalitional game value and $\bar{g}$ its centered extension. Generalizing Theorem \ref{thm::margoperatorunbound}, Proposition \ref{prop::coalqgameprop} indicates that, unlike their conditional analogs, marginal explanations
$f\mapsto \bar{g}_i(X; \vpdp,\mathcal{P},f)$ 
are not necessarily bounded in the $L^2(P_X)$-metric. A similar instability for game values motivated grouping the predictors as a remedy and resulted in Proposition \ref{prop::quotgameexpl}. Here, we provide better bounds for marginal coalitional values under the assumption that $g$ admits a two-step formulation with a particular type of intermediate games, i.e. those that appeared in the case of Owen or two-step Shapley value; see \eqref{modqgame} and \eqref{tshinterm}.

\begin{proposition}[{\bf bounds}]\label{prop::extra1} 

Let $g$ be a coalitional value with a two-step formulation with $h^{(1)}$, $h^{(2)}$, and $\hat{v}_T$ as in Definition \ref{def::2stepprop}. Denote the centered extension of $g$ by $\bar{g}$.

\begin{itemize}
    \item [(i)] Suppose $\hat{v}_T(N,v,\cP)=v^{\cP|T}$, $T\subseteq S_j$. Then for  $i\in S_j$ 
    \begin{equation*}
    \begin{split}
    \|\bar{g}_i(X; \vpdp,\mathcal{P},f)\|_{L^2(\PP)}\leq 
    &C\Big(\sum_{Q_A=\cup_{j\in A}S_j,A\subseteq M\setminus\{j\}}\|f\|_{L^2(P_{X_{Q_A}}\otimes P_{X_{-Q_A}})}\\
    &+\sum_{T\subseteq S_j}\sum_{Q_A=\cup_{j\in A}S_j,A\subseteq M\setminus\{j\}}\|f\|_{L^2(P_{X_{Q_A\cup T}}\otimes 
    P_{X_{-(Q_A\cup T)}})}\Big)
    \end{split}
    \end{equation*}
    where $C$ depends only on $h^{(1)}$, $h^{(2)}$ and $\mathcal{P}$.

    \item [(ii)] Suppose  $\hat{v}_T(N,v,\cP)(A)= \frac{|T|}{|S_j|}v^{\cP}(A) + |A|\big( v(T) - \frac{|T|}{|S_j|}v^{\cP}( \{j\}) \big)$, $T\subseteq S_j$. Then for $i\in S_j$ 
    $$
    \|\bar{g}_i(X; \vpdp,\mathcal{P},f)\|_{L^2(\PP)}\leq 
    C\Big(\sum_{Q_A=\cup_{j\in A}S_j,A\subseteq M}\|f\|_{L^2(P_{X_{Q_A}}\otimes P_{X_{-Q_A}})}
    +\sum_{T\subseteq S_j}\|f\|_{L^2(P_{X_T}\otimes P_{X_{-T}})}\Big)
    $$
    where $C$ depends only on $h^{(1)}$, $h^{(2)}$ and $\mathcal{P}$.
\end{itemize}
\end{proposition}
\begin{proof}
See Appendix \ref{app::extra}.
\end{proof}

We next state a corollary to the above proposition which assumes that the elements of the partition $\cP$ are independent. Then, any of the following assumptions yields the continuity of $f\mapsto \bar{g}_i(X; \vpdp,\mathcal{P},f)$, $i\in S_j$, in the $L^2(P_X)$-metric when either the predictors $\{X_i\}_{i \in S_j}$ are independent, or $S_j=\{i\}$.

\begin{corollary}\label{corr:extra3}
Let $X,\cP,f,f_0$ be as in Proposition \ref{prop::extra1}, and $g$ a coalitional game value that satisfies the assumptions of either part of the proposition. Suppose the group predictors $X_{S_1},\dots,X_{S_m}$ are independent. Then there exists a constant $C$ dependent only on game values and the partition such that 
for any $i\in S_j$ we have
$$
\|\bar{g}_i(X; \vpdp,\mathcal{P},f)\|_{L^2(\PP)}\leq C\Big(\|f\|_{L^2(P_X)}
      +\sum_{T\subseteq S_j}\|f\|_{L^2(P_{X_T}\otimes P_{X_{-T}})}\Big).
$$
If furthermore either the predictors $\{X_i\}_{i \in S_j}$ are independent or $S_j=\{i\}$, then 
$$
\|\bar{g}_i(X; \vpdp,\mathcal{P},f)\|_{L^2(\PP)}\leq C\|f\|_{L^2(P_X)}
$$
for another such constant $C$. Finally, in the case of $S_j=\{i\}$, if $h^{(2)}$ is efficient and $h^{(1)}$ satisfies the assumptions of Corollary \ref{corr::cond_operator_cont}$(ii)$, then one can take $C=1$ in the above inequality.
\end{corollary}
\begin{proof}
See Appendix \ref{app::extra}.
\end{proof}

%%%%%%%%%%%%%%%%%%%%%%%%%%%%%%%%%%%%%%%%%%%%%%%%%%%%%%%%%%%%%%%%%%%%%%%%%%%%%%%
%%%%%%%%%%%%%%%%%%%%%%%%%%%%%%%%%%%%%%%%%%%%%%%%%%%%%%%%%%%%%%%%%%%%%%%%%%%%%%%
% Quotient game formulation
%%%%%%%%%%%%%%%%%%%%%%%%%%%%%%%%%%%%%%%%%%%%%%%%%%%%%%%%%%%%%%%%%%%%%%%%%%%%%%%
%%%%%%%%%%%%%%%%%%%%%%%%%%%%%%%%%%%%%%%%%%%%%%%%%%%%%%%%%%%%%%%%%%%%%%%%%%%%%%%

\subsubsection{\bf Trivial group explainers with two-step formulation }\label{sec::qp_prop}

Trivial group explanations (see Definitions \ref{def::trivgexp} and \ref{def::coalexpl}) may differ for the marginal and conditional games, which in turn may break the continuity of marginal explanations with respect to models in $L^2(P_X)$. For game values, for instance, to remedy the situation Proposition \ref{prop::additivesumshap} required the model $f$ to be additive across partition (see \eqref{additiveP}), which is a very stringent requirement for ML models. It turns out that for coalitional values with a two-step formulation, there is no need to impose any conditions on the form of a model as long as $h^{(2)}$ is proportional to an efficient game value. To this end, we provide the following result.

\begin{proposition}\label{prop::unifcoalexpl} 
Let $g$ be a coalitional value with a two-step formulation with $h^{(1)}$, $h^{(2)}$ as in Definition \ref{def::2stepprop}. Suppose $h^{(2)}=\alpha h^{(2)}_*$ where $\alpha \in \RR$ and the game value $h^{(2)}_*$ satisfies \hyperref[axiom:EP]{(EP)}. Then:

\begin{itemize}

\item [$(i)$] For any $(N,v,\cP)$ we have
\begin{equation}\label{QP2stepprelim}
\sum_{i \in S_j} g_i[N,v,\mathcal{P}] = \alpha h^{(1)}_j[M,v^{\mathcal{P}}] = g[M,v^{\cP},\bar{M}]
\end{equation}
Consequently, the marginal trivial group explanations satisfy the improved bounds
\begin{equation}\label{trivgroupcoalbounds1}
\|\bar{g}_{S_j}(X;\vpdp,\cP,f)\|_{L^2(\PP)} \leq 
|\alpha|\cdot C\big(h^{(1)},\cP \big) \cdot \Big(\sum_{ A\subseteq M\setminus\{j\},Q_A=\cup_{j\in A}S_j}\|f\|_{L^2(P_{X_{Q_A}}\otimes P_{X_{-Q_A}})} \Big).
\end{equation}

\item[$(ii)$] If the predictor unions   $X_{S_1},X_{S_2},\dots,X_{S_m}$ are independent, then for $f \in L^2(P_X)$
\begin{equation}\label{QPexplanations}
\bar{g}_{S_j}(X;\vpdp,\cP,f) = \bar{g}_{S_j}(X;\vce,\cP,f), \quad S_j \in \cP.
\end{equation}
Consequently, the linear maps $f \mapsto \bar{g}_{S_j}(X;\vme,\cP,f)$ is bounded, with the Lipschitz constant that depend on $\alpha$, $h^{(1)}$, and $\cP$, and hence it is continuous on  $L^2(P_X)$.

\item [(iii)] 

Suppose $r_{Q_A}=\tfrac{d P_{X_{Q_A}} \otimes P_{X_{-Q_A}}}{d P_X}$, $Q_A=\cup_{j \in A} S_j$, exists and belongs to 
$L^\infty(P_X)$ for each $A \subseteq M$. Then $\big(L^2(\intP_{X,\cP}),\|\cdot\|_{L^2(P_X)}\big) = L^2(P_X)$ and for $f \in L^2(P_X)$
\begin{equation*}%\label{qgameunif}
\bar{g}_{S_j}(X;\vce,\cP,f)=\bar{g}_{S_j}(X;\vpdp,\cP,f) + \mathcal{I}(f,\{r_{Q_A}\}_{A \subseteq M}) \quad in \quad L^2(\P),
\end{equation*}
with the error term $\mathcal{I}$ satisfying the bound
\[
 \|\mathcal{I}(f,\{r_{Q_A}\}_{A \subseteq M})\|_{L^2(\P)} \leq C(g,\cP) \cdot \Big( \max_{A \subseteq M} \|r_{Q_A}-1\|_{L^{\infty}(P_X)} \Big) \cdot \| f \|_{L^2(P_X)}.
\]
\end{itemize}
\end{proposition}

\begin{proof} For any $j \in M$ we have
\[
\sum_{i\in S_j}g[N,v,\cP] =\sum_{i \in S_j} h^{(2)}_i[S_j,v^{(j)}]=\alpha_* v^{(j)}(S_j)=\alpha_*h^{(1)}_j[M,v^{\cP}]
\]
where we used the efficiency of $h_*^{(2)}$. Since $h^{(2)}_*$ satisfies  \hyperref[axiom:EP]{(EP)}, Lemma \ref{lmm::twostepconncoal}$(ii)$ implies that $\alpha_*h^{(1)}_j[M,v^{\cP}]=g[M,v^{\cP},\bar{M}]$. This proves \eqref{QP2stepprelim}. The bound \eqref{trivgroupcoalbounds1} follows from \eqref{QP2stepprelim} and the definition of the quotient game. This establishes $(i)$. If $X_{S_1},X_{S_2},\dots,X_{S_m}$ are independent, then  \eqref{QP2stepprelim} and Proposition \ref{prop::quotgameexpl}$(i)$ give $(ii)$. Finally, Proposition \ref{prop::quotgameexpl}$(i)$ and \eqref{QP2stepprelim} imply $(iii)$.
\end{proof}

\begin{proposition}\label{prop::addflow2step} 
Let $g$ be a coalitional value with a two-step formulation with $h^{(1)}$, $h^{(2)}$ as in Definition \ref{def::2stepprop}. Suppose $h^{(1)}$, $h^{(2)}$ satisfy \hyperref[axiom:EP]{(EP)}. Then:

\begin{itemize}
\item [$(i)$] The coalitional value $g$ satisfies \hyperref[axiom:EP]{(EP)}. 

\item [$(ii)$] Let $h^{(1)}$ have the form \eqref{lingameform} and satisfy the properties outlined in Corollary \ref{corr::cond_operator_cont}$(ii)$. Then 
\begin{equation}\label{trivgroupcoalbounds2}
 \sum_{j \in M}\|\bar{g}_{S_j}(X;\vce,\cP,f)\|^2_{L^2(\PP)}  \leq \|f-f_0\|^2_{L^2(P_X)} \leq \|f\|^2_{L^2(P_X)}, \quad f \in L^2(P_X). 
\end{equation}
Consequently, if unions $X_{S_1},\dots,X_{S_m}$ are independent, the same bound holds for $\bar{g}_{S_j}(X;\vme,\cP,f)$.
\end{itemize}
\end{proposition}
\begin{proof}
Let $\{v^{(j)}\}_{j=1}^M$ be as in Definition \ref{def::2stepprop}. Suppose that $h^{(1)},h^{(2)}$ are efficient, then we have
\[
\sum_{i\in N}g[N,v,\cP] =\sum_{j \in M}\sum_{i \in S_j} h^{(2)}_i[S_j,v^{(j)}]=\sum_{j \in M}v^{(j)}(S_j)=\sum_{j \in M}h^{(1)}_j[M,v^{\cP}]=v^{\cP}(M)=v(N),
\]
where we used the property $\hat{v}_{S_j}=v^{\cP}$ hence $(i)$.

Since $h^{(2)}$ is efficient, Proposition \ref{prop::unifcoalexpl}$(i)$ implies $\bar{g}_{S_j}(X;\vce,\cP,f) = h^{(1)}_j[M,v^{\CE,\cP}]$,  $j\in M$. This together with Lemma \ref{lmm::boundquot}$(iii)$ and the properties of $h^{(1)}$ implies \eqref{trivgroupcoalbounds2}, which proves $(ii)$.
\end{proof}

\begin{remark}\rm

Given a scaled efficient game value $h^{(2)}$, Proposition \ref{prop::unifcoalexpl}$(i)$ implies that
\begin{equation}\tag{QP}\label{quotientgame}
\sum_{i \in S_j} g_i[N,v,\mathcal{P}]=g_j[M,v^{\mathcal{P}},\bar{M}], \quad S_j \in \cP, \quad \cP=\{S_1,S_2,\dots,S_m\}
\end{equation}
called the quotient game property. Thus, the two-step formulation of $g$, together with the efficiency of $h^{(2)}$, is equivalent to $g$ satisfying \eqref{quotientgame}. Note that we could have imposed the condition \eqref{quotientgame} on a coalitional value in order to obtain the subsequent stability results. However, we elected not to do this and rather work with the two-step formulation setup. There are several reasons for this: 1) the two-step formulation allows for a simpler requirement on $h^{(2)}$ for a coalitional value to obtain \eqref{quotientgame}; 2) by fixing an efficient $h^{(2)}$, one can engineer a large  class of coalitional values with the \eqref{quotientgame} property by varying $h^{(1)}$ and the intermediate game; 3) the two-step formulation is helpful when designing recursive values for a given partition tree. 
\end{remark}

\noindent {\bf Role of coalitional values in the design of recursive values with additive flows.}
If $h^{(1)}$  and $h^{(2)}$ are both efficient, the coalitional value $g[N,v,\cP]$ satisfies simultaneously \eqref{quotientgame} and \hyperref[axiom:EP]{(EP)}, which induces a recursive coalitional values with an additive flows along any combinatorial tree and which allows to design group explainers based on any parametrized partition tree, which we do in \S \ref{sec::ExTree}. 

To understand the role of the two-step formulation in building recursive values, note that that a coalitional value $g[N,v,P]$  with a normalized two-step formulation can be associated with a combinatorial tree of depth two, except when $\cP=\bar{N}$; see Figure \ref{fig::two_step_trees}. The root of the tree contains $v(N)=g_1[\{1\},v^{N} ,\{\{1\}\}]$, its children contain $h^{(1)}_j[M,v^{\cP}]=g_j[M,v^{\cP},\bar{M}]$, $j\in M$, associated with each element of  $\cP=\{S_1,\dots, S_m\}$, and terminal nodes correspond to $g_i[N,v,\cP]=h^{(2)}[S_j,v^{(j)}]$, $i \in N$. Each terminal node $i$ which is a child of the root corresponds to a singleton $S_j={i}$ for some $j\in M$ and satisfies $g_i [N,v,\cP]=h^{(1)}_j [M,v^{\cP}]=g_j[M,v^{\cP},\bar{M}]$. 

Thus, if $h^{(1)}$  and $h^{(2)}$  are  efficient, for every non-terminal node, the sum of the values in its children equals to the value of the node. This gives an additive flow along the combinatorial tree of depth two. In \S \ref{sec::ExTree} we generalize such coalitional values to a combinatorial tree of any depth and design  recursive values with additive flows. Since one can choose a game $\hat{v}_T$ with any desired properties, one can construct a large (infinite) collection of recursive coalitional values with additive flows and then use it to construct group explainers based on a given parametrized partition tree, see \S \ref{sec::parttree1}.

\vspace{5pt}

\noindent{\bf Examples.} Note that the Shapley value, viewed as a coalitional value, fails to satisfy the quotient game property; see \cite{Lorenzo-Freire}. However, the two-step Shapley values and Owen values satisfy the quotient game property, which can be verified by direct calculations using the two-step formulations and the efficiency of $h^{(2)}=\varphi$:
\[
\begin{aligned}
\sum_{i \in S_j} TSh_i[N,v,\cP] = \sum_{i \in S_j} \varphi[S_j,v^{(j)}_{TSh}] = v^{(j)}_{TSh}(S_j)=\varphi_j[M,v^{\cP}]=TSh_j[M,v^{\cP},\bar{M}]\\
\sum_{i \in S_j} Ow_i[N,v,\cP] = \sum_{i \in S_j} \varphi[S_j,v^{(j)}] = v^{(j)}_{Ow}(S_j)=\varphi_j[M,v^{\cP}]=Ow_j[M,v^{\cP},\bar{M}].
\end{aligned}
\]

%% Removed the following example:
% Another example of a coalitional value that satisfies the quotient game property is the symmetric Banzhaf value introduced in \citet{Alonso2002} and defined by the formula
% \begin{equation}\label{Bzsym}
% Bz^{sym}_i[N,v,\mathcal{P}] = \sum_{R \subseteq M \setminus \{j\} } \sum_{T \subseteq S_j \setminus \{i\} } 
% \frac{1}{2^{m-1}} \frac{t!(s_j-t-1)!}{s_j!}\big[ v(Q \cup T \cup \{i\})-v(Q\cup T)\big].
% \end{equation}
% To see this, note that the two-step formulation for this coalitional value reads
% \begin{equation*}%\label{2stepBzsym}
% Bz^{sym}_i[N,v,\cP]=\varphi_i[S_j,v^{(j)}_{Bz^{sym}}],  \quad v^{(j)}_{Bz^{sym}}(T)=Bz_j[M,v^{\cP|T}]
% \end{equation*}
% where $h^{(1)}=Bz$, and $h^{(2)}=\varphi$ satisfies the efficiency property. 

%%%%%%%%%%%%%%%%%%%%%%%%%%%%%%%%%%%%%%%%%%%%%%%%%%%%%%%%%%%%%%%%%%%%%%%%%%%%%%%%
% Part 4.1.1: grouping predictors intro
%%%%%%%%%%%%%%%%%%%%%%%%%%%%%%%%%%%%%%%%%%%%%%%%%%%%%%%%%%%%%%%%%%%%%%%%%%%%%%%%

\section{Information-theoretic hierarchical clustering of predictors}\label{sec::Grouping}
Given predictors $X\in \RR^n$, the first step in constructing group explainers is to identify disjoint sets $S_j \subseteq N$ that yield a partition $\mathcal{P}=\{S_1,S_2,\dots,S_r\}$ of 
predictor indices, so that $X_{S_1},X_{S_2},\dots,X_{S_r}$ form (weakly) independent unions such that within each group the predictors share a significant amount of mutual information \cite{Cover}. Such partitioning would effectively reduce the dimensionality of the problem and, consequently, lower the complexity of explanations, while also alleviating the issue of explanation splitting. Moreover, as we shall see in \S \ref{sec::qgameexpl}, grouping unifies the conditional and marginal explanations.

% In this section, we discuss an information-theoretic grouping technique, which is  the first step towards constructing group explainers. Given predictors $X=(X_1,X_2,\dots,X_n)$, the objective is to identify disjoint sets 
% $S_j \subseteq \{1,2,\dots,n\}$ that yield a partition $\mathcal{P}=\{S_1,S_2,\dots,S_r\}$ of 
% predictors' indices, so that $X_{S_1},X_{S_2},\dots,X_{S_r}$ form independent,  or weakly independent, unions such that within each group the predictors share a significant amount of mutual information \cite{Cover}. Each union $X_{S_j}$ may be viewed as a single predictor, called a group predictor. Such partitioning would effectively reduce the dimensionality of the problem and consequently, the complexity of explanations. Moreover, as we shall see in \S \ref{sec::qgameexpl}, grouping in a sense unifies conditional and marginal explanations.

% Grouping predictors and group attribution methods have been discussed previously in the context of linear or simple functional dependencies  \citep{aas}. However, in real datasets the dependencies are often highly non-linear and not necessarily functional. Therefore, grouping methods must be based on techniques that measure non-linear dependencies accurately. {\color{blue} For this reason, to construct dependency-based partition of predictors}, we propose to employ a variable hierarchical clustering technique in conjunction with a state-of-the-art measure of dependence called the Maximal Information Coefficient (MIC), introduced in \citet{Reshef11,Reshef16}. 

Group attribution methods have been discussed previously in the context of linear or simple functional dependencies \citep{aas}. In real datasets, however, the dependencies are often highly non-linear and not necessarily functional. For this reason, to construct a dependency-based partition of predictors, we propose to employ a variable hierarchical clustering technique in conjunction with a state-of-the-art measure of dependence called the Maximal Information Coefficient (MIC), that overcomes the disadvantages of traditional measures and was introduced in \citet{Reshef11,Reshef16}. In what follows, we introduce this measure and describe hierarchical clustering methods based on it. An example that demonstrates the advantage of using MIC in clustering is provided in \S \ref{app::example_var_clust}.

%%%%%%%%%%%%%%%%%%%%%%%%%%%%%%%%%%%%%%%%%%%%%%%%%%%%%%%%%%%%%%%%%%%%%%%%%%%%%%%%
% Part 4.1.2: MIC definition
%%%%%%%%%%%%%%%%%%%%%%%%%%%%%%%%%%%%%%%%%%%%%%%%%%%%%%%%%%%%%%%%%%%%%%%%%%%%%%%%

\subsection{Maximal information coefficient as a measure of dependence}\label{sec::MIC}

Immense progress has been made in recent years in designing powerful statistics for measuring
variable dependence. Most notable measures of dependence are investigated in the following works:  \citet{Kraskov2004}, \citet{Zenga2019} and \citet{Paninski2003} on the estimation of mutual information; \citet{Renyi1959} and \citet{Breiman1985} on maximal correlation;  \citet{Szekely2007} and \citet{ Szekely2009} on distance correlation; \citet{Reshef11,Reshef16,Reshef16b} on maximal information coefficient (MIC); \cite{Gretton2005,Gretton2012} on the Hilbert-Schmidt independence criterion,  \citet{Lopez-Paz2013} on the randomized dependence coefficient; \citet{Heller2013} on the Heller-Heller-Gorfine distance, \citet{Heller2016} on $S^{\text{\tiny \it DDP}}$.

\citet{Reshef16} introduced the information-theoretic measure of dependence called
MIC$_*$, the population value of the MIC statistic, defined as a regularized form of mutual information between a pair of random variables.
\begin{definition}[\citet{Reshef16}] Let $(X,Y)$ be jointly distributed random variables. The population maximal information coefficient$_*$ (${\rm{MIC}}_*$) of $(X,Y)$ is defined by
  \[
  {\rm{MIC}}_*(X,Y)=\sup_G \frac{I\big((X,Y)|_G\big)}{\log \|G\|}.
  \]
  Here $G$ denotes a two-dimensional grid, $\|G\|$ denotes the minimum of the number of rows of $G$ and the number of columns of $G$, $I((X,Y)|_G)$ denotes the discrete mutual information of $(X,Y)|_G:=(col_G(X),col_G(Y))$.
\end{definition}
MIC$_*$ has the following remarkable properties:
\begin{itemize}
    \item MIC$_*$ applies to pairs of random variables and returns a value in $[0,1]$ that represents the strength of the relationship between them. That value is $0$ if and only if the variables are independent;
    \item it is {\it transitive} in the sense that it provides a similar value between transformed variables,
    ${\rm{MIC}}_*(X,Y)={\rm{MIC}}_*(g(X),h(Y))$, where $g,h$ are strictly monotonic;
    \item it is {\it equitable}, that is, it outputs a similar value between pairs of variables that exhibit similar noise levels, 
    ${\rm{MIC}}_*(X,Y) \approx {\rm{MIC}}_*(Z,W)$  $\Rightarrow$ ${\rm{MIC}}_*(X+\eps_1,Y+\eps_2) \approx {\rm{MIC}}_*(Z+\eps_1,W+\eps_2)$.
\end{itemize}

There are two statistics, MIC and $\MICe$, that can be used to estimate $\MICstar$. While both of the statistics are consistent estimators, MIC introduced in \citet{Reshef11} can be computed only via an inefficient heuristic approximation, while MIC$_e$ introduced in \citet{Reshef16}  can be computed exactly and efficiently using an appropriate optimization technique which yields a fast algorithm that allows one to estimate ${\rm{MIC}}_*$ in linear time; see Definition \ref{MIC_def} of $\MICe$ and Corollary \ref{corr::mic_complexity} that discusses its complexity. 

\subsection{ Dependency-based hierarchical clustering}

A partition of data points into $K<n$ clusters can be characterized by a grouping map $C:\{1,2,\dots,n\}\to \{1,2,\dots,K\}$ that assigns each observation to a cluster $k \in \{1,2,\dots,K\}$ following a certain rule. A clustering algorithm’s objective is to identify an optimal grouping map that solves a minimization problem $\min_C W(C;d)$ for some energy function $W(C;d)$ based upon the {\it dissimilarity} measure $d(p_i,p_j)$  between points; see \citep[Section 14.3]{Hastie et al}.

% where
% \[
% W(C)=\sum_{k \in \{1,\dots,K\}} \sum_{C(i)=k} \sum_{C(j)=k} d(p_i,p_j)
% \]
% is an energy function and 

% $d(p_i,p_j)$ is a measure of {\it dissimilarity} between points.

Hierarchical clustering algorithms produce hierarchical representations called \textit{dendrograms}; for example, see Figure \ref{fig::hierclus}. In addition to the dissimilarity measure, these algorithms require as an input a measure of dissimilarity between disjoint clusters, called intergroup dissimilarity. A well-known intergroup dissimilarity measure is the group average linkage (GA) given by $d_{GA}(S_1,S_2)=\tfrac{1}{|S_1||S_2|} \sum_{i\in S_1} \sum_{ j \in S_2} d(p_i,p_j)$, which satisfies the statistical consistency property. Two other popular measures are single linkage and complete linkage, that estimate the smallest and the largest pairwise distances, respectively, between points in two clusters.

Agglomerative methods, or recursive merging, start at the bottom where each single data point represents a cluster and at each new level merge a selected pair of clusters into a single one. The clusters picked for merging are those for which intergroup dissimilarity achieves the smallest value. This procedure yields a binary partition tree (a parameterized tree in which exactly two branches coalesce) where the height of each node is proportional to the value of the intergroup dissimilarity between the two child nodes with the terminal nodes located at zero height; see \citep[Section 14.3.12]{Hastie et al}.

In our work, we seek to generate a dendrogram that accurately encodes the strength of dependencies between predictors. To this end, we propose to use the dissimilarity measure between predictors based on regularized mutual information given by
\begin{equation*}%\label{MICstardissim}
d_{\MICstar}(X_i,X_j)=1-\MICstar(X_i,X_j) \in [0,1],
\end{equation*}
estimated by the statistic $1-\MICe\big(\{x_i^{(\ell)}\}_{\ell=1}^M,\{x_j^{(\ell)}\}_{\ell=1}^M\big)$ based on observations $\{(x_1^{(\ell)},\dots,x_n^{(\ell)})\}_{\ell=1}^M$.

The advantage of the MIC-based clustering algorithm is that properties of MIC are carried over to the partition tree. In particular, the shape of the tree has the following desirable properties: (a) the tree height, representing the strength of dependencies in predictors, is always $\leq 1$; (b) in light of transitivity, the geometry of the tree is invariant under strictly monotone transformations; and (c) in light of equitability, the height of each subtree reflects information about the noise level among predictors corresponding to the terminal nodes of the subtree.

Given an MIC-based dendrogram of height $h$, the parameter  $\alpha \in (0,h)$ characterizing the strength of dependencies induces a partition of predictors 
$\mathcal{P}_{\alpha}=\{S_1^{\alpha},S_2^{\alpha},\dots,S_{m_{\alpha}}^{\alpha}\}$
whose elements correspond to the terminal nodes of subtrees obtained by cross-sectioning the tree at height $\alpha$. Under the assumption that coalescence of branches happens at distinct heights, $\alpha \mapsto \cP_{\alpha}$ is a left-continuous partition map which characterizes the dendrogram and gives rise to a nested sequence of partitions starting at singletons $\{\{X_1\},\{X_2\},\dots,\{X_n\}\}$ and terminating at the grand coalition $\{X_1,\dots,X_n\}$; for details on the  construction and properties of the partition map, see \S \ref{sec::parttree}. In what follows, these dependency-based partitions are used to construct group explainers based on coalitional values, which incorporate the partition into their structure. For an illustration on the use of hierarchical clustering to produce partitions and construct group explainers see the example in \S \ref{app::example_var_clust}.

\section{ Numerical examples}\label{sec::examples}
This section contains examples that illustrate the theoretical aspects discussed in \S\ref{sec::obser_inter_expl} and  \S\ref{sec::Group_expl}. In the following computations we replace the deterministic marginal explanations with their empirical analogs, which are evaluated as corresponding game values for the empirical marginal game $\hat{v}^{\ME}$ defined by
\begin{equation}\label{empmarggame}
    \hat{v}^{\ME}(S;x,f, \bar{D}_X) := \frac{1}{|\bar{D}_X|}\sum_{ \tilde{x} \in \bar{D}_X} f(x_S,\tilde{x}_{-S}),
\end{equation}
where $x$ is an observation, $f$ is a model, and $\bar{D}_X=\{\tilde{x}^{(k)}\}_{k=1}^K$ is a background dataset of predictor observations used for averaging.

\subsection{Pedagogical example on instability of marginal explanations}\label{sec::pedag_example}

The results of \S \ref{sec::obser_inter_expl} show that marginal explanations viewed as linear operators may not be well-defined or stable in $L^2(P_X)$. Demonstrating the instability numerically is not a trivial task because the space $L^2(P_X)$ of models is much larger than any class of models obtained via training. Nevertheless, we numerically investigate the stability of marginal explanations by training a small collection of different models on the same dataset, or on perturbed datasets, and then compare the differences between the resulting explanations as well as between the predictions of those models.

We start with a pedagogical example that showcases the instability of marginal explanations under predictor dependencies. Consider the following data generating model. Let $X=(X_1,X_2,X_3)$ be predictors such that the pair $(X_1,X_2)$ is independent of $X_3$, with the distribution given by 
\begin{equation}\label{ex1_preds_}
\begin{aligned}
Z &   \sim { Unif}(-1,1) \\
X_1 & = Z + \epsilon_1, \quad \epsilon_1 \sim \mathcal{N}(0,\delta),\\
X_2 & = \sqrt{2}\sin( Z (\pi/4)) + \epsilon_2, \quad \epsilon_2 \sim \mathcal{N}(0,\delta),\\
X_3 & \sim { Unif}\big([-1,-0.5] \cup [0.5,1] \big).
\end{aligned}
\end{equation}
where $\delta>0$ is chosen later. The model for the output variable is assumed to be 
\begin{equation}\label{ex1_resp_model_}
Y=f_*(X_1,X_2,X_3)+\epsilon_3=3X_2X_3+\epsilon_3, \quad \epsilon_3\sim {Unif}(-0.05,0.05).
\end{equation}

Note that in the true regressor $f_*$ the variable $X_1$ is a dummy variable (it is not explicitly used). For this reason, the marginal explanation approach will assign zero attribution to this variable.

By design, the dependencies in predictors allow for the existence of many models from $L^2(P_X)$ that approximate the response variable well but have different representations. In what follows we demonstrate that the generated explanations differ in such cases where different models with distinct representations approximate the data well.

\begin{figure}
  \centering
  \begin{subfigure}[t]{0.3\textwidth}
    \centering
    \includegraphics[width=\textwidth]{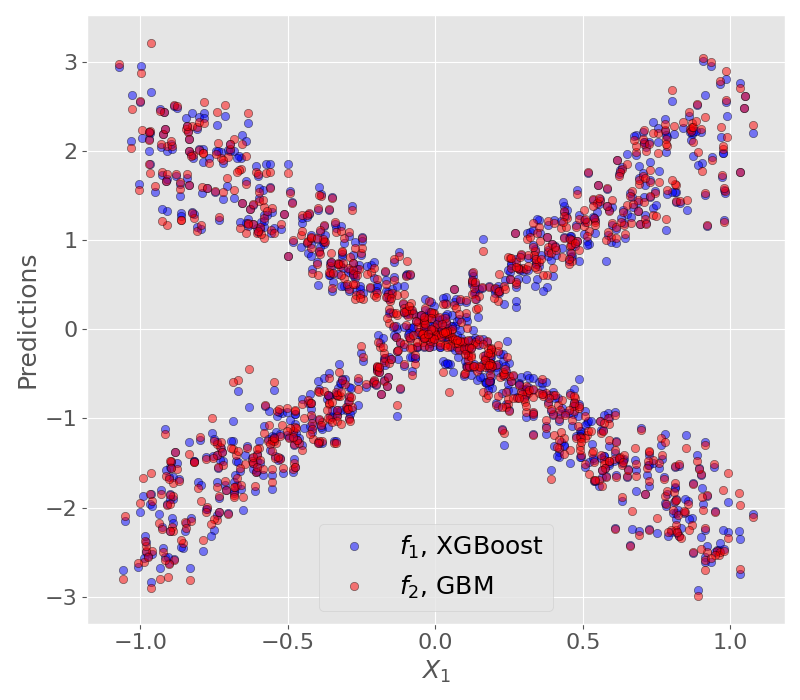}\caption{\footnotesize Predictions vs $X_1$.}\label{fig::ex1_pred_vs_X1}
  \end{subfigure}
    \begin{subfigure}[t]{0.3\textwidth}
    \centering
    \includegraphics[width=\textwidth]{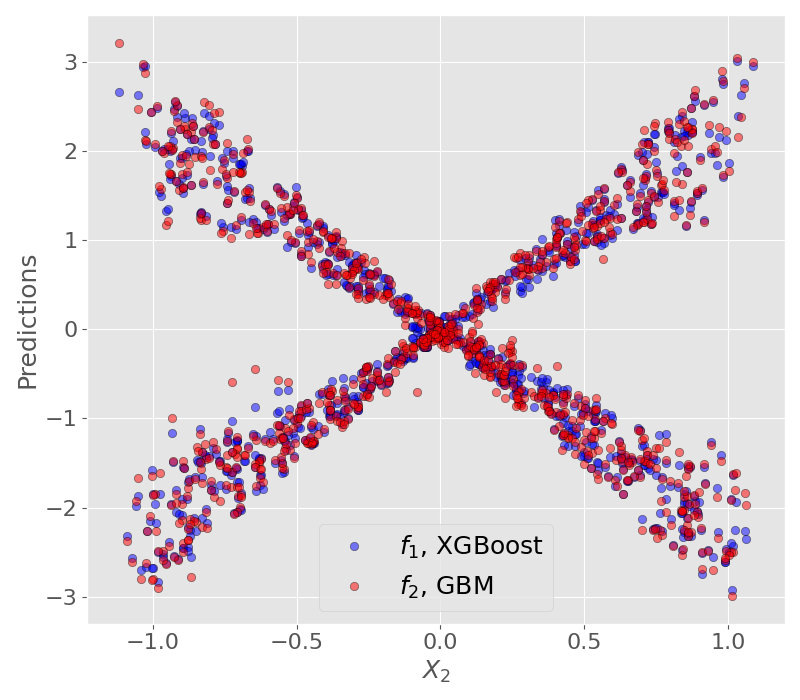}\caption{\footnotesize Predictions vs $X_2$.}\label{fig::ex1_pred_vs_X2}
  \end{subfigure}   
  \begin{subfigure}[t]{0.3\textwidth}
    \centering
    \includegraphics[width=\textwidth]{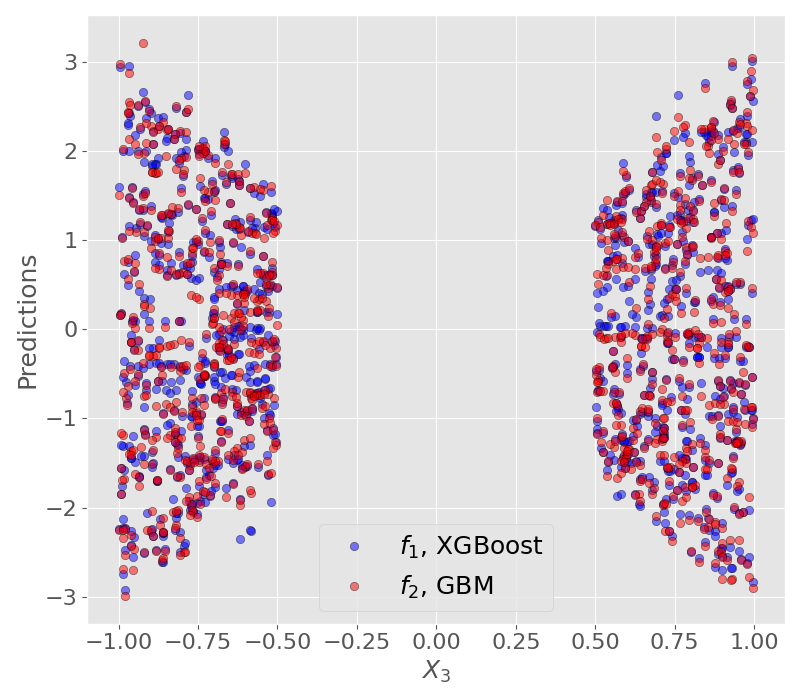}\caption{\footnotesize Predictions vs $X_3$.}\label{fig::ex1_pred_vs_X3}
  \end{subfigure}     
  \begin{subfigure}[t]{0.3\textwidth}
    \centering
    \includegraphics[width=\textwidth]{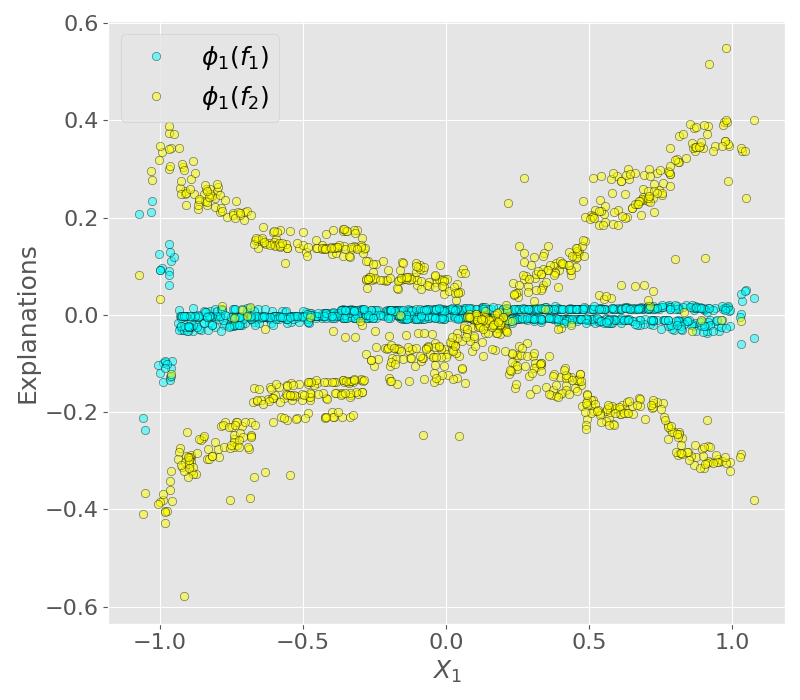}\caption{\footnotesize Explanations $\varphi_1$ vs $X_1$.}\label{fig::ex1_explan_vs_X1}
  \end{subfigure}
  \begin{subfigure}[t]{0.3\textwidth}
    \centering
    \includegraphics[width=\textwidth]{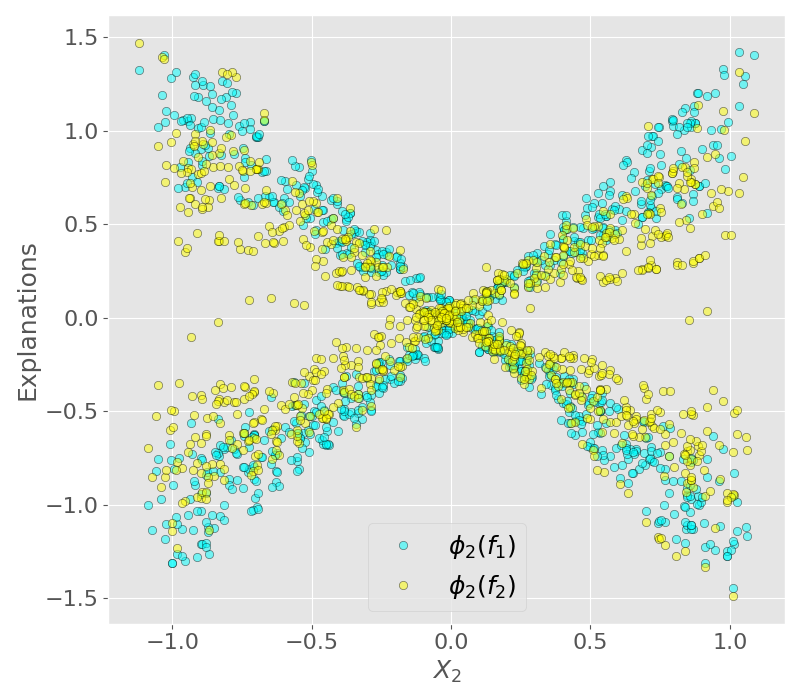}\caption{\footnotesize Explanatons $\varphi_2$ vs $X_2$.}\label{fig::ex1_explan_vs_X2}
  \end{subfigure}
  \begin{subfigure}[t]{0.3\textwidth}
    \centering
    \includegraphics[width=\textwidth]{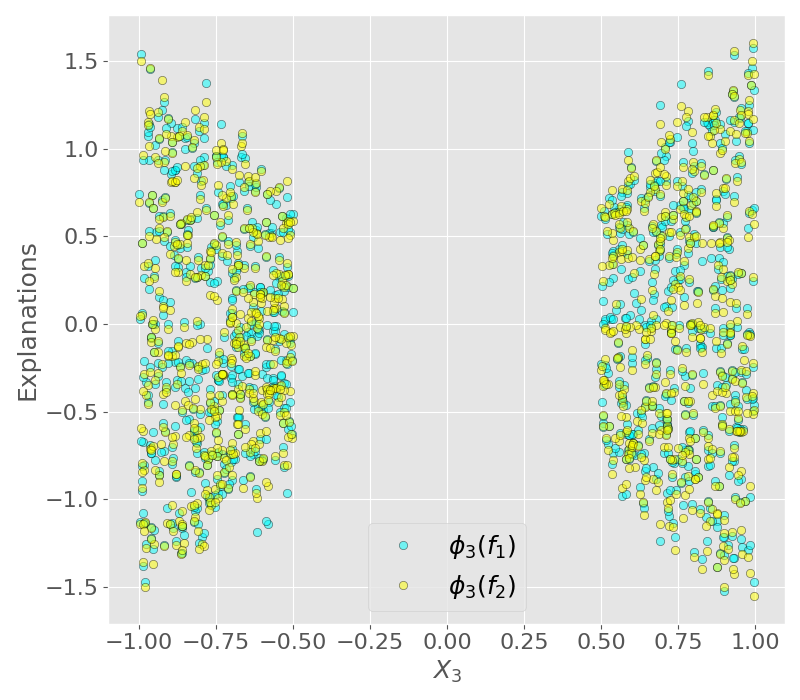}\caption{\footnotesize Explanatons $\varphi_3$ vs $X_3$.}\label{fig::ex1_explan_vs_X3}
  \end{subfigure}
  \caption{ Predictions and marginal Shapley values for XGBoost ($f_1$) and GBM ($f_2$) models.}\label{fig::ex1_pred_explan} 
\end{figure}

\subsubsection{Case 1: Two models trained on the same dataset.}

In our first experiment, we set $\delta=0.05$ and use a training dataset $D_X^{(train)}$ with $25000$ samples drawn from the data generating model \eqref{ex1_resp_model_} and train two regressors $f_1$ and $f_2$ using the XGBoost and Gradient Boosting (GBM) machine learning algorithms, respectively. The GBM model was trained using the following parameters: {n\_estimators}=$500$, {min\_samples\_split}=$5$, {subsample}=$1.0$, {learning\_rate}=$0.1$. For the XGBoost model we used: {n\_estimators}=$300$, {max\_depth}=$5$, {subsample}=$1.0$, {learning\_rate}=$0.1$, {alpha}=$10$, {lambda}=$10$.

Performance metrics for the two models on the training and test datasets, both having 25000 samples, were evaluated. Specifically, the relative $L^2$-errors are approximately $0.058$ and $0.038$, respectively, with the model norms satisfying $\|f_1\|_{L^2(P_X)} \approx 1.379$, $\|f_2\|_{L^2(P_X)} \approx 1.378$. The estimated relative $L^2$-difference between the two models is $0.065$.

We next pick $m=1000$ samples at random from the training dataset, constructing the dataset $D^{(e)}_X$ of predictor observations used for explanation, and use the two regressors to predict the response variable. Figures \ref{fig::ex1_pred_vs_X1}-\ref{fig::ex1_pred_vs_X2} depict the predicted values for each model versus predictors $X_1$ and $X_2$, respectively, where we see that both trained models have similar predictions.

We then evaluate the marginal explanations of each predictor for the two models, along each sample $x \in D^{(e)}_X$. To accomplish this, we make use of the empirical game $\hat{v}^{\ME}(\cdot,x; f_k,\bar{D}_X)$ defined in \eqref{empmarggame} with a background dataset $\bar{D}_X$ used for averaging, which is constructed by randomly drawing $1000$ samples from the training dataset. Specifically, we compute the Shapley values $\varphi_i[N,\hat{v}^{\ME}](x)$, where $i \in N=\{1,2,3\}$, for each observation $x \in D^{(e)}_X$ and each model $f_k$, $k \in \{1,2\}$. Figures \ref{fig::ex1_explan_vs_X1}-\ref{fig::ex1_explan_vs_X3} depict the distribution of the marginal explanations for predictors $X_1, X_2, X_3$ for each model, across the dataset $D^{(e)}_X$, where we see that the XGBoost model, due to regularization, treats the first predictor as a dummy variable, while the representation of the GBM model relies heavily on the predictor $X_1$.

\begin{figure}
  \centering
    \begin{subfigure}[t]{0.35\textwidth}
    \centering
    \includegraphics[width=\textwidth]{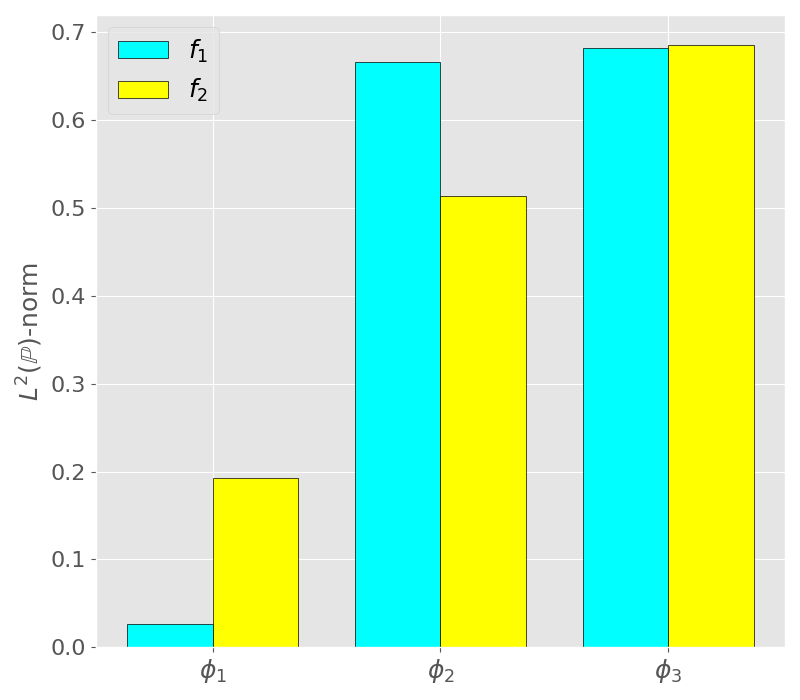} \caption{\footnotesize Explanation norms.}\label{fig::ex1_norm_expl}
  \end{subfigure}  
  ~~
  \begin{subfigure}[t]{0.35\textwidth}
    \centering
    \includegraphics[width=\textwidth]{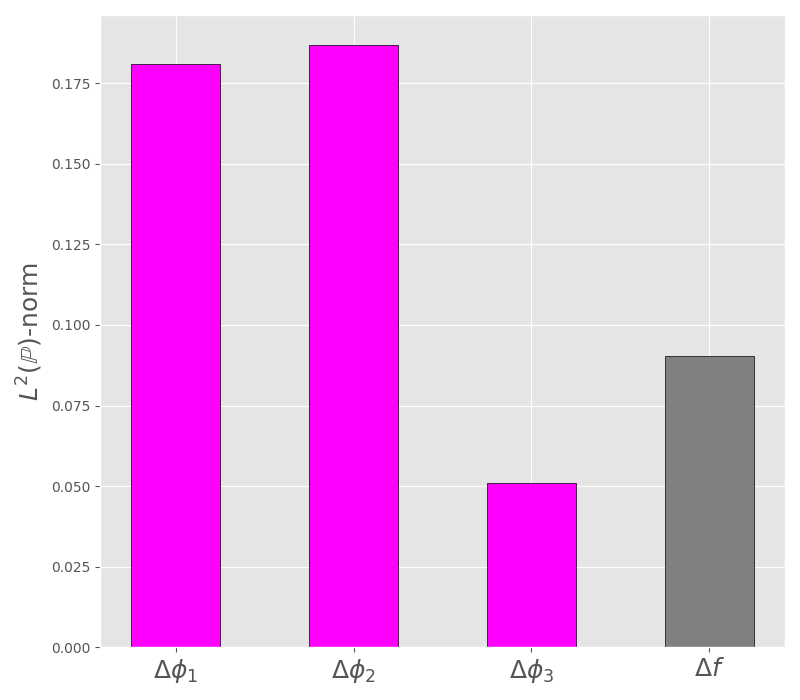} \caption{\footnotesize Explanation norms of $f_1-f_2$.}\label{fig::ex1_quant_expl_indiv}
  \end{subfigure}
~~
  \begin{subfigure}[t]{0.35\textwidth}
    \centering
    \includegraphics[width=\textwidth]{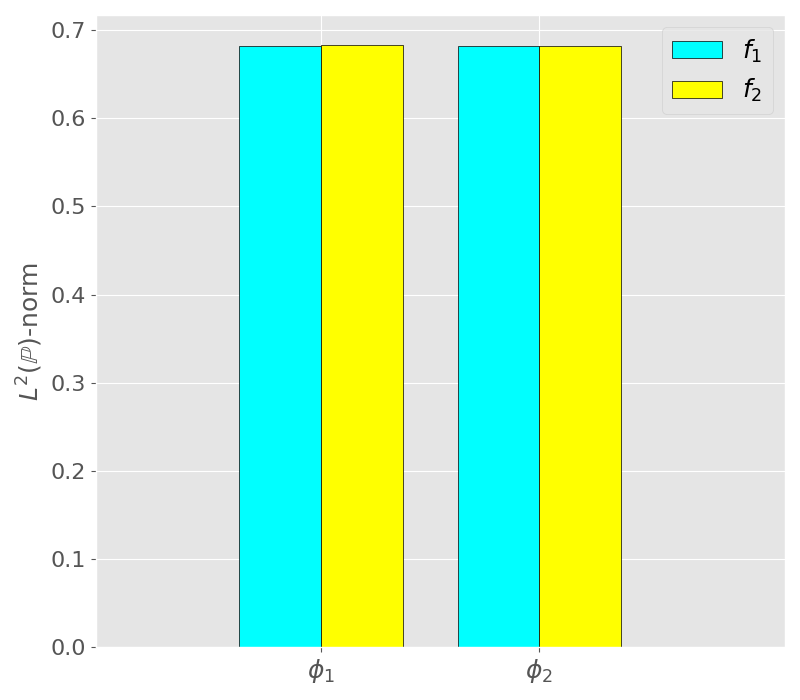} \caption{\footnotesize Quotient explanation norms.}\label{fig::ex1_quant_expl_group}
  \end{subfigure}
  ~~
  \begin{subfigure}[t]{0.35\textwidth}
    \centering
    \includegraphics[width=\textwidth]{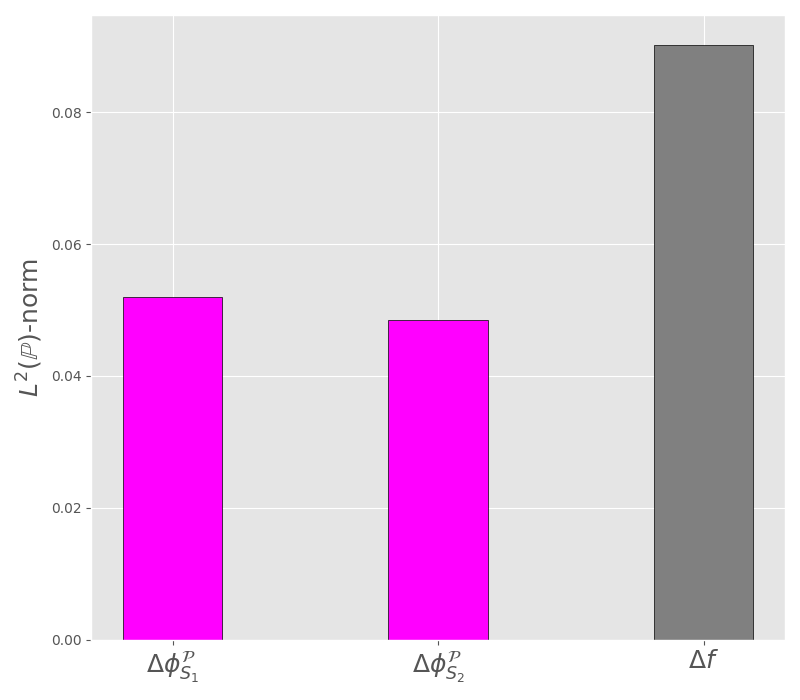} \caption{\footnotesize Quotient explanations of $f_1-f_2$.}\label{fig::ex1_quant_expl_group2}
  \end{subfigure}
  \caption{ Global individual and quotient explanations.}\label{fig::ex1_quant_expl}
\end{figure}

\begin{figure}
  \centering
  \begin{subfigure}[t]{0.35\textwidth}
    \centering
      \includegraphics[width=\textwidth]{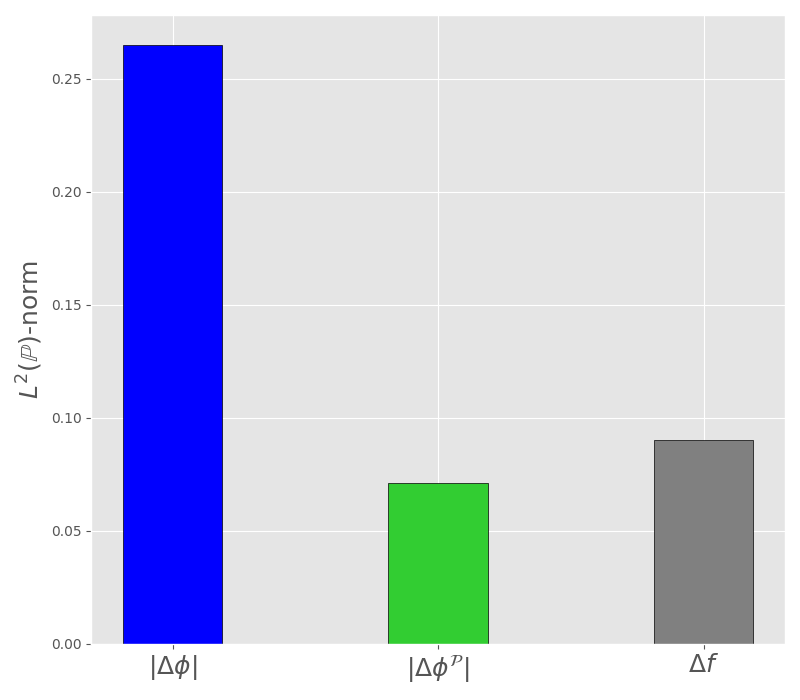} \caption{\footnotesize Total gain in stability.}\label{fig::ex1_quant_expl_gain_stab_tot}
  \end{subfigure}
  ~~
  \begin{subfigure}[t]{0.35\textwidth}
    \centering
    \includegraphics[width=\textwidth]{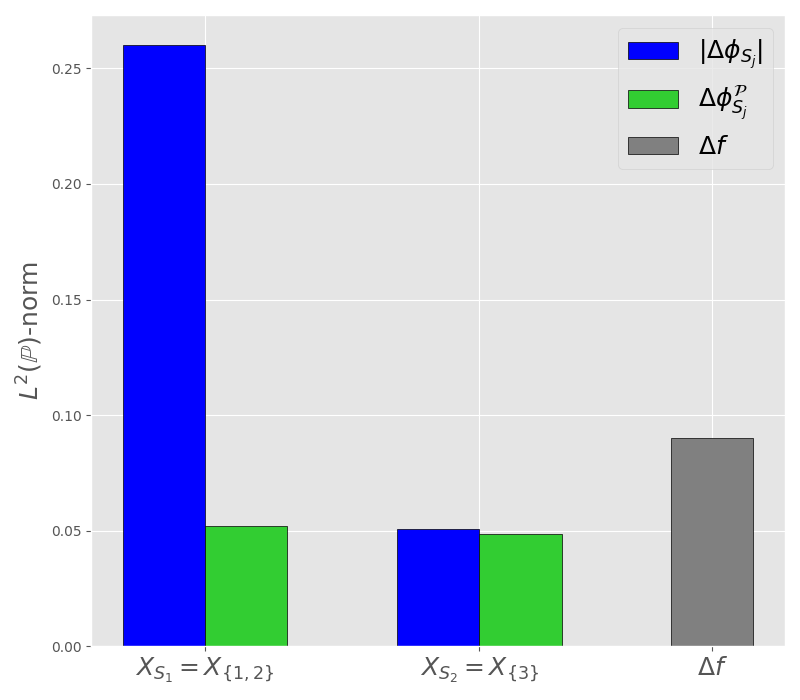} \caption{\footnotesize Gain in stability across groups.}\label{fig::ex1_quant_expl_gain_stab_group}
  \end{subfigure}
  \caption{Gain stability.}\label{fig::ex1_gain_stab}
\end{figure}

To quantify the global attribution of each predictor, we estimate the $L^2$-norms of the marginal Shapley values for each model, $\beta_i(f_k,\vhat^{\ME}):=\|\varphi_i(X;f_k,\vhat^{\ME})\|_{L^2(\P)}$, $i \in N$, which are depicted in Figure \ref{fig::ex1_norm_expl} and recorded in Table \ref{table::exp1_norms}. The values demonstrate that the features $X_1,X_2$ are utilized differently in each model, suggesting that they have different functional representations.

Recall that by Corollary \ref{corr::cond_operator_cont}$(ii)$ (due to the efficiency property of $\varphi$) the conditional Shapley operator is a linear, bounded operator with norm bounded by one and, hence, the conditional Shapley value satisfies $|\beta(f_1-f_2,\vce)|/\|f_1-f_2\|_{L^2(P_X)} \leq 1$, where $\beta:=(\beta_1,\beta_2,\beta_3)$. This bound ensures that the total distance $|\beta(f_1-f_2,\vce)|$ between these explanations is always smaller than the $L^2(P_X)$-distance between the models, and the same is true for any component and sub-vector of the vector $\beta(f_1-f_2,\vce)$. Meanwhile, in theory, in the presence of dependencies, the bound for the marginal explanations may in general be infinite or significantly larger than one, which depends on the relationship between $P_X$ and $\intP_X$.

To understand the degree of the instability in marginal explanations, we estimate the norm of the difference of the marginal Shapley values for the two models. Given the linearity of the marginal operator, this norm is equal to the norm of the Shapley values for the model difference $\beta_i(f_1-f_2,\vhat^{\ME})=\|\varphi_i(X;f_1-f_2,\hat{v}^{\ME})\|_{L^2(\P)}$, whose estimate is given by
\begin{equation}\label{ex1_shap_diff_}
\beta(f_1-f_2,\vhat^{\ME})=(\beta_1,\beta_2,\beta_3)(f_1-f_2,\vhat^{\ME})\approx (0.181, 0.186, 0.051), \,\, \|f_1-f_2\|_{L^2(P_X)} \approx 0.09.
\end{equation}
Observe that the total distance between marginal explanations $|\beta(f_1-f_2,\vhat^{\ME})|_2=0.265$ is approximately three times larger than the $L^2(P_X)$-distance between models. Moreover, the distances between explanations for features $X_1$ and $X_2$ are also approximately twice that of the models; see Figure \ref{fig::ex1_quant_expl_indiv}.  We also note that the total differences between explanations is significant and constitutes about 20\% of the train models' norm; for comparison see Table \ref{table::exp1_norms}.

To understand the effect of grouping, we construct quotient marginal explanations of the trained models for each sample $x \in D^{(e)}_X$. To accomplish this, we employ the empirical quotient marginal game and generate explanations corresponding to the partition $\cP$ based on dependencies, given by $\mathcal{P}=\{\{1,2\},\{3\}\} = \{S_1,S_2\}$. This is done by evaluating quotient Shapley values $\varphi_j[M,\hat{v}^{\ME,\cP}](x)$, $j \in M=\{1,2\}$, for each observation $x \in D^{(e)}_X$ and each model $f_k$, $k \in \{1,2\}$. 

We then use these explanations to quantify the global attribution of predictor groups by estimating the norms $\beta^{\cP}_j(f_k,\vhat^{\ME}):=\|\varphi^{\cP}_{S_j}(X;f_k,\vhat^{\ME})\|_{L^2(\P)}$, $j \in M$, which are depicted in Figure \ref{fig::ex1_quant_expl_group} and recorded in Table \ref{table::exp1_norms}. These values indicate that grouping by dependencies yields (on average) similar group explanations regardless of the functional representation.

To assess the level of the instabilities in quotient explanations,  we quantify the difference between quotient explanations and compare it with that of between the models. The $L^2$-distance between marginal quotient explanations is given by $\beta^{\cP}(f_1-f_2,\vhat^{\ME})=(\beta^{\cP}_1,\beta^{\cP}_2)(f_1-f_2,\vhat^{\ME})\approx (0.052, 0.049)$. Figure \ref{fig::ex1_quant_expl_group2} compares these distances with those of the models given in \eqref{ex1_shap_diff_}. As a result of grouping by dependencies, these distances are approximately twice smaller than the distances between the models, unlike the global attributions $\beta_i(f_1-f_2,\vhat^{\ME})$ of individual explanations; for comparison see Table \ref{table::exp1_norms}. Moreover, the total distance between quotient marginal explanations $|\beta^{\cP}(f_1-f_2,\vhat^{\ME})|=0.071$  is strictly smaller than the $L^2(P_X)$-distance between models. Thus, due to grouping, the unit bound in \eqref{effvalbound} is satisfied leading to increased stability in $L^2(P_X)$. Finally, the splitting of explanations across dependent predictors does not occur anymore, as was seen in the GBM model.

To estimate the gain in stability due to grouping, we introduce a method that will be useful when dealing with large datasets and where dependencies are not that obvious. Recall that Corollary \ref{corr::cond_operator_cont}$(ii)$   implies $|\beta(f_1-f_2,\vce)| \leq \|f_1-f_2\|$ while Proposition \ref{lmm::boundquot} implies $|\beta^{\cP}(f_1-f_2,\vce)|\leq\|f_1-f_2\|_{L^2(P_X)}$. Thus, to quantify the total gain in stability for marginal explanations across all features we propose to compare $|\beta(f_1-f_2,\vhat^{\ME})|$ and $|\beta^{\cP}(f_1-f_2,\vhat^{\ME})|$ with $\|f_1-f_2\|_{L^2(P_X)}$, which is accomplished in Figure \ref{fig::ex1_quant_expl_gain_stab_tot}. To quantify the gain in stability across each group in $\cP=\{S_1,S_2\}$, we compare the norm of the sub-vector $\beta_{S_j}(f_1-f_2,\vhat^{\ME})$, measuring the aggregated difference across the group $S_j$, with that of $\beta^{\cP}_j(f_1-f_2,\vhat^{\ME})$.  Figure \ref{fig::ex1_quant_expl_gain_stab_group} illustrates that the differences in aggregated explanations drop significantly after grouping, which showcases the gain in stability. It also illustrates that the unit bound is not satisfied for the aggregated explanations, while it is for the quotient ones.

\begin{table}
\begin{center}
\begin{tabular}{|c| c | c | c | c | c | c | c | c |}
 \hline
  & $\|\cdot\|$  & $\beta_1$ & $\beta_2$ & $\beta_3$ & $|\beta|$ & $\beta_1^{\cP}$ & $\beta_2^{\cP}$ & $|\beta^{\cP}|$\\[0.5ex]
 \hline
   $f_1$ & {\color{blue}$1.379$} & $0.026$ & $0.667$ & $0.683$ & {\color{blue}$0.955$}  & $0.682$ & $0.681$ &  {\color{blue}$0.964$} \\ 
 \hline 
 $f_2$ & {\color{blue}$1.378$} & $0.192$ & $0.514$ & $0.685$ & {\color{blue}$0.878$} & $0.683$ & $0.682$  & {\color{blue}$0.965$} \\ 
 \hline
 $f_1-f_2$ & {\color{blue}$0.090$} & $0.181$ & $0.187$ & $0.051$ & {\color{blue}$0.265$} & $0.052$ & $0.049$ & {\color{blue}$0.071$}  \\ 
 \hline
\end{tabular}
\caption{Global marginal Shapley attributions.}\label{table::exp1_norms}
\end{center}
\end{table}

%%%%%%%%%%%%%%%%%%%%%%%%%%%%%%%%%%%%%%%%%%%%%%%%%%%%%%%%%%%%%%%%%%%%%%%%%%%%%%%
\subsubsection{Case 2: Models on perturbed datasets}
%%%%%%%%%%%%%%%%%%%%%%%%%%%%%%%%%%%%%%%%%%%%%%%%%%%%%%%%%%%%%%%%%%%%%%%%%%%%%%%

In our second experiment, we construct five distinct datasets by varying the level of noise in the predictors from the previous subsection, and train five corresponding ML models. We then construct a test dataset as a mixture of the five training sets and use its observations for both explanations and averaging. This experiment demonstrates that the models with similar predictive power on the test dataset, which in turn is close in distribution to the training sets, have widely different explanations. It also illustrates how grouping features based on dependencies rectifies the explanation instabilities. The details of the experiment are provided below.

First, for each $\delta \in \{ \delta_i \}_{i=1}^5=\{0.0, 0.001, 0.0025, 0.005, 0.01\}$, which represents the noise level in predictors, we construct a corresponding dataset $D(\delta)=\{(x_{\delta}^{(k)},y_{\delta}^{(k)})\}_{k=1}^K$, containing $K=25000$ observations sampled from the distribution $(X_{\delta},Y_{\delta})$ where $X_{\delta}$ is given by \eqref{ex1_preds_} with  noise $\epsilon_1,\epsilon_2 \sim \mathcal{N}(0,\delta)$, and $Y_{\delta}$ is constructed using the response model \eqref{ex1_resp_model_}. Then for each $i \in \{1,\dots, 5\}$ an XGBoost regressor $f_i(x)$ is trained on the dataset $D(\delta_i)$, utilizing the same hyperparameters as in the previous experiment.

To compare the explanations of these models, a test dataset $D=\{(x^{(k)},y^{(k)})\}_{k=1}^K$ is constructed by drawing  $K=25000$ samples from the distribution $(X,Y)$ such that $X=\sum_{i=1}^5 1_{\{C=i\}} \cdot X_{\delta_i}$ is a mixture, where $C$ is a random variable satisfying $\P(C=i)=0.2$, and $Y$ is obtained using the response model \eqref{ex1_resp_model_}.

Performance metrics for the XGBoost models on the mixture dataset were evaluated. Specifically, the relative $L^2$-errors for the five models are approximately $0.051$, $0.045$, $0.041$, $0.052$ and $0.046$, respectively, with the norms $\|f_i\|_{L^2(P_{X})}$  of the models recorded in Table \ref{table::exp1b_norms}, which illustrates that all trained models have similar predictive power on the test set.

We next evaluate the $L^2(P_{X})$-distance between the true model $f_*$ and each trained model $f_k$, $k\in\{1,\dots,5\}$. The estimated values of the distances are given by
\begin{equation}\label{ex1_pred_diff_mix_}
\|f_k-f_* \|_{L^2(P_{X})}\approx (0.069,0.062,0.056,0.071,0.064), \quad \|f_* \|_{L^2(P_X)}\approx 1.37,
\end{equation}
and also recorded in Table \ref{table::exp1b_norms}. Thus, the predictions of the trained models on the mixture dataset are close in an $L^2$-sense to those of $f_*$. In particular, this implies that $\{f_k\}_{k=1}^5$ live in an $(L^2,\epsilon)$-Rashomon set of models  about $f_*$ (defined in \S\ref{sec::notation}) with $\epsilon=0.071$, which constitutes about 5\% relative $L^2$-distance.

\begin{figure}
  \centering
  \begin{subfigure}[t]{0.4\textwidth}
    \centering
    \includegraphics[width=\textwidth]{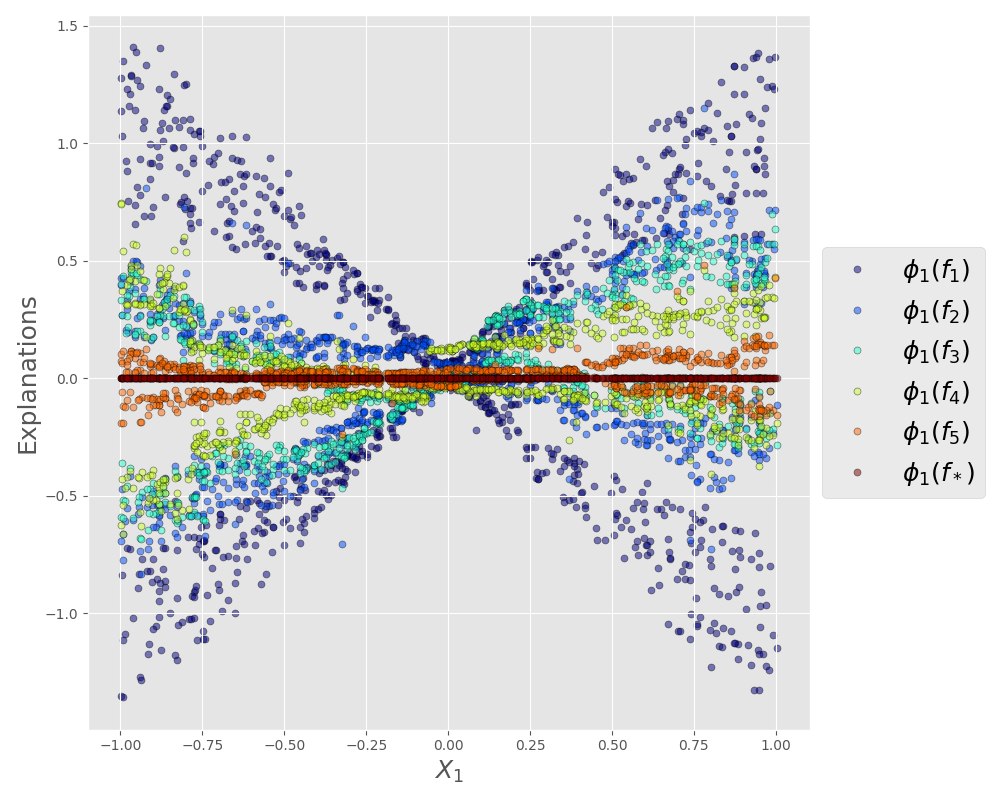}\caption{\footnotesize Explanations $\varphi_1$ vs $X_1$.}\label{fig::ex1b_explan_vs_X1}
  \end{subfigure}
  ~~
  \begin{subfigure}[t]{0.4\textwidth}
    \centering
    \includegraphics[width=\textwidth]{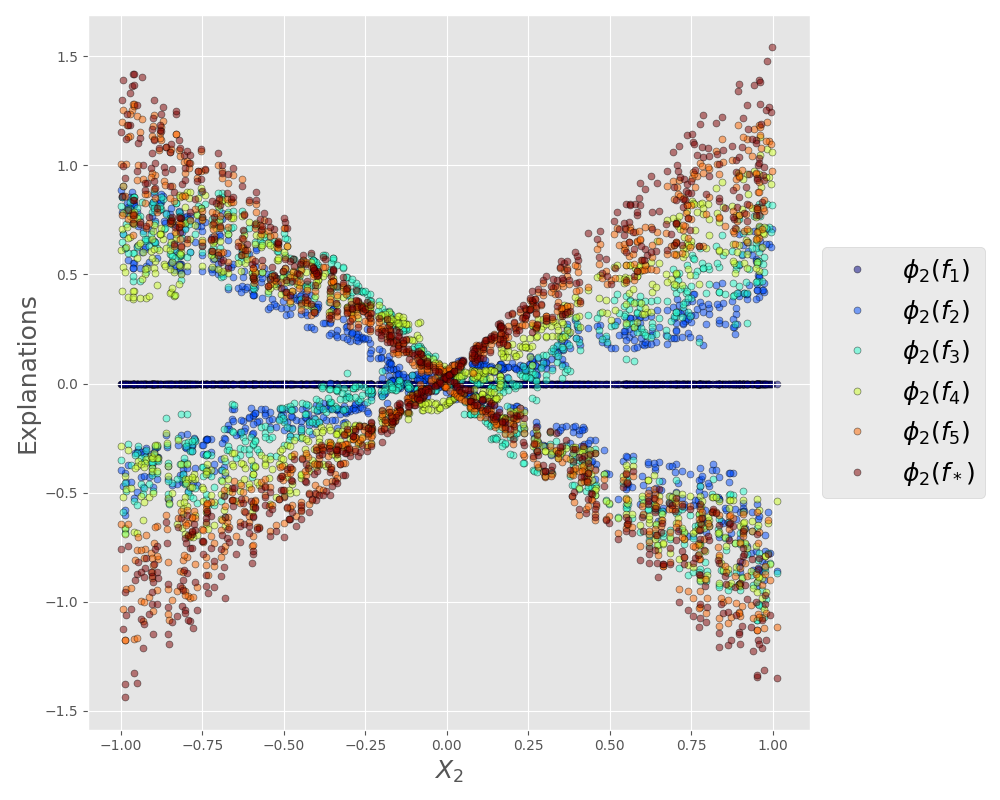}\caption{\footnotesize Explanations $\varphi_2$ vs $X_2$.}\label{fig::ex1b_explan_vs_X2}
   \end{subfigure}
  ~~\\
\begin{subfigure}[t]{0.4\textwidth}
    \centering
    \includegraphics[width=\textwidth]{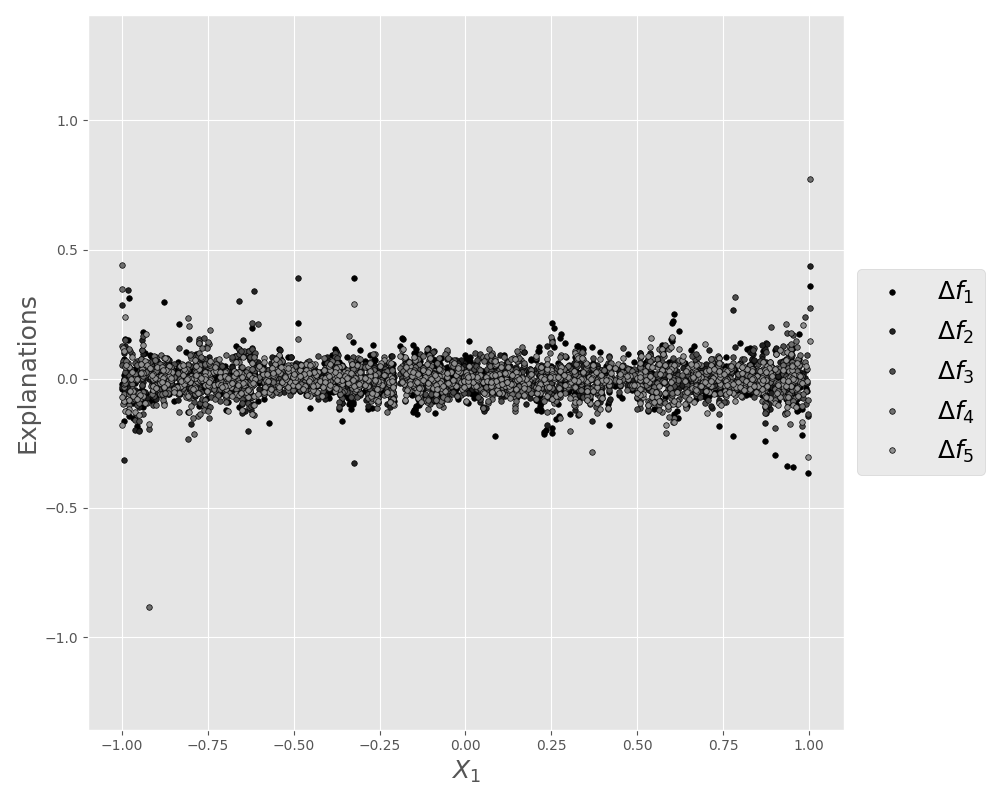}\caption{\footnotesize Differences of predictions vs $X_1$.}\label{fig::ex1b_diff_explan_vs_X1}
  \end{subfigure}
  ~~
  \begin{subfigure}[t]{0.4\textwidth}
    \centering
    \includegraphics[width=\textwidth]{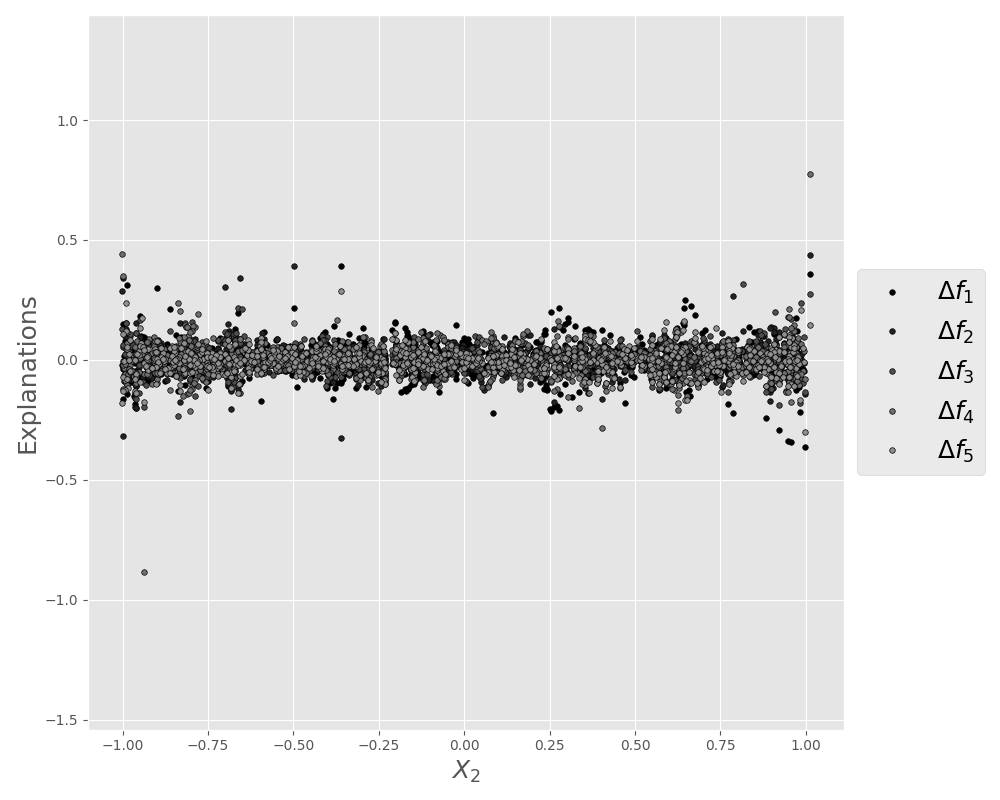}\caption{\footnotesize Differences of predictions vs $X_2$.}\label{fig::ex1b_diff_explan_vs_X2}
  \end{subfigure}
  \begin{subfigure}[t]{0.4\textwidth}
    \centering
    \includegraphics[width=\textwidth]{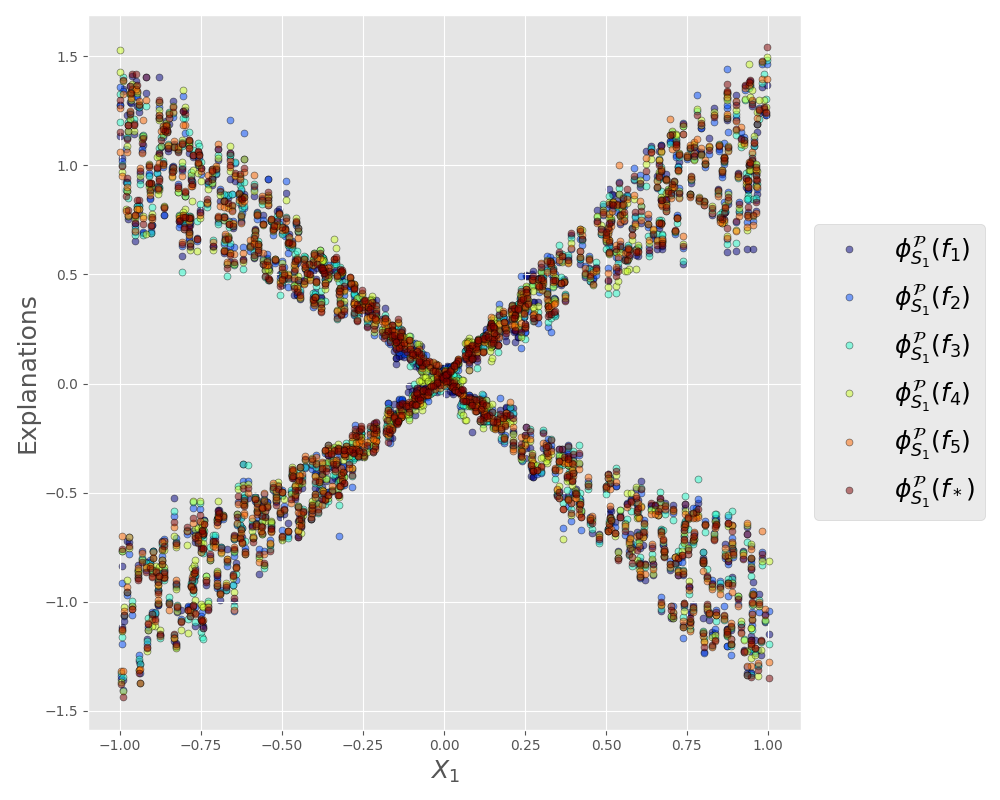}\caption{\footnotesize Explanations $\varphi_{S_1}^{\cP}$ vs $X_1$.}\label{fig::ex1b_quot_explan_vs_X1}
  \end{subfigure}
  ~~
  \begin{subfigure}[t]{0.4\textwidth}
    \centering
    \includegraphics[width=\textwidth]{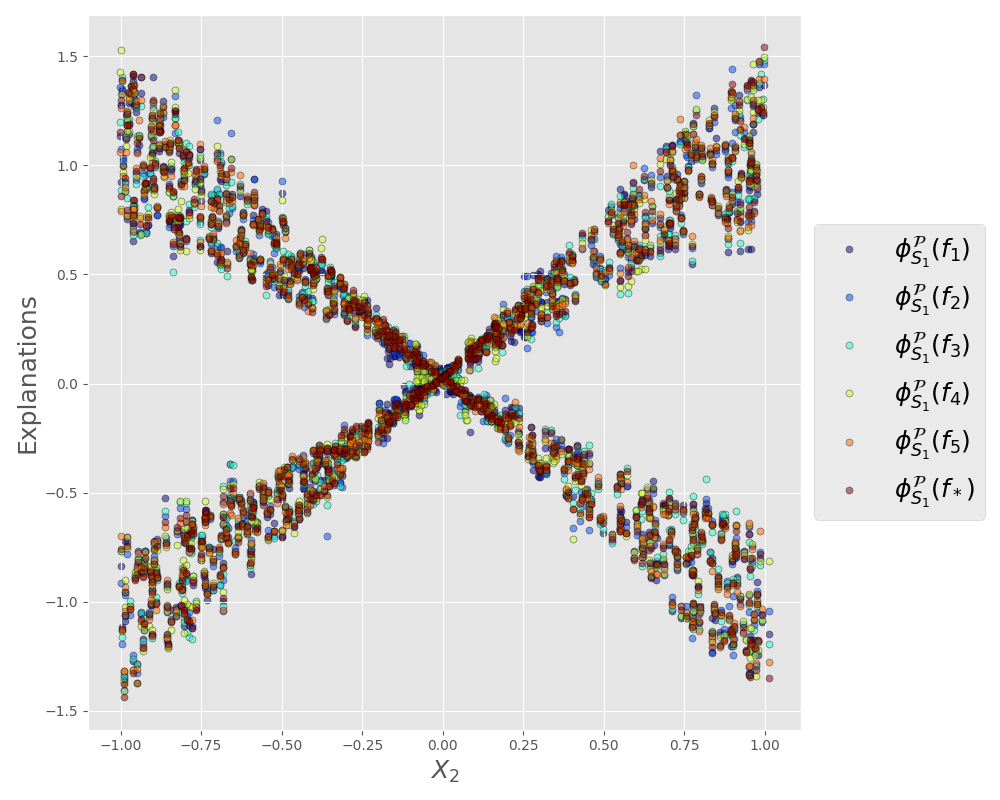}\caption{\footnotesize Explanations $\varphi_{S_2}^{\cP}$ vs $X_2$.}\label{fig::ex1b_quot_explan_vs_X2}
  \end{subfigure}
\caption{\footnotesize Individual and quotient marginal explanations.} \label{fig::ex1b_expl_mix}
\end{figure}

\begin{figure}
  \centering
    \begin{subfigure}[t]{0.4\textwidth}
    \centering
    \includegraphics[width=\textwidth]{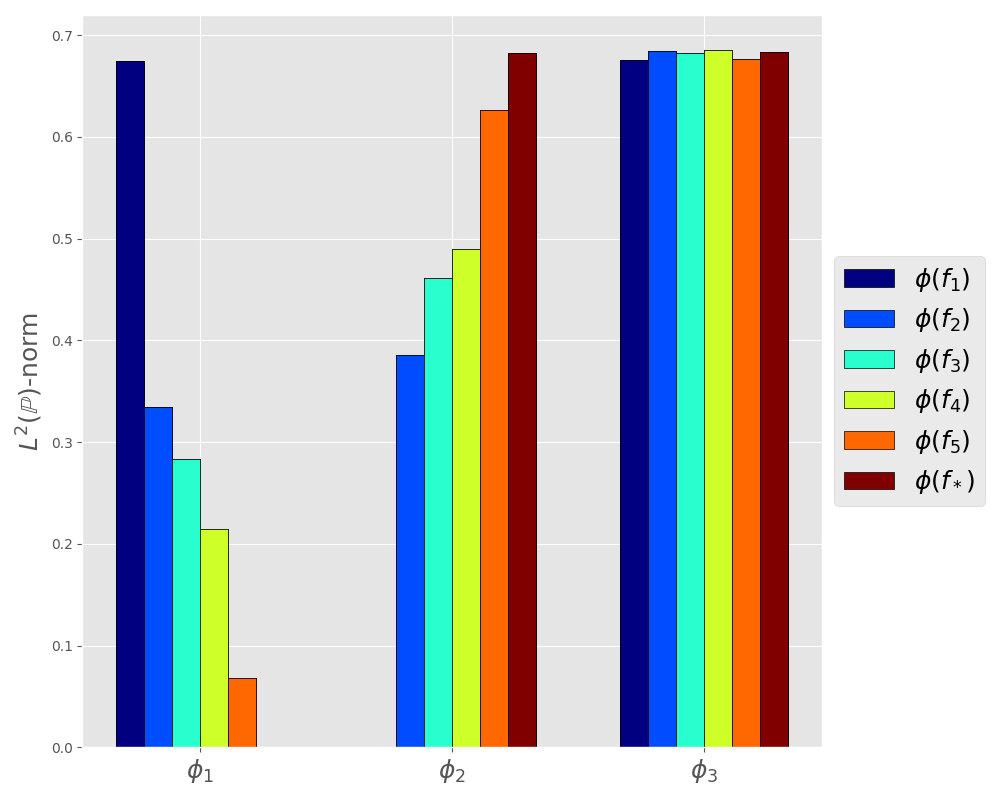} \caption{\footnotesize Explanation norms.}\label{fig::ex1b_norm_expl_mix}
  \end{subfigure}  
~~
  \begin{subfigure}[t]{0.4\textwidth}
    \centering
    \includegraphics[width=\textwidth]{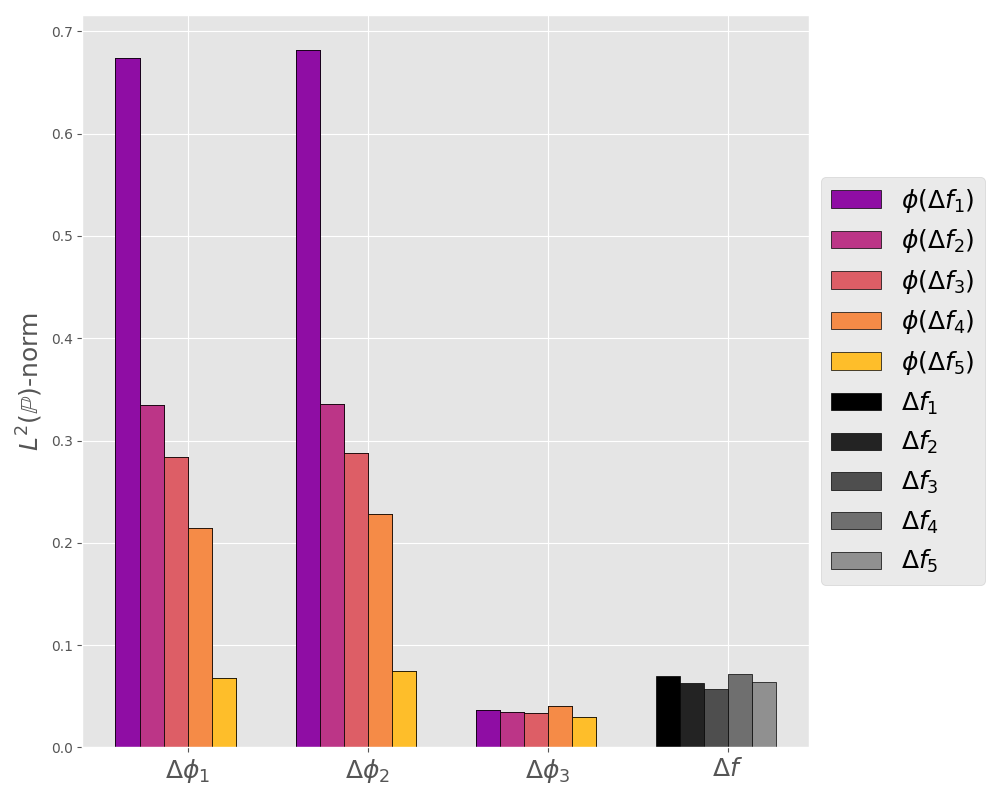} \caption{\footnotesize Global explanations of $\Delta f_i$.}\label{fig::ex1b_quant_expl_indiv_mix}
  \end{subfigure}
~~\\
  \begin{subfigure}[t]{0.4\textwidth}
    \centering
    \includegraphics[width=\textwidth]{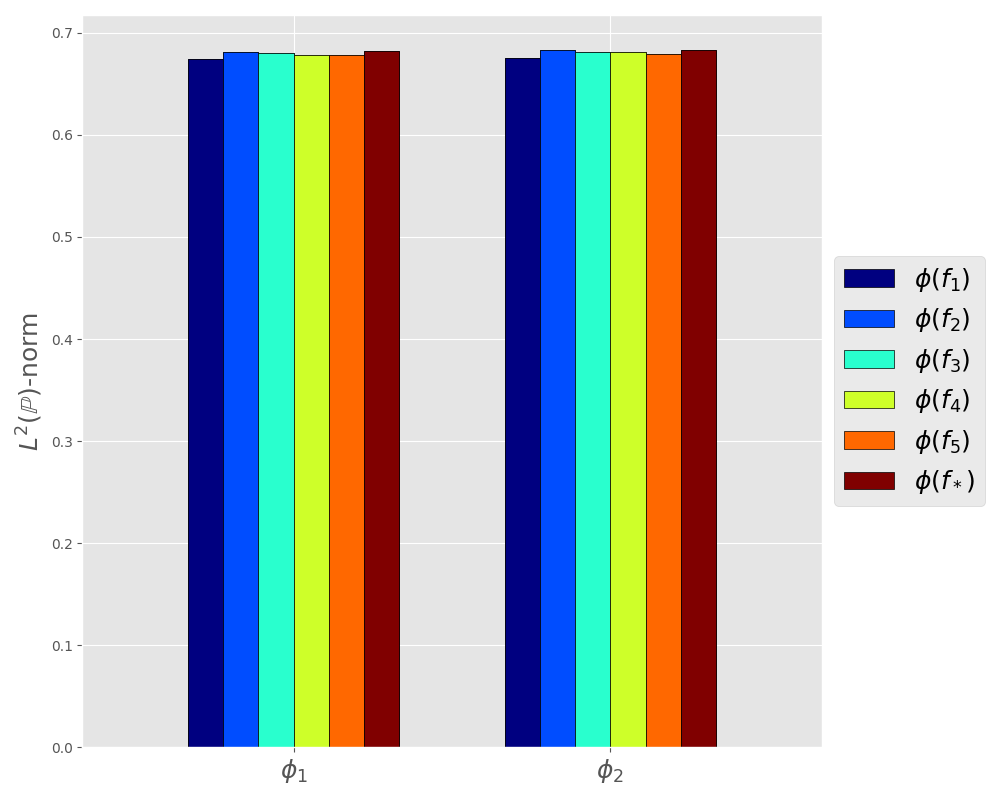} \caption{\footnotesize Quotient explanation norms.}\label{fig::ex1b_quant_expl_group}
  \end{subfigure}
~~
  \begin{subfigure}[t]{0.4\textwidth}
    \centering
    \includegraphics[width=\textwidth]{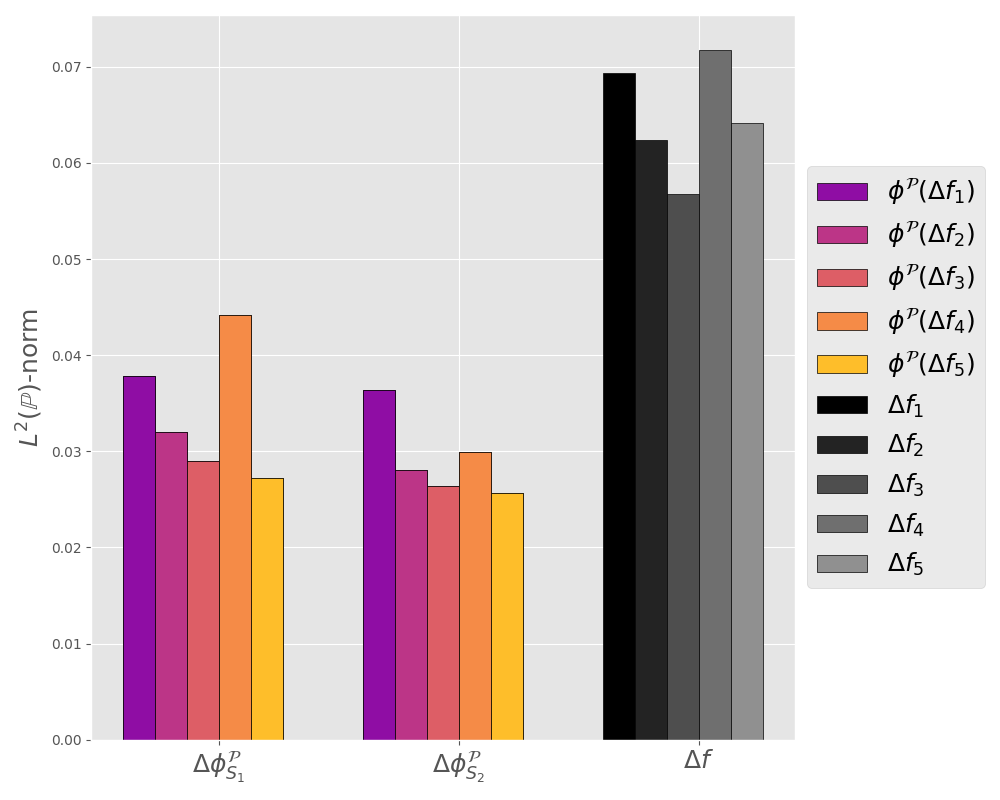} \caption{\footnotesize Global quotient explanations of $\Delta f_i$.}\label{fig::ex1b_quant_expl_group2}
  \end{subfigure}
  \caption{ Individual and quotient explanation norms.}\label{fig::ex1b_quant_expl}
\end{figure}

We next pick $m=1000$ samples at random from the mixture dataset, to construct the dataset  $D^{(e)}_{X}$ of predictor observations used for explanations. 
We also subsample the predictors from the mixture set and obtain a background dataset $\bar{D}_{X}$ with $1000$ samples. The background dataset is used for construction of the empirical marginal game $\hat{v}^{\ME}(S;x,f, \bar{D}_X)$ defined in \eqref{empmarggame} where {\color{blue}$x \in D^{(e)}_{X}$} is an observation and $f \in \{f_*,f_1,\dots f_5\}$.

We then evaluate the empirical marginal explanations $\varphi_i[N,\hat{v}^{\ME}](x)$ for each observation $x \in D_X^{(e)}$ and each predictor across the six models, the true model and the five XGBoost models. The computations are done by means of the interventional TreeSHAP method \cite{LundbergIntTreeShap}, which computes empirical marginal Shapley values for tree-based models such as XGBoost.

Figures \ref{fig::ex1b_explan_vs_X1}-\ref{fig::ex1b_diff_explan_vs_X2} depict the scatterplots of explanations (and their differences) for each model across the dataset $D^{(e)}_{X}$, where we see that explanations differ substantially, indicating that the trained models have different functional representations. In particular,  $f_1$ treats both predictors $X_1,X_2$ similarly due to the strong dependence between them, while $f_5$ treats the first predictor as a dummy variable which is similar to the model $f_*$.

\begin{figure}
  \centering
  \begin{subfigure}[t]{0.4\textwidth}
    \centering
    \includegraphics[width=\textwidth]{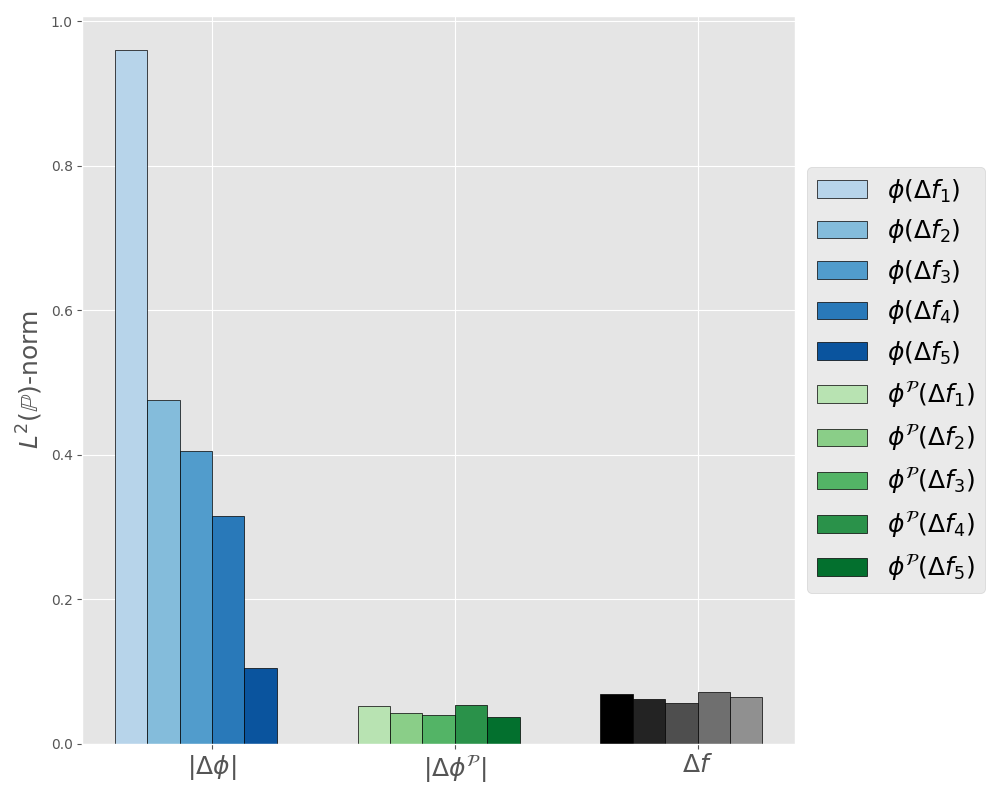} \caption{\footnotesize Total gain in stability.}\label{fig::ex1b_quant_expl_totgain_stab_abs_mix}
  \end{subfigure}  
  ~~
  \begin{subfigure}[t]{0.4\textwidth}
    \centering
    \includegraphics[width=\textwidth]{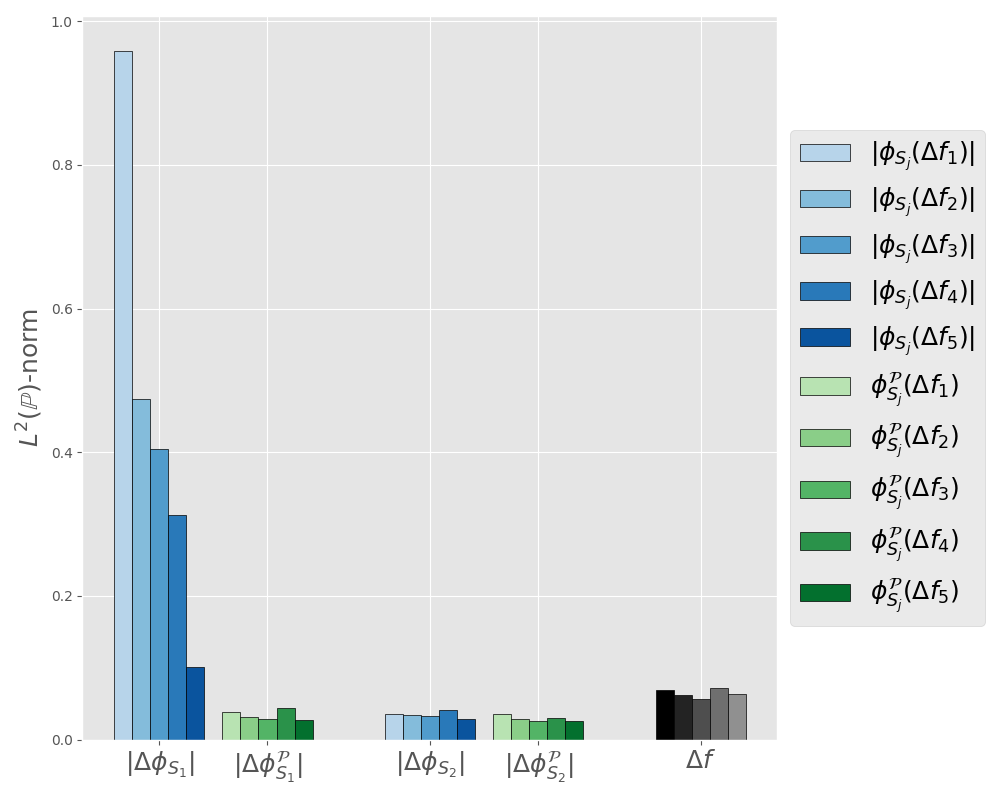} \caption{\footnotesize Gain in stability.}\label{fig::ex1b_quant_expl_gain_stab_abs_mix}
  \end{subfigure}
  \caption{Explanation norms and effect of grouping.}\label{fig::ex1b_gain_stab_mix}
\end{figure}

Recall that by Corollary \ref{corr::cond_operator_cont}$(ii)$ (due to the efficiency property of $\varphi$) the conditional Shapley operator is a linear, bounded operator with norm bounded by one and, hence, the conditional Shapley value satisfies $|\beta(f_1-f_2,\vce)|/\|f_1-f_2\|_{L^2(P_X)} \leq 1$, where $\beta:=(\beta_1,\beta_2,\beta_3)$. This bound ensures that the total distance $|\beta(f_1-f_2,\vce)|$ between these explanations is always smaller than the $L^2(P_X)$-distance between the models, and the same is true for any component and sub-vector of the vector $\beta(f_1-f_2,\vce)$. Meanwhile, in theory, in the presence of dependencies, the bound for the marginal explanations may in general be infinite or significantly larger than one, which depends on the relationship between $P_X$ and $\intP_X$.

 To quantify the global attribution of each predictor, we estimate the $L^2$-norms of the marginal Shapley values for each model, $\beta_i(f_k,\vhat^{\ME})=\|\varphi_i(X;f_k,\vhat^{\ME})\|_{L^2(\P)}$, $i \in N$, which are depicted in Figure \ref{fig::ex1b_norm_expl_mix} and recorded in Table \ref{table::exp1b_norms}. These values also demonstrate that the features $X_1,X_2$ are utilized differently across the models.

To understand the degree of the instability in marginal explanations, we estimate the distance between the marginal Shapley values of the reference model $f_*$ and $f_k$ for every $k \in \{1,\dots,5\}$ and each predictor $X_i, i\in\{1,2,3\}$, which are equal to the norm of the Shapley values for the model difference $\beta_i(f_i-f_*,\vhat^{\ME})=\|\varphi_i(X;f_i-f_*,\hat{v}^{\ME})\|_{L^2(\PP)}$, and then compare with those of the model. Figure \ref{fig::ex1b_quant_expl_indiv_mix}, where $f_k - f_*$ is denoted as $\Delta f_k$, showcases the comparison between the distances of the individual feature explanations and the model distances, again for each trained model.

\begin{table}
\begin{center}
\begin{tabular}{|c| c | c | c | c | c | c | c | c |}
 \hline
  & $\|\cdot\|$  & $\beta_1$ & $\beta_2$ & $\beta_3$ & $|\beta|$ & $\beta_1^{\cP}$ & $\beta_2^{\cP}$ & $|\beta^{\cP}|$\\[0.5ex]
 \hline
   $f_1$ & {\color{blue}$1.370$}  & $0.674$ & $0.000$  & $0.675$ & 
           {\color{blue}$0.954$}  & $0.674$ & $0.675$  & {\color{blue}$0.954$} \\ 
 \hline 
 $f_2$   & {\color{blue}$1.375$}  & $0.334$ & $0.385$  & $0.685$ & 
           {\color{blue}$0.854$}  & $0.681$ & $0.683$  & {\color{blue}$0.964$} \\ 
 \hline 
 $f_3$   & {\color{blue}$1.374$}  & $0.285$ & $0.461$  & $0.682$ & 
           {\color{blue}$0.871$}  & $0.680$ & $0.681$  & {\color{blue}$0.962$} \\ 
 \hline 
 $f_4$   & {\color{blue}$1.375$}  & $0.214$ & $0.228$ & $0.040$ & 
           {\color{blue}$0.315$}  & $0.679$ & $0.681$ & {\color{blue}$0.962$} \\ 
 \hline 
 $f_5$   & {\color{blue}$1.374$}  & $0.068$ & $0.627$  & $0.677$ & 
           {\color{blue}$0.925$}  & $0.678$ & $0.680$  & {\color{blue}$0.960$} \\ 
 \hline
 $f_*$   & {\color{blue}$1.380$}  & $0.000$ & $0.682$  & $0.682$ & 
           {\color{blue}$0.965$}  & $0.682$ & $0.683$  & {\color{blue}$0.965$} \\ 
 \hline
 $f_1-f_*$ & {\color{blue}$0.069$} & $0.674$ & $0.682$ & $0.036$ &
             {\color{blue}$0.960$} & $0.038$ & $0.036$ & {\color{blue}$0.052$}  \\ 
 \hline
 $f_2-f_*$ & {\color{blue}$0.062$} & $0.334$ & $0.336$ & $0.035$ & 
             {\color{blue}$0.475$} & $0.032$ & $0.028$ & {\color{blue}$0.042$}  \\ 
 \hline
 $f_3-f_*$ & {\color{blue}$0.056$} & $0.284$ & $0.288$ & $0.033$ & 
             {\color{blue}$0.406$} & $0.029$ & $0.026$ & {\color{blue}$0.039$}  \\ 
 \hline
 $f_4-f_*$ & {\color{blue}$0.071$} & $0.214$ & $0.228$ & $0.041$ & 
             {\color{blue}$0.315$} & $0.044$ & $0.029$ & {\color{blue}$0.053$}  \\ 
 \hline
 $f_5-f_*$ & {\color{blue}$0.064$} & $0.067$ & $0.075$ & $0.029$ & 
             {\color{blue}$0.104$} & $0.027$ & $0.026$ & {\color{blue}$0.037$}  \\ 
 \hline
\end{tabular}
\caption{ Global marginal Shapley attributions.}\label{table::exp1b_norms}
\end{center}
\end{table}
As in the previous experiment, we contrast the unit operator bound in \eqref{effvalbound} for conditional explanations in relation to the change in empirical marginal explanations with respect to the $L^2(P_X)$-distance between models. Specifically, the ratio of the  marginal explanation distance to the distance between models varies from approximately $1$ to $10$; see Figure \ref{fig::ex1b_quant_expl_indiv_mix}. Note that the differences between explanations are significant and for some models constitute about 50\% of the true model's norm. Observe also, that the total distances between the vectors of global marginal explanations satisfy $\{|\beta(f_i-f_*,\vhat^{\ME})|\}_{i=1}^5=\{0.960,0.475,0.406,0.315,0.104\}$ and are approximately two-to-fourteen times larger than the $L^2(P_X)$-distance between models; see Figure \ref{fig::ex1b_quant_expl_gain_stab_abs_mix}. We note that the total distance between explanations is significant and, in particular, for the model $f_1$ it constitutes about 60\% of the trained models' norm; see Table \ref{table::exp1b_norms}.

We next construct the quotient marginal explanations for each model. Figure \ref{fig::ex1b_quot_explan_vs_X1}-\ref{fig::ex1b_quot_explan_vs_X2} depict the scatterplots of quotient explanations for each model across the dataset $D_{X}$, where we see that the explanations between the models are similar.

To quantify the difference between quotient explanations, we estimate $L^2(P_X)$-norms of quotient marginal explanations and the $L^2(P_X)$-distances between marginal quotient explanations for the partition $\cP=\{\{1,2\},\{3\}\}=\{S_1,S_2\}$, denoted by $\beta^{\cP}_j(f_i,\vhat^{\ME})$ and $\beta^{\cP}_j(f_i-f_*,\vhat^{\ME})$, $j \in \{1,2\}$, respectively. Figure \ref{fig::ex1b_quant_expl_group} illustrates the former and Figure \ref{fig::ex1b_quant_expl_group2} compares the latter with distances between the models given in \eqref{ex1_pred_diff_mix_}. As in the previous case, again due to grouping, we see that these distances are approximately twice smaller than the distances between the models compared to individual explanations, showcasing the consistency with the bound for conditional explanations. The contrast between Figures \ref{fig::ex1b_quant_expl_indiv_mix} and \ref{fig::ex1b_quant_expl_group2}, as before, demonstrates that grouping by dependencies reduces the splitting of explanations across dependent predictors.

To quantify the gain in stability due to grouping, we apply the approach outlined in the previous experiment. Specifically, we compare the norm of explanation vectors $|\beta(f_i-f_*,\vhat^{\ME})|$ and $|\beta^{\cP}(f_i-f_*,\vhat^{\ME})|$ to quantify the total gain in stability (across all features simultaneously), which is depicted in Figure \ref{fig::ex1b_quant_expl_totgain_stab_abs_mix}. We also compare the norms of the quotient explanations' differences for each $j\in M$ with the length of corresponding subvectors $|\beta_{S_j}(f_i-f_*,\vhat^{\ME})|$. Figure \ref{fig::ex1b_quant_expl_gain_stab_abs_mix} illustrates that the differences in aggregated individual explanations drop significantly after grouping, and well below the $L^2(P_X)$-norm of the model difference, which showcases the gain in stability across each group.

\subsection{Experiments with public datasets}

In this section, we apply the group explanation techniques to public datasets.
We start our investigation with the Default of Credit Card Clients dataset \cite{default_credit_dataset} from the UCI Machine Learning Repository. This dataset contains $30000$ instances, 23 features and a dependent binary variable $Y$ that indicates if an individual defaulted on a payment, where the default is denoted by $Y=1$. The protected attributes `sex', `marriage', and `age' were removed in order to be consistent with regulatory practices. The remaining twenty predictors were used for model training, where we use the training dataset $D_{train}$ with $27000$ samples to build a classification score $p_*(x):=\widehat{\PP}(Y=1|X=x)$ using the CatBoost algorithm, whose corresponding  population minimizer is defined by $f_*(x)=\logit(p_*(x))$. For training we use the following parameters: iterations=$200$, min\_data\_in\_leaf=$5$, depth=$5$, subsample=$0.8$, and learning\_rate=$0.1$.

Performance metrics for the model on the trained dataset, and test dataset with $3000$ samples, were evaluated. Specifically, the mean logloss on the train and test set is approximately $0.40$ and $0.41$ respectively, and the AUC is $0.82$ and $0.80$ respectively.

\begin{figure}
    \centering
    \begin{subfigure}[t]{0.9\textwidth}
        \centering
        \includegraphics[width=\textwidth]{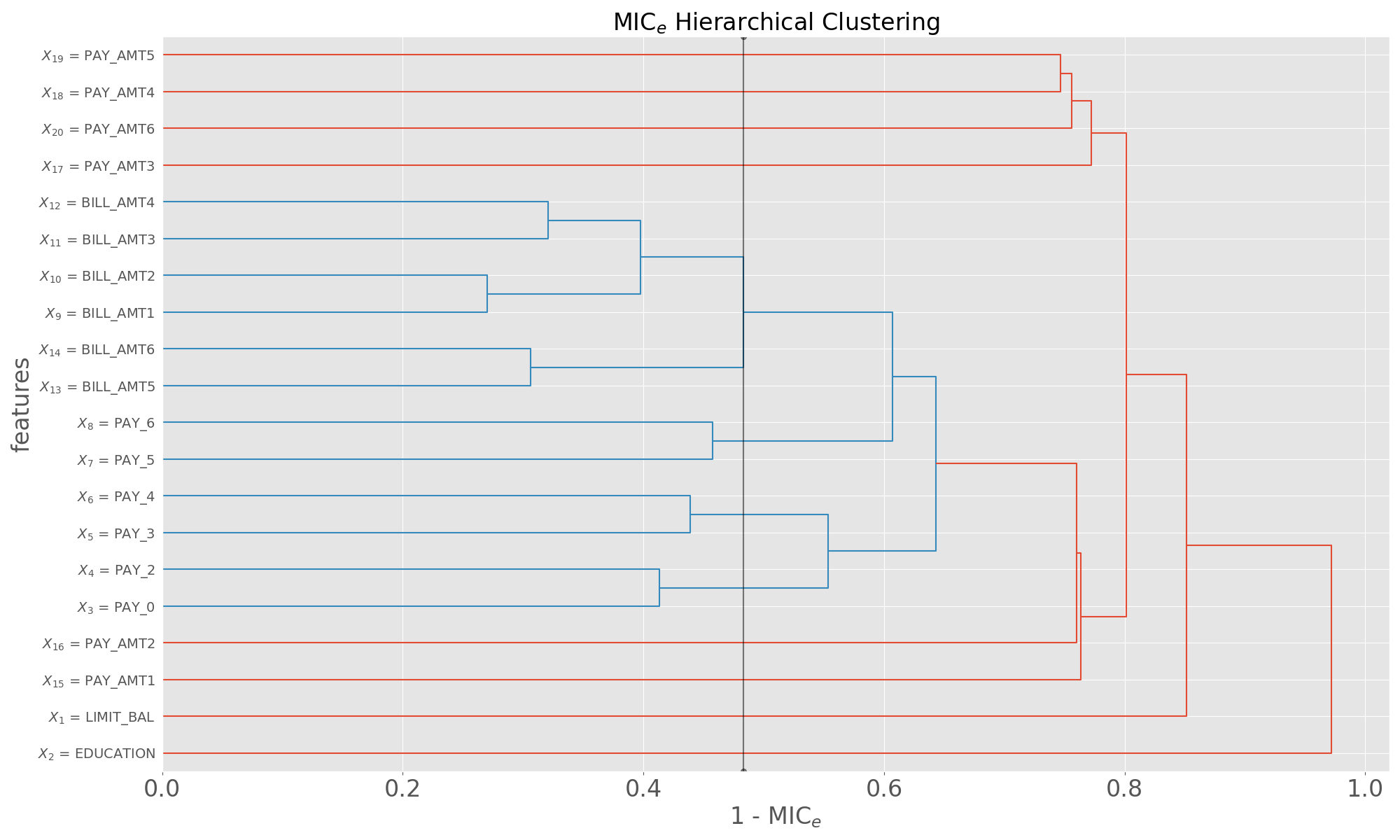}
        % \caption{\footnotesize MIC-based hierarchical clustering with GA linkage. \vspace{5pt}}
        % \label{fig::MIC_tree_credit}
    \end{subfigure}
    \caption{\footnotesize MIC-based hierarchical clustering for the Default of Credit Card Clients dataset.}\label{fig::MIC_tree_credit_}
\end{figure}

To assess the dependencies, we build a dendrogram based on the MIC-metric and investigate the level of dependence that exists among the twenty predictors. As seen in Figure \ref{fig::MIC_tree_credit_}, the dependencies are not as extreme as the ones designed in the synthetic dataset of \S\ref{sec::pedag_example}. There we were able to showcase the drastic Rashomon effect because the models we designed exhibit stronger dependencies between predictors. In particular, Figures \ref{fig::ex1_norm_expl} and \ref{fig::ex1b_norm_expl_mix} in  \S\ref{sec::pedag_example} depicting the norms of marginal explanations for each predictor illustrate the ``importance'' of these predictors for each model. Some models viewed the first two predictors as being similar, while others put more emphasis on one of them. As we will see, the Rashomon effect is still present in the models trained on the Default of Credit Clients dataset and grouping leads to improved stability, but not as drastically in light of the lack of strong dependencies.

In what follows, we will compute the Owen values (see \eqref{BzOw}) $\{Ow(x;\empvme, \cP, f_*)\}_{i=1}^{20}$ of the empirical marginal game for different partitions $\cP$ obtained by thresholding the tree. Recall from \S\ref{sec::ExTree} that the tree can be viewed as a coalescent process parameterized by $\alpha=1-\MICe$. This yields a sequence of nested partitions $\{\cP^{(k)}\}_{k=0}^{19}$, with $\cP^{(k)}$ corresponding to the $k$-th coalescent, and having $(20-k)$ groups of predictors.

In \S\ref{sec::obser_inter_expl}, we discussed two phenomena associated with efficient game values under dependencies: 1) the energy of the model (i.e. its  squared norm) splitting among explanations of dependent predictors (see \eqref{energyineq}, and Examples \ref{ex::energycons}, \ref{ex::energy_dissip_dep} and \ref{ex::energy_dissip_int}), and 2) similar models that approximate the same data can have very different marginal explanations (see Examples \ref{ex::marg_illposed} and \ref{ex::marg_instab}, and Theorem \ref{thm::margoperatorunbound}) in view of the Rashomon effect. As we will see in the context of marginal explanations, both phenomena may take place simultaneously. In what follows, we will explore these issues by first investigating the issue of energy splits and then discussing the Rashomon effect.

\subsubsection{Ranking and energy splitting of explanations}\label{sec::num_ranking}

Here, we consider partition $\cP^{(8)}=\cP_{0.49}$, that can be obtained by thresholding the dendrogram at $\alpha=0.49$. This partition contains the following groups: \{PAY\_0,PAY\_2\}, \{PAY\_3,PAY\_4\}, \{PAY\_5,PAY\_6\}, \{BILL\_AMT1,\dots,BILL\_AMT6\}, and the rest are singletons; see Figure \ref{fig::MIC_tree_credit_}.

We then compute the (global) empirical marginal Owen explanations of the population minimizer $f_*(x)$, the values $\beta_i(f_*,\vhat^{\ME})=\|Ow_i(X;\vhat^{\ME},\cP^{(8)},f_*)\|_{L^2(\P)}$, $i \in N$, which are depicted in Figure \ref{fig::expl_norms_credit_}.   To accomplish this, we use the empirical game defined in \eqref{empmarggame} with a background dataset $\bar{D}_X:=D_{train}$. To compute the explanations of the population minimizer, given the dimensions of the background dataset, we use the fast, exact algorithm introduced in Filom et al. \cite{FilomTBMarg}, which is designed specifically for the computation of empirical marginal coalitional values of CatBoost ensembles.

We then compute the sums of explanations over each group  (the trivial group explainations introduced in Definition \ref{def::coalexpl}) to obtain the global contributions of the groups themselves, that is, the values $\beta_j^{\cP^{(8)}}(f_*,\vhat^{\ME})=\|Ow_{S_j}(X;\vhat^{\ME},\cP^{(8)},f_*)\|_{L^2(\P)}$, $S_j \in \cP^{(8)}$, which are depicted in Figure \ref{fig::expl_norms_credit_}. Recall from \S\ref{sec::qp_prop} that the group sums are equal to the quotient Shapley values in view of \eqref{quotientgame}. Since groups are not fully independent, the quotient marginal Shapley values are only crude approximants of the conditional ones.

To observe the splits, it is sufficient to compare the contributions of highly dependent predictors that form the coalition with that of the coalition itself as well as with contributions of independent (or almost independent) predictors that form singletons and whose marginal explanations, according to Proposition \ref{prop::unifcoalexpl}, are equal to (or approximate well) the corresponding conditional ones.

The splits are prominent in Figure \ref{fig::expl_norms_credit_} which presents the norms of contributions of the individual predictors together with the corresponding groups. Observe the energy splitting occurring in the predictor group with BILL\_AMT's and contrast the individual and group explanations with, for example, the explanation of LIMIT\_BAL. When one attempts to rank order predictors based on their contributions, LIMIT\_BAL will be placed higher in the ranking compared to each BILL\_AMT. However, when ranking groups, \{LIMIT\_BAL\}, as a singleton, will be placed lower than the group containing the BILL\_AMT predictors. This clearly indicates the issue caused by energy splits to rank ordering based on contributions of individual predictors.

% computing empirical marginal Shapley values and compare them with those of sums of marginal  Shapley values for the model $f_*$;

To complete our study, we repeat the same experiment by 
computing the $L^2$-norms of empirical marginal Shapley values of the model $f_*$ and comparing them with those of their sums within groups (i.e. the trivial group explainations, cf. Definition \ref{def::trivgexp}); see Figure \ref{fig::expl_norms_shap_credit_}. It can be seen that the global marginal Owen and Shapley explanations are extremely similar. We believe this is due to the fact that $f_*$ is close to `additive' (since $f_*=\beta +\alpha(\sum \mathcal{T}_j)$ where each oblivious tree $\mathcal{T}_j$ is a function of at most $5$ variables) and the dependencies between predictors are not very strong. Furthermore, the linearity of the game value implies that Owen and Shapley explanations are linear combinations of explanations for each individual tree. Although the marginal Owen and Shapley values for a given tree in general differ (since a tree is not an additive function), the trees in this model are oblivious and not very deep, which apparently leads to values being very similar.

\begin{figure}
    \centering
    \begin{subfigure}[t]{0.8\textwidth}
        \centering
        \includegraphics[width=\textwidth]{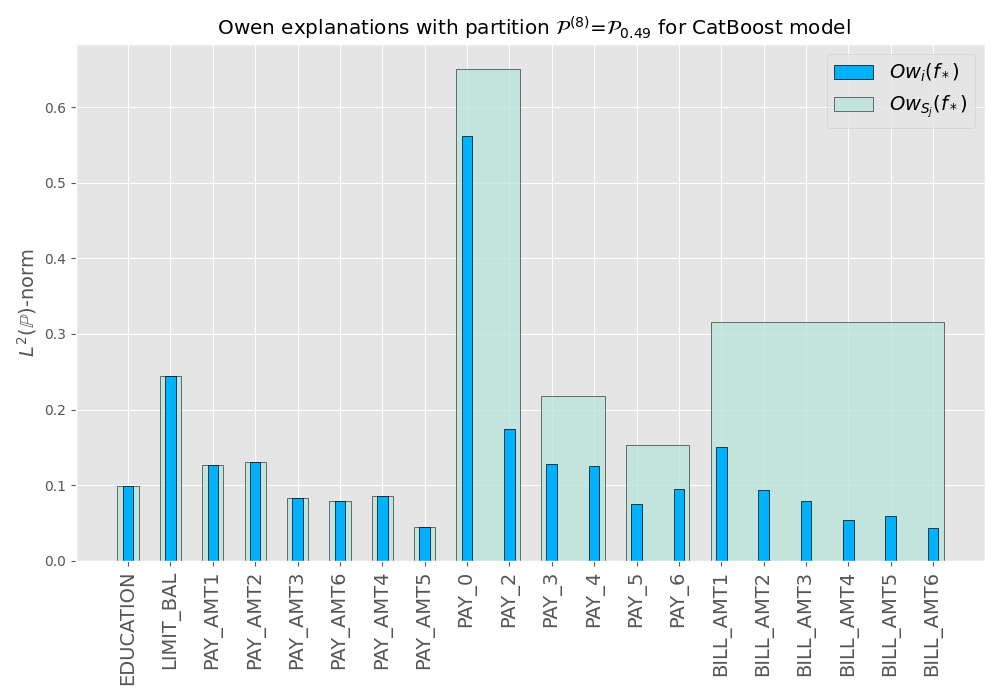}
        % \caption{\footnotesize Explanation norms. \vspace{5pt}}
        % \label{fig::expl_norms_credit}
    \end{subfigure}
    \caption{\footnotesize Global Owen explanation for $\cp^{(8)}=\cP_{0.49}$.}\label{fig::expl_norms_credit_}
\end{figure}

\begin{figure}
    \centering
    \begin{subfigure}[t]{0.8\textwidth}
        \centering
        \includegraphics[width=\textwidth]{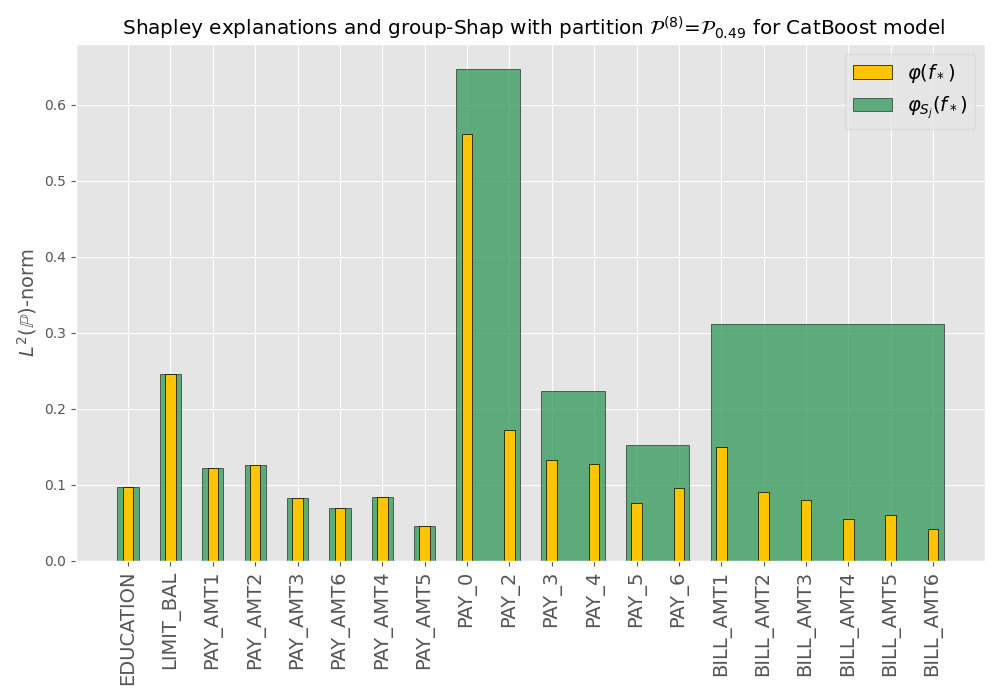}
        % \caption{\footnotesize Explanation norms. \vspace{5pt}}
        % \label{fig::expl_norms_credit}
    \end{subfigure}
    \caption{\footnotesize Global Shapley explanation for $\cP^{(8)}=\cP_{0.49}$.}\label{fig::expl_norms_shap_credit_}
\end{figure}

\subsubsection{Grouping effect on stability}\label{sec::num_stability}

In this section, we continue to explore the Rashomon effect by measuring and comparing the stability of explanations before and after grouping on the Default of Credit Card Clients dataset \cite{default_credit_dataset}.

To understand how the dependencies between groups affect stability, we design the following experiment. Given the reference model $f_*$ and the population minimizer described in Section \ref{sec::num_ranking}, we train a series of new models whose predictions are close to those of $f_*$ by varying the hyperparameters. Specifically, we pick the following parameters at random from the given intervals: iterations $\in [50,300]$, subsample $\in [0.5,1.0]$, depth $\in [2,10]$, learning\_rate $\in[0.025,0.25]$, rsm $\in[0.5,0.1]$, with the rest of the parameters being the same as for $f_*$. We then train a new model $f$ and accept or reject based on the following principle. Given a threshold $\tau \in [0,1]$, we accept the model if it is in the Rashomon ball of relative radius $\tau$, meaning that $\hat{\E}[|f(X)-f_*(X)|^2] \leq \tau \|f_*\|_{L^2(P_X)}$, otherwise it is rejected; here $\|f_*\|_{L^2(P_X)} \approx 1.84$ which is estimated on the test set, and we choose $\tau = 0.1$. We continue this procedure until we train 20 models $\{f_k\}_{k=1}^{20}$.

We next consider partitions of features with varying degrees of dependence, from moderate to strongly dependent. This task is accomplished by cross-sectioning the partition tree of dependencies for a given dataset at different heights. In this analysis the partitions considered are $\cP_{0.49}$, $\cP_{0.62}$, $\cP_{0.65}$ and $\cP_{0.77}$, containing $12$, $10$, $9$, and $5$ groups (see Figure \ref{fig::MIC_tree_credit_2}), respectively, with the subscript indicating the cutoff threshold.

Given these four partitions we evaluate the empirical marginal Owen explanations for each predictor and for each group (via summation) per model. Subsequently, we compute the global explanation of the model difference $\|f_*-f_k\|_{L^2(\P)}$  from individual explanations. Specifically, we first evaluate the explanations $Ow_i (X,f_*-f_k,\empvme,\cP)$  of the model difference $\Delta f_k = f_*-f_k$ and then compute the corresponding norms $\beta_i (f_*-f_k,\cP)=\|Ow_i (x,f_*-f_k,\cP)\|_{L^2(\P)}$, $i \in N$, for each $k\in\{1,\dots,20\}$ and  $\cP \in \{P_{0.49},P_{0.62},P_{0.65},P_{0.77}\}$. We do the same for the group explanations, $\beta_{j}^{\cP} (f_k-f_*,\cP)=\|Ow_{S_j} (X,f_k-f_*,\empvme,\cP)\|_{L^2(\P)}$, $j \in M$.

To contrast the stabilization effect between individual and group explanations, we evaluate the length of the vectors $\beta(f_*-f_k,\cP)=\{\beta_i(f_*-f_k,\cP)\}_{i \in N}$ and $\beta^{\cP}(f_*-f_k,\cP)=\{\beta^{\cP}_{j}(f_*-f_k,\cP)\}_{j \in M}$ and compute the maximum of these quantities across all models $k\in\{1,\dots,20\}$. These are plotted in Figure \ref{fig::Expl_Diff_Tot_Ind_Group_singl} together with the norm of the maximum model difference. As we see in the plot, the total group explanation differences are smaller than the respective total individual ones, showcasing the gain in stability when considering explanations of groups. Furthermore, as we reach partition $\cP_{0.77}$, observe that the total group explanation difference becomes approximately equal to the norm of the difference of the models, illustrating the alleviation of the Rashomon effect due to weaker dependencies between the groups. Note that in Figure \ref{fig::Expl_Diff_Tot_Ind_Group_nosingl} we have removed the energy contributed by singletons from both individual and group explanation vectors, since these do not have any effect on the norm evaluation between the two.

Given that the dependencies are not very strong in the Default of Credit Clients dataset, as seen in the partition tree in Figure \ref{fig::MIC_tree_credit_2}, note that the Rashomon effect in general is not as prominent as in the synthetic example from \S\ref{sec::pedag_example}, but it is still present. Nevertheless, Figure \ref{fig::Expl_Diff_Tot_Ind_Group_} still portrays this effect and its alleviation when predictor groups are considered based on dependencies.

\begin{table}
\begin{center}
\begin{tabular}{|c| c | c | c | c | c | c | c | c |}
\hline
& $\|f_*\|$ & $\max_k \|\Delta f_k\|$ &  $\max_k|\beta(\Delta f_k,\cP)|$ & $\max_k|\beta^{\cP}(\Delta f_k,\cP)|$\\[0.5ex]
\hline
$\cP_{0.49}$   &$1.839$ & $0.267$ & $0.388$ &  $0.364$  \\ 
 \hline 
 $\cP_{0.62}$  &$1.839$ & $0.267$ & $0.390$  & $0.388$   \\ 
 \hline 
 $\cP_{0.65}$  &$1.839$ & $0.267$ & $0.389$  & $0.340$   \\ 
 \hline 
 $\cP_{0.77}$  &$1.839$ & $0.267$ & $0.386$  & $0.269$   \\ 
 \hline
\end{tabular}
\caption{ Global marginal Owen attributions for Default of Credit Card Clients dataset.}\label{table::exp2a_norms}
\end{center}
\end{table}

\begin{figure}
    \centering
    \begin{subfigure}[t]{0.9\textwidth}
        \centering
        \includegraphics[width=\textwidth]{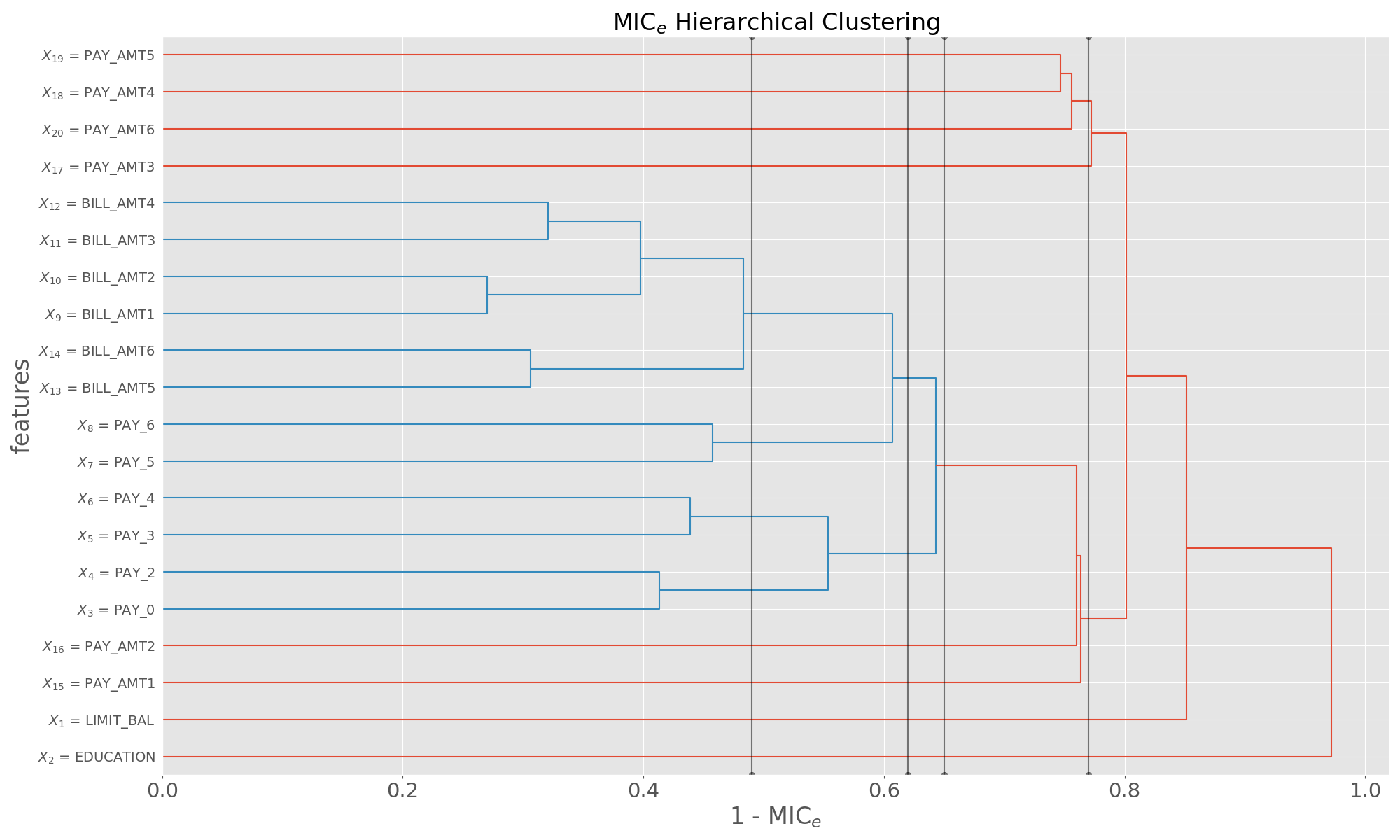}
    \end{subfigure}
    \caption{\footnotesize  MIC-based hierarchical clustering for the Default of Credit Card Clients dataset with 4 cutoffs.}\label{fig::MIC_tree_credit_2}
\end{figure}

\begin{figure}
    \centering
    \begin{subfigure}[t]{0.45\textwidth}
        \centering
        \includegraphics[width=\textwidth]{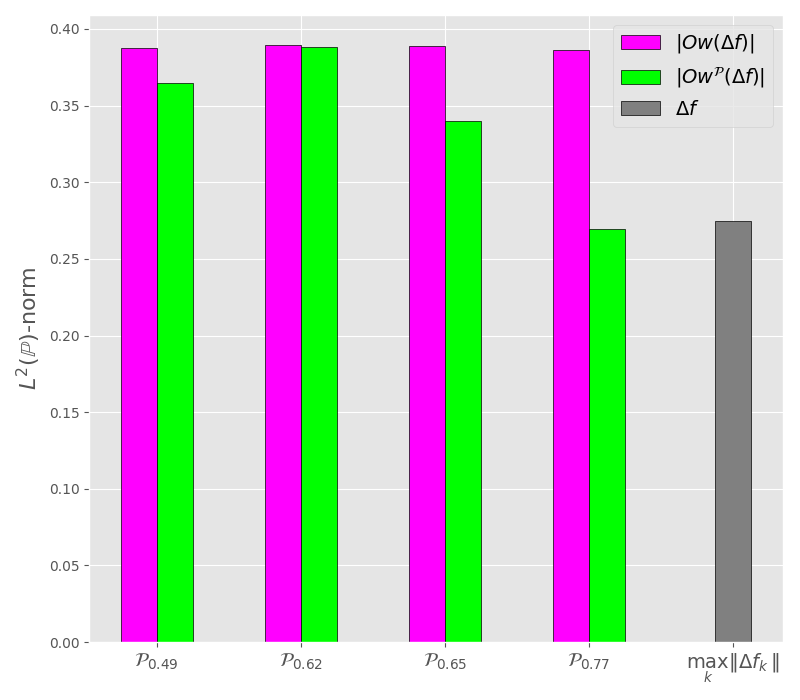}
        \caption{\footnotesize Norms of explanation differences (with singletons).}\label{fig::Expl_Diff_Tot_Ind_Group_singl}
    \end{subfigure}
    ~~
    \begin{subfigure}[t]{0.45\textwidth}
        \centering
        \includegraphics[width=\textwidth]{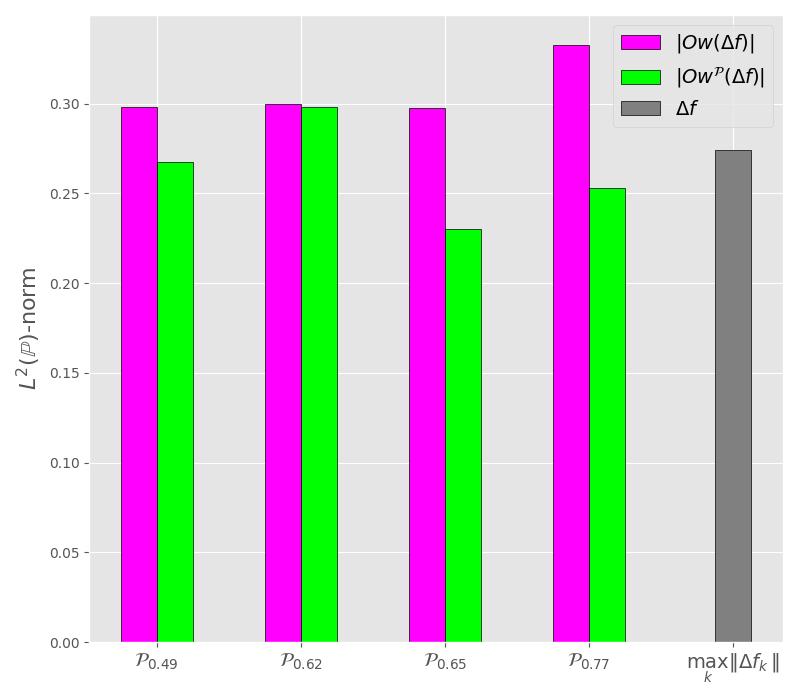}
        \caption{\footnotesize Norms of explanation differences (without singletons).}\label{fig::Expl_Diff_Tot_Ind_Group_nosingl}
    \end{subfigure}    
    \caption{\footnotesize Gain in stability, $|\beta|$ versus $|\beta^{\cP}|$. Deafult of Credt Card Clients dataset.}\label{fig::Expl_Diff_Tot_Ind_Group_}
\end{figure}

We next consider the Superconductivity dataset \cite{superconductivty_dataset}, a regression dataset where the superconductivity critical temperature is predicted based on $81$ features extracted from the superconductor’s chemical formula. The original dataset has $21263$ instances. As before, we first construct a hierarchical clustering tree of feature dependencies using the MIC-based metric (see Figure \ref{fig::MIC_tree_superconductivity}) in order to form partitions. The dataset is then randomly split into training and test sets in $90$:$10$ proportions, and we train a (reference) regressor model $f_*(x)=\hat{\E}[Y|X=x]$ using the CatBoost algorithm. For training we use the following parameters: iterations=$300$, min\_data\_in\_leaf=$5$, depth=$8$, subsample=$0.8$, and learning\_rate=$0.1$. 

% {\color{red}KK: Made into one paragraph.}

Performance metrics for the model on the trained and test datasets, the latter with $2126$ samples, were evaluated. Specifically, the mean square error estimate on the training and test sets is approximately $7.70$ and $9.45$ respectively, which constitutes about $16\%$ and $20\%$ of relative error given that the $L^2$-norm estimate of the reference model is $\|f_*\|_{L^2(P_X)} \approx 47.48$ on the test dataset.

Following the above methodology, we train a series of new models whose predictions are close to the predictions of $f_*$. Specifically, we pick the following parameters at random from the given intervals: iterations $ \in [100,500]$, subsample $\in[0.5,1.0]$, depth $\in[4,10]$, learning\_rate $\in[0.025,0.25]$, rsm $\in [0.5,0.1]$, with the rest of the parameters being the same as for $f_*$. We then train a new model $f$ and accept it if it is in the Rashomon ball centered at $f_*$  of relative size $\tau=0.06$, meaning if $\hat{\E}[|f(X)-f_*(X)|^2] \leq \tau \|f_*\|_{L^2(P_X)}$, or reject otherwise. We continue this procedure until we construct 25 models $\{f_k\}_{k=1}^{25}$.

Similar to the previous dataset, the partitions considered in this analysis are $\cP_{0.3}$, $\cP_{0.4}$, $\cP_{0.5}$, $\cP_{0.56}$, $\cP_{0.60}$ and $\cP_{0.65}$; see Figure \ref{fig::MIC_tree_superconductivity}. Given these six partitions we evaluate, as before, the empirical marginal Owen explanations for each predictor and for each group (via summation) per model and then evaluate the length of their global explanations. These are plotted in Figure \ref{fig::Expl_Diff_Tot_Ind_Group_sc_singl} together with the norm of the maximum model difference. Once again, the total group explanation differences are smaller than the respective total individual ones, showcasing the gain in stability when considering explanations of groups. Furthermore, as we reach partition $\cP_{0.77}$, observe that the total group explanation difference becomes approximately equal to the norm of the difference of the models, illustrating again the alleviation of the Rashomon effect.

\begin{table}
\begin{center}
\begin{tabular}{|c| c | c | c | c | c | c | c | c |}
\hline
& $\|f_*\|$ & $\max_k \|\Delta f_k\|$ &  $\max_k|\beta(\Delta f_k,\cP)|$ & $\max_k|\beta^{\cP}(\Delta f_k,\cP)|$\\[0.5ex]
\hline
$\cP_{0.3}$   &$47.482$ & $2.808$ & $11.286$  & $7.207$  \\ 
\hline 
$\cP_{0.4}$   &$47.482$ & $2.808$ & $11.269$  & $6.479$  \\ 
\hline 
$\cP_{0.5}$   &$47.482$ & $2.808$ & $11.119$  & $5.018$  \\ 
\hline 
$\cP_{0.55}$  &$47.482$ & $2.808$ & $11.146$  & $4.971$  \\ 
\hline
$\cP_{0.60}$  &$47.482$ & $2.808$ & $11.150$  & $4.054$  \\ 
\hline
$\cP_{0.66}$  &$47.482$ & $2.808$ & $11.167$  & $2.454$  \\ 
\hline
\end{tabular}
\caption{ Global marginal Owen attributions for Superconductivity dataset.}\label{table::exp2b_norms}
\end{center}
\end{table}

We would like to contrast this dataset with the Default of Credit Card Clients dataset. Note that due to the stronger dependencies among the features of the Superconductivity dataset, the Rashomon effect is much more apparent in this case compared to the previous dataset, which also means that the alleviation of the Rashomon effect due to evaluating group explanations is also more striking.

\begin{figure}
    \centering
    \begin{subfigure}[t]{0.9\textwidth}
        \centering
        \includegraphics[width=\textwidth]{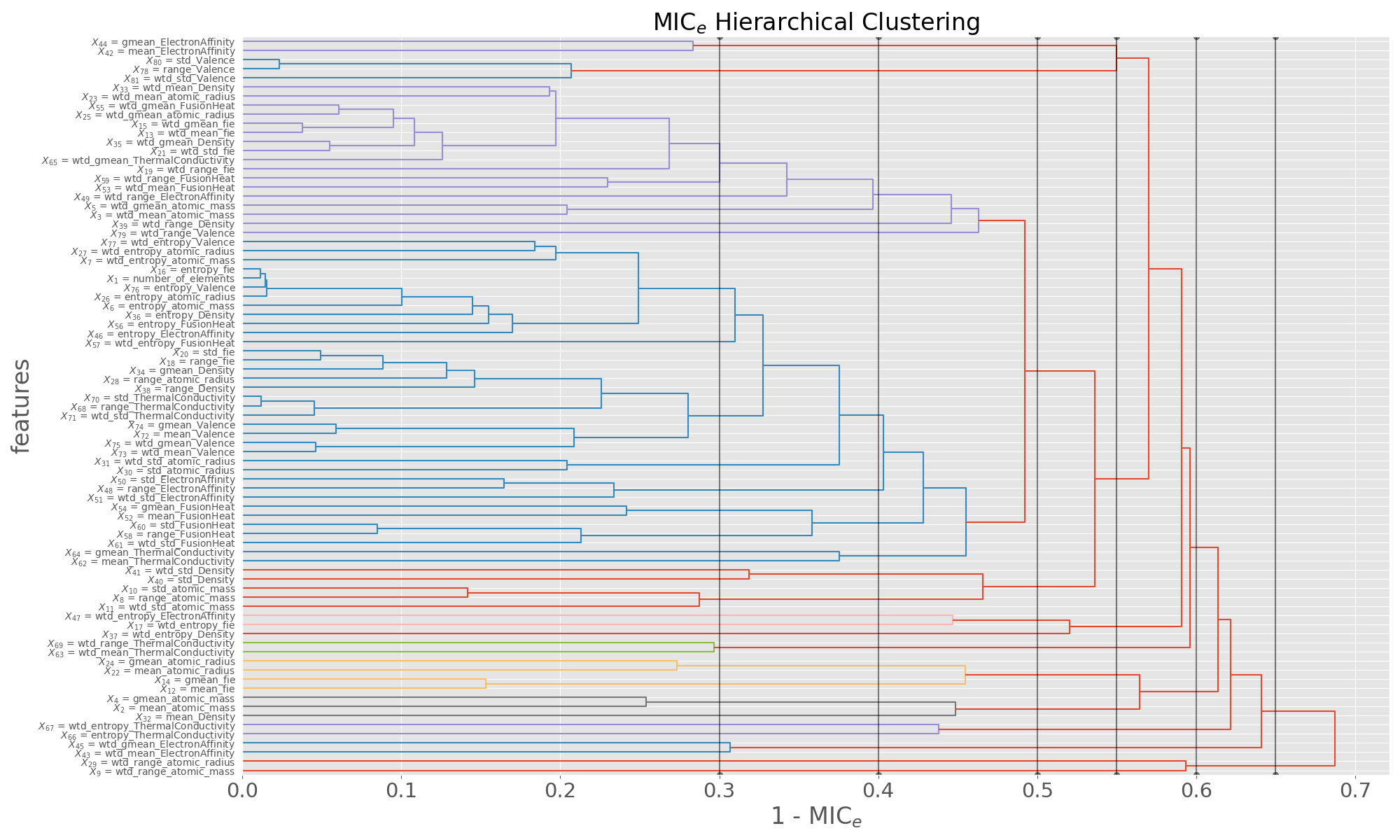}
    \end{subfigure}
    \caption{\footnotesize MIC-based hierarchical clustering for the Superconductivity dataset with $6$ cutoffs.}\label{fig::MIC_tree_superconductivity}
\end{figure}

\begin{figure}
    \centering
    \begin{subfigure}[t]{0.45\textwidth}
        \centering
        \includegraphics[width=\textwidth]{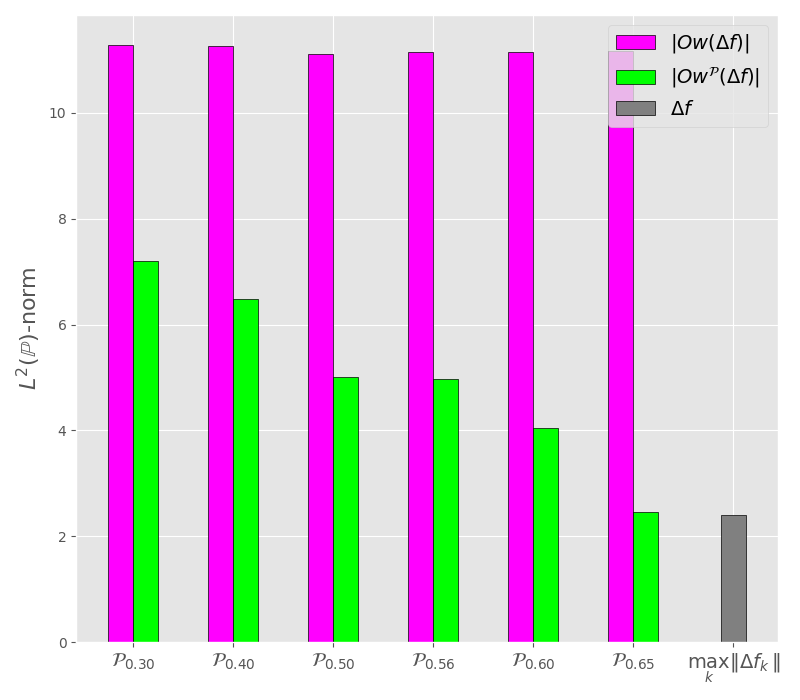}
        \caption{\footnotesize Norms of explanation differences (with singletons).}\label{fig::Expl_Diff_Tot_Ind_Group_sc_singl}
    \end{subfigure}
    ~~
    \begin{subfigure}[t]{0.45\textwidth}
        \centering
        \includegraphics[width=\textwidth]{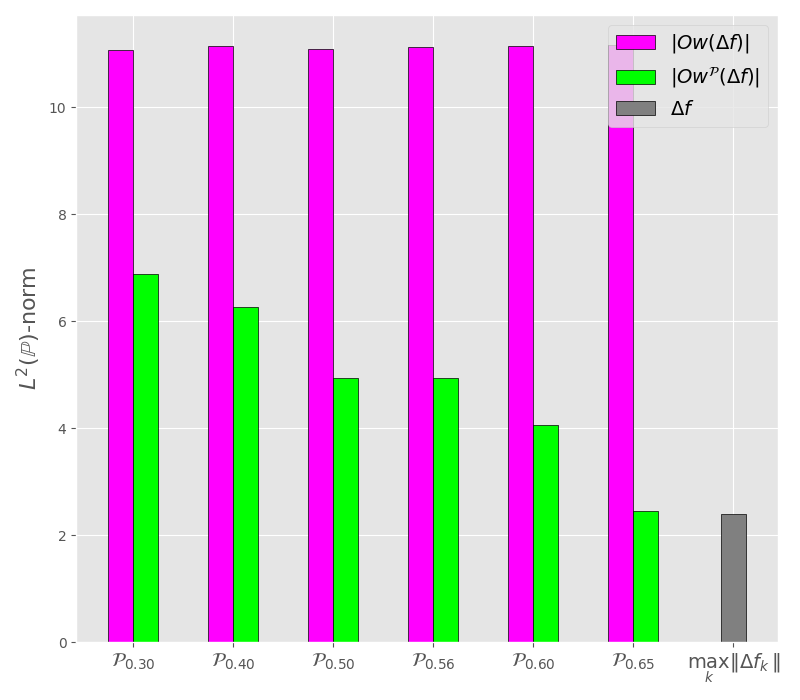}
        \caption{\footnotesize Norms of explanation differences (without singletons).}\label{fig::Expl_Diff_Tot_Ind_Group_sc_nosingl}
    \end{subfigure}    
    \caption{\footnotesize Gain in stability, $|\beta|$ versus $|\beta^{\cP}|$. Superconductivity dataset.}\label{fig::Expl_Diff_Tot_Ind_Group_sc}
\end{figure}

\section{Conclusion}\label{sec::conclusion}

In this work, we presented a comprehensive and rigorous treatment of machine learning explainers arising from the cooperative game theory by utilizing tools from functional analysis. We carefully set up explainers obtained from applying a linear game value to the conditional or marginal games associated with a machine learning model as appropriate linear operators, and we investigated their continuity. This highlighted the differences between the two games: the former takes into account the joint distribution of the predictors, whereas the latter highly depends on the structure of the model and may result in explanations which are unstable in the natural data-based metric. We provided numerous examples illuminating our theoretical results.

The conditional and marginal explanations often differ when the predictors are dependent (which is almost always the case). 
To unify the two paradigms and to address the instability of marginal explanations, we proposed partitioning the predictors based on dependency and then using coalitional game values. This approach also lowers the computational complexity of generating explanations. Various game-theoretical group explainers were constructed, and we showed that many of them coincide once the predictors are partitioned into independent groups.

In practice, for partitioning the predictors, we proposed a variable hierarchical clustering technique that employs a state-of-the-art measure of dependence called the maximal information coefficient, a regularized version of mutual information that can capture non-linear dependencies. This approach was tested on models trained on various datasets. The experiments showcased the benefits of grouping the predictors for generating feature attributions.

\section*{Acknowledgment}
{The authors would like to thank Steve Dickerson (former CAO, Decision Management at Discover Financial Services (DFS)) and Raghu Kulkarni (SVP, Chef Data Scientist at DFS) for formulation of the problem as well as helpful business and compliance insights. We also would like to thank professors Markos Katsoulakis and Robin Young from the University of Massachusetts Amherst, and Hangjie Ji from North Carolina State University for their valuable comments and suggestions that aided us in writing this article.}

\begin{appendices}

\section*{Appendix}

%%%%%%%%%%%%%%%%%%%%%%%%%%%%%%%%%%%%%%%%%%%%%%%%%%%%%%%%%%%%%%%%%%%%%%%%%%%%%
%%%%%%%%%%%%%%%%%%%%%%%%%%%%%%%%%%%%%%%%%%%%%%%%%%%%%%%%%%%%%%%%%%%%%%%%%%%%%
%%%%%%%%%%%%%%%%%%%%%%%%%%%%%%%%%%%%%%%%%%%%%%%%%%%%%%%%%%%%%%%%%%%%%%%%%%%%%
%%%%%%%%%%%%%%%%%%%%%%%%%%%%%%%%%%%%%%%%%%%%%%%%%%%%%%%%%%%%%%%%%%%%%%%%%%%%%
%%%%%%%%%%%%%%%%%%%%%%%%%%%%%%%%%%%%%%%%%%%%%%%%%%%%%%%%%%%%%%%%%%%%%%%%%%%%%
%%%%%%%%%%%%%%%%%%%%%%%%%%%%%%%%%%%%%%%%%%%%%%%%%%%%%%%%%%%%%%%%%%%%%%%%%%%%%

% super-additive

\section{On game values}\label{app::gamevalues}

\subsection{Game value axioms}\label{app::gameaxioms}

A cooperative game is a pair $(N,v)$ defined by the finite set of players $N \subset \mathbb{N}$ (typically, $N=\{1,2,\dots,n\}$) and a set function $v$ defined on the collection of all subsets $S \subseteq N$, which satisfies $v(\varnothing)=0$. 
A set $T \subseteq N$ is called a carrier of $v$ if $v(S)=v(S \cap T)$ for all $S \subseteq N$.
A game value is a map $(N,v) \mapsto h[N,v]=(h_i[N,v])_{i \in N}$.

We now list  some of useful game value properties:
\begin{enumerate}

    \item [(LP)]\label{axiom:LP} (linearity)  
    For two cooperative games $(N,v)$ and $(N,w)$ we have
\begin{equation}\label{linp}
    h[N,a v+w] = a h[N,v]+ h[N,w], \ a\in \R.  
\end{equation}
    
\item [(EP)]\label{axiom:EP}  (efficiency) The sum of the values is equal to the value of the game
\begin{equation}\label{effp}
\sum_{i\in N} h_i[N,v] = v(N).
\end{equation}

\item [(SP)]\label{axiom:SP} (symmetry) For any permutation $\pi$ on $N$ and game $(N,v)$
\begin{equation}\label{symp}
h_{\pi(i)}[N,\pi v]=h_i[N,v],\quad \pi v(\cdot)=v(\pi^{-1}\cdot).
\end{equation}

\item [(TPP)]\label{axiom:TPP} (total power) 
\begin{equation}\label{tpp}
\textstyle    \sum_{i=1}^n h_i[N,v] = \frac{1}{2^{n-1}} \sum_{i=1}^n \sum_{S \subseteq N \backslash \{i\}} [v(S \cup \{i\}) - v(S)]
\end{equation}

\item [(NPP)]\label{axiom:NPP} (null player) A null player $i\in N$ is a player that adds no worth to the game $v$, which means 
\begin{equation}\label{npp}
    v(S\cup \{i\})=v(S), \quad S\subseteq N\setminus \{i\}.
\end{equation}
We say that  $h$ satisfies the null-player property if $h_i[N,v]=0$ whenever $i \in N$ is a null player.

\item [(CDP)]\label{axiom:CDP} {(carrier dependence)} For any $v$ with a carrier $T \subseteq N$, 
\begin{equation}\label{cip}
h_i[N,v] = h_i[T,v], \quad i \in T.
\end{equation}

\item [(TPG)]\label{axiom:TPG} (total payoff growth) There exists strictly increasing $g:\RR_+ \times \mathbb{N} \to \RR_+$ satisfying $g(0,n)=0$ and $g(a,n)>0$ for $a>0$ such that for all games $v$ with the carrier $N$ 
\begin{equation}\label{growth}
    \textstyle\sum_{i=1}^n |h_i[N,v]| \geq g( |v(N)|,|N|) \geq 0.
\end{equation}

\item [(SEP)]\label{axiom:SEP} (singleton efficiency property) For any $i \in \mathbb{N}$, 
\begin{equation}\label{SEP}
    h[\{i\},v]=v(\{i\}).
\end{equation}

\item [(NN)]\label{axiom:NN} $h[N,v]$ is a linear game value in the marginalist form \eqref{lingameform} with weights satisfying $w(S,n)\geq 0, S \subset N$.

\item [(NVA)]\label{axiom:NVA} $h[N,v]$ is a linear game value in the marginalist form \eqref{lingameform} with weights satisfying 
\[
\Big(\sum_{S \subseteq N \setminus \{i\}} w(S,n) \Big) \neq 0, \quad i \in N.
\]

\end{enumerate}

A coalitional game value is a finer assignment $(N,v,\mathcal{P})\mapsto g[N,v,\mathcal{P}]$ whose inputs are coalitional games 
$(N,v,\mathcal{P})$ where $\mathcal{P}$ is a partition of the set of players $N$. 
Such objects appear in \S\ref{sec::ExCoal}. All axioms for game values discussed above have immediate generalizations to coalitional game values. 
Nevertheless, we shall need the following property of coalitional values $g$.
\begin{enumerate}
\item [(SIP)]\label{axiom:SIP} (singleton invariance property) 
If $(N,\tilde{u})$ is a unit cooperative game  (i.e. $\tilde{u}(S)=1$ for all
$\varnothing\neq S\subseteq N$), for any $i \in \mathbb{N}$ one has 
\begin{equation}\label{SIprop}
g_i[\{i\},\tilde{u},\{\{i\}\}]=g_1[\{1\},\tilde{u},\{\{1\}\}].
\end{equation}
\end{enumerate}

\subsection{Canonical representation of linear game values}\label{app::gamevalrepr}

\begin{lemma}\label{lmm::lingamevalrepr}

Let $h$ be a linear game value and $g$ a linear coalitional value. Then

\begin{itemize}

\item [$(i)$] For each $N \subset \mathbb{N}$, there exist constants $\{\gamma(i,N,S)\}_{i \in N,S \subseteq N}$ such that
\begin{equation}\label{lingamevalrepr}
h_i[N,v]=\sum_{S \subseteq N} \gamma(i,N,S)v(S).
\end{equation}
When $v$ is non-cooperative, extensions $\bar{h}$ of $h$ admit a similar representation provided that numbers $\gamma(i,N,\varnothing)$ are the same as $\gamma_i$'s from Lemma \ref{lmm:extlingame}.
  \item [$(ii)$] For each $N \subset \mathbb{N}$ and a partition $\mathcal{P}$ of $N$, there exist constants $\{\gamma(i,N,\mathcal{P},S)\}_{i \in N, S \subseteq N}$ such that
\begin{equation}\label{lincoalgamevalrepr}
g_i[N,v,\mathcal{P}]=\sum_{S \subseteq N} \gamma(i,N,\mathcal{P},S)v(S).
\end{equation}

\end{itemize}
\end{lemma}

\begin{proof}
For each non-empty set $S \subseteq N$ define the game $u_S(A)=\1_{\{S=A\}}(A)$, where $A \subseteq N$. Let also $u_{\varnothing}=0$. Then $v=\sum_{S \subseteq N} v(S)u_S$ and hence by the linearity of $h$ we obtain
\[
h[N,v]=\sum_{S \subseteq N} h[N,u_S] v(S).
\]
Setting $\gamma(i,N,S)= h_i[N,u_S]$, proves $(i)$. The proof of $(ii)$ is similar.
\end{proof}

\begin{lemma}\label{lmm::nonessential}
Let $h$ be a linear game value that satisfies efficiency and null-player properties. Let $(N,v)$ be a non-essential cooperative game, that is, $v(S)=\sum_{i \in S} v(\{i\})$. Then $h_i[N,v]=v(\{i\})$.
\end{lemma}
\begin{proof}
For each $i \in N$, define a game $v_i$ by $v_i(S)=v(S \cap \{i\})$. Then $\{i\}$ is a carrier of $v_i$. Since $h$ satisfies the null-player property, one has 
$h_k[N,v_i]=0$ for $k \neq i$. Then, by the efficiency property, we obtain
\[
v(\{i\})=v_i(N)=\sum_{k \in N} h_k[N,v_i]=h_i[N,v_i].
\]
Since the game is non-essential, for any non-empty $S \subseteq N$ we have
\[
v(S)=\sum_{i \in S} v(\{i\}) = \sum_{i \in S} v(S \cap \{i\}) = \sum_{i \in S} v_i(S) = \sum_{i \in N} v_i(S)
\]
where we used the fact that $v_i(S)=0$ for $i \notin S$. Thus, using the linearity of $h$,  we conclude
\[
h_i[N,v]=\sum_{k \in N} h_i[N,v_k]=h_i[N,v_i]=v(\{i\}).
\]
\end{proof}

\subsection{On the Radon-Nikodym derivative of probability measures} \label{app::radon-nykodim} 

Let $\mathcal{B}(\RR^k)$ denote the $\sigma$-algebra of Borel sets. The space of all Borel probability measures on $\RR^k$ is denoted by $\mathscr{P}(\RR^k)$. The space of probability measure with finite $q$-th moment is denoted by 
\[
\mathscr{P}_q(\RR^k)=\Big\{ \mu\in \mathscr{P}(\RR^k): \int_{\RR^k} |x|^q d\mu(x) < \infty \Big\}.
\]

\begin{definition}[{\bf push-forward}]\label{def::push-forward}\rm

  Let $\PP$ be a probability measure on a measurable space $(\Omega,\mathcal{F})$. Let $X \in \RR^n$ be a random vector defined on $(\Omega,\mathcal{F},\PP)$. The push-forward probability distribution of $\PP$ by $X$ is defined by
  \[
  P_X(A):=\PP\big( \{ \omega \in \Omega: X(\omega) \in A \} \big), \quad A\in \mathcal{B}(\RR^n).
  \]
\end{definition}

\begin{definition}[{\bf absolute continuity}]
Let $\mu,\nu $ be measures on a measurable space $(\Omega,\mathcal{F})$. $\mu$ is said to be absolutely continuous with respect to $\nu$, denoted as $\mu \ll \nu$, if $\nu(A)=0$ implies $\mu(A)=0$ for $A \in \mathcal{F}$.
\end{definition}

\begin{theorem}[\bf Radon-Nikodym derivative]
Suppose that $\mu,\nu$ are two $\sigma$-finite measures defined on a measurable space $(\Omega,\mathcal{F})$. If $\mu \ll \nu$, then there exists an $\mathcal{F}$-measurable function $r:\Omega \to [0,\infty)$, written as $r=\frac{d \mu}{d \nu}$, such that for any measurable set $A \in \mathcal{F}$, $\mu(A)=\int_A r(x) \, \nu(dx)$.
\end{theorem}
\begin{proof}
See Royden and Fitzpatrick \cite{Royden2010}.
\end{proof}

\begin{corollary}\label{app::lmm::rad_nik_prob_meas} Suppose that $\mu,\nu$ are two probability measures defined on a measurable space $(\Omega,\mathcal{F})$. If $\mu \ll \nu$, then the Radon-Nikodym derivative $\frac{d \mu}{d \nu}$ belongs to $L^1(\Omega,\mathcal{F},\nu)$ and is of norm $1$.
\end{corollary}

\begin{lemma}\label{app::lmm::rad_nik_equiv}
Let $\mu,\nu $ be probability measures on a measurable space $(\Omega,\mathcal{F})$. Suppose that $\mu \ll \nu $. Then the following statements are equivalent:

\begin{itemize}

\item [$(i)$] $\frac{d\mu}{d \nu} \in L^{\infty}(\Omega,\mathcal{F},\nu) $ in which case $\mu(A) \leq \big\| \frac{d\mu}{d \nu} \big\|_{L^{\infty}(\Omega,\mathcal{F},\nu)} \cdot \nu(A)$, $A \in \mathcal{F}$.

\item [$(ii)$] There exists $M > 0$ such that $\mu(A) \leq M \cdot \nu(A)$, all $A \in \mathcal{F}$.

\end{itemize}

\end{lemma}
\begin{proof} Suppose $(i)$ holds. Then for any $A \in \mathcal{F}$ as $\mu \ll \nu$, we have 
\[
\mu(A) \leq \int_A \frac{d\mu}{d\nu}(x) \nu(dx) \leq \Big\| \frac{d\mu}{d\nu}\Big\|_{L^{\infty}(\Omega,\mathcal{F},\nu)} \cdot \nu(A).
\]
This proves that $(i)$ implies $(ii)$.

Suppose $(ii)$ holds. Suppose there exists  $B \in \mathcal{F}$ of positive $\nu$-measure such that $\frac{d\mu}{d\nu}>M$ on $B$. Then
\[
M \cdot \nu(B) \geq \mu(B) = \int_B \frac{d\mu}{d\nu}(x) \, \nu(dx) > M \cdot \nu(B),
\]
which is a contradiction. Hence $\frac{d\mu}{d\nu}\leq M$ $\nu$-almost surely. Since $\nu$ and $\mu$ are probability measures, we must have $\frac{d\mu}{d\nu}\geq 0$ $\nu$-almost surely. This implies $(i)$.
\end{proof}

\begin{lemma}\label{app::lmm::cond_marg_rel}
Let $X=(X_1,\dots,X_n)$, $Z=(Z_1,\dots,Z_m)$ be random vectors on a measurable space $(\Omega,\mathcal{F},\PP)$ such that  $P_X \otimes P_Z \ll P_{(X,Z)} $. Suppose $f \in L^2_{1+r^2}(P_{(X,Z)})$, where $r:=\frac{d P_X \otimes P_Z }{ dP_{(X,Z)}}$. Then
\begin{equation}\label{cond_marg_meas_bound}
\int \Big(\int f(x,z) P_X(dx) - \int f(x,z) P_{X|Z=z}(dx) \Big)^2 P_Z(dz) \leq \| (r-1) \cdot f \|^2_{L^2(P_{(X,Z)})}.
\end{equation}
\end{lemma}
\begin{proof}
Take $B \in \mathcal{B}(\RR^m)$. Then, by definition of Radon-Nikodym derivative, we have
\[
\int 1_{B}(z) \cdot f(x,z) \, P_X \otimes P_Z (dx,dz) = \int  1_{B}(z) \cdot f(x,z) r(x,z) \, P_{(X,Z)} (dx,dz)
\]
and hence
\[
\begin{aligned}
&\int_B  \Big( \int f(x,z) \, P_X(dx) - \int  f(x,z)  P_{X|Z=z} (dx) \Big) P_Z(dz) \\
&=\int_B \Big( \int f(x,z) (r(x,z)-1)  P_{X|Z=z}(dx) \Big) P_Z(dz).
\end{aligned}
\]

Since $B \in \mathcal{B}(\RR^m)$ is arbitrary, we conclude that for $P_Z$-almost all $z$
\[
\int f(x,z) \, P_X(dx) - \int  f(x,z)  P_{X|Z=z} (dx) = \int f(x,z) (r(x,z)-1)  P_{X|Z=z}(dx).
\]
This implies \eqref{cond_marg_meas_bound}.
\end{proof}

% \begin{definition}[\bf divergences]
% Let $\mu,\nu$ be probability measures on $(\Omega,\mathcal{F})$, and $\mu \ll \nu$.
% \begin{itemize}

%   \item[$(i)$] The Kullback-Leibler divergence from $\nu$ to $\mu$ is defined to be
%   \[
%   D_{KL}(\mu || \nu ) = \int_{\Omega} \log\Big( \frac{d \mu}{ d \nu} \Big) d \mu.
%   \]

%   \item[$(ii)$] For $\alpha>0, \alpha \neq 1$ the R\'{e}nyi divergence of order $\alpha$ from $\nu$ to $\mu$ is defined to be
%   \[
%   D_{\alpha}(\mu || \nu ) = \frac{1}{\alpha - 1} \log \Big( \int_{\Omega} \Big( \frac{d \mu}{ d \nu} \Big)^{\alpha-1} d \mu ).
%   \]
% \end{itemize}
% \end{definition}

\begin{definition}[\bf Wasserstein]
The Wasserstein distance $W_1$ on $\mathcal{P}_1(\RR^k)$ is given by \cite{Kantorovich1958}
\[
W_1(\mu,\nu)= \sup \Big\{ \int \psi(x) [\mu-\nu](dx), \quad \psi \in Lip_1(\RR^k)=\big\{u:|u(x)-u(x')|\leq |x-x'| \big\} \Big\}.
\]
\end{definition}

\begin{lemma}[\bf Wasserstein bound]\label{app::lmm::W1_bound}
Let $\mu,\nu \in \mathscr{P}_1(\RR^k)$. Suppose $\mu \ll \nu$. Then 
\begin{equation}\label{W1_bound}
W_1(\mu, \nu) \leq \int |x| \cdot |r(x)-1| \, \nu(dx) < \infty, \quad r:=\frac{d \mu}{d \nu}.
\end{equation}
\end{lemma}
\begin{proof}
Take $\psi \in Lip_1(\RR^{k})$. Then, by definition of the Radon-Nikodym derivative, we have 
\[
\int \psi(x) \, \mu (dx) - \int \psi(x) \, \nu (dx) = \int (\psi(x)-\psi(0))(r-1) \, \nu(dx)
\]
and hence
\[
W_1(\mu,\nu) = \sup \Big\{ \int (\psi(x)-\psi(0))(r-1) \, \nu(dx), \,\, \psi \in Lip_1(\RR^k) \Big\}.
\]
Since $\varphi \in Lip_1(\RR^k)$, $|\psi(x)-\psi(0)| \leq |x|$,
which implies \eqref{W1_bound}.
\end{proof}

\subsubsection{Proof of Lemma \ref{lmm::_marg_value_bound_px} } \label{app::lmm::_marg_value_bound_px}

\begin{proof}
Since $\intP_X \ll P_X$, for each $S \subseteq N$ we have $P_{X_S} \otimes P_{X_{-S}} \ll P_X$ and hence, by Corollary \ref{app::lmm::rad_nik_prob_meas},  
the Radon-Nikodym derivative exists and satisfies $0 \leq r_S:=\frac{d P_{X_{S}} \otimes P_{X_{-S}}}{d P_X} \in L^1(P_X)$. Then for any $A \in \mathcal{B}(\RR^n)$ we have
\[
\int_A r(x) \, P_X(dx) = \intP_X(A)= \frac{1}{2^n} \sum_{S \subseteq N} P_{X_S} \otimes P_{X_{-S}}(A) = \frac{1}{2^n} \sum_{S \subseteq N} \int_A r_S(x) \, P_X(dx).
\]
Since $A \in \mathcal{B}(\RR^n)$ is arbitrary, we conclude $r=\frac{1}{2^n}\sum_{S \subseteq N} r_S$. By Corollary \ref{app::lmm::rad_nik_prob_meas}, we have $\|r_S\|_{L^1(P_X)}=1$.
\end{proof}

\subsection{Properties of conditional and marginal game operators}\label{app::gameops}

\subsubsection{Proof of Lemma \ref{lmm::marg_game_wellposed}} \label{app::lmm::marg_game_wellposed}
\begin{proof}
Suppose that the map $[f] \in H_X \mapsto \vme(S;f,X) \in L^2(\P)$ is well-defined for every $S \subseteq N$. Suppose that $\intP_X \not \ll P_X$. Then there exists $A \subset \mathcal{B}(\RR^n)$ and $S \subset N$ such that $P_X(A)=0$ and $P_{X_S}\otimes P_{X_{-S}}(A)>0.$ Set $f_*(x)=1_A(x)$. Since $\|f_*\|^2_{H_X}=P_X(A)=0$, we conclude $f \in [0]_{H_X}$.  Hence $\vme(S;f_*,X)=\vme(S;0,X)$ $\PP$-almost surely. This however leads to a contradiction because
\[
\E[\vme(S;f_*,X)] = \int \int 1_A(x) \, P_{X_S} (dx_S) P_{X_{-S}} dx_{-S}) = P_{X_S} \otimes P_{X_{-S}}(A) > 0.
\]

Suppose that $\intP_X \ll P_X$. Any $f \in [0]_{H_X}$ is $P_X$-almost surely zero; it is thus 
almost surely zero with respect to $\tilde{P}_X$, in particular with respect to 
any probability measure $P_{X_S} \otimes P_{X_{-S}}$ where $S\subseteq N$. This implies that $\vme(S;f,X) \in L^2(\P)$ is zero:
$$
\E\big[|\vpdp(S;X;f)|\big]= \int |f(x_S,x_{-S})|\,[P_{X_S} \otimes P_{X_{-S}}](d x_S, dx_{-S})=0.
$$
\end{proof}

\subsubsection{Proof of Lemma \ref{lmm::cond_marg_value_bound}} \label{app::lmm::cond_marg_value_bound}

\begin{proof}

Take $f \in L^2_{r^2}(P_X)$. Take $S \subset N$. Then, by Lemma \ref{lmm::_marg_value_bound_px} and Lemma \ref{app::lmm::cond_marg_rel}, we have $f \in L^2_{1+r^2_S}(P_X)$ and \[
\begin{aligned}
& \E\big[ \vce(S;X,f) - \vme(S;X,f) \big]^2 = \E_{x_S \sim P_{X_S}}\big[ \big(\E[f(x_S,X_{-S})|X_S=x_S] - \E[f(x_S,X_{-S})]\big)^2 \big] \\
&= \int \Big( \int f(x_S,x_{-S}) P_{X_{S}|X_{-S}=x_{-S}}(dx_{-S}) - \int f(x_S, x_{-S}) P_{X_{-S}}(dx_{-S})\Big)^2 P_{X_S}(dx_S)\\
& \leq \| (r_S-1) \cdot f \|^2_{L^2(P_X)}.
\end{aligned}
\]
This proves $(i)$. The item $(ii)$ follows directly from $(i)$ and the  representation \eqref{lingameform} of $h$.
\end{proof}

\subsubsection{Proof of Proposition \ref{prop::cond_marg_value_bound}} \label{app::prop::cond_marg_value_bound}

\begin{proof}
Let $r \in L^{\infty}(P_X)$. Then by Proposition \ref{prop::margoperatorwellpos} we have 
$H_X \cong (L^2(\intP_X),\|\cdot\|_{L^2(P_X)}) \subseteq L^2(P_X)$. Take $f \in L^2(P_X)$. Then, by the definition of Radon-Nikodym derivative, we obtain 
\[
\int |f(x)|^2 \intP_X(dx) = \int r(x)|f(x)|^2 \, P_X(dx) \leq \|r\|_{L^{\infty}(P_X)} \int |f(x)|^2 \, P_X(dx) < \infty, 
\]
and hence $f \in H_X$. This proves that $H_X=L^2(P_X)$. 

The remaining part of the statement follows  from Lemma \ref{lmm::_marg_value_bound_px} and Lemma \ref{lmm::cond_marg_value_bound}$(ii)$.
\end{proof}

\subsubsection{Proof of Theorem \ref{prop::condoperator}  and related corollaries} \label{app::prop::condoperator_proof}

\begin{proof}
The linearity of the operator is a consequence of the linearity of the expected value. To estimate the norm, observe that
\[
\begin{aligned}
\E\big[\E[Z|X_{S\cup \{i\}}] - \E[Z|X_{S}]\big]^2 &= \E\big[\E[Z|X_{S\cup \{i\}}] - \E[\E[Z|X_{S\cup \{i\}}]|X_S]\big]^2 \\
&\leq Var( \E[Z|X_{S \cup \{i\}}]) \leq \|Z\|^2_{L^2(\PP)}.
\end{aligned}
\]
Hence $\|\oper_i^{\CE}[Z]\|_{L^2(\PP)} \leq \sum_{S \subseteq N: i \in N}|w(S,n)|\cdot\|Z\|_{L^2(\PP)}$ which gives the estimate of the operator norm \eqref{condoperatorstab}.

Next, note that the  operator $\oper_i^{\CE}$ can be expressed as
\begin{equation}\label{projform}
\oper_i^{\CE} = \sum_{S \subseteq N\setminus\{i\}} w(S,n) \big( P_{S\cup\{i\}} - P_{S} \big),
\end{equation}
where $P_S$ is the orthogonal projection operator with values in $L^2(\Omega,\sigma(X_{S}),\P)$ defined by $P_S[Z]:=\E[Z|X_S]$. Since $P_S$ and $I-P_S$ project on orthogonal spaces, we have
\[
\langle P_S[Z_1], Z_2 \rangle = \langle P_S[Z_1], P_S [Z_2] \rangle = \langle Z_1, P_S[Z_2] \rangle \quad \text{for all} \quad Z_1,Z_2 \in L^2(\Omega,\mathcal{F},\P),
\]
and hence, using \eqref{projform}, we conclude that $\oper_i^{\CE}={\oper_i^{\CE}}^*$. This proves $(i)$.

Suppose that $w(S,n) \geq 0$ for all $S \subseteq N$, and $X_i$ is independent of $X_{N \setminus \{i\}}$. Then
\[
\|\oper_i^{\CE}[X_i - \E[X_i]]\|_{L^2(\PP)} = \Big(\sum_{S \subseteq N\setminus\{i\}}w(S,n)\Big) \|X_i - \E[X_i]\|_{L^2(\PP)} \geq \|\oper_i^{\CE}\| \cdot \|X_i - \E[X_i]\|_{L^2(\PP)} 
\]
which implies $(ii)$. 

The first inclusion in $(iii)$ is obvious and the second one follows from \eqref{condoperator}.
Part $(iv)$ follows from $(iii)$ because ${\rm Ker}(\oper^{\CE})=\bigcap_{i=1}^n{\rm Ker}(\oper_i^{\CE})$ contains subspaces 
\begin{equation*}
\begin{split}
&\bigcap_{i=1}^n\{Z\in L^2(\PP):\, \E[Z|X_{S \cup \{i\}}] = \E[Z|X_{S}], \,\, S \subseteq N \setminus \{i\} \}\supseteq \\
&\{Z\in L^2(\PP): \E[Z|X]=const \,\, \text{$\P$-a.s.} \} \supseteq\{Z\in L^2(\PP):\, Z \indep X\}. 
\end{split}    
\end{equation*}
Suppose next the \hyperref[axiom:TPG]{(TPG)} property, i.e. \eqref{growth}, holds. Then for any $Z \in {\rm Ker}(\oper^{\CE})$:
\[
0=\sum_{i=1}^n |\oper_i^{\CE}[Z - \E[Z]]|\geq g(|\E[Z|X]-\E[Z]|,|N|) \geq 0.
\]
Since $g(a,n)=0$ if and only if $a=0$, we obtain $\E[Z|X]=\E[Z]$ $\PP$-a.s.  This concludes the proof of $(v)$. The property $(vi)$ follows directly from the efficiency property \hyperref[axiom:EP]{(EP)}.

\end{proof}

\begin{proof}[Proof of Corollary \ref{corr::cond_operator_cons}]
Follows immediately from that fact that $\epsilon\in {\rm Ker}(\oper^{\CE})$ due to Theorem \ref{prop::condoperator}$(iv)$.
\end{proof}

\begin{proof}[Proof of Corollary \ref{corr::cond_operator_cont}]\label{app::corr::cond_operator_cont}
Parts $(i)$ and $(iii)$ of the corollary follow immediately from $\bar{\oper}^{\CE}[f]=\oper^{\CE}[f(X);h,X]$, and parts $(i)$, $(iv)$ and $(v)$ of Theorem \ref{prop::condoperator}. 
Part $(ii)$ is more subtle: By Theorem \ref{prop::condoperator}$(vi)$, the efficiency property puts a constraint on
$\bar{\oper}^{\CE}[f]=\big(\bar{\oper}_1^{\CE}[f],\dots,\bar{\oper}_n^{\CE}[f]\big)$; its components should add up to $f(X)-\E[f(X)]$. As we shall see, this constraint allows for a better estimation of the norm of this vector. There is no loss of generality in assuming that $\E[f(X)]=0$ since constant functions lie in the kernel. Now it suffices to establish
$\|\bar{\oper}^{\CE}[f]\|_{L^{2}(\P)^n}\leq \|f\|_{L^2(P_X)}$.
Notice that
\begin{equation}\label{giant1}
\begin{split}
\|f\|^2_{L^2(P_X)}
&=\|f(X)\|^2_{L^2(\PP)}=\langle f(X),\sum_{i=1}^n\bar{\oper}_i^{\CE}[f]\rangle_{L^2(\PP)}
=\sum_{i=1}^n \langle f(X),\bar{\oper}_i^{\CE}[f]\rangle_{L^2(\PP)}\\
&=\sum_{i=1}^n\Big(\sum_{S \subseteq N\setminus\{i\}} w(S,n)\langle f(X),\E[f(X)|X_{S \cup \{i\}}]-\E[f(X)|X_S]\rangle_{L^2(\PP)}\Big)\\
&=\sum_{i=1}^n\Big(\sum_{S \subseteq N\setminus\{i\}} w(S,n)\|\E[f(X)|X_{S \cup \{i\}}]-\E[f(X)|X_S]\|^2_{L^2(\PP)}\Big).
\end{split}    
\end{equation}
The last equality is based on interpreting conditional expectation as orthogonal projections which indicates that the inner products
$\langle f(X)-\E[f(X)|X_S],\E[f(X)|X_S]\rangle_{L^2(\PP)}$, $\langle f(X)-\E[f(X)|X_{S \cup \{i\}}],\E[f(X)|X_{S \cup \{i\}}]\rangle_{L^2(\PP)}$ and
$\langle\E[f(X)|X_{S \cup \{i\}}]-\E[f(X)|X_S],\E[f(X)|X_S]\rangle_{L^2(\PP)}$
are all zero. The number $\|f\|^2_{L^2(P_X)}$, as described above, is not smaller than $\|\bar{\oper}^{\CE}[f]\|^2_{L^{2}(\P)^n}$ due to:
\begin{equation}\label{giant2}
\begin{split}
&\|\bar{\oper}^{\CE}[f]\|^2_{L^{2}(\P)^n}
=\sum_{i=1}^n\|\bar{\oper}_i^{\CE}[f]\|^2_{L^{2}(\P)}
=\sum_{i=1}^n \big\|\sum_{S \subseteq N\setminus\{i\}} w(S,n) 
\big[ \E[f(X)|X_{S\cup\{i\}}] - \E[f(X)|X_{S}] \big]\big\|^2_{L^{2}(\P)}\\
&\quad\leq \sum_{i=1}^n\bigg(\Big(\sum_{S \subseteq N\setminus\{i\}} w(S,n)\Big)
\Big(\sum_{S \subseteq N\setminus\{i\}} w(S,n)\|\E[f(X)|X_{S\cup\{i\}}] - \E[f(X)|X_{S}]\|^2_{L^{2}(\P)}\Big)\bigg)\\
&\quad= \sum_{i=1}^n\Big(\sum_{S \subseteq N\setminus\{i\}} w(S,n)\|\E[f(X)|X_{S \cup \{i\}}]-\E[f(X)|X_S]\|^2_{L^2(\PP)}\Big);
\end{split}    
\end{equation}
where on the second line we used Cauchy-Schwarz along with $w(S,n)\geq 0$ while the third line relies on  
$\sum_{S \subseteq N\setminus\{i\}} w(S,n)=1$ which follows from the efficiency property.
\end{proof}

\begin{proof}[Proof of Lemma \ref{lemm::cond_operator_split}]\label{app::cond_operator_split}
One just needs to examine the part of the proof of Corollary \ref{corr::cond_operator_cont}$(ii)$ which established 
\eqref{energyineq} in the case of $f_0= \E[f(X)]=0$ (and hence generally). That argument was based on expanding 
$\|f\|^2_{L^2(P_X)}=\sum_{i=1}^n\langle f(X),\bar{\oper}_i^{\CE}[f]\rangle_{L^2(\PP)}$
in \eqref{giant1}, and inequalities $\|\bar{\oper}_i^{\CE}[f]\|^2_{L^{2}(\P)}\leq\langle f(X),\bar{\oper}_i^{\CE}[f]\rangle_{L^2(\PP)}$
in \eqref{giant2}. Thus the equality in 
$$\|\bar{\oper}^{\CE}[f]\|^2_{L^2(\PP)^n}=\sum_{i=1}^n\|\bar{\oper}_i^{\CE}[f]\|^2_{L^{2}(\P)}\leq 
\sum_{i=1}^n\langle f(X),\bar{\oper}_i^{\CE}[f]\rangle_{L^2(\PP)}=\|f\|^2_{L^2(P_X)}$$
is achieved if and only if $\|\bar{\oper}_i^{\CE}[f]\|^2_{L^{2}(\P)}=\langle f(X),\bar{\oper}_i^{\CE}[f]\rangle_{L^2(\PP)}$
for all $i\in N$. The general case (when $f_0 \neq 0$) follows from applying this result to $f-f_0$ and using the fact that $\langle f_0 , \bar{\oper}^{\CE}_i[f;h,X]\rangle_{L^2(\P)} = f_0 \cdot  \E[\bar{\oper}^{\CE}_i[f;h,X]]=0$ if $h$ has the form \eqref{lingameform}, and the fact that constants are in the kernel of $\bar{\oper}^{\CE}$.
\end{proof}

\subsubsection{Proof of Theorem \ref{prop::margoperator} } \label{app::prop::margoperator_proof}

\begin{proof}
If $f=f_*$ $\tP_X$-a.s., then for any $S\subseteq N$
$$
v^{\ME}(S\cup\{i\};X,f)=v^{\ME}(S\cup\{i\};X,f_*)\,\, \text{$\P$-a.s.}
$$
which implies, in view of \eqref{lingameform}, that $\bar{\oper}^{\ME}$ is well-defined on $L^2(\intP_X)$.

Now let $\bar{f}_S(x_S)=\E[f(x_S,X_{-S})]$. Then
\[
\begin{aligned}
\| \bar{\oper}^{\ME}_i[f] \|_{L^2(\PP)} 
&\leq \sum_{S \subseteq N \setminus \{i\}} |w(S,n)| \cdot\| \bar{f}_{S \cup \{i\}}(X_{S \cup \{i\}}) - \bar{f}_{S}(X_{S}) \|_{L^2(\PP)} \\
&\leq \Big(\sum_{S \subseteq N \setminus \{i\}}w^2(S,n)\Big)^{\frac{1}{2}} 
\Big(\sum_{S \subseteq N \setminus \{i\}}\|\bar{f}_{S \cup \{i\}}(X_{S \cup \{i\}}) - \bar{f}_{S}(X_{S}) \|^2_{L^2(\PP)}\Big)^{\frac{1}{2}}
\\
% &\leq\Big(\sum_{S \subseteq N \setminus \{i\}}w^2(S,n)\Big)^{\frac{1}{2}} 
% \Big(2\sum_{S \subseteq N \setminus \{i\}}
% \big(\|\bar{f}_{S \cup \{i\}}(X_{S \cup \{i\}})\|^2_{L^2(\PP)} + \|\bar{f}_{S}(X_{S}) \|^2_{L^2(\PP)}\big)\Big)^{\frac{1}{2}}
% \\
% &\leq \Big(\sum_{S \subseteq N \setminus \{i\}}w^2(S,n)\Big)^{\frac{1}{2}}
% \Big(2\sum_{S \subseteq N}
% \|\bar{f}_{S}(X_{S})\|^2_{L^2(\PP)}\Big)^{\frac{1}{2}}
% \\
&= \Big(\sum_{S \subseteq N \setminus \{i\}}w^2(S,n)\Big)^{\frac{1}{2}}
\Big(2\sum_{S \subseteq N}\|f\|^2_{L^2(P_{X_S} \otimes P_{X_{-S}})}\Big)^{\frac{1}{2}}
\\
&=2^{\frac{n+1}{2}}\Big(\sum_{S \subseteq N \setminus \{i\}} w^2(S,n)\Big)^{\frac{1}{2}}\cdot \|f\|_{L^2(\intP_X)}.
\end{aligned}
\]
which establishes  $(i)$.

We next prove $(ii)$. Suppose $f=c$ $\tP_X$-a.s. for some constant $c \in \RR$. Let $f_*(x):=c$ for each $x \in \RR^n$. Note that for any $S \subseteq N$, including $S=\varnothing$, we have
\[
\vpdp(S \cup \{i\};X,f_*)-\vpdp(S;X,f_*)=0 \quad \text{$\PP$-a.s.},
\]
and from $\eqref{lingameform}$ it follows that $\bar{\oper}^{\ME}[f_*]=0 \in L^2(\PP)$.
Note that $f=f_*$ $\tP_X$-a.s. and hence, using the fact that $\bar{\oper}^{\ME}$ is well-defined, we conclude that 
 $\bar{\oper}^{\ME}[f]=0 \in L^2(\PP)$ which establishes $(ii)$.

Suppose that $f \in {\rm Ker}(\bar{\oper}^{\ME})$ and \eqref{growth} holds. Then
\[
0 = \sum_{i=1}^n |\bar{\oper}_i^{\ME}[f - \E[f(X)] | \geq g(|f(X)-\E[f(X)]|,n) \geq 0,
\]
and hence $f=\E[f(X)]$ $P_X$-a.s., which proves $(iii)$.

Suppose that $\tP_X \ll P_X$ and \eqref{growth} holds. Then for any constant $c \in \RR$, $f=c$ $P_X$-a.s. implies $f=c$ $\tP_X$-a.s. and hence, using $(ii)$ and $(iii)$, we obtain $(iv)$.

Next, if $f(x)=f(x_{N \setminus\{i\}})$, then $\bar{f}_{S \cup \{i\}}(X_{S \cup \{i\}})=\bar{f}_{S}(X_{S})$ and hence $\bar{\oper}_i^{\ME}[f]=0 \in L^2(\PP)$, which gives $(v)$ and $(vi)$. 

Finally, property $(vii)$ follows directly from the efficiency property  \hyperref[axiom:EP]{(EP)}.
\end{proof}

\subsubsection{Proof of Lemma \ref{lemm::discont} } \label{app::lemm::discont}
\begin{proof} % proof is checked.
Let us first assume that $\E[f_1(X)-f_2(X)]=0$. Then, without loss of generality, we can assume that $\E[\eta_i(X_i)]=0$; otherwise,
we can define $\tilde{\eta}_i:=\eta_i-\E[\eta_i(X_i)]$ and write $f_1-f_2=\sum_{i=1}^n \tilde{\eta}_i(x_i)$. 

 Let $\gamma_i:=\sum_{S \subseteq N\setminus\{i\}}w(S,n)$, $i \in N$. For each $i \in N$, we have
\[
\|\bar{\oper}^{\ME}_i[f_1-f_2;h,X]\|_{L^2(\PP)}
=|\gamma_i|\|\eta_i\|_{L^2(P_{X_i})}=\frac{|\gamma_i|}{2^n}\sum_{S \subseteq N}\|\eta_i\|_{L^2( P_{X_S} \otimes P_{X_{-S}} )}=|\gamma_i|\|\eta_i\|_{L^2(\intP_{X})}. 
\]
In light of \hyperref[axiom:NVA]{(NVA)}, $\gamma_i \neq 0, i\in N$. Hence 
\[
\|\bar{\oper}^{\ME}[f_1]-\bar{\oper}^{\ME}[f_2]\|^2_{L^2(\PP)^n} = \sum_{i=1}^n \gamma^2_i\|\eta_i\|^2_{L^2(\intP_X)} \geq C \|f_1-f_2\|^2_{L^2(\intP_X)}
\]
for some $C>0$ that depends on $n$ and $\{\gamma_i\}_{i\in N}$ only.

Suppose now $\E[f_1(X)-f_2(X)]=c$. Then, using the above inequality and Theorem \ref{prop::margoperator}$(iii)$,
we obtain
\[
\begin{aligned}
\|\bar{\oper}^{\ME}[f_1]-\bar{\oper}^{\ME}[f_2]\|_{L^2(\PP)^n} &= \|\bar{\oper}^{\ME}[f_1-f_2 - c]\|_{L^2(\PP)^n} \\
&\geq \sqrt{C}\|f_1-f_2-c\|_{L^2(\intP_X)} \geq \sqrt{C} \big(\|f_1-f_2\|_{L^2(\intP_X)} - |c| \big)
\end{aligned}
\]
which proves the statement.
\end{proof}

\subsubsection{Proof of Theorem \ref{prop::margoperatorwellpos} 
(well-posedness) }  \label{app::prop::margoperatorwellpos}

\begin{proof}
Given the definition of $\tP_X$ in \eqref{Ptilde}, one has $P_X(A)\leq 2^n\cdot\tP_X(A)$ for any Borel subset of $\Bbb{R}^n$. In particular, $P_X\ll\tP_X$ and there exists a well-defined bounded linear map $\tilde{I}:L^2(\tP_X)\rightarrow L^2(P_X)$ that takes the $L^2(\tP_X)$-class of a function to its $L^2(P_X)$-class; notice that $\tilde{I}$ is not necessarily injective or surjective in general. Observe that $H_X$ is the image of $\tilde{I}$; and recall that ${\rm{Im}}(\tilde{I})$ can be identified with $L^2(\tP_X)/{\rm{Ker}}(\tilde{I})$ as vector spaces. 
Thus, the well-defined operator $(\bar{\oper}^{\ME},L^2(\tP_X))$ 
can be pushforwarded via $\tilde{I}$ to  a well-defined operator  $(\bar{\oper}^{\ME},H_X)$ 
if and only if 
$$H_X^0={\rm{Ker}}(\tilde{I})=\big\{f \in L^2(\tP_X): \|f\|_{L^2(P_X)}=0 \big\}\subseteq {\rm Ker}(\bar{\oper}^{\ME},L^2(\tP_X)).$$  
Part $(ii)$ describes the situation where $H_X^0$ is non-trivial while part $(i)$ addresses the case where $\tilde{I}$ is an embedding onto the subspace $H_X$. The latter happens precisely when $\tP_X\ll P_X$.  This finishes the proof.
\end{proof}

\subsubsection{Proof of Lemma \ref{lmm::marggameunbound} (game boundedness)}  \label{app::lmm::marggameunbound}

    \begin{proof}
    
By Lemma \ref{app::lmm::rad_nik_equiv} the condition \eqref{boundradon} is equivalent to $r=\frac{d \intP_X }{dP_X} \in L^{\infty}(P_X)$. Then for any $S \subseteq N$
\begin{equation*}
\begin{aligned}
\|v(S;X,f)\|^2_{L^2(\PP)} &\leq \int |f(x)|^2 \, [P_{X_S} \otimes P_{X_{-S}}](dx_S,dx_{-S}) \\
&=\int r_S |f(x)|^2 P_X(dx) \leq  2^n \|r\|_{L^{\infty}(P_X)} \int |f(x)|^2 \, P_X(dx)
\end{aligned}
\end{equation*}
where $r_S$ is given by \eqref{radon_decomp}. This proves $(i)$.

Let $\varnothing \neq S \subset N$. First, suppose that the condition \eqref{blowupcondgame} holds. Suppose that $A \in \mathcal{B}(\RR^{|S|})$, $B \in \mathcal{B}(\RR^{|-S|})$, and $P_X(A \times B)>0$. Set $f(x)=1_{A}(x_S) \cdot 1_B(x_{-S})$. Then
\[
 \begin{aligned}
 &\frac{\E[\vme(S;X,f)]^2}{\|f\|^2_{L^2}(P_X)} = \frac{1}{P_X( A \times B)}\int \Big(\int 1_{A}(x_S) \cdot 1_B(x_{-S}) \, P_{X_{-S}}(dx_{-S}) \Big)^2 P_{X_S}(dx_S) \\
 &=\frac{P_{X_S}(A)\cdot (P_{X_{-S}}(B))^2}{P_X(A \times B)}=
 \frac{[P_{X_S} \otimes P_{X_{-S}}](A \times B)}{P_X(A \times B)} \cdot P_{X_{-S}}(B).
\end{aligned}
\]
Then \eqref{blowupcondgame} and the relationship above imply that the map $f \in H_X \mapsto \vme(S;X,f) \in L^2(\PP)$ is unbounded. This proves the first part of $(ii)$.

To prove the second part, suppose the map $f \in H_X  \mapsto \vme(S;X,f) \in L^2(\P)$ is bounded. Then there exists $c_*>0$ such that for any $f \in H_X$ we have
\[
\int \Big( \int f(x_S,x_{-S}) P_{X_{-S}} (dx_{-S})\Big)^2 P_{X_{S}}(dx_{S}) \leq c_* \int f^2(x) P_X(dx).
\]
Let $A \in \mathcal{B}(\R^{|S|})$. Then, by above and the definition of $r_S \geq 0$, we obtain
\[
\int_A \rho^2(x_S) P_{X_S}(dx_S) \! = \! \int_A \! \Big( \int r_S^{1/2}(x_S,x_{-S}) P_{X_{-S}} (dx_{-S})\Big)^2 P_{X_{S}}(dx_S) \leq  c_* \! \int 1_{A}(x_S) P_{X_S}(dx_S).
\]
Since $A$ was arbitrary, $0\leq \rho^2 \leq c_*$ $P_{X_S}$-almost surely. This proves the second part of $(ii)$.
\end{proof}

\subsubsection{Proof of Theorem \ref{thm::margoperatorunbound} (boundedness)}  \label{app::thm::margoperatorunbound}

\begin{proof}

By Lemma \ref{app::lmm::rad_nik_equiv} the condition \eqref{boundradon} is equivalent to $\intP_X \ll P_X$ with $\frac{d \intP_X }{dP_X} \in L^{\infty}(P_X)$. Hence by Proposition \eqref{prop::cond_marg_value_bound} we have $H_X = L^2(P_X)$ and therefore for every $f \in L^2(P_X)$ we have
\[
\| \bar{\oper}^{\ME}_i[f; h,X] \|_{L^2(\PP)} \leq \| \bar{\oper}^{\ME}_i[f; h,X]  -  \bar{\oper}^{\CE}_i[f; h,X] \|_{L^2(\PP)} + \| \bar{\oper}^{\CE}_i[f; h,X] \|_{L^2(\PP)}=:I_1 + I_2.
\]
Combining the bound for $I_1$ given by Proposition \ref{prop::cond_marg_value_bound} and the bound for $I_2$ obtained from  Theorem \ref{prop::condoperator}$(i)$ together with the definition \ref{def::condoperator}, we obtain \ref{marg_value_bound}. This proves $(i)$.

Suppose next \eqref{blowupcond} holds for some distinct $i,j\in N$. 
Let $w(S,n)$ be the weights as in \eqref{lingameform}. Define 
\begin{equation}\label{cocarcoef}
w_{\{i,j\}}:= \sum_{S\subset N: i \notin S, j \notin S} w(S,n), \quad w_{\{i\}}:=\sum_{S\subset N: i \notin S, j \in S} w(S,n), \quad w_{\{j\}}:=\sum_{S\subset N: i\in S, j \notin S} w(S,n).
\end{equation}

Suppose \hyperref[axiom:NN]{(NN)} holds, that is, $w(S,n)\geq0$ for $S \subset N$. Suppose also \eqref{weightsblowup} holds for the indices $i,j$. Then 

\begin{equation}\label{coeflingamehyp}
\underline{w}_{i,j} := \min \{ |w_{\{i\}}|, |w_{\{j\}}|, |w_{\{i,j\}}| \} > 0.
\end{equation}
% which holds true if $w(S,n)>0$ for each non-empty $S \subseteq N$. 
For instance, for the Shapley value, one always has 
$w_{\{i\}}=w_{\{i,j\}}=\frac{1}{2}$ (which allows one to simplify some of the computations below; cf. Remark \ref{BCisSharp}). 

First, consider a special case $n=2$. In that case, we have $i=1,j=2$. Let $R=A \times B \subseteq \RR^2$  where $A,B$ are Borel sets. Denote $f_R:=\1_R(x_1,x_2)=\1_A(x_1) \1_B(x_2)$. Then, by \eqref{lingameform}  for $i=1$, we obtain
\[
\begin{aligned}
\bar{\oper}^{\ME}_1[f_R] = h_1[\vpdp(\cdot;X,f_R)]&=w(\varnothing) \big[ \vpdp(\{1\};X,\1_R) -  \vpdp(\varnothing;X,\1_R) \big]\\
&\quad  + w(\{2\}) \big[ \vpdp(\{1,2\};X,\1_R) -  \vpdp(\{2\};X,\1_R) \big]\\
&=w(\varnothing) \big( \1_A(X_1) P_{X_2}(B) -  P_X(R)) \big)\\
&\quad  + w(\{2\}) \big( \1_R(X_1,X_2) -  \1_B(X_2) P_{X_1}(A) \big)\\
\end{aligned}
\]
where we suppress the dependence on $n$ in the coefficients $w(S,n)$. 

Let us denote $p:=P_X(R)$, $\alpha:=P_{X_1}(A)$, and $\beta := P_{X_2}(B)$. Then
\[
\begin{aligned}
 (\bar{\oper}^{\ME}_1[f_R])^2 &  
% &=w^2(\{1\}) \big(  \1_A(x_1) \beta  - p \big)^2 + w^2(\{1,2\}) \big( \1_R -  \1_B(x_2) \alpha  \big)^2 \\
% & \quad + 2w(\{1\})w(\{1,2\})(\1_A(x_1)\beta - p)(1_R - 1_B(x_2) \alpha)\\
% &=w^2(\{1\}) \big(  \1_A(x_1) \beta^2 + p^2 - 2 \1_A(x_1) \beta p \big) \\
% & \quad + w^2(\{1,2\}) \big( \1_R(x_1,x_2) +  \1_B(x_2) \alpha^2   -  \1_R(x_1,x_2) 2\alpha  \big)^2 \\
% & \quad + 2w(\{1\})w(\{1,2\})( \1_R(x_1,x_2) \beta - \1_R(x_1,x_2) \alpha \beta  - p \1_R(x_1,x_2) + \1_B(x_2) \alpha p) \\
=w^2(\varnothing) \big(  \1_A(X_1) \beta^2 + p^2 - \1_A(X_1) 2 \beta p \big) \\
& \quad + w^2(\{2\}) \big( \1_R(X_1,X_2) +  \1_B(X_2) \alpha^2   -  \1_R(X_1,X_2) 2\alpha  \big) \\
& \quad + 2w(\varnothing)w(\{2\})( \1_R(X_1,X_2) \beta - \1_R(X_1,X_2) \alpha \beta  - p \1_R(X_1,X_2) + \1_B(X_2) \alpha p). 
\end{aligned}
\]
Then, taking the expectation we obtain
\[
\begin{aligned}
& \E[(\bar{\oper}^{\ME}_1[f_R])^2]  \\
&=w^2(\varnothing) \big(  \alpha \beta^2 + p (p- 2 \alpha \beta)  \big) \\
& \quad + w^2(\{2\}) \big( p(1-2\alpha) +  \beta \alpha^2   \big) \\
& \quad + 2w(\varnothing)w(\{2\}) p ( \beta - p )\\
& \geq  w^2(\varnothing) \alpha \beta^2 + w^2(\{2\}) \beta \alpha^2 - 2p (w^2(\varnothing)+w^2(\{2\})) - 2p|w(\varnothing)w(\{2\})|.
\end{aligned}
\]

Note that $\|f_R\|^2_{L^2(P_X)}=P_X(R)=p$ and hence, assuming that $p>0$, we conclude
\begin{equation*}
\begin{aligned}
& \frac{1}{\|f_R\|^2_{L^2(P_X)}}\E[(\bar{\oper}^{\ME}_1[f_R])^2]  \geq  \underline{w} ^2 \bigg( \frac{\alpha \beta^2 + \beta \alpha^2}{p} \bigg) - 6 \overline{w}^2 
\end{aligned}
\end{equation*}
where $\underline{w}:=\min_{S \subset N} |w(S)|$ and $\overline{w}:=\max_{S \subset N} |w(S)|$. 
Now if \eqref{blowupcond} and \eqref{weightsblowup} hold for $i=1$ and $j=2$, then \eqref{coeflingamehyp} holds for $i=1$ and $j=2$ and hence $\underline{w}^2>0$ in the inequality above. Then the right-hand side of the inequality is unbounded and, hence, $(\bar{\oper}_1^{\ME},H_X)$ is unbounded.

Performing similar calculations for $\bar{\oper}_2^{\ME}$, we come to the  conclusion that if \eqref{blowupcond} and \eqref{weightsblowup} hold for $i=1$ and $j=2$, then $(\bar{\oper}_2^{\ME},H_X)$ is unbounded. This proves  part $(iii)$ for $n=2$.

Next, consider a general case of $n \geq 2$.  Suppose \eqref{blowupcond} holds with for some distinct $i,j \in \{1,2,\dots,n\}$. Let $R=A \times B \subseteq \RR^2$, where $A,B$ are Borel sets. Define a function of $n$ variables as follows $f_R(x_1,x_2,\dots,x_n):=\1_{R}(x_i,x_j)$. By construction, $f_R$ does not depend explicitly on $x_k$ for each $k\in N\setminus \{i,j\}$, and hence by Theorem \ref{prop::margoperator}$(vi)$, 
 $T:=\{i,j\}$ is a carrier for  $\vpdp(\cdot; X, f_R)$. 
Hence, by \eqref{lingameform}, we obtain
\[
\begin{aligned}
\bar{\oper}_i^{\ME}[f_R] &= w_{\{i,j\}} \Big(\vpdp(\{i\};X,\1_R)-\vpdp(\varnothing;X,\1_R) \Big)\\
& \quad + w_{\{i\}}\Big(\vpdp(\{i,j\};X,f_R)-\vpdp(\{j\};X,f_R)\Big) \\
\bar{\oper}_j^{\ME}[f_R] &= w_{\{i,j\}} \Big(\vpdp(\{j\};X,f_R)-\vpdp(\varnothing;X,f_R) \Big)\\
& \quad + w_{\{j\}}\Big(\vpdp(\{i,j\};X,f_R)-\vpdp(\{i\};X,f_R)\Big)
\end{aligned}
\]
where $w_{\{i\}},w_{\{j\}}$, and $w_{\{i,j\}}$ are defined in \eqref{cocarcoef}.

Note that for each $S \subseteq T=\{i,j\}$
\[
\vpdp(S;X,f_R)=\vpdp(S;(X_i,X_j),\1_R(x_i,x_j)).
\]
Then, denoting $\alpha:=P_{X_i}(A)$, $\beta:=P_{X_j}(B)$, $p:=P_{(X_i,X_j)}(R)$ and proceeding as in the case $n=2$, we obtain
 \begin{equation}\label{lowboundngen}
\begin{aligned}
& \frac{\E[(\bar{\oper}^{\ME}_i[f_R])^2]}{\|f_R\|^2_{L^2(P_X)}},\frac{\E[(\bar{\oper}^{\ME}_j[f_R])^2]}{\|f_R\|^2_{L^2(P_X)}} \geq 
 \underline{w}_{i,j}^2 \bigg( \frac{\alpha \beta^2 + \beta \alpha^2}{p} \bigg) - 6 \overline{w}^2_{i,j},
\end{aligned}
\end{equation}
where $\underline{w}_{i,j}$ is defined in \eqref{coeflingamehyp}, 
and $\overline{w}_{i,j}:=\max \{ |w_{\{i\}}|,|w_{\{j\}}|,|w_{\{i,j\}}|\}$, 
where we have assumed that $\|f_R\|^2_{L^2(P_X)}=P_{(X_i,X_j)}(R)=p>0$. 

Note that by \eqref{coeflingamehyp} the coefficient $\underline{w}_{i,j}>0$. Then,  \eqref{blowupcond} for the given distinct indices $i,j \in N$ together with \eqref{lowboundngen} imply that $\bar{\oper}_i^{\ME}$, $\bar{\oper}_j^{\ME}$, and $\bar{\oper}^{\ME}$ are unbounded on $H_X$. This proves $(iii)$.
\end{proof}

Theorem \ref{thm::margoperatorunbound} on the boundedness/unboundedness of marginal game operators can be extended to linear game values that are not in the form of \eqref{lingameform}, or more generally, to coalitional game values; cf. Proposition \ref{prop::coalqgameprop}$(iii)$.
\begin{proposition}\label{prop::coalvalblowup}
With the notation as before, denote the predictors by $X=(X_1,\dots,X_n)$. Take $\mathcal{P}$ to be a partition of the predictors. 
Let $g$ be a linear coalitional value and $\bar{g}$ an extension of it. Therefore, there are constants $\{\gamma(i,N,\mathcal{P},S)\}_{i \in N, S \subseteq N}$ such that  
\begin{equation*}%\label{app::lingamevalrepr}
\bar{g}_i[N,v,\mathcal{P}]=\sum_{S \subseteq N} \gamma(i,N,\mathcal{P},S)v(S), \quad i \in N,
\end{equation*}
for any game $v$ (cooperative or non-cooperative); see Lemma \ref{lmm::lingamevalrepr}.  Suppose that \eqref{blowupcond} holds for distinct indices $i,j \in N$, and that for some $k\in N$
% Suppose that there exist distinct indices $i,j \in N$ for which \eqref{blowupcond} holds, 
\begin{equation}\label{coeflingamehyp_gen}
\underline{\gamma}^{(k)}_{i,j}:=\min \{ |\gamma_{\{i\}}^{(k)}|, |\gamma_{\{k\}}^{(k)}|, |\gamma_{\{i,j\}}^{(k)}| \} > 0
\end{equation}
where
\begin{equation*}\label{cocarcoef_gen}
\begin{aligned}
&\gamma_{\{i\}}^{(k)}:= \sum_{S\subseteq N: i \in S, j \notin S} \gamma(k,N,\mathcal{P},S), \,\, 
\gamma_{\{j\}}^{(k)}:=\sum_{S\subseteq N: i \notin S, j \in S } \gamma(k,N,\mathcal{P},S), \,\, \\
&\gamma_{\{i,j\}}^{(k)}:=\sum_{S\subseteq N: i,j \in S} \gamma(k,N,\mathcal{P},S).
\end{aligned}
\end{equation*}
Suppose the coalitional marginal game operator $f \mapsto \bar{\oper}^{\ME}[f; g,X] := \bar{g}[N,\vpdp(\cdot;X,f),\mathcal{P}]$ is well-defined on $H_X$. Then $(\bar{\oper}_k^{\ME}[\cdot; g,X],H_X)$, and thus $(\bar{\oper}^{\ME}[\cdot; g,X],H_X)$,  are unbounded.
\end{proposition}

\begin{proof}
As in the proof of Theorem \ref{prop::margoperatorwellpos}, take $R$ to be a rectangle $A\times B$ and set $f_R(x_1,x_2,\dots,x_n):=\1_{R}(x_i,x_j)$. As before, we denote $P_{X_i}(A)$, $P_{X_j}(B)$ and $P_{(X_i,X_j)}(R)$ by 
$\alpha$, $\beta$ and $p$ respectively. Notice that for any $S\subseteq N$ one has  
$$
\vpdp(S;X,f_R)=\vpdp(S\cap\{i,j\};X,f_R),
$$
and this term becomes $p=\|f_R\|^2_{L^2(P_X)}$ if $S\cap\{i,j\}=\varnothing$. Define
$$
\hat{\gamma}^{(k)}_{i,j}:=\sum_{S\subseteq N: i,j \notin S}\gamma(k,N,\mathcal{P},S).
$$
We now have
\begin{equation*}
\begin{aligned}
&\bar{\oper}^{\ME}_k[f_R; g,X]-\hat{\gamma}^{(k)}_{i,j}\cdot\|f_R\|^2_{L^2(P_X)}\\
&=\gamma_{\{i\}}^{(k)}\cdot\vpdp(\{i\};X,f_R)+\gamma_{\{j\}}^{(k)}\cdot\vpdp(\{j\};X,f_R)+\gamma_{\{i,j\}}^{(k)}\cdot\vpdp(\{i,j\};X,f_R)\\
&=\gamma_{\{i\}}^{(k)}\cdot \1_A(X_i)P_{X_j}(B)+\gamma_{\{j\}}^{(k)}\cdot P_{X_i}(A)\1_B(X_j)+\gamma_{\{i,j\}}^{(k)}\cdot \1_R(X_i,X_j),
\end{aligned}
\end{equation*}
which implies 
\begin{equation*}
\begin{split}
\E\Big[\big|\bar{\oper}^{\ME}_k[f_R; g,X]-\hat{\gamma}^{(k)}_{i,j}\cdot\|f_R\|^2_{L^2(P_X)}\big|^2\Big]
&= \big(\gamma_{\{i\}}^{(k)}\big)^2\alpha\beta^2+\big(\gamma_{\{j\}}^{(k)}\big)^2\alpha^2\beta + \big(\gamma_{\{i,j\}}^{(k)}\big)^2 p\\
&\quad+2\gamma_{\{i\}}^{(k)}\gamma_{\{j\}}^{(k)}\alpha\beta p+2\gamma_{\{i\}}^{(k)}\gamma_{\{i,j\}}^{(k)}\beta p
+2\gamma_{\{j\}}^{(k)}\gamma_{\{i,j\}}^{(k)}\alpha p\\
&\geq \big({\underline{\gamma}^{(k)}_{i,j}}\big)^2(\alpha\beta^2+\alpha^2\beta+p)-6\big({\overline{\gamma}^{(k)}_{i,j}}\big)^2p,
\end{split}    
\end{equation*}
where $\overline{\gamma}^{(k)}_{i,j}:=\max \{ |\gamma_{\{i\}}^{(k)}|, |\gamma_{\{k\}}^{(k)}|, |\gamma_{\{i,j\}}^{(k)}| \}$.
We conclude that
\begin{equation*}
\begin{split}
\frac{\|\bar{\oper}^{\ME}_k[f_R; g,X]\|^2_{L^2(\PP)}}{\|f_R\|^2_{L^2(P_X)}}
&\geq 
\frac{1}{2}\cdot\frac{\E\Big[\big|\bar{\oper}^{\ME}_k[f_R; g,X]-\hat{\gamma}^{(k)}_{i,j}\cdot\|f_R\|^2_{L^2(P_X)}\big|^2\Big]
}{\|f_R\|^2_{L^2(P_X)}}-\big(\hat{\gamma}^{(k)}_{i,j}\big)^2p\\
&\geq \frac{1}{2}\big({\underline{\gamma}^{(k)}_{i,j}}\big)^2\big(\frac{\alpha^2\beta+\alpha\beta^2}{p}\big)
+\frac{1}{2}\Big(\big({\underline{\gamma}^{(k)}_{i,j}}\big)^2-6\big({\overline{\gamma}^{(k)}_{i,j}}\big)^2\Big)-\big(\hat{\gamma}^{(k)}_{i,j}\big)^2.
\end{split}
\end{equation*}
The terms appearing on the last line are all constants except $\frac{\alpha^2\beta+\alpha\beta^2}{p}$ which can become arbitrarily large according to \eqref{blowupcond}. Given that ${\underline{\gamma}^{(k)}_{i,j}}>0$ due to \eqref{coeflingamehyp_gen}, the same is true about   
$\frac{\|\bar{\oper}^{\ME}_k[f_R; g,X]\|^2_{L^2(\PP)}}{\|f_R\|^2_{L^2(P_X)}}$, hence the unboundedness of 
$f\mapsto \bar{\oper}^{\ME}_k[f; g,X]$.
\end{proof}

\subsection{On condition (UO)}\label{app::on_condition_uo}
In our investigation of the marginal explanation operators, in Theorem 3.16, we set forth a condition that, if true, causes the operators to be unbounded with respect to the $\|\cdot\|_{L^2(P_X)}$ norm even when $\intP_X\ll P_X$. Recall the \eqref{blowupcond} (Unbounded Operator): 
\begin{equation}\label{suppl::blowupcond} %\tag{UO}
\sup \bigg\{ \frac{[P_{X_i} \otimes P_{X_j}](A \times B)}{P_{(X_i,X_j)}(A \times B)} \cdot P_{X_j}(B), \,\,\, A,B \in \mathcal{B}(\RR), P_{(X_i,X_j)}(A \times B)>0 \bigg\} = \infty.
\end{equation}
Theorem 3.16 asserts that, given predictors $X=(X_1,\dots,X_n)$ and a game value $(N,v)\mapsto h[N,v]=(h_i[N,v])_{i\in N}$ whose coefficients satisfy a positivity condition specified therein, if \eqref{suppl::blowupcond} is satisfied for distinct indices $i,j\in N$, 
then the associated maps $f\mapsto \bar{\oper}^{\ME}_i[f;h,X]$ and $f\mapsto \bar{\oper}^{\ME}_j[f;h,X]$ are unbounded when the domain is equipped with $\|\cdot\|_{L^2(P_X)}$. Here, we point out that the expression in \eqref{suppl::blowupcond} emerges naturally when $h$ is the Shapley value $\varphi$ (whose coefficients are of course positive). 
With $R=A\times B$, and setting $f_R(x):=\mathbbm{1}_R(x_i,x_j)$, we shall argue that 
\begin{equation}\label{aux}
\frac{\|\bar{\oper}_i^{\ME}[f_R;\varphi,X]\|_{L^2(\PP)}^2}{\|f_R\|_{L^2(P_X)}^2}
=\frac{1}{4}\frac{[P_{X_i} \otimes P_{X_j}](R)}{P_{(X_i,X_j)}(R)}\big(P_{X_i}(A)+P_{X_j}(B)\big)+O(1)    
\end{equation}
as $A$ and $B$ vary among Borel subsets of $\RR$ with $P_{(X_i,X_j)}(A\times B)>0$.
This will indicate that for the Shapley value, the unboundedness of marginal explanations, at least once restricted to indicator functions, results in condition \eqref{suppl::blowupcond} from the paper---hence motivating condition \eqref{suppl::blowupcond}. To establish the equality above, we revisit the following from the proof of Theorem 3.16:
\begin{equation*}
\begin{split}
\bar{\oper}_i^{\ME}[f_R;h,X] &= w_{\{i,j\}} \Big(\vpdp(\{i\};X,f_R)-\vpdp(\varnothing;X,f_R) \Big)\\
& \quad + w_{\{i\}}\Big(\vpdp(\{i,j\};X,f_R)-\vpdp(\{j\};X,f_R)\Big)\\
&=w_{\{i,j\}} \Big(\mathbbm{1}_A(X_i)P_{X_j}(B)-P_{(X_i,X_j)}(R)\Big)\\
& \quad + w_{\{i\}}\Big(\mathbbm{1}_R(X_i,X_j)-\mathbbm{1}_B(X_j)P_{X_i}(A)\Big)
\end{split}
\end{equation*}
where the $w_{\{i\}}$ and $w_{\{i,j\}}$ are defined it terms of the coefficients $w(S,n)\,(S\subset N)$ of the game value $h$ as:  
\begin{equation*}
w_{\{i,j\}}:= \sum_{S\subset N: i \notin S, j \notin S} w(S,n), \quad w_{\{i\}}:=\sum_{S\subset N: i \notin S, j \in S} w(S,n).
\end{equation*}
When $h=\varphi$, the coefficients are given by $w(S,n)=\frac{1}{n\binom{n-1}{|S|}}$, and:
\begin{equation*}
\begin{split}
&w_{\{i,j\}}=\sum_{s=0}^{n-2}\frac{1}{n\binom{n-1}{s}}\cdot\binom{n-2}{s}=\sum_{s=0}^{n-2}\frac{n-s-1}{n(n-1)}=\frac{(n-1)+\cdots+1}{n(n-1)}=\frac{1}{2},\\
&w_{\{i\}}=\sum_{s=1}^{n-1}\frac{1}{n\binom{n-1}{s}}\cdot\binom{n-2}{s-1}=\sum_{s=1}^{n-1}\frac{s}{n(n-1)}=\frac{1+\cdots+(n-1)}{n(n-1)}=\frac{1}{2}.
\end{split}
\end{equation*}
Substituting in the formula above, we have 
\begin{equation*}
\|\bar{\oper}_i^{\ME}[f_R;\varphi,X]\|_{L^2(\PP)}^2=
\frac{1}{4}\cdot\Bbb{E}\Big[\big(\mathbbm{1}_A(X_i)P_{X_j}(B)-P_{(X_i,X_j)}(R)+\mathbbm{1}_R(X_i,X_j)-\mathbbm{1}_B(X_j)P_{X_i}
(A)\big)^2\Big]
\end{equation*}
which can be simplified as 
\begin{equation*}
\frac{1}{4}\Big(P_{X_i}(A)P_{X_j}(B)^2+P_{X_i}(A)^2P_{X_j}(B)\Big)+P_{(X_i,X_j)}(R)\cdot(\text{a bounded term})
\end{equation*}
where the  bounded term in parentheses  is 
$$
\frac{1}{4}\Big(1-P_{(X_i,X_j)}(R)+2P_{X_j}(B)-2P_{X_i}(A)-2P_{X_i}(A)P_{X_j}(B)\Big)\in (-1,1).
$$
Dividing by $\|f_R\|_{L^2(P_X)}^2=P_{(X_i,X_j)}(R)$, we arrive at \eqref{aux}, as desired.

\subsection{On $H_X$ and the Radon-Nikodym derivative $r=\frac{d \intP_X}{ dP_X}$}\label{suppl::rn_boundedness}

In Theorem \ref{thm::margoperatorunbound} we established that if $r=\frac{d \intP_X}{ dP_X}$ exists and belongs to $L^{\infty}(P_X)$, then $H_X=L^2(P_X)$ where 
\begin{equation*} %\label{datasubspace_suppl}
\begin{aligned}
H_X &:= \bigg( \Big\{[f]: [f]=\big\{\tilde{f}: \text{$\tilde{f}=f$ $P_X$-a.s. and } \int |\tilde{f}(x)|^2 \intP_X(dx) < \infty \big\}\Big\}, \, \| \cdot \|_{L^2(P_X)} \bigg) \\
& \hookrightarrow L^2(P_X).
\end{aligned}
\end{equation*}

It turns out that the reverse is true as well. Specifically, we have the following.

\begin{lemma}\label{suppl::rn_noundedness}
Suppose $\intP_X \ll P_X$ and $r:=\frac{d \intP_X}{ dP_X}$. The following statements are equivalent:
\begin{itemize}

  \item[$(i)$] $r \in L^{\infty}(P_X)$.

  \item[$(ii)$] $H_X=L^2(P_X)$.
\end{itemize}
  \end{lemma}
\begin{proof}
First, suppose $r \in L^{\infty}(P_X)$. By construction, $H_X$ is a subset  of $L^2(P_X)$. Thus, to show that $H_X=L^2(P_X)$ it suffices to show that $L^2(P_X) \subseteq L^2(\intP_X)$. Pick any $f \in L^2(P_X)$. For any $k>0$ we have
\[
\begin{aligned}
\int 1_{\{|f|\leq k\}} f^2(x) \intP_X(dx) & = \int 1_{\{|f|\leq k\}} r(x) f^2(x) P_X(dx) \\
&\leq \|r\|_{L^{\infty}(P_X)} \int f^2(x) P_X(dx) < \infty.
\end{aligned}
\]
Then sending $k \to \infty$ and using the monotone convergence theorem we conclude that $f \in L^2(\tilde{P}_X)$. Thus, $L^2(P_X) \subseteq L^2(\intP_X)$. This proves that $H_X = L^2(P_X)$.

Next, suppose that $H_X=L^2(P_X)$. Then for every $f \in L^2(P_X)$ we have
\[
\infty > \int f^2(x) \intP_X(dx) = \int f^2(x) r(x) P_X(dx).
\]
Thus, for every $f \in L^2(P_X)$, we have $f r^{1/2} \in L^2(P_X)$. 

Set $A_k:=\{x \in \R^n: r(x)\geq k  \text{ $P_X$-a.s.}\}$ for every nonnegative integer $k\geq 0$. Suppose $r$ is not $P_X$-essentially bounded. Then $P_X(A_k)>0$ for every $k \geq 0$ and 
\[
f_*(x) := \bigg( \sum_{k=1}^{\infty} \frac{1}{k^2} 1_{A_k}(x) \frac{1}{P_X(A_k)} \bigg)^{1/2}
\]
is well-defined. Then, by the monotone convergence theorem we have
\[
\int f_*^2(x) P_X(dx) = \int \bigg(\sum_{k=0}^{\infty} \frac{1}{k^2} 1_{A_k}(x) \frac{1}{P_X(A_k)} \bigg) P_X(dx) = \sum_{k=0}^{\infty} \frac{1}{k^2} < \infty.
\]
Thus, $f_* \in L^2(P_X)$. However, for every $K \geq 0$ we have
\[
\int f_*^2(x) r(x) P_X(dx) \geq \int \big(\sum_{k=0}^{K} \frac{1}{k^2} 1_{A_k}(x) \frac{1}{P_X(A_k)}\bigg) r(x) P_X(dx) \geq \sum_{k=0}^{K} \frac{1}{k}.
\]
Sending $K \to \infty$, we conclude that $f_* r^{1/2} \notin L^2(P_X)$, which is a contradiciton. Hence $r$ is $P_X$-essentially bounded.
\end{proof}
In the above proof, to construct $f_*$, we used help from MathOverflow \cite{MathOverflow2}. As \cite{MathOverflow2} points out, an alternative proof is to show that $r$ induces a bounded linear functional on $L^1(P_X)$ using the uniform boundedness principle and then apply the Riesz representation theorem.

\subsection{On the relationship between probability measures $P_X$ and $\intP_X$}\label{app::measures_relationship}
The comparison of  probability measures $P_X$ and $\intP_X:=\frac{1}{2^n}\sum_{S\subseteq N}P_{X_S}\otimes P_{X_{-S}}$ lies at the heart of the analysis of conditional and marginal explanations carried out in this paper. Recall that the former is the joint probability distribution of predictors $X=(X_1,\dots,X_n)$ while the latter probability measure on $\Bbb{R}^n$ emerged naturally in our investigation of marginal explanations. 

\begin{proposition}
The following three statements are equivalent. 
\begin{enumerate}[label=(\alph*)]
\item The predictors are independent. 
\item $P_{X_S}\otimes P_{X_{-S}}$ coincides with $P_X$ for every $S\subseteq N$. 
\item $\intP_X$ coincides with $P_X$. 
\end{enumerate}
\end{proposition}
\begin{proof}
 Obviously $(a)\implies(b)\implies(c)$.  
It remains to show that $(c)\implies (a)$. We prove this by induction on $n$. First, we claim that if $\intP_X=P_X$ where 
$X=(X_1,\dots,X_{n})$, then any $n-1$ of these random variables are independent. By symmetry, it suffices to show that $X_1,\dots,X_{n-1}$ are independent. Let $\pi:\Bbb{R}^n\rightarrow\Bbb{R}^{n-1}$ denote the projection onto the first $n-1$ coordinates. Then 
$\pi_*P_{X}=P_{X'}$ where $X':=(X_1,\dots,X_{n-1})$. Also the pushforward of 
$\intP_X=\frac{1}{2^n}\sum_{S\subseteq N}P_{X_S}\otimes P_{X_{N\setminus S}}$ by $\pi$ is equal to $\intP_{X'}=\frac{1}{2^{n-1}}\sum_{S\subseteq N'}P_{X_S}\otimes P_{X_{N'\setminus S}}$ where $N':=\{1,\dots,n-1\}$.
This is due to the fact that for every $S\subseteq N'$, $P_{X_S}\otimes P_{X_{N'\setminus S}}$ can be realized as the pushforward of two terms in $\tP_X$: $P_{X_{S}}\otimes P_{X_{N\setminus S}}$ and $P_{X_{S\cup\{n\}}}\otimes P_{X_{N\setminus (S\cup\{n\})}}$. Consequently, applying $\pi_*$ to $\intP_X=P_X$ yields $\intP_{X'}=P_{X'}$, and thus by the induction hypothesis, the independence of $X_1,\dots,X_{n-1}$. Now since any $n-1$ of the random variables $X_1,\dots,X_n$ are independent, for any non-empty and proper subset $S$ of $N$ we have $P_{X_{S}}\otimes P_{X_{N\setminus S}}=P_{X_1}\otimes\cdots\otimes P_{X_n}$. When $S=\varnothing\text{ or }N$, the measure $P_{X_{S}}\otimes P_{X_{N\setminus S}}$ coincides with $P_X$. Therefore, $\intP_X=P_X$ amounts to 
$$
\frac{1}{2^n}\big((2^n-2)P_{X_1}\otimes\cdots\otimes P_{X_n}+2P_X\big)=P_X
$$
which results in $P_{X_1}\otimes\cdots\otimes P_{X_n}=P_X$, i.e. random variables $X_1,\dots,X_n$ are independent.
\end{proof}

Next, we move from equality $\intP_X=P_X$ to the continuity condition $\intP_X\ll P_X$.
The probability measure $\intP_X$ is a convex combination of the product measures $P_{X_S}\otimes P_{X_{-S}}$. The latter is $P_X$ when $S=\varnothing\text{ or }N$ which immediately indicates that  the other direction holds:
$P_X\ll \intP_X$. The condition $\intP_X\ll P_X$ amounts to $P_{X_S}\otimes P_{X_{-S}}\ll P_X$ for all $S\subseteq N$. 
As discussed extensively in the paper, this condition appears when it comes to setting up marginal explanations as well-defined operators. The goal here is to elaborate on it through providing some examples and non-examples.\footnote{Inspired by this problem, we had raised a question on MathOverflow \cite{MathOverflow}.} Especially, we elucidate this condition by relating it to the shape of the support of $P_X$. Recall that the  support ${\rm{supp}}(\mu)$ of a Borel measure $\mu$ on a metric space is the set of points whose every open neighborhood has a positive measure \cite{MeasureBook}. Its complement is thus the union of all measure zero open subsets. Hence ${\rm{supp}}(\mu)$ is automatically closed; and in the case of a separable space such as $\Bbb{R}^n$, the support can be characterized as the complement of the largest open subset of measure zero. 
\begin{lemma}
One always has ${\rm{supp}}(P_X)\subseteq{\rm{supp}}(\intP_X)$ and the supports coincide if $\intP_X\ll P_X$. Moreover, if ${\rm{supp}}(P_X)={\rm{supp}}(\intP_X)$, then for any $\varnothing\neq S\subset N$, they coincide with 
${\rm{supp}}\big(P_{X_S}\otimes P_{X_{-S}}\big)$ and $\pi_S\big({\rm{supp}}(P_X)\big)\times\pi_{-S}\big({\rm{supp}}(P_X)\big)$ where $\pi_S:\Bbb{R}^n\rightarrow\Bbb{R}^{|S|}$ and $\pi_{-S}:\Bbb{R}^n\rightarrow\Bbb{R}^{n-|S|}$ are projections onto coordinates belonging or not belonging to $S$ respectively.\footnote{Following our convention, ignoring the order of coordinates, a vector $x\in\Bbb{R}^n$ may be written as $(x_S,x_{-S})$, and this is how 
${\rm{supp}}\big(P_{X_S}\otimes P_{X_{-S}}\big)=\pi_S\big({\rm{supp}}(X)\big)\times\pi_{-S}\big({\rm{supp}}(X)\big)$
should be understood.}
\end{lemma}

\begin{proof}
For any two Borel measures $\mu$ and $\nu$ on $\Bbb{R}^n$, $\mu\ll\nu$ implies ${\rm{supp}}(\mu)\subseteq{\rm{supp}}(\nu).$ Thus  ${\rm{supp}}(P_X)\subseteq{\rm{supp}}(\intP_X)$ due to $P_X\ll\intP_X$; and also $\intP_X\ll P_X$ yields 
${\rm{supp}}(\intP_X)\subseteq{\rm{supp}}(P_X)$, and hence ${\rm{supp}}(P_X)={\rm{supp}}(\intP_X)$. 
Next, suppose ${\rm{supp}}(P_X)={\rm{supp}}(\intP_X)$. These sets should contain ${\rm{supp}}(P_{X_S}\otimes P_{X_{-S}})$ 
for any $S$ because $P_{X_S}\otimes P_{X_{-S}}\ll \intP_X$.
It follows easily from the definition of a measure's support that 
${\rm{supp}}(P_{X_S}\otimes P_{X_{-S}})={\rm{supp}}(P_{X_S})\times{\rm{supp}}(P_{X_{-S}})$
and 
${\rm{supp}}(P_{X_{\pm S}})\supseteq\pi_{\pm S}\big({\rm{supp}}(P_X)\big)$. Therefore:
\begin{equation*}\label{aux0}
\pi_S\big({\rm{supp}}(P_X)\big)\times\pi_{-S}\big({\rm{supp}}(P_X)\big)\subseteq 
{\rm{supp}}(P_{X_S}\otimes P_{X_{-S}})\subseteq{\rm{supp}}(\intP_X)={\rm{supp}}(P_X).    
\end{equation*}
But clearly ${\rm{supp}}(P_X)\subseteq\pi_S\big({\rm{supp}}(P_X)\big)\times\pi_{-S}\big({\rm{supp}}(P_X)\big)$. Consequently, all the subsets appeared above coincide.  
\end{proof}

The lemma clearly shows that $\intP_X\ll P_X$ requires the support of $P_X$ to have a ``product structure''.
\begin{corollary}
If ${\rm{supp}}(P_X)={\rm{supp}}(\intP_X)$, then ${\rm{supp}}(X)=\prod_{i\in N}\pi_i\big({\rm{supp}}(X)\big)$ where $\pi_i$ denotes the projection onto the $i^{{th}}$ coordinate. In particular, this holds when $\intP_X\ll P_X$.
\end{corollary}

\begin{proof}
Follows from fact that ${\rm{supp}}(P_X)=\pi_S\big({\rm{supp}}(P_X)\big)\times\pi_{-S}\big({\rm{supp}}(P_X)\big)$ for all  subsets
$\varnothing\neq S\subset N$ if ${\rm{supp}}(P_X)={\rm{supp}}(\intP_X)$. 
\end{proof}

The product structure ${\rm{supp}}(X)=\prod_{i\in N}\pi_i\big({\rm{supp}}(X)\big)$ puts a constraint on the support: Its projections to coordinate axes must be closed\footnote{Choosing arbitrary points $a_i\in\pi_i\big({\rm{supp}}(X)\big)$, due to this product decomposition, each $\pi_i\big({\rm{supp}}(X)\big)$ is the preimage of the closed subset ${\rm{supp}}(X)$ under the continuous map $\Bbb{R}\rightarrow\Bbb{R}^n:t\mapsto(a_1,\dots,a_{i-1},t,a_{i+1},\dots,a_n)$.}, something which does not hold generally for an arbitrary closed subset of $\Bbb{R}^n$.
In terms of the joint probability, the product structure means that the predictors take their values ``heterogenously'': Given numbers $a_1,\dots,a_n$, if for every $\epsilon>0$ there is a positive probability of $X_i$ lying in $(a_i-\epsilon,a_i+\epsilon)$, then the probability of $(X_1,\dots,X_n)$ belonging to any given open neighborhood of $(a_1,\dots,a_n)$ is non-zero. In contrast,
when the data lies on a ``complicated'' lower-dimensional submanifold of $\Bbb{R}^n$, we are in
a different regime where $\intP_X\ll P_X$ fails. This last assertion is made rigorous below:

\begin{corollary}
If ${\rm{supp}}(P_X)\subseteq\Bbb{R}^n$ is not a Cartesian product of $n$ subsets of $\Bbb{R}$, then $\intP_X$ cannot be absolutely continuous with respect to $P_X$. In particular, when ${\rm{supp}}(P_X)$ is connected, the continuity fails unless 
${\rm{supp}}(P_X)$ is a (possibly degenerate or unbounded or both) rectangular cube.  
\end{corollary}
\begin{proof}
As established above, $\intP_X\ll P_X$ yields the equality  ${\rm{supp}}(X)=\prod_{i\in N}\pi_i\big({\rm{supp}}(X)\big)$, which requires all subsets appearing on the right-hand side to be closed. If the support is connected, each projection $\pi_i\big({\rm{supp}}(X)\big)$ of it must be a connected subset of $\Bbb{R}$, i.e. an interval (closed and possibly degenerate). Therefore, 
${\rm{supp}}(X)$ is a product of intervals in that case.  
\end{proof}

Finally, we discuss the converse implication: Can the continuity of measures be deduced from assumptions about the supports? 
As a matter of fact, the equality of supports ${\rm{supp}}(P_X)={\rm{supp}}(\intP_X)$--which as we saw is a necessary condition for $\intP_X\ll P_X$, and implies that ${\rm{supp}}(P_X)$ has a product structure--can yield $\intP_X\ll P_X$ if the features are discrete, or admit a density function (with a small caveat, see below).

\begin{proposition}
The equality of supports ${\rm{supp}}(P_X)={\rm{supp}}(\intP_X)$ implies the continuity of measures $\intP_X\ll P_X$ under any of the following assumptions on the predictors:
\begin{enumerate}[label=(\roman*)]
\item The support of each $X_i$ is a discrete subset of $\Bbb{R}$. 
\item The joint probability distribution $P_{X}$ of $(X_1,\dots,X_n)$ admits a density function which is Lebesgue a.e. positive on 
${\rm{supp}}(P_X)$. 
\end{enumerate}
\end{proposition}

\begin{proof} 
When the closed subset ${\rm{supp}}(X_i)$ is discrete, the probability of $X_i$ belonging to a Borel subset of $\Bbb{R}$ is positive if and only if it intersects ${\rm{supp}}(X_i)$. The same is true for any random vector $X_S$ ($S\subseteq N$) in place of $X_i$ because 
${\rm{supp}}(X_S)$ (being contained in $\prod_{i\in S}{\rm{supp}}(X_i)$) is discrete too. Pick a subset $\varnothing\neq S\subset N$. It suffices to show $P_{X_S}\otimes P_{X_{-S}}\ll P_X$; that is, $P_{X_S}\otimes P_{X_{-S}}(B)=0$ for any Borel subset $B$ of $\Bbb{R}^n$ with $P_X(B)=0$. As discussed above, $B$ does not intersect ${\rm{supp}}(P_X)$. But this subset, according to the lemma, coincides with ${\rm{supp}}(P_{X_S}\otimes P_{X_{-S}})$ because the hypothesis is that ${\rm{supp}}(P_X)={\rm{supp}}(\intP_X)$.
So $B$ cannot intersect ${\rm{supp}}(P_{X_S}\otimes P_{X_{-S}})$ either. This support is discrete as well 
(being equal to ${\rm{supp}}(P_{X_S})\times{\rm{supp}}(P_{X_{-S}})$). We deduce that $P_{X_S}\otimes P_{X_{-S}}(B)=0$, as desired. 

For the second part, let $\rho$ be a density for $P_X$, a Borel measurable function $\rho:\Bbb{R}^n\rightarrow[0,\infty)$.  
Fix a subset $\varnothing\neq S\subset N$. The product measure $P_{X_S}\otimes P_{X_{-S}}$ admits a density function of form
$x\mapsto\rho_S(x_S)\rho_{-S}(x_{-S})$ where $\rho_S(x_S):=\int\rho(x_S,x_{-S})dx_{-S}$ and $\rho_{-S}(x_{-S}):=\int\rho(x_S,x_{-S})dx_S$. 
When a density exists, the measure of a Borel subset is zero if and only the density vanishes at Lebesgue-almost every point of it.  
Therefore, to establish $P_{X_S}\otimes P_{X_{-S}}\ll P_X$, we only need to show that 
$P_{X_S}\otimes P_{X_{-S}}\big(\big\{x\in\Bbb{R}^n\mid\rho(x)=0\big\}\big)=0$, or equivalently the Lebesgue measure of  
$\big\{x\in\Bbb{R}^n\mid\rho(x)=0,\rho_S(x_S)\rho_{-S}(x_{-S})\neq 0\big\}$ is zero. This subset is contained in the union
$$
\big\{x\in{\rm{supp}}(P_X)\mid\rho(x)=0\big\}\cup
\big\{x\in\Bbb{R}^n\setminus{\rm{supp}}(P_X)\mid\rho_S(x_S)\rho_{-S}(x_{-S})\neq 0\big\}.
$$
The first subset is of Lebesgue measure zero due to our assumption. Proving the same for the second one concludes the proof. As argued previously in this proof,  ${\rm{supp}}(P_X)$ coincides with 
${\rm{supp}}(P_{X_S}\otimes P_{X_{-S}})={\rm{supp}}(P_{X_S})\times{\rm{supp}}(P_{X_{-S}})$ because of ${\rm{supp}}(P_X)={\rm{supp}}(\intP_X)$. 
Hence $\big\{x\in\Bbb{R}^n\setminus{\rm{supp}}(P_X)\mid\rho_S(x_S)\rho_{-S}(x_{-S})\neq 0\big\}$ is contained in the union
$$
\big\{x\in\Bbb{R}^n\mid x_S\notin{\rm{supp}}(P_{X_S}), \rho_S(x_S)\neq 0\big\}\cup
\big\{x\in\Bbb{R}^n\mid x_{-S}\notin{\rm{supp}}(P_{X_{-S}}),\rho_{-S}(x_{-S})\neq 0\big\}.
$$
They are both of Lebesgue measure zero in $\Bbb{R}^n$ since subsets $\{\rho_S\neq 0\}\setminus{\rm{supp}}(P_{X_S})$ and 
$\{\rho_{-S}\neq 0\}\setminus{\rm{supp}}(P_{X_{-S}})$ are of Lebesgue measure zero in the corresponding Euclidean spaces $\Bbb{R}^{|S|}$ and $\Bbb{R}^{n-|S|}$ due to the fact that $\rho_S$ and $\rho_{-S}$ are respectively density functions for probability measures $P_{X_S}$ on $\Bbb{R}^{|S|}$ and  
$P_{X_{-S}}$ on $\Bbb{R}^{n-|S|}$. 
\end{proof}

\begin{example}\rm
We provide an example to show that the condition from the second part of theorem above on the values that the density function assumes  on the support is necessary.  Let $C\subset [0,1]$ be a ``fat'' Cantor set, i.e. a Cantor set of positive Lebesgue measure $\alpha\in (0,1)$. Let the density function of $X=(X_1,X_2)$ be $\rho:=\frac{1}{1-\alpha^2}\cdot\mathbbm{1}_{[0,1]^2\setminus C^2}$. So the probability distribution $P_X$ is continuous with respect to the Lebesgue measure, and its support is the whole square $[0,1]^2$ because $C^2$ is a closed and nowhere-dense subset of the square. But $\rho$ vanishes on the subset $C^2$ which is of positive Lebesgue measure. We argue that $P_{X_1}\otimes P_{X_2}(C^2)$, unlike $P_X(C^2)$, is non-zero. A density function for 
$P_{X_1}\otimes P_{X_2}$ is $(x_1,x_2)\mapsto\tilde{\rho}(x_1)\tilde{\rho}(x_2)$ where 
$$
\tilde{\rho}(t):=\frac{1}{1-\alpha^2}\cdot\begin{cases}
1 & t\in [0,1]\setminus C,\\
1-\alpha & t\in C.
\end{cases}
$$ 
This density of $P_{X_1}\otimes P_{X_2}$ is positive at every point of $[0,1]^2$ which yields 
${\rm{supp}}(P_{X_1}\otimes P_{X_2})=[0,1]^2$, and $P_{X_1}\otimes P_{X_2}(C^2)>0$ because the two-dimensional Lebesgue measure of $C^2$ is positive. Consequently, continuous probability distributions $P_X$ and $\intP_X=\frac{1}{2}\Big(P_X+P_{X_1}\otimes P_{X_2}\Big)$ have the same support $[0,1]^2$ while $\tP_X\not\ll P_X$ due to the fact that   
$$
P_X(C^2)=0<\intP_X(C^2).
$$
\end{example}

\subsection{Coalitional values with two-step formulation}

\subsubsection{Canonical representation of coalition values with two-step formulation}\label{app::twostepcanon}

\begin{lemma}\label{lmm::canonical}
Let $g$ be a coalitional value with a two-step formulation with $h^{(1)}$, $h^{(2)}$ and  
the intermediate game $\hat{v}_T(N,v,\cP)$  
as in Definition \ref{def::2stepprop} and suppose that $g_1[\{1\},\tilde{u},\{\{1\}\}] \neq 0$. Then $g$ satisfies \hyperref[axiom:SIP]{(SIP)} if and only if there exists a unique constant $\alpha_* \neq 0$ and unique game values $h_*^{(1)}$, $h_*^{(2)}$ independent of $\hat{v}_T$ such that
\begin{equation}\label{canonrep}
g_i[N,v,\cP] = \alpha_* h^{(2)}_{*,i}[S_j,v^{(j)}], \quad v^{(j)}(T)=h_{*,j}^{(1)}[M,\hat{v}_T], \quad i \in S_j, 
\end{equation}
where $h_{*,1}^{(1)}[\{1\},v]=v(1)$ and $h_*^{(2)}$ satisfies \hyperref[axiom:SEP]{(SEP)}. As a consequence, we have
\begin{equation}\label{renormconn}
g[N,v,\bar{N}]=\alpha_* h_*^{(1)}[N,v], \quad g[N,v,\{N\}]=\alpha_* h_*^{(2)}[N,v].
\end{equation}
\end{lemma}

\begin{proof}
By the two-step formulation definition and linearity of game values we obtain
\begin{equation}\label{gonsingletons}
g_i[\{i\},\tilde{u},\{\{i\}\}] =h^{(1)}_1[\{1\},\tilde{u}] h^{(2)}_i[\{i\},\tilde{u}], \quad i \in \mathbb{N}.
\end{equation}

Suppose $g$ satisfies \eqref{SIprop}. Then from \eqref{gonsingletons}  it follows that $h^{(k)}_1[\{1\},u] \neq 0$, $k \in \{1,2\}$. Then, this allows us to define 
\[
\alpha_*=h^{(1)}_1[\{1\},\tilde{u}] h^{(2)}_1[\{1\},\tilde{u}], \quad h_*^{(k)}=\frac{h^{(k)}}{h^{(k)}_1[\{1\},\tilde{u}]}, \,\, k \in \{1,2\}.
\]
Hence, by the linearity of $h^{(1)}$, $h^{(2)}$ we obtain \eqref{canonrep}. 

Note that by construction $h_{*,1}^{(1)}[\{1\},v]=v(1)$. Furthermore, by \eqref{SIprop} and $\eqref{gonsingletons}$ we have $h_{*,i}^{(2)}[\{i\},\tilde{u}]=h_{*,1}^{(2)}[\{1\},\tilde{u}]$ and hence $h_*^{(2)}$ satisfies \hyperref[axiom:SEP]{(SEP)}. Finally, the uniqueness of the representation and the independence of $h^{(k)}_*$ from $\hat{v}_T$ is a consequence of Lemma \ref{lmm::twostepconncoal}.

Conversely, suppose there exist $\alpha_* \neq 0$ and game values $h_*^{(1)}$, $h_*^{(2)}$ such that \eqref{canonrep} holds where $h_{*,1}^{(1)}[\{1\},v]=v(1)$ and $h_*^{(2)}$ satisfies \hyperref[axiom:SEP]{(SEP)}. Then, the two-step formulation implies $g_i[\{i\},\tilde{u},\{\{i\}\}] =\alpha_*$ which gives \hyperref[axiom:SIP]{(SIP)}.
\end{proof}

\begin{remark}\label{canonicalRemark}
\rm
The lemma above implies that if $g$ has a representation \eqref{canonrep},  one can choose two game values $h^{(1)}$, $h^{(2)}$  in the two-step formulation of $g$, by absorbing $\alpha_*$ either in $h_*^{(1)}$, or in $h_*^{2}$, or split between the two games.
If we absorb $\alpha_*$ in $h_*^{(1)}$, i.e. $h^{(1)}=\alpha_*h_*^{(1)}$ and $h^{(2)}=h_*^{(2)}$ then, according to Lemma \ref{lmm::twostepconncoal}, $g$ for singletons is equal to $h^{(1)}$. Alternatively, if we absorb $\alpha_*$ in $h_*^{(2)}$, i.e. $h^{(1)}=h_*^{(1)}$ and $h^{(2)}=\alpha_*h_*^{(2)}$, then $g$ for the grand coalition is equal to $h^{(2)}$.
\end{remark}

\begin{lemma}\label{lmm::qp2step}
Let $g$ be a coalitional value with a two-step formulation with $h^{(1)}$, $h^{(2)}$ as in Definition \ref{def::2stepprop}. Suppose that $h^{(2)}$ satisfies \hyperref[axiom:SEP]{(SEP)}. Then:
\begin{itemize}
\item [(i)] if $h^{(2)}$ satisfies \hyperref[axiom:EP]{(EP)} then $g$ satisfies \eqref{quotientgame}.

\item [(ii)] If $g$ satisfies \hyperref[axiom:EP]{(EP)} then $g$ satisfies \eqref{quotientgame}. 

\item [(iii)] If $g$ satisfies \eqref{quotientgame} then
 \begin{equation*}%\label{qp2step}
  \sum_{i\in S_j}h^{(2)}_i[S_j,v^{(j)}] =g_j[M,v^{\cP},\bar{M}]  = h^{(1)}_j[M,v^{\cP}]=v^{(j)}(S_j), \quad j \in M.
 \end{equation*}

\end{itemize}
\end{lemma}

\begin{proof}
The results $(i)$-$(iii)$ follow from Lemma \ref{lmm::twostepconncoal} and Lemma \ref{lmm::efficg2step}. 
% The proof of $(iv)$ follows from $(iii)$ and Proposition \ref{prop::quotgameexpl}$(i)$.
\end{proof}

\begin{lemma}\label{lmm::efficg2step}
Let $g$ be a coalitional value with a two-step formulation with $h^{(1)}$, $h^{(2)}$ as in Definition \ref{def::2stepprop}. 

\begin{itemize}
  \item [(i)] Suppose  $h^{(1)}$ and $h^{(2)}$ satisfy \hyperref[axiom:EP]{(EP)}. Then $g$ satisfies \hyperref[axiom:EP]{(EP)}.

  \item [(ii)] Suppose $g$ satisfies \hyperref[axiom:EP]{(EP)}. If either $h^{(1)}$ or  $h^{(2)}$ satisfies \hyperref[axiom:SEP]{(SEP)}, then $h^{(1)}$, $h^{(2)}$ satisfy \hyperref[axiom:EP]{(EP)}.
\end{itemize}
\end{lemma}
\begin{proof}
Let $\{v^{(j)}\}_{j=1}^M$ be as in Definition \ref{def::2stepprop}. Suppose that $h^{(1)},h^{(2)}$ are efficient, then we have
\[
\sum_{i\in N}g[N,v,\cP] =\sum_{j \in M}\sum_{i \in S_j} h^{(2)}_i[S_j,v^{(j)}]=\sum_{j \in M}v^{(j)}(S_j)=\sum_{j \in M}h^{(1)}_j[M,v^{\cP}]=v^{\cP}(M)=v(N),
\]
where we used the property $\hat{v}_{S_j}=v^{\cP}$ hence $(i)$.
Part $(ii)$ follows from Lemma \ref{lmm::twostepconncoal} and the efficiency of $g$.
\end{proof}

\subsubsection{Proof of the results from \S\ref{sec::extra} }  \label{app::extra}
\begin{proof}[Proof of Proposition \ref{prop::extra1}]
By the two-step formulation and Lemma \ref{lmm::lingamevalrepr}, for any game $(N,v)$ and $i\in S_j$ we have 
\begin{equation}\label{auxiliary} 
|g_i[N,v,\cP]|\leq C\sum_{T\subseteq S_j}|v^{(j)}(T)|\leq C \sum_{T\subseteq S_j}\sum_{R\subseteq M}|\hat{v}_T(R)| 
\end{equation}
where $C$ denotes appropriate constants depending only on game values $h^{(1)}$ and $h^{(2)}$.
If $\hat{v}_T=v^{\cP|T}$, then $\hat{v}_T(R)=v(Q)$ if $j\notin M$ while $\hat{v}_T(R)=v(Q\cup T)$ otherwise ($Q=\cup_{r\in R}S_r$ as in the lemma). Then the inequality above may be rewritten as 
$$
|g_i[N,v,\cP]|\leq C \sum_{R\subseteq M\setminus\{j\}} \sum_{T\subseteq S_j}(|v(Q)|+|v(Q\cup T)|) 
\leq C \big(\sum_{R\subseteq M\setminus\{j\}}|v(Q)|+\sum_{R\subseteq M\setminus\{j\}} \sum_{T\subseteq S_j}|v(Q\cup T)|\big) 
$$
where $C$ denotes a generic constant (the last one can be expressed in terms of the one before as $2^{|S_j|}C$).
Now substituting the cooperative game $\vpdp(\cdot\,;X,f-f_0)$ for $v$ in the last inequality implies  part $(i)$.

In the case of $(ii)$, one has 
$$
|\hat{v}_T(R)|\leq |v^\cP(R)|+m\big(|v(T)|+|v^\cP(\{j\})|\big)=|v(Q)|+m\big(|v(T)|+|v(S_j)|\big).
$$
Combining with \eqref{auxiliary} yields 
$$
|g_i[N,v,\cP]| 
\leq C \big(\sum_{R\subseteq M}|v(Q)|+ \sum_{T\subseteq S_j}|v(T)|\big) 
$$
for a suitable $C$. Plugging the cooperative game $\vpdp(\cdot\,;X,f-f_0)$ for $v$ then implies part $(ii)$.
\end{proof}

\begin{proof}[Proof of Corollary \ref{corr:extra3}]
Notice that $Q$ is a union of $S_r$'s (i.e. $Q=\cup_{r\in R}S_r$). So the independence of $X_{S_1},\dots,X_{S_m}$ implies 
that in the inequalities appearing in Proposition \ref{prop::extra1} $P_{X_Q}\otimes P_{X_{-Q}}$ can be replaced with 
$\otimes_{r\in M}P_{X_{S_r}}=P_X$ while, in part $(i)$ of the lemma where $Q\cap T=\varnothing$, $P_{X_{Q\cup T}}\otimes P_{X_{-(Q\cup T)}}$ is the same as $P_{X_T}\otimes P_{X_{-T}}$. Substituting in those inequalities yields the first inequality in  Corollary \ref{corr:extra3}.
To obtain the second assertion, notice that if predictors $X_{S_j}$ are independent, then for any $T\subseteq S_j$:
$$
P_{X_T}\otimes P_{X_{-T}}=P_{X_T}\otimes P_{X_{S_j-T}}\otimes P_{X_{-S_j}}=P_{X_{S_j}}\otimes P_{X_{-S_j}}=P_X.
$$
The same is true when $S_j$ is a singleton as then $P_{X_T}\otimes P_{X_{-T}}$ from above coincides with 
$P_{X_{S_j}}\otimes P_{X_{-S_j}}=P_X$ because $T=\varnothing \text{ or }S_j$.
Finally, if $h^{(2)}$ is efficient and $S_j=\{i\}$, $g$ satisfies the quotient game property by Lemma \ref{lmm::qp2step}$(i)$, and then Proposition \ref{prop::unifcoalexpl} implies that 
$$
\bar{g}_i(X;\vpdp,\cP,f)=\bar{g}_{S_j}(X;\vpdp,\cP,f) = \bar{g}_{S_j}(X;\vce,\cP,f) =\bar{g}_i(X;\vce,\cP,f)
$$
where we have used $S_j=\{i\}$. But due to the two-step formulation:
$$
\bar{g}_i(X;\vce,\cP,f)=g_i(X;\vce,\cP,f-f_0)= h_{j}^{(1)}[M,\vceP(.;X,f-f_0)]
$$
which admits the Lipschitz constant $1$ due to Lemma \ref{lmm::boundquot}$(iii)$.
\end{proof}

\section{Maximal information coefficient}\label{app::HierClust}

\begin{figure}
    \centering
    \begin{subfigure}[t]{0.8\textwidth}
        \centering
        \includegraphics[width=\textwidth]{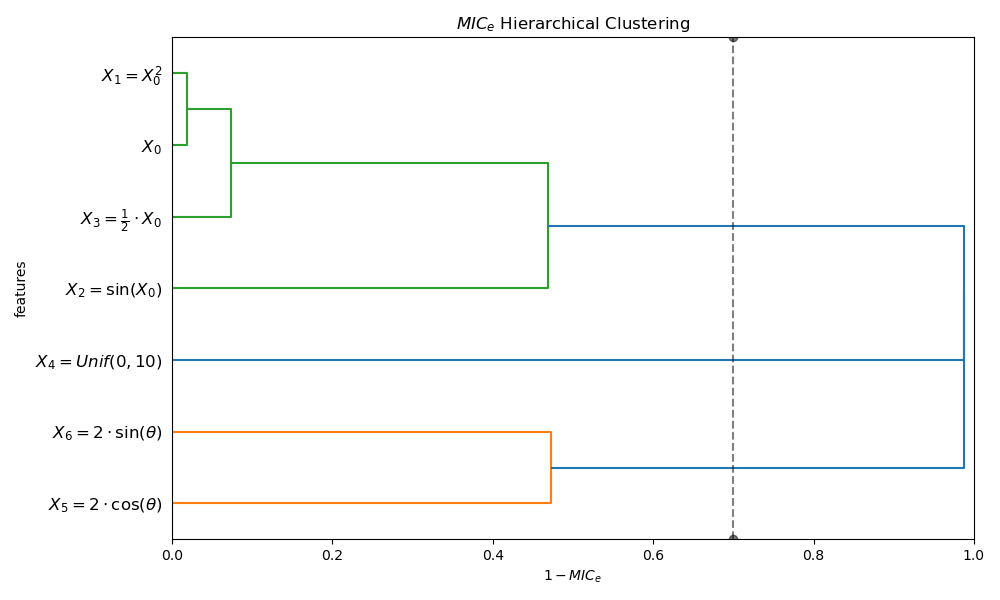}
        \caption{\footnotesize MIC-based hierarchical clustering with GA linkage.  \vspace{5pt}}\label{fig::MIC_clus}
    \end{subfigure}
    ~~
    \begin{subfigure}[t]{0.8\textwidth}
        \centering
        \includegraphics[width=\textwidth]{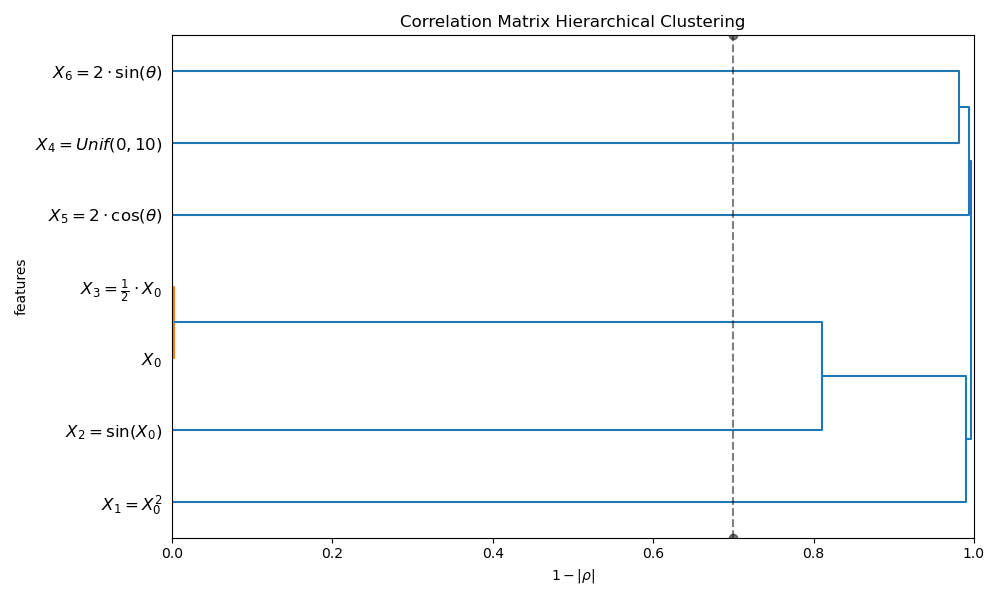}
        \caption{\footnotesize Correlation-based hierarchical clustering with GA linkage. }\label{fig::PCA_clus}                        
    \end{subfigure}
    \caption{\footnotesize Variable hierarchical clustering for the model \eqref{mictestmodel}. The dotted vertical line is based on a dissimilarity threshold; the predictors that have remained together on the left of it end up in same groups.}\label{fig::hierclus}
\end{figure}

\begin{figure}
    \centering
    \begin{subfigure}[t]{0.35\textwidth}
        \centering
        \includegraphics[width=\textwidth]{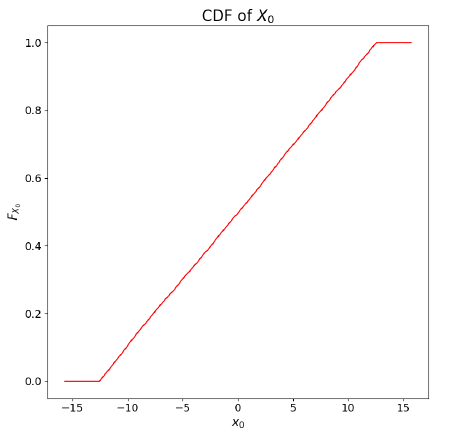}
        \caption{\footnotesize CDF of $X_0$ }\label{fig::cdf_x0}
    \end{subfigure}
    ~~
    \begin{subfigure}[t]{0.35\textwidth}
        \centering
        \includegraphics[width=\textwidth]{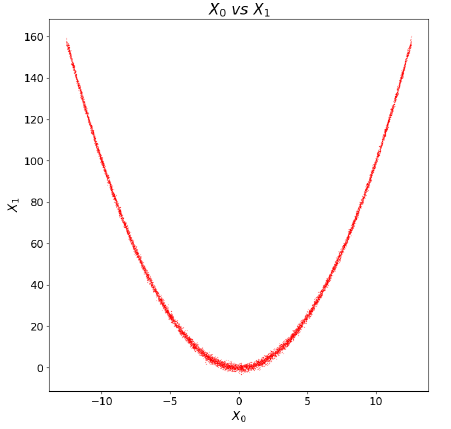}
        \caption{\footnotesize Quadratic relationship with noise \vspace{5pt}}\label{fig::quadratic_noise}
    \end{subfigure}
    ~~
    \begin{subfigure}[t]{0.35\textwidth}
        \centering
        \includegraphics[width=\textwidth]{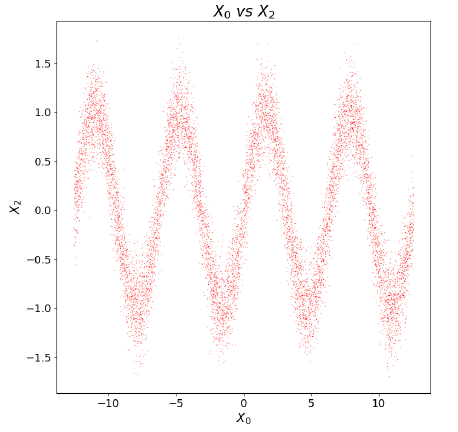}
        \caption{\footnotesize Sine relationship with noise }\label{fig::sin_noise}
    \end{subfigure}
    ~~
    \begin{subfigure}[t]{0.35\textwidth}
        \centering
        \includegraphics[width=\textwidth]{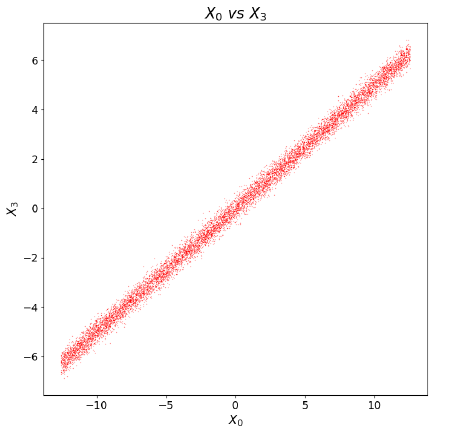}
        \caption{\footnotesize Linear relationship with noise \vspace{5pt}}\label{fig::line_noise}
    \end{subfigure}
    ~~
    \begin{subfigure}[t]{0.35\textwidth}
        \centering
        \includegraphics[width=\textwidth]{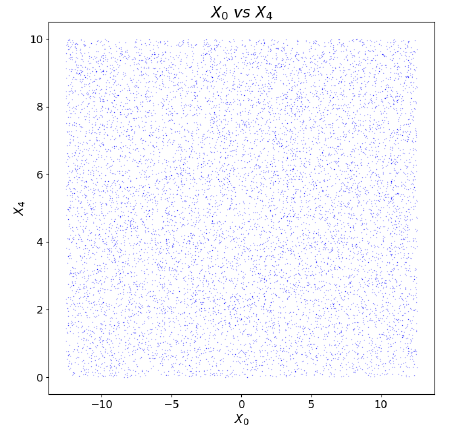}
        \caption{\footnotesize Independent relationship }\label{fig::unif}
    \end{subfigure}
    ~~
\begin{subfigure}[t]{0.35\textwidth}
    \centering
    \includegraphics[width=\textwidth]{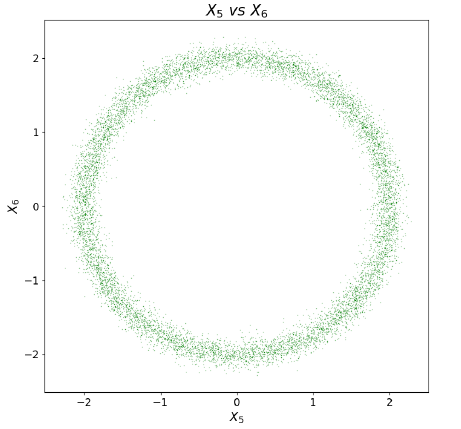}
    \caption{\footnotesize circle}\label{fig::circle}
\end{subfigure}
    \caption{\footnotesize Scatter plots showing the dependencies in the distribution of \eqref{mictestmodel}. }\label{fig::mictestmodel}
\end{figure}

\subsection{ $\MICe$ statistic}\label{app::MIC}

\begin{definition}[\bf \citet{Reshef16}]\label{MIC_def}
    Let $D_n$ be a dataset drawn from $(X,Y)$, with $|D_n|=n$. Let $B(n)$ be a tuning function that tends to $\infty$ as $n \to \infty$. Then
    \begin{equation*}\label{MIC_val}
        {\rm{MIC}}_e(D_n; B(n)) := \max_{k\ell < B(n)}\Big\{\frac{\1_{\{k<l\}}I^{[\ast]}(D_n,k,[\ell])+\1_{\{k \geq l\}}I^{[\ast]}(D_n,[k],\ell)}{\log(\min\{k,\ell\})}\Big\},
    \end{equation*}    
    where $I^{[\ast]}(D_n,k,[\ell])=\max_{G \in G(k,[l])} I(X,Y)|_G$ and $G(k,[l])$ is the set of $k$-by-$l$ grids whose $y$-axis partition is an equipartition of size $l$.    
\end{definition}

\citet[Corrollary 26]{Reshef16} establishes that $\MICe$ is a consistent estimator of $\MICstar$ provided that $\omega(1)<B(n)\le O(n^{1-\epsilon})$ for some $\epsilon \in (0,1)$. Furthermore, \citet[Theorem 28]{Reshef16} shows that $\MICe$ can be computed in time $O(n+n^{5(1-\epsilon)/2})$ when $B(n)=O(n^{1-\eps})$, which in turn implies the following.

\begin{corollary}[\bf \citet{Reshef16}]\label{corr::mic_complexity}
${\rm{MIC}}_*$ can be estimated consistently in linear time.
\end{corollary}

\subsection{Example of variable hierarchical clustering based on $\MICstar$}\label{app::example_var_clust}

% We next

In this section we perform variable clustering using MIC$_e$ and compare it with that based on the correlation for the model:
\begin{equation}\label{mictestmodel}
\begin{aligned}
    &X_0\sim Unif(-4\pi,4\pi),& &X_1 = X_0^2 + \epsilon_1,& &X_2 = \sin(X_0) + \epsilon_2,& &X_3 = 0.5 X_0 + \epsilon_3,&\\
    &X_4 \sim Unif(0,10),& &X_5 = 2\cos(\theta)+\epsilon_5,& &X_6 = 2\sin(\theta) + \epsilon_6,&
\end{aligned}
\end{equation}
where 
\[
\epsilon_1\sim \mathcal{N}(0,1), \, \epsilon_2\sim \mathcal{N}(0,\frac{1}{4}), \, \epsilon_3\sim \mathcal{N}(0,\frac{1}{4}), \, \epsilon_5\sim \mathcal{N}(0,\frac{1}{10}), \epsilon_6\sim \mathcal{N}(0,\frac{1}{10}), \,  \theta\sim Unif(0,2\pi).
\]

By construction, there are three independent groups of variables in the model \eqref{mictestmodel} 
\begin{equation}\label{mictestmodelgroupingtheor}
X_{S_1}=(X_0,X_1,X_2,X_3), \quad X_{S_2}=X_4, \quad X_{S_3}=(X_5,X_6),
\end{equation}
% \begin{equation}\label{mictestmodelgroupingtheor}
% S_1=\{0,1,2,3\}, \quad S_2=\{4\}, \quad S_3=\{5,6\},
% \end{equation}
such that within each group the variables have strong dependencies. Figure \ref{fig::mictestmodel} displays scatter plots of $10^4$ samples of paired variables from the joint distribution \eqref{mictestmodel} that visually confirms the grouping \eqref{mictestmodelgroupingtheor}.

Figure \ref{fig::MIC_clus} displays a dendrogram generated by the MIC$_e$-based dissimilarity measure, whose geometry is in accordance with our intuition on how predictors should be grouped with each other based on their dependencies and the accompanying noise level. Using the dendrogram as a guide, setting  the dissimilarity threshold $\alpha=0.7 \geq 1-\MICe$, we conclude that the variables are partitioned into groups $\mathcal{P}^{\tiny{\rm {MIC_e}}}_{\alpha=0.7}=\{S_1,S_2,S_3\}$ with $S_i$ given by \eqref{mictestmodelgroupingtheor}, which coincides with the built-in grouping.

In contrast, according to the dendrogram on Figure \ref{fig::PCA_clus}, the correlation-based clustering fails to capture non-linear dependencies as it  ignores the sine functional dependence and captures weak dependencies between $X_5$ and $X_6$ that form a noisy circle, placing them in different clusters. Setting the dissimilarity threshold $\alpha = 0.7 \geq 1-|\rho|$ with $\rho$ the Pearson correlation, we obtain 
$\mathcal{P}^{\rho}_{\alpha=0.7}=\{ \{0,3\}, \{1\},\{2\},\{4\},\{5\},\{6\}\} \}
$, which is drastically different from the designed grouping \eqref{mictestmodelgroupingtheor}.

%%%%%%%%%%%%%%%%%%%%%%%%%%%%%%%%%%%%%%%%%%%%%%%%%%%%%%%%%%%%%%%%%%%%%%%%%%%%%%%
%%%%%%%%%%%%%%%%%%%%%%%%%%%%%%%%%%%%%%%%%%%%%%%%%%%%%%%%%%%%%%%%%%%%%%%%%%%%%%%
% Explainers with coalition structure under partition tree
%%%%%%%%%%%%%%%%%%%%%%%%%%%%%%%%%%%%%%%%%%%%%%%%%%%%%%%%%%%%%%%%%%%%%%%%%%%%%%%%%%%%%%%%%%%%%%%%%%%%%%%%%%%%%%%%%%%%%%%%%%%%%%%%%%%%%%%%%%%%%%%%%%%%%%%%%%%%%%

\section{Explainers with coalition structure under partition tree}\label{sec::ExTree}

\begin{figure}
  \centering
 \begin{subfigure}[t]{0.5\textwidth}
    \centering
    \includegraphics[width=\textwidth]{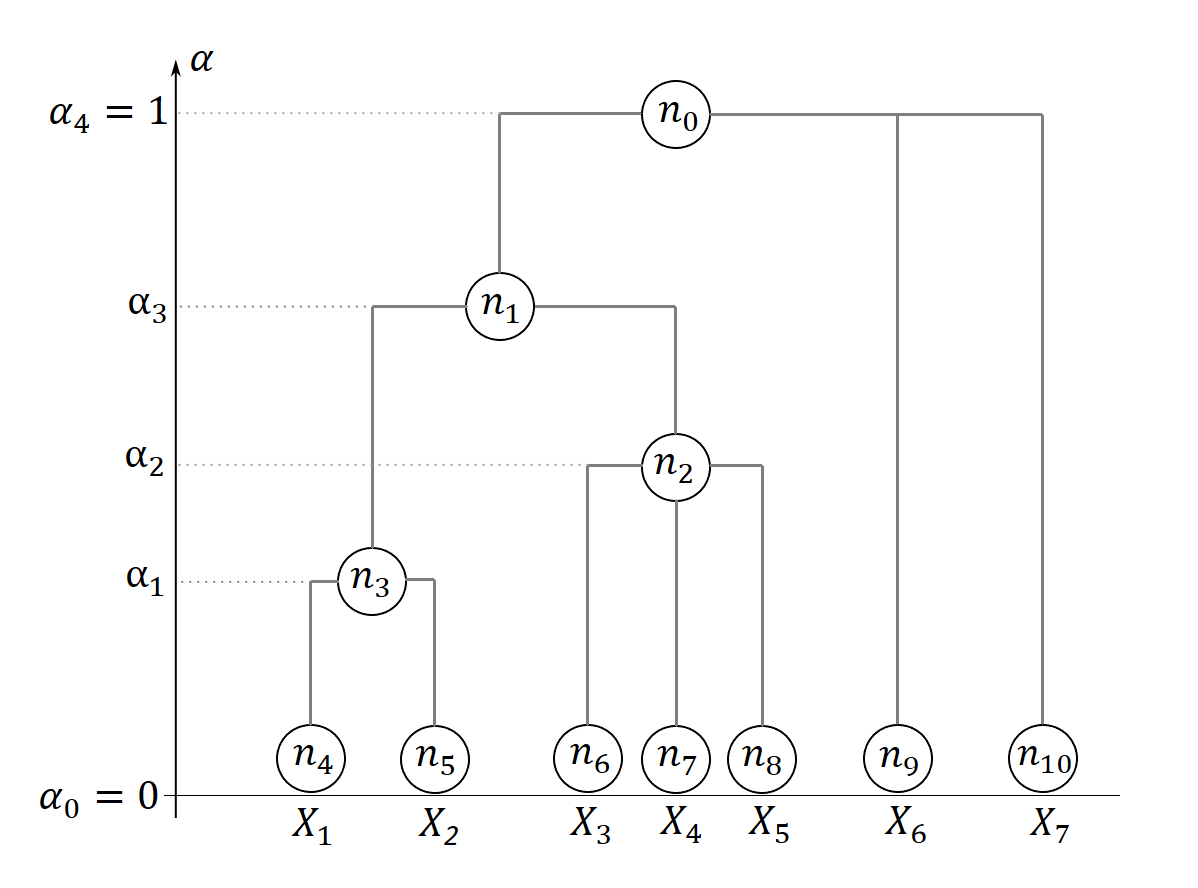} 
  \end{subfigure}
  \caption{\footnotesize Partition tree} \label{fig::coalescent} 
\end{figure}

\subsection{Values with levels structure and games with hierarchy}

Earlier in the text we explored predictor explanations based on values with coalition structure, which is constructed in our work by partitioning features based on dependencies. However, having just one set of coalitions may not adequately express the varying strength of predictor dependencies. In that case, a nested sequence of partitions could be used to further encode information about the level of dependencies in the joint distribution.

For example, suppose a model has four predictors $\{X_1,X_2,X_3,X_4\}$, where the first three are strongly dependent and jointly independent of $X_4$. The high-level partition can be defined as $\cP_1=\{\{1,2,3\},\{4\}\}$. Suppose that $X_1$  and $X_2$ are strong proxies of each other, but $X_3$ is not. Then $\cP_2=\{\{1,2\},\{3\},\{4\}\}$ may be set as a refinement of $\cP_1$, making the sequence $\{\cP_1,\cP_2\}$ capture further details of the dependencies.

More generally, a ``levels structure'' on $N=\{1,2,\dots,n\}$ is defined to be a finite sequence of partitions $\cL=\{P_1,P_2,\dots,P_l\}$ where $\cP_{k+1}$  is a refinement of $\cP_k$, that is, if $S\in \cP_{k+1}$, then $S\subseteq T$ for some $T \in \cP_k$. The article of Winter \cite{Winter1989} generalizes values with coalition structure to those with levels structures. Given the set $\fL$ of all levels structures on $N$ and any  collection $\cG$ of games on $N$, \cite{Winter1989} defines a value on $\cG$ with a levels structure from $\cL\in \fL$ as an operator $u: \cG \times \fL\mapsto \R^n$  which assigns a payoff vector $u[N,v,\cL] \in \RR^n$ to any pair $(v,\cL)$ of a game $v\in \cG$ and levels structure  $\cL \in \fL$.

In particular, \cite{Winter1989} generalizes the Owen value \cite{Owen} to a value with levels structure, called the Winter value. The derivation is axiomatic, that is, the Winter value is shown to be the only value with levels structure that is efficient, additive, coalitionally symmetric and symmetric within coalitions. In the case of a one-level partition, it reduces to the Owen value. By replacing additivity with marginality (meaning, the value explicitly depends on the terms $\{v(S\cup\{i\})-v(S)\}_{S \subseteq N \setminus \{i\}}$), the work of Khmelnitskaya and Yanovskaya\cite{Khmelnitskaya2007} shows that the Winter value is again the only value with those properties.

In some applications, besides the payoffs of each player in $N$, one may be interested in evaluating payoffs of each partition's element in the levels structure $\cL=\{\cP_0,\cP_1,\cP_2,\dots,\cP_l\}$, where $\cP_0=\{N\}$ is the grand coalition. $\cL$ naturally induces a partition tree of depth up to $n-1$, where each node corresponds to an element from one of the partitions in the levels structure, and where the subtree height determines the hierarchy level. Assigning a game to each non-terminal node in the tree (played on the set of the node's children) yields games with hierarchy induced by the tree; see \cite{Algaba2019}. The payoff in each node is then obtained by computing the game value based on the game associated with the node's parent.

Let us point out the differences between values with levels structure and games with hierarchy. The former defines a solution concept (an operator) that provides a payoff for each player given a levels structure. Since a partition tree defines a levels structure, the value with a levels structure can trivially assign payoffs to every node. This can be done, for instance, via summations across each partition or using quotient games. The latter requires an explicit a priori assignment of games to every node and typically the same game value is applied to these games. The two concepts are clearly related but not equivalent {\footnote{ A coalitional value with levels structure naturally induces games with a hierarchy, and vice versa, given games with hierarchy, applying a game value to the parent of every terminal node yields a coalitional value with levels structure. However, in the latter case, the games at every node can be unrelated to the games in the children or the parent, while this is clearly not the case for a value with a levels structure.}}.

Games with hierarchy have found numerous applications in the ML explainability literature. It is relatively easy to set them up by using subgames, games obtained by restricting a fixed game to a subset of players. For example, a setup like this is used in \cite{Teneggi2023} which proposes hierarchical explanations of images using an appropriate base game. In this context the pixels are the players and a levels structure is given by a 4-partition tree, which splits every node (the portion of the image) into four parts. A subgame is obtained at each node by restricting the base game \cite{Teneggi2023} to the portion of the image associated with that node. Other works on this include text classification explanations \cite{HChen2020} as well as explanations of neural networks \cite{Singh2018} and the article on explanations of structured data \cite{ChenLShapley2019}. Finally, the PartitionExplainer\footnote{ The online documentation states that these values are recursive Owen values for binary trees. Our own analysis indicates that the values are in fact recursive two-step Shapley, not Owen values.} is an algorithm that provides values with binary levels structure (using, for example, the marginal game as the utility when explaining a tree-based model), but the lack of documentation precludes us from concluding that the values are generalized Owen values.

Although setting up games with hierarchies via subgames is easy, the coalitions in that case are isolated, meaning they lack communication in the sense discussed by \cite{Owen1986}. In fact, this setup, in cohort with the Shapley value, leads to a generalization of the Aumann-Dr\'{e}ze (AD) value \cite{Aumann1974} where each player's payoff is obtained by computing the Shapley value of the subgame restricted to the coalition in the partition. To allow for communication between coalitions, which may be crucial in some applications (such as ours), it is necessary to avoid subgames. In this case, initialization of each node in the partition tree with a game might be algorithmically complex , so values with levels structures can be utilized because they naturally induce games with hierarchy.

The game-theoretic approach \cite{Winter1989,Peleg2003,Algaba2019} to the design of values with levels structures is usually axiomatic. For this reason, in what follows, we design such values via direct recursive construction. Specifically, we show how to naturally generalize a two-step representation model (cf. Definition \ref{def::2stepprop}) for coalitional values to values with levels structures. This process produces both values with levels structures and games with hierarchies that are produced recursively via the use of the intermediate game. To the best of our knowledge a theoretical construction of such recursively-defined concept solutions has not been done before. It can be shown that the recursive Owen value obtained in such a way is in fact the Winter value, which follows from \cite[Theorem 2]{Khmelnitskaya2007}.

The family of recursive coalitional values, which we introduce
shortly, may be utilized to construct hierarchical group explainers based on predictor dependencies. Given a dependency-based levels structure induced by a parameterized partition tree, we evaluate games and their values recursively at every node in the tree. Then, for a given cross-section of the tree at height $\alpha \geq 0$ we define a group explainer associated with the partition induced by the cross-section (for example, see Figures \ref{fig::hierclus} and \ref{fig::coalescent}). This procedure yields a generalization of trivial explainers based on coalitional values introduced earlier. By requiring the two game values $h^{(1)}$  and $h^{(2)}$ in the two-step formulation to be efficient, we show that a recursive value  satisfies an additive flow property across the partition tree.

\subsection{Notation and assumptions for the partition tree}\label{sec::parttree}

A partition tree is a parameterized coalescent tree (not necessarily binary); see Figure \ref{fig::coalescent} and dendrograms in Figure \ref{fig::hierclus}. The leaves at the bottom of the tree, at the zero level, correspond to singletons; every time branches coalesce, this gives rise to a larger group via merging with the others. This defines a sequence of partitions corresponding to each coalescence. The height of each node in the tree can be thought of as the strength of some underlying relationship (e.g. dependency, meaning, or a combination of both)  in the group corresponding to the leaves of the corresponding subtree. In  Figure \ref{fig::coalescent}, for example, the sequence of nested partitions dictated by the order of coalescence is given by $\{\{1\},\{2\},\{3\},\{4\},\{5\},\{6\},\{7\}\}$, $\{\{1,2\},\{3\},\{4\},\{5\},\{6\},\{7\}\}$, $\{\{1,2\},\{3,4,5\},\{6\},\{7\}\}$, $\{\{1,2,3,4,5\}$,$\{6\}$,$\{7\}$\}, \{$1$,$2$,$3$,$4$,$5$,$6$,$7$\}; for details, see Section \ref{sec::parttree}. For the sake of the exposition, we work with partition trees of height one. Finally, we note that binary coalescent trees are common for modeling of a genealogical relationship in a population of constant size; see \cite{Kingman1982}, \cite{Wakeley2008}.

In this subsection, we generalize the concept of two-step formulation from \S\ref{sec::ExCoal} to recursive game values which incorporate a given partition tree  and then use them to design corresponding group explainers. It is worth noting that the recursive values we design utilize only the combinatorics of the partition tree, while the group explainers (our main objective) rely on the parametrization of the tree.

We adapt the following conventions: 
\begin{itemize}
\item We denote the nodes (both terminal and non-terminal) of the tree by $\mathcal{V} = \{n_0,n_1,\dots,n_R\}$ with $n_0$ being the root.
\item For $\nu \in \mathcal{V}$, let $\TT(\nu)$ be a subtree with the root node $\nu$ and let $\ell(\nu)$  denote the collection of leaves 
(terminal nodes) of $\TT(\nu)$, i.e. the final descendants of $\nu$.
\item The parent of each node $\nu$, except the root, is denoted by $p(\nu)$ while the collection of its children is denoted by $c(\nu)$.
If $\nu$ is not a leaf we assume that $|c(\nu)|>1$. 
\end{itemize}

In our setup, the partition tree is equipped with a height $h$ that measures the height of each node $\nu$ (or, equivalently, the height of the subtree $\TT(\nu)$) and that satisfies
\begin{itemize}
\item The values of $h$ are in $[0,1]$; and  $h(n_0)=1$ at the root while  $h(\nu)=0$ if $\nu$ is a leaf (i.e. if $\nu \in \ell(n_0)$).
\item The heights of the children of  a non-terminal node $\nu$ are smaller than the height of $\nu$; that is $h(\xi)<h(\nu)$ for any $\xi \in c(\nu)$. Another technical requirement is that the heights of all non-terminal nodes are distinct (see Remark \ref{tecreqheight}).

\end{itemize}

\begin{remark}\label{tecreqheight}\rm
The requirement for the heights to be distinct  uniquely determines a sequence of nested partitions which corresponds to the order in which branches coalesce. Once this requirement is dropped, the sequence of nested partitions is no longer uniquely defined. In principle, however, this requirement can be dropped. In this case, if several nodes are at the same height one has to manually assign the order of coalescence associated with these nodes. 
\end{remark}

\begin{remark}\label{rem::dendr}\rm
Examples of parameterized trees have come up before in \S\ref{sec::Grouping}: 
 Given a a suitable normalized dissimilarity measure, dendrograms  obtained from hierarchical clustering can be thought of as parameterized binary trees. In practice, dendrograms constructed using the dissimilarity measure $1-{\rm{MIC}}_*$, the heights associated with non-terminal nodes are always distinct except some degenerate cases. Thus, a dendrogram can be viewed as a binary coalescent tree.
 \end{remark}

To relate trees to the machine learning setting, let us denote the predictors by $X_1,\dots,X_n$ as before.
We shall focus on  trees $\TT$ in which every node corresponds to a partition of a subset of predictors with its children (if any) corresponding to a partition of a set belonging to that former partition. 
The root corresponds to the partition $\{\{X_1,\dots,X_n\}\}$ determined by the grand coalition while at each leaf  we have the unique partition of some $\{X_i\}$ which may be identified with the predictor $X_i$ (or with $i\in N$).   
We assume that $\TT$ has $n$ leaves and there is a bijective enumeration map  
$\pi:N \to \ell(n_0)$ such that predictor $X_i$ (or  index $i\in N$) corresponds to the leaf $\pi(i)$.
\begin{itemize}
\item For a node $\nu$, $S(\nu):=\pi^{-1}(\ell(\nu))$ is the set of predictors appearing as the leaves of the rooted subtree $\TT(\nu)$ emanating from $\nu$. Assuming that $\nu$ is non-terminal, the partition of $S(\nu)$ determined at $\nu$ is 
$$
\cP(\nu):= \{ S(\xi): \xi \in c(\nu) \}.
$$
\end{itemize}

Finally, recalling that $\TT$ is equipped with a height $h$: 
\begin{itemize}
\item For any $\alpha \in [0,1]$, we define the collection of nodes $\mathcal{N}(\alpha)$ immediately below the cross-section of the tree at the height $\alpha$ as follows: $\mathcal{N}(0)= \ell(n_0)$ and for $\alpha \in (0,1]$ we set
\[
\mathcal{N}(\alpha)=\{\nu: h(\nu) < \alpha \leq h(p(\nu)) \}.
\]
Finally, for $\alpha>1$ we set $\mathcal{N}(\alpha)=\{n_0\}$.
Thus, each $\alpha \in [0,\infty)$ defines a partition $\cP_{\alpha}$ given by 
$$
\cP_{\alpha}= \{ S(\nu)\}_{ \nu \in \mathcal{N}(\mathcal{\alpha})},
$$
where  $\alpha \mapsto \cP_{\alpha}$ is a left-continuous partition map which represents the partition tree and gives rise to a nested sequence of partitions starting at singletons $\{\bar{N}\}$ and terminating at the grand coalition $\{N\}$ containing one element.

% \item For any $\alpha \in [0,1]$, we define the collection of nodes $\mathcal{N}(\alpha)$ immediately below the cross-section of the tree at the height $\alpha$ as
% $$
% \mathcal{N}(\alpha)=\{\nu: h(\nu) < \alpha \leq h(p(\nu)) \},
% $$
% if $0<\alpha<1$, and as $\ell(n_0)$ or $\{n_0\}$ if $\alpha=0$ or  $\alpha=1$ respectively.
% Thus, each $\alpha \in [0,1]$ defines a partition $\cP_{\alpha}$ given by 
% $$
% \cP_{\alpha}= \{ S(\nu)\}_{ \nu \in \mathcal{N}(\mathcal{\alpha})}.
% $$
\end{itemize}

Notice that by identifying the set of predictors with $N=\{1,\dots,n\}$, one can work with all the concepts defined so far in the context of nested partitions of a finite set $N$. In \S\ref{sec::parttree1}, we define recursive game values that are defined based on a game $(N,v)$ and the underlying combinatorial graph of the partition tree $\TT$ whose leaves are in bijection with elements of $N$. Following the same procedure as before, these game values can then be utilized to construct group explainers associated with a parameterized tree partitioning the predictors. This is the content of \S\ref{sec::parttree2}.

\subsection{Game values under partition tree}\label{sec::parttree1}

Given a coalitional game value $g$ with two-step formulation along with a cooperative game $(N,v)$ and a tree $\TT$ partitioning $N$, here, we shall generalize the two-step formulation by defining numbers $\hat{g}^{(\nu)}[N,v,\TT]$, $\nu$ being a node of $\TT$, based on all this information.
Let us first state the hypotheses we assume for $g$ and $\TT$:
\begin{enumerate}
\renewcommand{\labelenumi}{\textbf{(\theenumi)}}
\renewcommand{\theenumi}{H\arabic{enumi}}
\renewcommand{\labelenumi}{(\theenumi)}
\renewcommand{\theenumi}{H\arabic{enumi}}

%\item \label{hyp1::sip} $g$ satisfies (SIP) and $g_1[\{1\},u,\{\{1\}\}]\ne 0$.
\item \label{hyp2::norm2step} $g$ has a two-step formulation with $h^{(1)},h^{(2)}$ and intermediate games $\hat{v}_T$ as in Definition \ref{def::2stepprop}, where  $h^{(1)}$ and $h^{(2)}$ are linear, $h^{(1)}$ is symmetric, and $h^{(2)}$ satisfies \hyperref[axiom:SEP]{(SEP)}.
\item \label{hyp3::parttree} As in \S\ref{sec::parttree}, $\TT$ corresponds to a family  of nested partitions of $N$ with its leaves in bijection with elements of $N$ via 
$\pi:N\rightarrow\ell(n_0)$. We further consider enumeration maps for each non-terminal node $\nu_*$: Suppose that $\nu_*$ has $m$ children, $c(\nu_*) = \{n_{r_j}\}_{j=1}^m$, with some arbitrary enumeration encoded by the bijective map  $\bar{\pi}_{\nu_*}: M_{\nu_*}=\{1,2,\dots,m\} \to c(\nu_*)$ such that $n_{r_j}=\bar{\pi}_{\nu_*}(j)$. 
Following the notation introduced in \S\ref{sec::parttree}, the partition $\cP(\nu_*)$ of $S(\nu_*)\subseteq N$ is given by  
$\{ S(n_{r_1}),S(n_{r_2}),\dots,S(n_{r_m}) \}$.
\end{enumerate}
\begin{remark}
\rm
Above, we assumed that $g$ comes with a two-step formulation in which $h^{(2)}$ satisfies \hyperref[axiom:SEP]{(SEP)}. The logic of this assumption shall be explained shortly. However, in view of Lemma \ref{lmm::canonical} and Remark \ref{canonicalRemark}, if $g$ satisfies \hyperref[axiom:SIP]{(SIP)} and $g_1[\{1\},\tilde{u},\{\{1\}\}]\ne 0$, then there is a unique two-step formulation with such a property.
\end{remark}

With these conventions in mind, the definition of values $\hat{g}^{(\nu)}[N,v,\TT]$ below utilizes certain games assigned to nodes of $\TT$. Starting from $v$ itself assigned to the root, such games are defined inductively, and, in the vein of Definition \ref{def::2stepprop}, with the help of intermediate games. To elaborate, as in \eqref{hyp3::parttree}, consider a node $\nu_*$. The game $v^{(\nu_*)}$ is played on  the subset $S(\nu_*)$ of $N$ for which the tree structure provides a partition 
$\cP(\nu_*)= \{ S(n_{r_1}),S(n_{r_2}),\dots,S(n_{r_m})\}$. 
The game $v^{(n_{r_j})}$ corresponding to a child $n_{r_j}$ of $\nu_*$ is defined as 
$$
v^{(n_{r_j})}(T):=h^{(1)}_j[M_{\nu_*},\hat{v}_T],\quad T \subseteq S(n_{r_j}),
$$
where 
$$
\hat{v}_T=\hat{v}_T\left(S(\nu_*),v^{(\nu_*)},\cP(\nu_*)\right), \quad T \subseteq S(n_{r_j}),\,\, j\in\{1,\dots,r\},
$$
are intermediate games (as appear in the two-step formulation of $g$) played on $M_{\nu_*}$.

\begin{definition}
With conventions from  \eqref{hyp2::norm2step} and \eqref{hyp3::parttree},  
let $g$ be a coalitional value, $(N,v)$ a cooperative game and $\TT$ a partition tree for $N$.  
For each node $\nu$ of $\TT$, let $(S(\nu),v^{(\nu)})$  be the game assigned to $\nu$ via the inductive construction described before this definition.  
In what follows, we define the recursive values $\hat{g}^{(\nu)}[N,v,\TT]$ where $\nu$ is a node of $\TT$. 
If $\nu$ is the root node $n_0$, we set $\hat{g}^{(n_0)}[N,v,\TT]=v(N)$. Next, given a non-terminal node $\nu_*$, for a child $\nu \in c(\nu_*)$, we set
\begin{equation}\label{recursval}
\hat{g}^{(\nu)}[N,v,\TT]:=
\begin{cases} 
v^{(\nu)}(S(\nu))=h^{(1)}_{j}[M_{\nu_*}, \big(v^{(\nu_*)}\big)^{\cP(\nu_*)}], \, j=\bar{\pi}^{-1}_{\nu_*}(\nu), \,\, \text{if $\nu_* =n_0$ or $c(\nu_*) \neq \ell(\nu_*)$;}\\ 
h^{(2)}_{i}[S(\nu_*),v^{(\nu_*)}], \, i=\pi^{-1}(\nu), \,\, \text{otherwise.}
\end{cases}
\end{equation}
In addition, coalitional values under the partition tree $\TT$ are defined to be 
\begin{equation*}%\label{coalvaltree}
u_i[N,v,\TT,g]=\hat{g}^{(\pi(i))}[N,v,\TT], \quad i \in N;
\end{equation*}
which are the recursive values corresponding to leaves of $\TT$. 
\end{definition}

\begin{figure}[ht]
  \centering
  \begin{subfigure}[t]{0.31\textwidth}
    \centering
    \includegraphics[width=\textwidth]{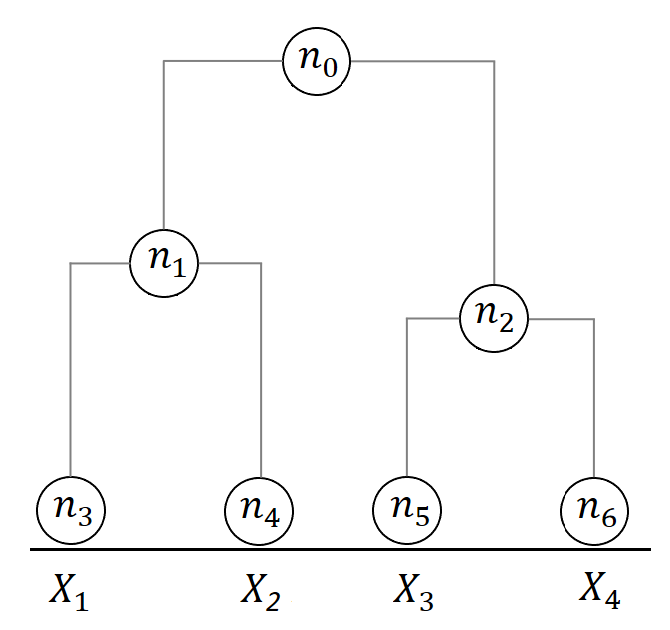} \caption{\footnotesize$\cP(n_0)=\{\{1,2\},\{3,4\}\}$}\label{fig::tree_2_2}
  \end{subfigure}    
~~
\begin{subfigure}[t]{0.31\textwidth}
    \centering
    \includegraphics[width=\textwidth]{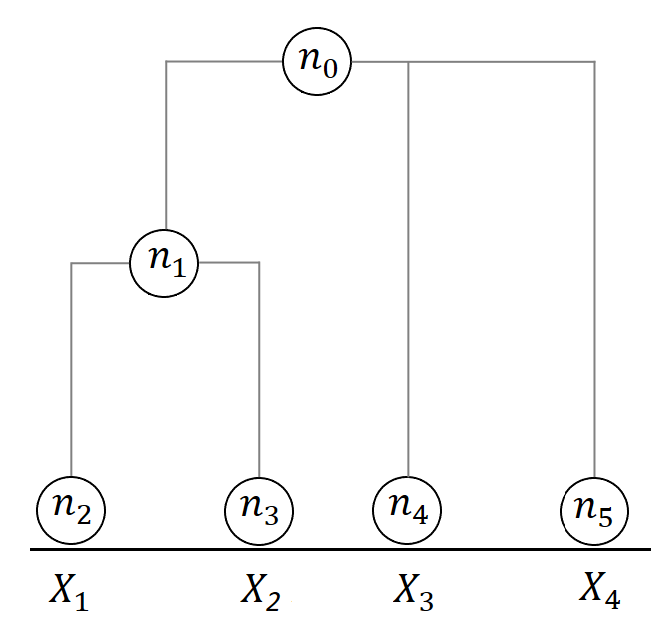} \caption{\footnotesize$\cP(n_0)=\{\{1,2\},\{3\},\{4\}\}$}\label{fig::tree_2_1_1}  
  \end{subfigure}
~~
 \begin{subfigure}[t]{0.31\textwidth}
    \centering
    \includegraphics[width=\textwidth]{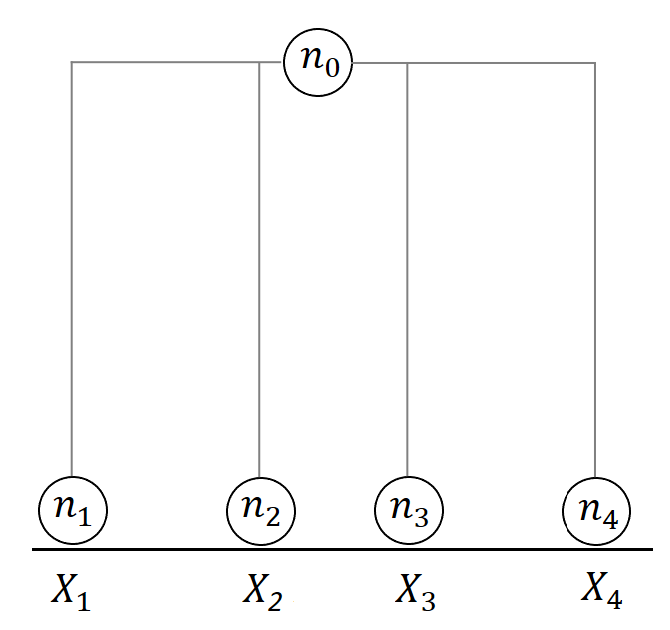} \caption{\footnotesize$\cP(n_0)=\{\{1\},\{2\},\{3\},\{4\}\}$}\label{fig::tree_1_1_1_1} 
  \end{subfigure}
  \caption{\footnotesize Different trees of depth at most two illustrating different partitions of $\{1,2,3,4\}$; see Remark \ref{depthatmost2}.} \label{fig::two_step_trees} 
\end{figure}

\begin{remark}
\rm Let us elaborate on the conditioning in \eqref{recursval}. The equality on the first line of \eqref{recursval} follows from property $(ii)$ in Definition \ref{def::2stepprop}. Furthermore, $h^{(1)}_{j}[M_{\nu_*}, \big(v^{(\nu_*)}\big)^{\cP(\nu_*)}]$ is equal to 
$g_{j}[M_{\nu_*},(v^{\nu_*})^{\cP(\nu_*)},\bar{M}_{\nu_*}]$ due to the fact that, per \eqref{hyp2::norm2step}, $h^{(2)}$ satisfies \hyperref[axiom:SEP]{(SEP)}. Consequently, if the parent $\nu_*$ of $\nu$ satisfies $\nu_* =n_0$ or $c(\nu_*) \neq \ell(\nu_*)$:
\begin{equation*}%\label{intnodegconnection}
  \hat{g}^{(\nu)}[N,v,\TT] = g_{j}[M_{\nu_*},(v^{\nu_*})^{\cP(\nu_*)},\bar{M}_{\nu_*}], \quad j=\bar{\pi}^{-1}_{\nu_*}(\nu).
\end{equation*}
All nodes $\nu$ satisfy one of the former conditions except those leaves whose all siblings are leaves too and are not adjacent to the root. 
The recursive value is defined via the second line of \eqref{recursval} in such a situation.
\end{remark}

\begin{remark}\label{depthatmost2}
\rm
To see how the definition above generalizes the two-step formulation, suppose the depth of $\TT$ does not exceed two.
The combinatorial graph of such a tree can be recovered uniquely from the  partition $\cP(n_0)$ of $N$; see  Figure \ref{fig::two_step_trees}. To elaborate, here the 
 root corresponds to $N$ with members of the partition as its children (which can be indexed by $M_{n_0}$). A child of the root is a leaf if it represents a singleton. Otherwise, its children are all leaves and correspond to the elements of that member of $\cP(n_0)$. 
Now let $\nu$ be a leaf corresponding to an element $i\in N$ and $\nu_*$ its parent. If $\nu_*\neq n_0$, then 
$v^{(\nu_*)}$ is the game $T\mapsto h^{(1)}_k[M_{n_0},\hat{v}_T(N,v^{(n_0)}=v,\cP(n_0))]$
where $k=\bar{\pi}_{n_0}^{-1}(\nu_*)$; thus 
$$u_i[N,v,\TT,g]=\hat{g}^{(\nu)}[N,v,\TT]=h^{(2)}_{i}[S(\nu_*),v^{(\nu_*)}]=g_i[N,v,\cP(n_0)].$$ %to the root
Alternatively, when the leaf $\nu$ is adjacent to the root, \eqref{recursval} yields $u_i[N,v,\TT,g]=h^{(1)}_i[M_{n_0},v^{\cP(n_0)}]$. The reader can easily check that this coincides with $g_i[N,v,\cP(n_0)]$ since $h^{(2)}$  satisfies \hyperref[axiom:SEP]{(SEP)} and $\{i\}\in\cP(n_0)$. 
We conclude that if $\TT$ is of depth at most two, then
\begin{equation*}
u_i[N,v,\TT,g]=g_i[N,v,\cP(n_0)].    
\end{equation*}
\end{remark}

The next result establishes a generalization of the quotient game property \eqref{quotientgame}.
\begin{lemma}[\bf additive flow]\label{additiveflow}
Let $g$ be a coalitional value, $\TT$ a partition tree, and \eqref{hyp2::norm2step}, \eqref{hyp3::parttree} hold. Let $\hat{g}[N,v,\TT]$ be defined by \eqref{recursval}. Suppose $g$ satisfies \hyperref[axiom:EP]{(EP)}. Then for each internal node $\nu_*$ of $\TT$ one has
\[
\sum_{\nu \in c(\nu_*)} \hat{g}^{(\nu)}[N,v,\TT] = \hat{g}^{(\nu_*)}[N,v,\TT].
\]
\end{lemma}

\begin{proof} Since $g$ is efficient, \eqref{hyp2::norm2step} together with Lemma \ref{lmm::efficg2step}$(ii)$ imply $h^{(1)}$ and $h^{(2)}$ are efficient.
    Suppose $\nu_*\ne n_0$ and $c(\nu_*)=\ell(\nu_*)$. By efficiency of $h^{(2)}$ and \eqref{recursval} we then have
    \begin{equation*}
        \sum_{\nu \in c(\nu_*)} \hat{g}^{(\nu)}[N,v,\TT] = \sum_{\nu \in c(\nu_*)} h^{(2)}_{\pi^{-1}(\nu)}[S(\nu_*),v^{(\nu_*)}] = v^{(\nu_*)}(S(\nu_*)) = \hat{g}^{(\nu_*)}[N,v,\TT].
    \end{equation*}
    Now suppose otherwise. By the efficiency of $h^{(1)}$ and \eqref{recursval} we then have
    \begin{equation*}
    \begin{aligned}
         \sum_{\nu \in c(\nu_*)} \hat{g}^{(\nu)}[N,v,\TT] &= \sum_{\nu \in c(\nu_*)} h^{(1)}_{\bar{\pi}_{\nu_*}^{-1}(\nu)}[M_{\nu_*},(v^{(\nu_*)})^{\cP(\nu_*)}] \\
        & = (v^{(\nu_*)})^{\cP(\nu_*)}(M_{\nu_*}) = v^{(\nu_*)}(S(\nu_*)) = \hat{g}^{(\nu_*)}[N,v,\TT].    
    \end{aligned}
    \end{equation*}
\end{proof}

Just like before, the construction can be extended to non-cooperative games as well.

\begin{definition}\label{def::ghatext}
Let $g$ be a coalitional value with its centered extension $\bar{g}$ as described in Lemma \ref{lmm::unitgamecoalval}. Let 
$$
\hat{g}[N,v,\TT]:=\left(\hat{g}^{(\nu)}[N,v,\TT]\right)_{\nu\in\mathcal{V}}
$$
be recursive values based on some partition tree $\TT$, and suppose \eqref{hyp2::norm2step}, \eqref{hyp3::parttree} hold. Let $u$ denote the unit, non-cooperative game. For non-cooperative games, we define the centered extension of $\bar{\hat{g}}$ of $\hat{g}$ along with extensions $\bar{u}_i$ of $u_i$ as 
\begin{equation*}%\label{centrecurs}
\bar{\hat{g}}[N,v,\TT] = \hat{g}[N,v-v(\varnothing)u,\TT] \quad \text{and} \quad \bar{u}_i[N,v,\TT,g]=\bar{\hat{g}}^{(\pi(i))}[N,v,\TT].
\end{equation*}
\end{definition}

\subsubsection{Group explainers under partition tree}\label{sec::parttree2}
We shall apply the game-theoretic machinery developed in \S\ref{sec::parttree1} to construct new group explainers for a machine learning model $f$ with predictors $X=(X_1,\dots,X_n)$. Let $\TT$ be a parameterized tree whose leaves are in bijection with the set of predictors $\{X_1,\dots,X_n\}$. 
The goal is to introduce a collection of group explainers parameterized by the height $\alpha \in [0,1]$.
As in \S\ref{sec::parttree}, for each $\alpha \in [0,1]$, the collection of nodes $\mathcal{N}(\alpha)$ yields a partition of $N$, denoted by
\[
\cP_{\alpha}=\{S_1^{\alpha},S_2^{\alpha},\dots,S_m^{\alpha}\}=\{ S(\nu):  \nu \in \mathcal{N}(\alpha) \}, \quad S_j^{\alpha} = S( \pi_{\alpha}(j)),
\]
where $m:=|\mathcal{N}(\alpha)|$ and $\pi_\alpha$ is an enumeration map $M=\{1,\dots,m\}\rightarrow\mathcal{N}(\alpha)$.

This leads to the following definition of the trivial and quotient game  explainers under a partition tree
generalizing Definition \ref{def::coalexpl}.
\begin{definition}\label{def::coalexplaintree}
Let $g$ be a coalitional value, $\TT$ a partition tree, and \eqref{hyp2::norm2step}, \eqref{hyp3::parttree} hold. For $\alpha \in [0,1]$, set
\[
\bar{u}_{S^{\alpha}_j}(X;v,\TT,f)=\sum_{i \in S_j^{\alpha}} \bar{u}_i[N,v,\TT,g], \quad \bar{u}^{\TT}_{S_j^{\alpha}}(X;v,f)=\bar{\hat{g}}^{(\pi_{\alpha}(j))}[N,v,\TT], \quad v \in \{\vce,\vpdp\},
\] 
where $\bar{u}$ and $\bar{\hat{g}}$ are as in Definition \ref{def::ghatext}.
\end{definition}

We conclude the section by presenting a result unifying the explainers introduced above for conditional and marginal games. 
As before, an assumption on the independence of certain unions of predictors is required. 
\begin{lemma}
Suppose for $\alpha_* \in (0,1)$ the partition $\cP_{\alpha_*}$ yields independent unions $X_{S_1^{\alpha_*}},X_{S_2^{\alpha_*}}$, $\dots,X_{S_{m_*}}^{\alpha_*}$. Then for all $\alpha>\alpha_*$ the partition $\cP_{\alpha}$ also yields independent unions.
\end{lemma}
\begin{proof}
Trivial, left to the reader.
\end{proof}

\begin{definition}\label{def::intquotgame} Let $g$ satisfy the two-step property and $\hat{v}$ be the intermediate game from Definition \ref{def::2stepprop}. We say that $\hat{v}$ is a quotient-like game if, for any $v$ and $\cP$, there exists a function $\beta$, independent of $v$, such that for each $A \subseteq M=\{1,2,\dots,|\cP|\}$ and any $T \subseteq S_j$
\begin{equation}\label{intquotgame}
\hat{v}_T[N,v,\cP](A) = \beta \big( \{v^{\cP}(B)\}_{B\subseteq M}, \{v^{\cP|T}(B)\}_{B \subseteq M}; A, T, j \big), \quad T \subseteq S_j.
\end{equation}
(See \eqref{modqgame} for the definition of games $v^{\cP|T}$.)
\end{definition}

\begin{proposition}\label{prop::recrsexpl} 
Consider the explainers introduced in Definition \ref{def::coalexplaintree}, and suppose the conditions of Definition \ref{def::intquotgame} hold. 
Furthermore, assume that there exists $\alpha_* \in (0,1)$ such that the partition $\cP_{\alpha_*}$ yields independent unions of predictors. Then:

\begin{itemize}
    \item [(i)] For each $\alpha>\alpha_*$ 
    \[
    \bar{u}_{S_j^{\alpha}}^{\TT}(X;\vce,f) = \bar{u}_{S_j^{\alpha}}^{\TT}(X;\vpdp,f).
    \]
    Consequently, the linear map $f \mapsto \bar{u}_{S_j^{\alpha}}(X;\vpdp,f)$ is bounded and hence continuous in $L^2(P_X)$.

    \item [(ii)] If $g$ satisfies \hyperref[axiom:EP]{(EP)}, then
\begin{equation}\label{explsumtree}
    \bar{u}_{S_j^{\alpha}}(X;\vpdp,\TT,f) = \bar{u}_{S_j^{\alpha}}(X;\vce,\TT,f)=\bar{u}_{S_j^{\alpha}}^{\TT}(X;\vpdp,f) = \bar{u}_{S_j^{\alpha}}^{\TT}(X;\vce,f).
\end{equation}    
Consequently, all explanations in \eqref{explsumtree} yield bounded linear operators in $L^2(P_X)$, and hence they are continuous in $L^2(P_X)$.

\end{itemize}
\end{proposition}

\begin{proof}
For any $\nu \in \mathcal{N}(\alpha)$ with $h(\nu)>\alpha_*$, whenever $T$ is a union of elements from $\cP(\nu)$, the intermediate game $\hat{v}_T[S(p(\nu)),v^{(p(\nu))},\cP(p(\nu))]$ associated with the parent $p(\nu)$ can be expressed as a function of values obtained by evaluating the original game $v$ (placed at the root node) at unions of sets that correspond to independent groups. This argument follows from induction, and makes use of the representation \eqref{intquotgame}. This yields the result in $(i)$.

The efficiency of $g$ implies that the game values $h^{(1)}$, $h^{(2)}$ used in the two-step formulation both satisfy \hyperref[axiom:EP]{(EP)}; 
cf. Lemma \ref{lmm::efficg2step}. 
This implies that the sum of all children values is equal to the recursive value of the corresponding parent; see Lemma \ref{additiveflow}. Hence, for any node $\nu \in \TT$, summing the recursive values over $S(\nu)$ gives the value $g^{(\nu)}$. 
Thus $\bar{u}_{S_j^{\alpha}}(X;\vpdp,\TT,f)=\bar{u}_{S_j^{\alpha}}^{\TT}(X;\vpdp,f)$
and $\bar{u}_{S_j^{\alpha}}(X;\vce,\TT,f)=\bar{u}_{S_j^{\alpha}}^{\TT}(X;\vce,f)$.
This together with $(i)$ implies $(ii)$.
\end{proof}

\end{appendices}

\end{document}